\documentclass[11pt,letter]{article}

\usepackage{amsfonts}
\usepackage{amsmath}
\usepackage{amssymb}
\usepackage{amsthm}
\usepackage[nottoc]{tocbibind}
\usepackage[margin=1in]{geometry}
\usepackage[utf8]{inputenc}
\usepackage{dsfont}
\usepackage{color}
\usepackage{bbm}
	\usepackage[procnumbered, linesnumbered,algoruled]{algorithm2e}
\usepackage[compact]{titlesec}
\usepackage{todonotes}
\usepackage[numbers,sort&compress]{natbib}
\usepackage[colorlinks=true,linkcolor=red,citecolor=blue]{hyperref}
\usepackage[nameinlink]{cleveref}
\usepackage{booktabs}

\newcommand{\ceil}[1]{{\left\lceil#1  \right\rceil}}

\newcommand{\tol}{{\textnormal{\textsf{tol}}}}
\newcommand{\LP}{{\textnormal{LP}}}
\newcommand{\demandLP}{{\textnormal{Demand-LP}}}
\newcommand{\MLP}{{\textnormal{MLP}}}
\newcommand{\CVZ}{{\textsf{SampleMatching}}}
\newcommand{\cA}{{\mathcal{A}}}

\newcommand{\cE}{{\mathcal{E}}}

\newcommand{\lr}{\textnormal{large}}
\newcommand{\sm}{\textnormal{small}}

\newcommand{\cF}{{\mathcal{F}}}

\newcommand{\iterround}{\textsf{Iterative Randomized Rounding}}

\newcommand{\abs}[1]{{\left| #1\right|}}

\newcommand{\appr}{\textnormal{\textsf{appr}}}
\newcommand{\cC}{{\mathcal{C}}}
\newcommand{\cD}{{\mathcal{D}}}
\newcommand{\cT}{{\mathcal{T}}}
\newcommand{\cS}{{\mathcal{S}}}
\newcommand{\cG}{{\mathcal{G}}}
\newcommand{\cM}{{\mathcal{M}}}

\renewcommand{\Pr}{\mathrm{Pr}}
\newcommand{\supp}{\textnormal{supp}}
\newcommand{\bx}{{\bar{x}}}

\newcommand{\bu}{{\bar{u}}}
\newcommand{\bv}{{\bar{v}}}
\newcommand{\bw}{{\bar{w}}}
\newcommand{\by}{{\bar{y}}}
\newcommand{\bz}{{\bar{z}}}
\newcommand{\bc}{{\bar{c}}}
\newcommand{\bd}{{\bar{d}}}
\newcommand{\bb}{{\bar{b}}}
\newcommand{\ba}{{\bar{a}}}
\newcommand{\br}{{\bar{r}}}

\newcommand{\bgam}{{\bar{\gamma}}}
\newcommand{\blam}{{\bar{\lambda}}}
\newcommand{\bmu}{{\bar{\mu}}}
\newcommand{\btau}{{\bar{\tau}}}
\newcommand{\hd}{{\hat{d}}}
\newcommand{\bt}{{\bar{t}}}
\newcommand{\tv}{{\tilde{v}}}

\newcommand{\bone}{{\bold{1}}}
\newcommand{\bzero} {{\bold{0}}}
\newcommand{\OPT}{\textnormal{OPT}}
\newcommand{\OPTf}{\textnormal{OPT}_f}

\newcommand{\eps}{{\varepsilon}}

\newcommand{\E}{{\mathbb{E}}}
\newcommand{\floor}[1]{\left\lfloor #1 \right\rfloor}

\newtheorem{thm}{Theorem}[section]

\newtheorem{obs}[thm]{Observation}
\newtheorem{cor}[thm]{Corollary}
 \newtheorem{lemma}[thm]{Lemma}
 \newtheorem{proposition}[thm]{Proposition}
 \newtheorem{claim}[thm]{Claim}
\newtheorem{defn}[thm]{Definition}
\def \II   {{\mathcal I}}
\newcommand{\one}{\mathbbm{1}}
\newcommand{\type}{\textnormal{\textsf{T}}}

\begin{document}

\pagenumbering{gobble}
\begin{titlepage}
 \pagestyle{empty}
\title{
	Improved Approximations for Vector Bin Packing\\ via Iterative Randomized Rounding
}
\author{Ariel Kulik\thanks{CISPA Helmholtz Center for Information Security, Germany. \texttt{ariel.kulik@cispa.de}} 
\and 
Matthias Mnich\thanks{Hamburg University of Technology, Institute for Algorithms and Complexity, Hamburg,
Germany. \texttt{matthias.mnich@tuhh.de}} 
\and
Hadas Shachnai\thanks{Computer Science Department, 
Technion, Haifa, Israel. \texttt{hadas@cs.technion.ac.il}}
}
\date{}
\maketitle

\begin{abstract}

We study the  {\sc $d$-dimensional Vector Bin Packing} ($d$VBP) problem, a  generalization of {\sc Bin Packing} with 
central applications in resource allocation and scheduling.
In $d$VBP, we are given a set of items, each of which is characterized by a $d$-dimensional {\em volume} vector; the objective is to partition the items into a minimum number of subsets (bins), such that the total volume of items in each subset is at most $1$ in each dimension. 

Our main result is an asymptotic approximation algorithm for $d$VBP that yields a ratio of $(1+\ln d-\chi(d) +\eps)$ for all $d \in \mathbb{N}$ and any $\eps > 0$; here, $\chi(d)$ is some strictly positive function.
This improves upon the best known asymptotic ratio of $ \left(1+ \ln d +\eps\right)$ due to Bansal, Caprara and Sviridenko (SICOMP 2010) for any $d >3$.
By slightly modifying our algorithm to include an initial matching phase and applying a tighter analysis we obtain an asymptotic approximation ratio of  $\left(\frac{4}{3}+\eps\right)$ for the special case of $d=2$, thus substantially improving the previous best ratio of $\left(\frac{3}{2}+\eps\right)$ due to Bansal, Eli{\'a}{\v{s}} and Khan (SODA 2016).

Our algorithm iteratively solves a configuration LP relaxation for the residual instance (from previous iterations) and samples a small number of configurations based on the solution for the configuration LP.
While iterative rounding was already used by Karmarkar and Karp (FOCS 1982) to establish their celebrated result for classic (one-dimensional) {\sc Bin Packing}, iterative {\em randomized}  rounding is used here  for the first time in the context of {\sc (Vector) Bin Packing}.
Our results show that iterative randomized rounding is a powerful tool for approximating $d$VBP, leading to simple algorithms with improved approximation guarantees.

\end{abstract}

\end{titlepage}
\tableofcontents
\newpage
\pagenumbering{arabic}
\section{Introduction}
\label{sec:into}
{\sc Bin Packing} is one of the most fundamental problems in combinatorial optimization.
An instance of {\sc Bin Packing} consists of a set $I$ of $n$ items with sizes in $(0,1]$, for which we seek the smallest number~$m$ of unit-size bins into which those items can be packed.
The extensive study of {\sc Bin Packing} since the early 1970's has had a great impact on the design and analysis of approximation algorithms (see, e.g., \cite{DeLaVegaL1981,KarmarkarK1982,CoffmanCGMV2013,HobergR2017}).

In this work we study a $d$-dimensional generalization of {\sc Bin Packing}, where both the items to be packed as well as bin capacities are given as $d$-dimensional vectors.
Formally, an instance~$\mathcal I$ of the {\sc $d$-Dimensional Vector Bin Packing ($d$VBP)} problem is a pair $(I,v)$, where $I$ is a set of~$n$ items and $v:I\rightarrow (0,1]^d $ is a $d$-dimensional volume function.\footnote{Instances with  $v:I\rightarrow [0,1]^d$ can be easily reduced to equivalent instances with $v:I\rightarrow (0,1]^d$.}
A \emph{solution} for the instance~$(I,v)$ is a collection of subsets of items $S_1,\ldots, S_m\subseteq I$ such that $v(S_b) =\sum_{i\in S_b} v(i) \leq (1,\hdots,1)$ for all~$b=1,\hdots, m$ and $\bigcup_{b=1}^{m} S_b = I$.\footnote{We say that $(a_1,\hdots,a_d)\leq (b_1,\hdots,b_d)$ if $a_i\leq b_i$ for $i = 1,\hdots,d$.}
The {\em size} of the solution is $m$.
Our objective is to find a solution of minimum size.  

As a natural generalization of {\sc Bin Packing}, and due to its wide range of applications, there has been extensive research on $d$VBP (see, e.g.,~\cite{GareyGJY1976,DeLaVegaL1981,Woeginger1997,KellererK2003,ChekuriK2004,ChangHP2005,MonaciT2006,BansalCS2010,BansalEK2016,AringhieriDGH2018,WeiLLH2020,Ray2021}).
Consider, for example, the allocation of computing services (items) to a minimum number of identical servers (bins), where each service 
requires the use of both CPU and memory.
A set of services allocated to a single server may not exceed the available memory and CPU capacity of the server.
This yields an instance of 2VBP.
For other applications see, e.g., \cite{Spieksma1994,PanigrahyTUW2011,YuG2012,TantawiS2019}.

Our goal in this paper is to design efficient polynomial-time approximation algorithms for~$d$VBP.
Let $\alpha \geq 1$ be a constant.
An algorithm $\cA$ is an \emph{asymptotic $\alpha$-approximation algorithm} for $d$VBP if for any instance $\II$ of $d$VBP it returns, in polynomial time, a solution of size at most $\alpha \OPT(\II) +o(\OPT(\II))$, where $\OPT(\II)$ is the optimal solution size for $\II$.
A weaker notion is that of a \emph{randomized asymptotic $\alpha$-approximation algorithm}; such an algorithm always returns a solution for $\II$ in polynomial time, but the solution size has to be at most $\alpha \OPT(\II) +o(\OPT(\II))$ with some constant probability. 
An {\em asymptotic polynomial-time approximation scheme (APTAS)} is an infinite family $\{\mathcal A_{\varepsilon}\}$ of asymptotic $(1+\eps)$-approximation algorithms, one for each $\eps > 0$.
Ray~\cite{Ray2021} showed that 2VBP does not admit an asymptotic approximation ratio better than~$\frac{600}{599}$, assuming $\mathsf{P}\neq \mathsf{NP}$, implying there is no APTAS already for $d=2$.\footnote{Ray's result addresses an oversight in an earlier proof of Woeginger \cite{Woeginger1997}.}

In~\cite{BansalCS2010} Bansal, Caprara and Sviridenko introduced the {\em Round\&Approx} framework, which yields an asymptotic $ \left(1+ \ln d +\eps\right)$-approximation for $d$VBP, for all $d\in \mathbb{N}$ and any $\eps>0$.
Their results are the best-known asymptotic approximation ratio for $d>3$.
For the special cases of $d=2$ and $d=3$, the best-known asymptotic approximation ratios, due to Bansal, Eli{\'a}{\v{s}} and Khan~\cite{BansalEK2016}, are $1.5+\eps$ and $2+\eps$, respectively, for all $\eps>0$. 

\subsection{Our Contribution}
\label{sec:ourcontribution}
Our main contribution is an asymptotic approximation algorithm for $d$VBP which improves upon the best-known ratio of~\cite{BansalCS2010} for all $d > 3$.
Specifically, we show the following result.
\begin{thm}
\label{thm:main_all}
 For all $d\in \mathbb{N}$ and any $\eps>0$ there is  a randomized $(1+\ln d-\chi(d) +\eps)$-asymptotic approximation algorithm for $d$VBP, where $\chi(d)=\left(\frac{1}{2}\cdot \ln d+ \frac{1}{\sqrt d}-1\right)\cdot \left( 1- \sqrt[2d]{\frac{1}{d}}\right)^d >0$.
\end{thm}
\Cref{thm:main_all} is derived via a simple iterative randomized rounding algorithm.
In fact, we show that our algorithm outperforms \emph{any} algorithm which follows the framework of Bansal et al.~\cite{BansalCS2010}.

For the case $d=2$, we provide a tighter analysis and an additionial {\em matching} subroutine prior to the iterative randomized rounding phase; together, they enable us to obtain a better bound.
\begin{thm}
\label{thm:main}
  For any $\eps>0$, there is a randomized asymptotic $\left(\frac{4}{3} +\eps\right)$-approximation algorithm for $2$VBP.  
\end{thm}

\Cref{fig:comparison} summarizes the previously known, as well as our new results for {\sc Vector Bin Packing}.

\renewcommand{\arraystretch}{1.3}
\begin{table}
  \begin{centering}
    \begin{tabular}{lcccc}
	  \toprule
      Reference & $d=2$ & $d=3$ & $d=4$ & arbitrary $d$\\
      \specialrule{2.5pt}{0 pt}{0pt}
       \cite{GareyGJY1976} & $2.7$ &3.7&4.7& $d+0.7$\\
      \midrule
 \cite{DeLaVegaL1981} & $2+\eps$&$3+\eps$& $4+\eps$& $d+\eps$\\
 \midrule
 \cite{ChekuriK2004} & & & &$O(\ln d)$\\
      \midrule
       \cite{KellererK2003}  & $2$ (absolute) &&&\\
      \midrule
       \cite{BansalCS2010} & 
       $\approx 1.69314$& 
        $\approx2.09861$ & 
        $\approx2.38629$   & $1+\ln d +\eps$\\
      \midrule
      \cite{BansalEK2016}  & $\frac{3}{2}+\eps~\approx~ 1.5$ & $2+\eps$  & $2.5+\eps$&$ \frac{d+1}{2} +\eps$\\
      & (absolute) &  &  & 
      \\
      \midrule
      {\bf This paper} & $\frac{4}{3}+\eps~\approx~ 1.3333$& $\approx 2.09801$ &$ \approx 2.38617$& $ 1+\ln d -\chi(d)+\eps$ \\
	  \bottomrule
    \end{tabular}
    \caption{Known and new results for $d$VBP.\label{fig:comparison} An entry of value $\beta$ indicates the paper in this row gives an asymptotic $\beta$-approximation for $d$VBP, where $d$ appears at the head of the column.}
  \end{centering}
\end{table}

\subsection{Related Work}
\label{sec:relatedwork}
The one-dimensional case ($1$VBP) is the classic {\sc Bin Packing} problem.
A simple reduction from {\sc Partition}~\cite[Ch. 9]{Vazirani2001} shows there is no $\alpha$-approximation for {\sc Bin Packing} with $\alpha<\frac{3}{2}$, assuming $\mathsf{P}\neq \mathsf{NP}$.
This motivates the study of asymptotic approximation algorithms for the problem, and in particular, the search for APTASs.
The first APTAS for {\sc Bin Packing} was proposed by {Fernandez} de la Vega and Lueker~\cite{DeLaVegaL1981}, who introduced the {\em linear grouping} technique. In their seminal work, 
Karmarker and Karp~\cite{KarmarkarK1982} give an approximation algorithm that uses at most $\OPT(\II)+O(\log^2(\OPT(\II)))$ bins. Their work introduced the concept of {\em Configuration-LP} to which they applied (deterministic) {\em iterative rounding}.
More recently, Hoberg and Rothvo{\ss}~\cite{HobergR2017} obtained a polynomial-time algorithm that returns a solution of size $\OPT(\II)+O(\log(\OPT(\II)))$.
Comprehensive surveys of algorithmic results for {\sc Bin Packing} are given, e.g., by Coffman et al.~\cite{CoffmanCGMV2013} and Delorme et al.~\cite{DelormeIM2016}.

To the best of our knowledge, the first asymptotic approximation algorithm for $d$VBP, due to Garey et al.~\cite{GareyGJY1976}, achieves the ratio $\left(d+\frac{7}{10}\right)$.
This ratio  was improved to an asymptotic $(d+\varepsilon)$-approximation by {Fernandez} de la Vega and Lueker \cite{DeLaVegaL1981}.  
The first algorithm to break the additive of $d$ in the approximation ratio is an asymptotic $(1+O(\ln d))$-algorithm due to Chekuri and Khanna~\cite{ChekuriK2004}.
An absolute (i.e., non-asymptotic) $2$-approximation ratio for the special case of $2$VBP was given by Kellerer and  Kotov \cite{KellererK2003}.

Bansal, Caprara and Sviridenko~\cite{BansalCS2010} introduced a powerful framework, based on randomized rounding, which they call \emph{Round\&Approx}.
They use it to obtain a randomized asymptotic $\left(1+\ln d+\eps\right)$-approximation for $d$VBP, for every $d \geq 2$ and any $\varepsilon > 0$.
The framework combines a {\em configuration LP} relaxation of the problem with a ``subset-oblivious'' approximation algorithm.
Informally, a $\beta$-subset oblivious algorithm for $d$VBP is an algorithm which, given a $d$VBP instance~$(I,v)$ and a random subset of items $S\subseteq I$, such that $\Pr(i\in S)\leq \gamma$ for all $i\in I$, returns a solution for $(I,v)$ using approximately $\beta \cdot \gamma \cdot \OPT(I,v)$ bins.
A (nearly-optimal) solution for the configuration LP is interpreted as a distribution over the configurations of the instance (i.e., subsets $S\subseteq I$ for which $v_t(S)\leq 1$ for every $t\in \{1,\ldots, d\}$).
This distribution is used to independently sample a set of configurations; items which do not belong to any of the sampled configurations are packed using the subset-oblivious approximation algorithm.
The properties of the subset-oblivious approximation algorithm combined with a concentration bound of McDiarmid~\cite{McDiarmid1989} then yield  the claimed approximation guarantee. 
{\em Round\&Approx} is the framework used to obtain the best approximation algorithms for {\sc 2-Dimensional Geometric  Bin Packing} and for {\sc Vector Bin Packing}.

Bansal, Eli{\'a}{\v{s}} and Khan~\cite{BansalEK2016} obtained an asymptotic $\left(\frac{d+1}{2}+\eps\right)$-approximation for $d$VBP, for all $d\in \mathbb{N}$ and any $\eps>0$.
Their algorithm is based on a rounding scheme which yields a packing with {\em resource augmentation} in all dimensions except one.
The rounding scheme is combined with the generation of an inflated solution of specific structure, which leaves some free volume in all dimensions but one. 
The free volume is used to balance the resource augmentation.
The authors prove the existence of such a solution, while the algorithm uses heavy enumeration to ``guess'' properties of the solution which suffice to reconstruct it.
The authors also attempted to combine this algorithm with the {\em Round\&Approx} framework of Bansal et al.~\cite{BansalCS2010} to obtain improved asympatotic approximation.
Unfortunately, there is a flaw in the analysis (we give the details in \Cref{sec:flaw}).\footnote{We contacted the authors and made them aware of this flaw~\cite{BansalEK2021}.}
For $d=2$, Bansal et al.~\cite{BansalEK2016} obtained an absolute $(3/2 + \eps)$-approximation using a combinatorial algorithm.

Recently, Sandeep~\cite{Sandeep2021} showed there is no asymptotic $o(\log d)$-approximation for $d$VBP.
For other results relating to $d$VBP see, e.g.,~\cite{Johnson2016} and
the excellent survey on multidimensional {\sc Bin Packing} problems by Christensen et al.~\cite{ChristensenKPT2017}.

Iterative rounding and randomized rounding are two powerful techniques used to obtain an integral solution from a fractional solution of an LP relaxation for a problem.
Iterative rounding generates an integral solution by iteratively assigning integral values to subsets of variables in the LP, and solving a suitably modified linear program (excluding these variables).
In contrast, randomized rounding is done in one shot, by interpreting the variable values as probabilities, and assigning an integral value to each variable via sampling according to these probabilities. An excellent survey on iterative rounding can be found in~\cite{LauRS2011} (see also~\cite{Bansal2014}).
For various applications of randomized rounding, see, e.g.,~\cite{Vazirani2001,WilliamsonS2011}.

One of the earliest and most sophisticated applications of iterative rounding appears in the analysis of Karmarkar and Karp~\cite{KarmarkarK1982} in their $\OPT(\II)+O(\log^2(\OPT(\II)))$-approximation for classic {\sc Bin Packing}.
Later works applied iterative {\em randomized} rounding for solving other problems, such as {\sc Steiner Tree}~\cite{ByrkaGRS2013}, makespan minimization on unrelated machines and degree-bounded minimum spanning trees~\cite{Bansal2019}, fair scheduling~\cite{ImM2020}, and $k$-Clustering Completion~\cite{HifiS2022}.
However, we are not aware of earlier use of iterative randomized rounding in solving classic {\sc Bin Packing} or its variants.

\subsection{The Algorithm}
\label{sec:alg_intro}
Given a $d$VBP instance $(I,v)$, a {\em configuration} is a subset $C\subseteq I$ of items such that $v(C)\leq \bone$.\footnote{We use the notation $\bone=(1,\ldots, 1)$ and $\bzero=(0,\ldots,0)$.} 
For each item $i\in I$, let $C(i) \in \{0, 1\}$ indicate whether the item $i$ appears in the configuration~$C$ or not.
We use $\cC$ to denote the set of all configurations.
That is, $\cC=\{C\subseteq I~|~v(C)\leq (1,\ldots, 1)\}$.
We use a variant of the standard configuration LP which only consider a subset of items $S\subseteq I$.
Given a Boolean expression ${\cD}$, we define $\one_{\cD} \in \{ 0,1\}$ such that $\one_{\cD}=1$ if~$\cD$ is true and $\one_{\cD}=0$ otherwise.
For every $S\subseteq I$ define
\begin{equation}
\label{eq:config_LP}
  \begin{aligned}
	\LP(S):~~~  & \min && \sum_{C\in \cC} \bx_C,\\
	&\forall i\in I: &&\sum_{C\in \cC} \bx_C\cdot C(i)= \one_{i\in S}\\\ 
	&\forall C\in \cC:~~~&& \bx_C\geq 0\enspace.
  \end{aligned}
\end{equation}
Each of the variables $\bx_C$ represents a (fractional) selection of the configuration $C$, where the first constraints ensure that each item $i\in S$ is covered.
It is well-known~\cite{BansalCS2010} that there is a PTAS for $\LP(S)$. 

For any vector $\bx\in [0,1]^{\cC}$ we associate a distribution over the configurations $\cC$.
We say that a random configuration $R\in\cC$ is \emph{distributed by $\bx$} (and use the notation $R\sim \bx$) if $\Pr(R=C)=\frac{\bx_C}{z}$ for every $C\in \cC$, where $z=\|\bx\|\equiv \sum_{C\in \cC}\bx_C$.  

Our main algorithm, $\iterround$, is given in \Cref{alg:basic_round_and_round}.
For arbitrary $d$, the algorithm is used with $S_0=I$; the distinction between $I$ and $S_0$ will be used later in our  improved algorithm for $2$VBP (see \Cref{alg:match_and_round}).
We note that \Cref{alg:basic_round_and_round} has a polynomial run time (for fixed $\delta$), and that it returns a solution for the $d$VBP instance $(S_0,v)$.
 \Cref{round:first_fit} of \Cref{alg:basic_round_and_round} uses a classic First-Fit approach to pack the remaining items (see \Cref{sec:prelim} for more details).
 
\begin{algorithm}[h]
  \SetAlgoLined
  \SetKwInOut{Input}{Input}\SetKwInOut{Config}{Parameters}
	
  \SetKwInOut{Output}{Output}
  \Config{$\delta\in(0,0.1)$, $\alpha = -\ln\left( 1-\delta\right)$ and $k=\ceil{\log_{1-\delta}(\delta)}$, where $\delta^{-1} \in \mathbb{N}$. } 
  \Input{A {$d$-VBP} instance $(I,v)$ and a subset $S_0\subseteq I$.}
  \Output{A solution for the instance $(S_0,v)$.}
  \DontPrintSemicolon
	
  \For{ $j=1,\ldots, k$ \label{round:loop}}
  {
	Find a $(1+\delta^2)$-approximate solution $\bx^j$ for $\LP(S_{j-1})$ and let $z_j= \|\bx^j\|$ be its value.
	\label{round:LPS}
	\;
	
	Independently sample $\rho_j=\ceil{ \alpha z_j}$ configurations $C^j_1,\ldots, C^j_{\rho_j}$, where $C^j_\ell \sim \bx^j$ for all $\ell \in \{1,\ldots,\rho_j\}$.\label{round:sample}\;

	Update $S_j\leftarrow S_{j-1}\setminus \left( \bigcup_{\ell =1}^{\rho_j} C^j_{\ell}\right)$. \;
  }	
  Pack $S_k$  into configurations $C^*_1,\ldots, C^*_{\rho^*}$ using First-Fit 
  \label{round:first_fit}
  \;
	
  Return $\left( \bigcup_{j=1}^{k} \{ C^j_1,\ldots, C^j_{\rho_j}\}\right)\cup\{C^*_1,\ldots, C^*_{\rho^*} \}$.
  \; 
  \caption{\iterround\label{alg:basic_round_and_round}}
\end{algorithm}

In the analysis we show that $\rho^*$ is negligible in comparison to $\OPT(I,v)$.
Thus, the solution generated by \Cref{alg:basic_round_and_round} consists predominately of configurations which are randomly sampled according to solutions for the configuration LP.

Furthermore, the algorithm repeatedly solves the configuration LP, each time using the set $S_j$ consisting of the items not covered in previous iterations.
This stands in contrast to algorithms associated with the {\em Round\&Approx} framework (e.g., \cite{BansalCS2010}) which solve the configuration LP once and utilize a subset-oblivious algorithm to generate a significant part of the solution following the random sampling stage.   

The above difference is the key for the improved approximation ratio.  
The analysis of \linebreak {\em Round\&Approx}  uses the fact that if $C_1,\ldots, C_\rho$  are independent random configurations  distributed by a (nearly) optimal solution $\bx$ for $\LP(I)$, then $\Pr(i \notin \bigcup_{\ell=1}^{\rho} C_{\ell }) \approx \exp\left(-\frac{\rho}{\OPT(\LP(I))}\right)$. For example, to have the random configurations $C_1,\ldots, C_\rho$ cover each item $i\in I$ (i.e., $i\in C_1\cup \ldots \cup C_{\rho}$) with probability~$\frac{1}{2}$, the number of sampled configurations has to be $\rho \approx \OPT(\LP(I))\cdot \ln(2)$. The core idea in our analysis is that if the configurations are sampled iteratively, as in  \Cref{alg:basic_round_and_round}, then the probability of an item to remain uncovered is $\frac{1}{2}$ after sampling strictly fewer configurations.

Bansal et al.~\cite{BansalCS2010} defined the notion of $\beta$-subset oblivious algorithms for $d$VBP; we give a formal definition of the term in \Cref{sec:ddim}.
The main result of Bansal et al.~\cite{BansalCS2010}, applied to $d$VBP, is  the following. 
\begin{thm}[{\em Round\&Approx} \cite{BansalCS2010}]
\label{thm:BCS}
  Let $d\in \mathbb{N}$ and $\beta\geq 1$. 
  If there is a polynomial-time $\beta$-subset oblivious algorithm for $d$VBP then there is a randomized asymptotic $(1+\ln \beta +\eps)$-approximation algorithm for $d$VBP for every $\eps>0$. 
\end{thm}
Bansal et al.~\cite{BansalCS2010} also presented subset-oblivious algorithms for $d$VBP, as stated in the next lemma.
\begin{lemma}
\label{lem:d_subset_oblivious}
  For every $\eps>0$  and $d\in \mathbb{N}$ there is a polynomial-time $(d+\eps)$-subset oblivious algorithm for $d$VBP. 
\end{lemma}
In particular, the  asymptotic $(1+\ln d +\eps)$-approximation for $d$VBP of Bansal et al.~\cite{BansalCS2010} is derived as an immediate consequence of \Cref{thm:BCS} and \Cref{lem:d_subset_oblivious}. 
The following theorem states that \Cref{alg:basic_round_and_round} is strictly better than any algorithm that is based on {\em Round\&Approx} (\Cref{thm:BCS}).
\begin{thm}
  \label{thm:better_than_randa}
  Let $\beta\geq 1$ and $d\in \mathbb{N}$.
  If there is a $\beta$-subset oblivious algorithm for $d$VBP then for every $\eps>0$ there exists $\delta>0$ such that \Cref{alg:basic_round_and_round} configured with~$\delta$ is a randomized asymptotic $(1+\ln \beta -\chi(\beta,d) +\eps )$-approximation algorithm for $d$VBP, where $\chi(\beta,d)=\left(\frac{1}{2}\cdot \ln \beta+ \frac{1}{\sqrt \beta}-1\right)\cdot \left( 1- \sqrt[2d]{\frac{1}{\beta}}\right)^d $.
\end{thm}
\Cref{thm:main_all} follows immediately from \Cref{thm:better_than_randa} and \Cref{lem:d_subset_oblivious}. Since $\chi(\beta,d)>0$ for all $\beta>1$ and $d>1$,
\Cref{thm:better_than_randa} implies that  \Cref{alg:basic_round_and_round} is strictly better than {\em Round\&Approx};
that is, it achieves an asymptotic approximation ratio smaller than that obtained by any {\em Round\&Approx}-based algorithm.
Furthermore, while the result of \Cref{thm:BCS} refers to an algorithm which uses as a subroutine a $\beta$-subset oblivious algorithm, 
the result of \Cref{thm:better_than_randa} uses the $\beta$-subset oblivious algorithm only as part of its proof.
Thus,~\Cref{thm:better_than_randa} does not require the subset-oblivious algorithm to run in polynomial time.
Finally, we note that the value of $\chi(\beta,d)$ in \Cref{thm:better_than_randa} is likely to be sub-optimal, and can probably be replaced by a larger value. Our main objective is to show \Cref{alg:basic_round_and_round} yields a better asymptotic approximation ratio in comparison to {\em Round\&Approx} in a simple  manner, possibly sacrificing  the  value of $\chi(\beta,d)$.

The proof of \Cref{thm:better_than_randa} utilizes an iteration-dependent bound on $\OPT(S_j,v)$.
Trivially, $\OPT(S_j,v) \leq \OPT(I,v)$.
The subset-oblivious algorithm is used to show that $\OPT(S_j,v)\lesssim \beta (1-\delta)^j \OPT(I,v)$ with high probability.
Together, these two bounds can be used to show that the asymptotic approximation ratio of \Cref{alg:basic_round_and_round} is approximately $(1+\ln \beta)$, matching the statement of~\Cref{thm:BCS}.
To show a strictly better approximation ratio, we consider a nearly optimal solution $A_1,\ldots, A_m$ of the instance, and use a simple rounding scheme to show that if $T_j\subseteq \{1,\ldots, m\}$ is a set of configurations such that $v_t(A_b\cap S_j )\leq 1-\delta$ for all $t=1,\ldots, d$, then the items in $\left(\bigcup_{b\in T_j} A_b \right)\cap S_r$ can be packed in strictly less than $\abs{T_j}$ configurations for $r>j$.
This, together with a lower bound on $\abs{T_j}$ for a specific iteration $j$, leads to a third upper bound on $\OPT(S_j,v)$, which is used to obtain the improved asymptotic approximation ratio.
The configurations in $T_j$ can be considered as ``easy'', and the lower bound on $\abs{T_j}$ can be interpreted as a guarantee that some configurations must ``become easy'' as the iterative rounding process progresses.

In \Cref{sec:ddim} we further show that the dependence of $\delta$ on $\eps$ derived from \Cref{thm:better_than_randa} is polynomial.
A simple consequence of this property is that, by appropriately setting $\delta$,  \Cref{alg:basic_round_and_round} is a randomized asymptotic $(1+\eps)$-approximation for {\sc Bin Packing} whose run time is polynomial in the input size and in $\frac{1}{\eps}$.
Thus, we have 
\begin{lemma}
\label{lem:AFTPAS}
  \Cref{alg:basic_round_and_round} is a randomized  asymptotic fully polynomial-time approximation scheme (AFPTAS)\footnote{A randomized AFPTAS	for a problem ${\cal P}$ is an infinite family $\{\mathcal A_{\varepsilon}\}$ of randomized asymptotic $(1+\eps)$-approximation algorithms for ${\cal P}$, one for each $\eps > 0$, whose run times are polynomial in the input size and in $\frac{1}{\eps}$.} for {\sc Bin Packing}. 
\end{lemma}

\subsection{Improved Algorithm for $2$VBP}
\label{sec:2vbp_improved}
For the special case where $d=2$, we strengthen our analysis to obtain a better approximation ratio.
To simplify our analysis, we may  assume our instances adhere to a specific structure. 
Given $\delta>0$, we say that an item $i\in I$ is $\delta$-{\em huge} if $v_1(i)\geq 1-\delta$ and $v_2(i)\geq 1-\delta$. 
The {\sc $\delta$-huge free $2$VBP} ($\delta$-2VBP) is the special case of $2$VBP in which there are no $\delta$-huge items.
In solving a general $2$VBP instance, we may restrict our attention to the corresponding $\delta$-huge free instance, as formalized in the next result.

\begin{lemma}
\label{lem:huge_free}
  For any $\alpha \geq 1$ and $\delta\in (0,0.1)$, if there is a randomized asymptotic $\alpha$-approximation for $\delta$-$2$VBP then there is a randomized asymptotic $(\alpha+4\delta)$-approximation for 2VBP.
\end{lemma}
The lemma follows by noting that each huge item can be packed in a separate bin.
This incurs only a small increase in the packing size (we omit the details).

The analysis of \Cref{alg:basic_round_and_round} (as part of our approximation algorithm for 2VBP) relies on an iteration-dependent bound on $\OPT(S_j,v)$ which holds with high probability. 
We use a classification of items and configurations into categories.
As in \Cref{alg:basic_round_and_round}, let $\delta\in (0,1)$ be such that $\delta^{-1}\in \mathbb{N}$.	 
We say that an item $i\in I$  is {\em $\delta$-large} if $v_1(i)>\delta$ or $v_2(i)>\delta$, and use $L\subseteq I$ to denote the set of $\delta$-large items ($\delta$ is commonly known by context). It can be easily shown that $|C\cap L|\leq 2\cdot \delta^{-1}$ for all $C\in \cC$.
For $h = 2,\hdots, 2\cdot \delta^{-1}$, we define 
\begin{equation}
\label{eq:Ch_def}
  \cC_h = \left\{ C\in \cC \mid v(C\cap L)>(1-\delta, 1-\delta) \textnormal{ and } |C\cap L| =h \right\} .
\end{equation}
Let $\cC_0 =\cC\setminus \left(\bigcup_{h=2}^{2\cdot \delta^{-1}} \cC_h\right)$ be the set of all remaining configurations.
As we assume that $(I,v)$ is an instance of $\delta$-2VBP (i.e., no $\delta$-huge items), it follows that for every $C\in \cC_0$ either $v_1(C\cap L)\leq 1-\delta$ or $v_2(C\cap L)\leq 1-\delta$.
 
For vectors $\bx, \bz \in [0,1]^{\cC}$, $\bx \cdot \bz = \sum_{C \in \cC} {\bx}_C \cdot {\bz}_C$ is the dot product of $\bx$ and $\bz$.
By applying a tigher analysis (in comparison to \Cref{thm:better_than_randa}),
it can be shown that if $\bx^*\in [0,1]^{\cC}$ is a solution for $\LP({S_0})$ then,with high probability, the solution returned by \Cref{alg:basic_round_and_round}  is of size at most
\begin{equation}
\label{eq:round_and_round_bound}
  \bx^* \cdot \one_{\cC_0} +\sum_{h=2}^{2\delta^{-1}} \frac{h+1}{h}\cdot \bx^* \cdot \one_{\cC_h}\leq \left( \frac{3}{2}+\delta \cdot O(1)\right) \cdot  \|\bx^*\| + O(1) \enspace .
\end{equation}
This implies that, given the input $S_0=I$, and by taking $\bx^*$ which corresponds to an optimal solution, \Cref{alg:basic_round_and_round} yields an asymptotic approximation ratio arbitrarily close to $\frac{3}{2}$.
While we do not include a proof of~\eqref{eq:round_and_round_bound}, the proof can be derived by modifying the proof of \Cref{lem:main_rnr} and using \Cref{lem:undercovered_simple_and_small}. 
 
Our analysis relies on structural properties of 2VBP instances (inspired by properties presented by Bansal et al.~\cite{BansalEK2016}) by which configurations in $\cC_0$ are ``easy'' (when selected by $\bx^*$) and configuration in $\cC\setminus \cC_0$ are ``difficult''. 
Intuitively, from the viewpoint of \Cref{alg:basic_round_and_round}, a configuration $C\in \cC\setminus \cC_0$ {\em becomes} easy at iteration $j$ if $C\cap L \not\subseteq S_j$, as in this case $C\cap S_j \in \cC_0$.
Our analysis exploits this intuition via the notion of {\em touched} and {\em untouched} configurations (see the formal definition in \Cref{sec:main_rnr}).

The bound \eqref{eq:round_and_round_bound} on the solution quality suggests that the most ``difficult'' configurations in $\bx^*$ are those in $\cC_2$; indeed, if we have an optimal solution containing no configuration in $\cC_2$ then we can obtain an approximation ratio of $\frac{4}{3}$.
Furthermore, if an optimal (integral) solution contains only configurations in $\cC_2$ then a nearly optimal solution can be easily constructed using {\em matching}.
As a solution may contain both configurations in $\cC_2$ and in $\cC\setminus \cC_0$, we use a sophisticated combination of a matching polytope and a configuration LP, along with the dependent sampling technique of Chekuri et al.~\cite{ChekuriVZ2011}.
In the execution of our algorithm \textsf{Match\&Round}, the solution for the resulting LP is (conceptually) partitioned into two parts: one which contains the configurations in~$\cC_2$ and handled using matching techniques, and another which contains the remaining configurations that is handled by \Cref{alg:basic_round_and_round}.

We define the \emph{$\delta$-matching graph} $G=(L,E)$ of $(I,v)$ as the graph whose vertex set $L$ consists of the $\delta$-large items of $(I,v)$, and whose edge set is $E=\{ \{i_1, i_2\}\subseteq L\mid\{i_1,i_2\} \in \cC_2\}$.  
We use $P_{\cM}(G)$ to denote the matching polytope of $G$.
We refer the reader to Schrijver's book \cite{Schrijver2003} for a formal definition of the matching polytope.
Given $\bx\in [0,1]^{\cC}$, we define the projection of $\bx$ on $E$ as the vector $\bar{p} \in \mathbb{R}_{\geq 0}^E$ where $\bar{p}_{e} = \sum_{C\in \cC \textnormal{ s.t. } e\subseteq C} \bx_C$.
Let $\cE(\bx)=\bar{p}$. 
We note that for any $C\in\cC$, there is at most a single edge $e\in E$ such that $e\subseteq C$. 

The {\em Matching Configuration LP} of the $\delta$-2VBP instance $(I,v)$ is the following optimization problem:
\begin{equation}
\label{eq:matching_LP}
  \begin{aligned}
	\MLP:~~~  & \min && \sum_{C\in \cC} \bx_C,\\
		&\forall i\in I: &&\sum_{C\in \cC} \bx_C\cdot C(i)= 1,\\\
		& && \cE(\bx)\in P_{\cM}(G),\\
		&\forall C\in \cC:~~~&& \bx_C\geq 0\enspace.
	\end{aligned}
\end{equation}
Thus, MLP takes as input a $\delta$-2VBP instance $(I,v)$, and a \emph{solution} for $(I,v)$ is a vector $\bx\in \mathbb{R}_{\geq 0}^\cC$ which satisfies the constraints in~\eqref{eq:matching_LP}.
The objective is to find a solution $\bx$ such that $\|\bx\|=\sum_{C\in \cC} \bx_C$ is minimized. 

Note that the Matching Configuration LP is a restriction of $\LP(I)$ in which we also require that $\cE(\bx)$ is in the matching polytope $P_{\cM}(G)$.
Observe that if $S_1,\ldots, S_m$ is a solution for~$(I,v)$ in which the sets~$S_1,\ldots, S_m$ are pairwise disjoint, then the vector $\bx\in \{0,1\}^{\cC}$ with $\bx_{S_{b}} = 1 $ for $b\in \{1,\hdots,m\}$ and $\bx_{C}=0$ for any other $C\in \cC$, is a feasible solution for $\MLP$.
This holds since the set $\left\{e\in E \mid\exists b\in \{1,\hdots,m\}: ~e\subseteq S_b \right\}$ forms a matching in the graph $G$.

Similar to the configuration LP, $\MLP$ can be approximated as well:
\begin{lemma}
\label{lem:matching_ptas}
  For any $\delta\in (0,0.1)$, there is a PTAS for the $\MLP$ problem.
\end{lemma}
We note that writing $P_{\cM}(G)$ as a linear program requires a super-polynomial number of constraints~\cite{Rothvoss2017}.
It follows that both $\MLP$ and its dual have super-polynomial number of variables and a super-polynomial number of constraints.
Thus, the standard method for solving configuration LPs using an approximate separation oracle for the dual program fails (the method can be traced back to Karmarker and Karp~\cite{KarmarkarK1982}), and more sophisticated tools are required to obtain a PTAS.
We give the proof of \Cref{lem:matching_ptas} in \Cref{sec:match_lp}.

Given $\bx$ such that $\bar{\beta}=\cE(\bx)\in P_{\cM}(G)$ and a parameter $\gamma>0$, we use a randomized algorithm of Chekuri, Vondr{\'a}k and Zenklusen~\cite{ChekuriVZ2011} called $\CVZ$.
This algorithm, for input $(\bar{\beta},\gamma)$ in polynomial time generates a random matching $\cM$ for which $\Pr(e\in \cM) = (1-\gamma)\bar{\beta}_e$.
Importantly, the algorithm also gives dimension-free Chernoff-like concentration bounds for $\cM$ (see \Cref{lem:cvz11} for details).

We refer to our algorithm for 2VBP as \textsf{Match\&Round}; its pseudocode is given in \Cref{alg:match_and_round}.\footnote{The idea to use matching algorithms is inspired by the work of Bansal et al. \cite{BansalEK2016}. However, matching plays different roles in the two algorithms. In particular, $\MLP$ is introduced in this paper.}
We note that \textsf{Match\&Round} is a polynomial-time algorithm which returns a solution for the instance~$(I,v)$. 

\begin{algorithm}[h]
  \SetAlgoLined
  \SetKwInOut{Input}{Input}\SetKwInOut{Config}{Parameters}
  \SetKwInOut{Output}{Output}

  \Config{$0<\delta<0.1$, where $\delta^{-1} \in \mathbb{N}$.} 
  \Input{A $\delta$-2VBP instance $(I,v)$.}
  \Output{A solution for the instance $(I,v)$.}

  \DontPrintSemicolon

  \label{match:MatchingLP}
  Find a $(1+\delta^2)$-approximate solution $\bx^{0}$ for $\MLP$. \;

  $\cM\leftarrow \CVZ\left(\cE(\bx^0), \delta^4\right)$, and set $S_0\leftarrow I\setminus \left(\bigcup_{e\in \cM} e\right)$. \;
	
  Run \Cref{alg:basic_round_and_round} on the instance $(I,v)$ with $S_0$ and the parameter $\delta$. Denote the returned solution by $D_1,\ldots, D_m$. 
  \label{match:invoke_round}\\	
	
  Return $\cM \cup \{D_1, \ldots , D_m\}$.
  \; 
  \caption{\textsf{Match\&Round}\label{alg:match_and_round}}
\end{algorithm}

Our main result for $2$VBP  follows from the next lemma.
\begin{lemma}
\label{lem:weakly_asym}
  For any $\delta\in (0,0.1)$, \Cref{alg:match_and_round} is a randomized asymptotic $\left( \frac{4}{3}+O(\delta) \right)$-approximation for $\delta$-$2$VBP.
\end{lemma}
Using \Cref{lem:weakly_asym} and~\Cref{lem:huge_free}, we obtain the statement of \Cref{thm:main}.
We use the standard notation of $\wedge$ for the element-wise minimum of two vectors.\footnote{That is, for $\br^1=(\br^1_1, \ldots , \br^1_k)$ and $\br^2=(\br^2_1, \ldots , \br^2_k)$, $(\br^1 \wedge \br^2)_i  
= \min\{ \br^1_i, \br^2_i\}$ for $i = 1,\hdots,k$.}
The analysis of \Cref{alg:match_and_round} is based on a partition of the solution $\bx^0$ obtained in \Cref{match:MatchingLP} into its two ``matching'' and ``fractional'' components: $\bx^0\wedge \one_{\cC_2}$ and $\bx^0\wedge \one_{\cC\setminus \cC_2}$.
We show that, with high probability, $|\cM|\lesssim  \bx^0\cdot \one_{\cC_2}$.
Furthermore, we exploit the fact that $\bx^0\wedge \one_{\cC\setminus \cC_2}$ does not select configurations in $\cC_2$ to show that the number of configurations returned by \Cref{alg:basic_round_and_round} (when invoked in Step~\ref{match:invoke_round} of \Cref{alg:match_and_round}) is bounded by $\approx \frac{4}{3}\cdot  \bx^0 \cdot \one_{\cC\setminus \cC_2} + \frac{1}{3} \cdot \bx^0\cdot \one_{\cC_2}$.

\subsection{Technical Contribution}
\label{sec:contribution}

Our main  technical contribution is the introduction of 
iterative {\em randomized} rounding in the context of {\sc Bin Packing}.
The ingenious randomized rounding techniques known for {\sc Bin Packing} problems (e.g., \cite{BansalCS2010}) rely on solving {\em once} a Configuration-LP and sampling a set of configurations according to the distribution induced by the Configuration-LP solution. In contrast, our  iterative randomized rounding approach is based on solving  a (modified) Configuration-LP  iteratively and sampling in each iteration a set of configurations using the distribution induced by the current LP solution.
While the resulting algorithms are simple and yield improved ratios, we are not aware of the use of iterative randomized rounding in previous studies of {\sc Bin Packing} problems.

Intuitively, we expect iterative randomized rounding to outperform non-iterative randomized rounding in the context of {\sc Bin Packing}. Indeed, the former is less likely to select many configurations containing the same item, presumably leading to a more efficient solution. Moreover, once a significant fraction
(say, $10\%$) of the items in $I$ is ``covered" (at random), we expect the Configuration-LP solution value to decrease. This can be used to obtain a better approximation ratio if we solve the {\em modified} Configuration-LP, and use the corresponding distribution to sample configurations.  However, formalizing the above intuition into a rigor proof is non-trivial. In the proof of \Cref{thm:better_than_randa} we provide a formal expression to the above intuition and prove that iterative randomized rounding is superior to any algorithm which follows the {\em Round\&Approx} framework. 

Our analysis for the case of arbitrary $d >2$ is fairly simple, leaving
 much room for improvement. For the special case of $d=2$ we use a tighter analysis. While many of the ingredients in this tighter analysis can be applied also to $d$VBP instances where $d >2$, and possibly to other {\sc Bin Packing} variants, some of the concepts exploit special properties of $2$VBP instances. This includes a strong structural property (\Cref{lem:structural}) on which we elaborate in \Cref{sec:twodim}, and the Matching-Configuration-LP. The strong structural property can potentially be extended to the $d$-dimensional case; however, such extension requires overcoming some technical challenges. We elaborate on these challenges in  \Cref{sec:discussion}.

\subsection{Organization}
\label{sec:organization}

In \Cref{sec:prelim} we give some definitions and notation. \Cref{sec:ddim} presents the analysis of \Cref{alg:basic_round_and_round} as well as the proof of \Cref{thm:better_than_randa} and \Cref{lem:AFTPAS}. \Cref{sec:twodim} gives the  results for $2$VBP  including the PTAS for the Matching-Configuration-LP (\Cref{lem:matching_ptas}). In \Cref{sec:prob_basic} we show basic properties which are used both in \Cref{sec:ddim,sec:twodim}.
 We conclude with a discussion in~\Cref{sec:discussion}.

\section{Preliminaries}
\label{sec:prelim}
In this section we give some basic definitions and properties that will be used in the proofs of \Cref{thm:better_than_randa} and \Cref{thm:main}.
Throughout the paper, for $x \in \mathbb{R}$, $\exp(x)=e^x$, where $e=2.718...$ is the base of the natural logarithm.

\subsection{Probability Space}
\label{sec:prelim_prob_space}
Our analysis refers to an execution of either \Cref{alg:basic_round_and_round} or \Cref{alg:match_and_round} on a $d$VBP instance~$(I,v)$.
For an execution of \Cref{alg:basic_round_and_round}, we have $S_0=I$. 
We use $(\Omega, \cF,\Pr)$ to denote the probability space generated by the algorithm.
Observe that as~$\delta<0.1$,
\begin{equation*}
  \rho_j \leq \ceil{\alpha z_j}\leq  \ceil{ \left(- \ln(1-\delta)\right)(1+\delta^2) \OPT}\leq \OPT, \qquad \text{for~} j = 1,\hdots,k \enspace .
\end{equation*}
Assume, without loss of generality, that \Cref{alg:basic_round_and_round} samples in each iteration  $\OPT$ configurations $C^j_1,\ldots, C^j_{\OPT}$ independently according to $\bx^j$, and ignores configurations $C^j_{\rho_j+1},\ldots, C^j_{\OPT}$. 
Furthermore, we may assume that $\Omega$ is finite. 
Define the random variables $P_0 = S_0$ and $P_j =(C^j_1,\ldots, C^j_{\OPT})$ for $j = 1,\hdots,k$.
Let $\cF_j = \sigma(P_0,P_1,\ldots,P_j)$ be the $\sigma$-algebra of the random variables $P_0,P_1,\ldots,P_j$.
We also define $\cF_{-1}=\{\emptyset,\Omega\}$.
It follows that $\cF_{-1}\subseteq \cF_0\subseteq \cF_1 \subseteq \ldots \subseteq \cF_k$.   
	
We use conditional expectations and probabilities given the $\sigma$-algebra $\cF_j$.
We refer the reader to standard textbooks on probability (e.g., by Chow and Teicher~\cite{ChowT1997}) for the formal definitions.
Intuitively, $\E\left[ X \middle| \cF_j\right]$ is the expectation of $X$ given the sample outcomes up to iteration $j$, and as such depends on the outcomes of the first $j$ iterations. 
	
The parameter $\alpha$ is set such that the probability of $i\in S_j$ decreases exponentially with~$j$, as stated in the next lemma.
\begin{lemma}
\label{lem:step_bound}
  For $j = 1,\hdots,k$ and $i\in I$ it holds that
  \begin{equation*}
  \Pr\left(i\in S_j~\middle|~\cF_{j-1} \right) = \one_{i\in S_{j-1}} \cdot  \left(1-\frac{1}{z_j} \right)^{\rho_j} \leq (1-\delta)\cdot \one_{i\in S_{j-1}}\enspace \enspace .
  \end{equation*}
\end{lemma}
\begin{proof}
  We can write
  \begin{equation}
  \label{eq:basic_decay}
    \begin{aligned}
	  &\Pr\left(\one_{i\in S_{j}} \middle |\cF_{j-1}\right)=\one_{i\in S_{j-1}}\cdot 
		\Pr\left( \forall \ell\in \rho_{j}:~i\notin C^j_{\ell}~\middle|~\cF_{j-1}  \right)
		= 	\one_{i\in S_{j-1}}\cdot  \prod_{\ell =1}^{\rho_{j}} \Pr\left(i\notin C^j_{\ell}~\middle|~\cF_{j-1}  \right)\\
		=~& \one_{i\in S_{j-1}}  \left(1-\frac{\one_{i\in S_{j-1}}} {z_{j}} \right)^{\rho_{j}} 
		= \one_{i\in S_{j-1}}  \left(1-\frac{1}{z_{j}} \right)^{\rho_{j}} 
		\leq  \one_{i\in S_{j-1}} \cdot \exp\left(-\alpha \right)= \one_{i\in S_{j-1}}\cdot (1-\delta) \enspace .
    \end{aligned}
  \end{equation}
  The first equality holds by the definition of $S_{j}$, and the second holds since $C^j_1,\ldots, C^j_{\rho_j}$ are conditionally independent given $\cF_{j-1}$ (note that $\rho_j$ is $\cF_{j-1}$-measurable).
  The third equality holds since~$\bx^j$ is a solution for $\LP(\one_{S_{j-1}})$ and $C^j_{\ell}\sim \bx^j$.
  The inequality in \eqref{eq:basic_decay} uses $\rho_j\geq \alpha z_j$ and $\left(1-\frac{1}{x}\right)^{x}\leq \exp(-1)$ for~$x\geq 1$. 
\end{proof}

\subsection{McDiarmid's Concentration Bound}

Our analysis heavily relies on concentration bounds. 
Let $A$ be an arbitrary  set, $m\in \mathbb{N}_+$ and $f:A^m\rightarrow \mathbb{R}$. For any $\eta\geq 0$, we say that $f$ is of {\em $\eta$-bounded difference} if for any $\bx, \bx'\in A^m$ and $r\in \{1,\hdots,m\}$ such that $\bx_{\ell} = \bx'_{\ell}$ for all $\ell\in \{1,\hdots,m\}\setminus \{r\}$ (i.e.,~$\bx$ and $\bx'$ differ only in the $r$-th entry)
it holds that $|f(\bx) - f(\bx')|\leq \eta$.
The next result is due to McDiarmid~\cite{McDiarmid1989}.
\begin{lemma}[McDiarmid]
	\label{lem:McDiarmid}
	Given a finite  arbitrary  set $A$, $m\in \mathbb{N}_+$ and $\eta>0$,
	let $f:A^m\rightarrow\mathbb{R}$ be a function of $\eta$-bounded difference. 
	Also, let $X_1,\ldots, X_m\in A$ be independent random variables. Then for any $t\geq 0$,
	\begin{equation*}
		\Pr \left( f(X_1,\ldots ,X_m)- \E\left[ f(X_1,\ldots, X_m)\right]>t\right)\leq \exp\left( -\frac{2\cdot t^2}{m\cdot \eta^2}\right)\enspace.
	\end{equation*}
\end{lemma}

To motivate our next lemma, consider the following example arising in our setting.
Let  $x_0,x_1,\ldots, x_k$  be random variables defined by $x_j = \sum_{t=1}^{d} v_t(S_j)$. That is, $x_j$ is the total volume of~$S_j$ in all dimensions.
Given $S_{j-1}$ and $\rho_j$ we 
we can express $x_j$  as a function of $C^j_1,\ldots, C^j_{\OPT}$.
For any $S\subseteq I$ and $\rho \in [\OPT]$ define $f_{S,\rho}:\cC^{\OPT} \rightarrow \mathbb{R}$ by
\begin{equation*} 
	f_{S,\rho} (C_1,\ldots, C_{\OPT}) =\sum_{t=1}^{d}v_t\left(S\setminus\left(\bigcup_{\ell=1}^{\rho} C_j\right) \right) \enspace .
\end{equation*}
Then it can be verified that 
$x_j
= g(C^j_1,\ldots, C^j_{\OPT})$ where $g=f_{S_{j-1}, \rho_j}$.
However, we cannot use \Cref{lem:McDiarmid} to show that 
$x_j
\approx \E[x_j]$ with high probability, since the random variables $C^j_1,\ldots C^j_{\OPT}$ are not independent, and the function $g$ is random. 

Nontheless, we note that at the end of iteration $j-1$ (Step~\ref{round:loop} of \Cref{alg:basic_round_and_round}) the values of~$S_{j-1}$ and~$\rho_j$ are known (while $\rho_j$ was not computed yet, its value does not depend on future random samples); thus, the function $g=f_{S_{j-1},\rho_j}$ is known at iteration $j$ of the algorithm.
Furthermore, the random variables $C^j_1,\ldots, C^j_{\OPT}$ are independent (by definition) assuming we have the random samples of the first $(j-1)$ iterations.
Therefore, we expect \Cref{lem:McDiarmid} to hold in this setting.
More formally, since $C^j_1,\ldots, C^j_{\OPT}$  are conditionally independent\footnote{See, e.g., the book by Chow and Teicher~\cite{ChowT1997} for a formal definition of conditional independence.}  given $\cF_{j-1}$, and $g=f_{S_{j-1}, \rho_j}$ is a random function that is $\cF_{j-1}$-measurable, we expect that $g(C^j_1,\ldots ,C^{j}_{\OPT})\approx \E[ g(C^j_1,\ldots ,C^{j}_{\OPT})|\cF_{j-1}]$. 
This is formalized in the next lemma.
\begin{lemma}[Generalized McDiarmid]
	\label{lem:Generalized_McDiarmid}
	Given a finite arbitrary set $A$, $m\in \mathbb{N}_+$ and  $\eta>0$,
	let $D$ be a finite family of $\eta$-bounded difference functions from $A^m$ to $\mathbb{R}$. 
	Let $(\Omega, \cF, \Pr)$ be a probability space for which $\Omega$ is finite, $\cG\subseteq \cF$ a $\sigma$-algebra, and $g\in D$ a $\cG$-measurable random function  (i.e., $g:\Omega\rightarrow D$ with $\{\omega\in \Omega|~g(\omega)\in U\}\in \cG$ for every $U\subseteq D$).
	Then, for a sequence of random variables $X_1,\ldots, X_m\in A$ which are conditionally independent given $\cG$, and any $t \geq 0$,
	\begin{equation*}
		\Pr\left(g(X_1,\ldots, X_m) - \E\left[ g(X_1,\ldots, X_m)|\cG\right] >t\right)~\leq~ \exp\left( -\frac{2\cdot t^2}{m\cdot \eta^2}\right) \enspace .
	\end{equation*}
\end{lemma}
\Cref{lem:Generalized_McDiarmid} can be derived from \Cref{lem:McDiarmid} using standard arguments from probability theory (we omit the details).

We use \Cref{lem:Generalized_McDiarmid} in the proofs of
Theorems~\ref{thm:better_than_randa} and~\ref{thm:main}.
For a set of items $S \subseteq I$, we denote by $\one_S \in \{ 0,1 \}^I$ an indicator vector in which entries corresponding to $i \in S$ are equal to `1', and all other entries are equal to `0'.\footnote{Similarly, for a set of configurations ${\cC'} \in \cC$, we use the indicator vector $\one_{\cC'} \in \{ 0,1 \}^{\cC}$ in which entries corresponding to $C \in {\cC'}$ are equal to `1'.}
The next lemma is used in the proofs of both theorems, and deals with random variables of the form $\one_{S_j}\cdot \bu$ where $\bu\in \mathbb{R}_{\geq 0}^I$.
Given $\bu\in \mathbb{R}^I$ define the {\em tolerance} of $\bu$ by $\tol(\bu) = \max_{C\in \cC}\left( \sum_{i\in C} \bu_i\ \right)$.
Intuitively, the vector $\bu$ associates with each item $i \in I$ some weight $\bu_i$; then $\tol(\bu)$ is the largest total weight of a configuration $C$ with respect to~$\bu$. 
\begin{lemma}
	\label{lem:concentration_multistep_prelim}
	Let $j\in \{0,1,\ldots, k-1\}$ and $t>0$.
	Also, let $\bu\in \mathbb{R}^I_{\geq 0}$ be an $\cF_{j}$-measurable random vector.
	Then,
	\begin{equation*}
		\Pr\left( \exists r\in \{j,\hdots,k\}:~\bu \cdot \one_{S_{r}} - (1-\delta)^{r-j}\cdot  \bu \cdot \one_{S_j} > t \cdot \tol(\bu) \right) ~\leq~ \delta^{-2}\cdot  \exp \left( -\frac{2\cdot\delta^{4} \cdot t^2}{\OPT} \right) \enspace .
	\end{equation*}
\end{lemma}
The proof of \Cref{lem:concentration_multistep_prelim}, given in \Cref{sec:prob_basic}, follows from \Cref{lem:Generalized_McDiarmid} and \Cref{lem:step_bound}.
	
	\subsection{First-Fit}
	
	In several places we use the following First-Fit strategy, which takes as input a $d$VBP instance~$(I,v)$ and a subset of items $S\subseteq I$.
	Throughout its execution, First-Fit maintains a set $A_1,\hdots,A_m\subseteq S$ of configurations, and iterates over the items in $S$.
	For each item $i\in S$, First-Fit examines the configurations sequentially, until it finds a configuration $A_i$ to which $i$ can be added without violating the volume constraints. 
	If no such configuration exists, First-Fit adds a new configuration $A_{m+1}= \{i\}$. 
	The next lemma follows from a simple analysis of First-Fit for {\sc Bin Packing} (see, e.g., Vazirani~\cite[Ch. 9]{Vazirani2001}), by taking for each item $i \in I$ in the $d$VBP instance ${\hat v} (i)= \max \{ v_1(i), \ldots, v_d(i)\}$, and considering the problem in single dimension.
	\begin{lemma}
		\label{lem:first_fit}
		Given a $d$VBP instance $(I,v)$ and a subset of items $S\subseteq I$, First-Fit returns a packing of $S$ in at most $2\cdot \left( \sum_{t=1}^{d} v_t(S)\right) +1$ bins. 
	\end{lemma}

	Recall that $\rho^*$ is the number of configurations used by the First-Fit strategy in Step~\ref{round:first_fit} of \Cref{alg:basic_round_and_round}.
	By \Cref{lem:step_bound}, it follows that $\E\left[ \sum_{t=1}^{d}v_t(S_k)\right]\leq (1-\delta)^{k} \left(\sum_{t=1}^{d} v_t(I)\right)\leq d\cdot \delta \OPT$, and by \Cref{lem:first_fit} we have $\E[\rho^*]\leq 2\cdot d\cdot\delta \OPT+1$.
	The next lemma uses \Cref{lem:concentration_multistep_prelim} to show that, with high probability, $\rho^*$ does not significantly deviate from its expectation.
	\begin{lemma}
		\label{lem:first_fit_bound}
		With probability at least $ 1-\delta^{-2}\cdot  \exp\left(-\delta^7 \cdot \OPT\right)$,  it holds that $\rho^* \leq  8\cdot d\cdot \delta \cdot \OPT+1$. 
	\end{lemma}
	\noindent
	The proof of \Cref{lem:first_fit_bound} is given in \Cref{sec:prob_basic}.   \Cref{lem:first_fit_bound} implies that the number of configurations added by the First-Fit  strategy in \Cref{round:first_fit} of \Cref{alg:basic_round_and_round} is negligible.

\section{Improved Asymptotic Approximation for $d$VBP}
\label{sec:ddim}
In this section we prove \Cref{thm:better_than_randa}.
That is, we  show that  \Cref{alg:basic_round_and_round} outperforms
{\em any} algorithm which falls into the {\em Round\&Approx} framework of Bansal et al.~\cite{BansalCS2010}. We also derive \Cref{lem:AFTPAS} as a simple consequence of the analysis of \Cref{alg:basic_round_and_round}.

As \Cref{thm:better_than_randa} refers to subset-oblivious algorithms, we first have to formally define this class of algorithms.
The following is a slight simplification of the definition of Bansal et al.~\cite[Definition 1]{BansalCS2010}.
\begin{defn}
\label{def:subset_oblivious}
  For every $d\in \mathbb{N}$ and  $\beta \geq 1$, an algorithm $\appr$
  is {\em $\beta$-subset oblivious} for $d$VBP if for every~$\eps>0$ there are $K\in \mathbb{N}$ and $\psi>0$ such that, for every $d$VBP instance $(I,v)$, there is a set of $K$ vectors~$\cS\subseteq \mathbb{R}^{I}_{\geq 0}$  which satisfies the following properties:
  \begin{enumerate}
    \item For any $\bu \in\cS$, it holds that $\tol(\bu)\leq \psi$. 
    \item $\OPT(I,v) \geq \max_{\bu \in \cS} \| \bu \| $.
	\item For any $Q\subseteq I$, given the  $d$VBP instance $(Q,v)$, $\appr$ returns a solution satisfying
	  $$\appr(I,v,Q)\leq \beta \cdot \max_{\bu \in \cS} \one_Q \cdot \bu +\eps \cdot \OPT(I,v) +K,$$
      where $\appr(I,v,Q)$ is  the number of bins used by the solution.
  \end{enumerate}
  We refer to $K$ and $\psi$ as the {\em $\eps$-parameters of $\appr$}, and to $\cS$ as the {\em  $\eps$-weight vectors of~$\appr$ and~$(I,v)$}. 
\end{defn}

Instead of \Cref{thm:better_than_randa} we prove a more specific result, which indicates also the dependencies between $\eps$ and $\delta$. 
\begin{thm}
\label{lem:better_than_randa}
  Let $\beta\geq 1$ and $d\in \mathbb{N}$.
  If there is a $\beta$-subset oblivious algorithm for $d$VBP then for every $\delta\in (0,0.1)$ such that $\delta^{-1}\in \mathbb{N}$ and $\delta <\min\left\{\frac{1}{28d^2},~\frac{1}{\beta}\right\}$ it holds that \Cref{alg:basic_round_and_round} configured with~$\delta$ is a randomized asymptotic $(1+\ln \beta -\chi(\beta,d) + 200 \cdot d^2 \cdot  \delta \cdot \beta )$-approximation algorithm for $d$VBP, where $\chi(\beta,d)=\left(\frac{1}{2}\cdot \ln \beta+ \frac{1}{\sqrt \beta}-1\right)\cdot \left( 1- \sqrt[2d]{\frac{1}{\beta}}\right)^d $.
\end{thm}
We give the proof of \Cref{lem:better_than_randa} in \Cref{sec:better_than_proof}.
We first use \Cref{lem:better_than_randa} to derive \Cref{lem:AFTPAS}.
\begin{proof}[Proof of \Cref{lem:AFTPAS}]
  Let $\eps\in (0,0.1)$ and $\delta = \frac{1}{\ceil{ 400\cdot \eps^{-1}} }$.
  Consider the execution of \Cref{alg:basic_round_and_round} with a {\sc Bin Packing} ($1$VBP) instance and the above parameter $\delta$.
  By \Cref{lem:d_subset_oblivious} there is a $(1+\delta)$-subset oblivious algorithm for {\sc Bin Packing}; thus, by \Cref{lem:better_than_randa}, \Cref{alg:basic_round_and_round} is a randomized asymptotic $\zeta$-approximation for {\sc Bin Packing}, where
  \begin{equation*}
	\zeta ~=~ \left(1+\ln (1+\delta )+200\cdot \delta \cdot (1+\delta) -\chi(1+\delta, 1 )\right) ~\leq~ \left(1+\delta+200\cdot \delta \cdot (1+\delta)\right)~\leq~ (1+\eps)\enspace. 
  \end{equation*}
  The first inequality holds as $\chi(1+\delta,1)\geq 0$ and $\ln(1+\delta \leq \delta)$.
  The last inequality follows from $400\cdot\delta \leq \eps$ by the definition of $\delta$.

  It is well-known that \eqref{eq:config_LP} admits an FPTAS for {\sc Bin Packing} instances.
  Indeed, in this case the separation oracle for the dual of \eqref{eq:config_LP} needs to solve an instance of the (standard) {\sc Knapsack} problem, for which there is an FPTAS (see, e.g., Vazirani's textbook~\cite{Vazirani2001}).
  Hence, the run time of each iteration in \Cref{round:loop} of \Cref{alg:basic_round_and_round} (given a {\sc Bin Packing} instance) is polynomial in the instance size and $\delta^{-2}$. As the total number of iterations is $k\leq \delta^{-2}$, it follows that the total run time is polynomial in the input size and in $\frac{1}{\delta}$.
  Since we defined $\delta$ to be polynomial in $\eps$, it follows that the run time is polynomial in the input size and $1/\varepsilon$.
  Thus, \Cref{alg:basic_round_and_round} is an AFPTAS for {\sc Bin Packing}.
\end{proof}

\subsection{Proof of \Cref{lem:better_than_randa}}
\label{sec:better_than_proof}
Let $(I,v)$ be a $d$VBP instance, and let $\appr$ be a $\beta$-subset oblivious algorithm for $d$VBP.
Also, let $\delta \in (0,0.1)$ such that $\delta \leq \frac{1}{28\cdot d^2}$, $\delta<\frac{1}{\beta}$, and $\delta^{-1}\in \mathbb{N}$.
We denote by $\OPT=\OPT(I,v)$ the value of  an optimal solution for the instance.
Consider an execution of \Cref{alg:basic_round_and_round} with the instance $(I,v)$, $S_0=I$ and the parameter $\delta$. 
We use notations such as $\rho_j$, $S_j$ and $C^j_{\ell}$ when referring to the corresponding variables in the execution of \Cref{alg:basic_round_and_round}. 
We also use the probability space $(\Omega,\Pr,\cF)$ and the filtration $\cF_{-1},\cF_0,\ldots, \cF_k$ as defined in \Cref{sec:prelim}. 
 
The size of the solution returned by \Cref{alg:basic_round_and_round} is $\sum_{j=1}^{k}\rho_j +\rho^*$.
By \Cref{lem:first_fit_bound}, the value of~$\rho^*$ is negligible with high probability.
Thus, we may focus in the analysis on $\sum_{j=1}^{k} \rho_j$.
This sum can be trivially upper bounded by
\begin{equation}
\label{eq:sum_rho}
  \begin{aligned}
    \sum_{j=1}^{k} \rho_j \leq&~\sum_{j=1}^{k} \ceil{\alpha \cdot z_k} \\
    \leq&~ k+\sum_{j=1}^{k} \alpha \cdot z_k\\
    \leq&~ \delta^{-2} +(1+\delta^2) (1+2\delta)\delta\sum_{j=1}^{k} \OPT(S_{j-1},v)\\
    \leq &~\delta^{-2} + (1+4\delta) \cdot \delta \sum_{j=1}^{k} \OPT(S_{j-1},v),
  \end{aligned}
\end{equation}
where the third inequality uses $\alpha = -\ln(1-\delta) \leq \delta \cdot (1+2\delta) $,  $k=\ceil{\log_{1-\delta}(\delta)}\leq \delta^{-2}$ and $$z_j
\leq (1+\delta^2) \cdot \OPT(S_{j-1},v)\enspace.$$
Following \eqref{eq:sum_rho}, we turn our attention to the expression $\delta \cdot \sum_{j=1}^{k} \OPT(S_{j-1},v)$.

We use the next trivial bound for small values of $j$.
\begin{obs}
\label{obs:trivial_bound}
  For $j=1,2,\ldots, k$ it holds that $\OPT(S_{j-1},v)\leq \OPT$. 
\end{obs}

We can use the subset-oblivious algorithm $\appr$ to obtain an additional bound on $\OPT(S_{j-1},v)$. 
Let $K$ and~$\psi$ be the $\delta^2$-parameters of $\appr$.
Observe that by \Cref{def:subset_oblivious}, it holds that $K$ and~$\psi$ depend solely on $\delta^2$, and are independent of the instance~$(I,v)$.
Without loss of generality, we assume that $\psi,K>1$. 

\begin{lemma}
\label{lem:so_bound}
  With probability at least $1-K\cdot \delta^2 \cdot \exp\left( -\frac{\delta^8}{\psi^2} \cdot \OPT\right)$, it holds that
  \begin{equation*}
    \forall j\in \{0,1,\ldots, k-1\}:~~~\OPT(S_j,v)\leq \beta \cdot (1-\delta)^j\cdot\OPT +2\cdot \delta^2 \cdot \beta \cdot \OPT + K\enspace. 
  \end{equation*}
\end{lemma}
\begin{proof}
  Let $\cS$ be the set of  $\delta^2$-weight vectors of $\appr$ and $(I,v)$. 
  The set $\cS$ is non-random, and is therefore $\cF_0$-measurable.
  Thus, by \Cref{lem:concentration_multistep_prelim}, for every $\bu\in \cS$ it holds that 
  \begin{equation}
  \label{eq:so_concentration}
    \begin{aligned}
	  &\Pr\left( \exists j\in \{0,1,\ldots, k\}: \bu\cdot \one_{S_j}>(1-\delta)^j \cdot \| \bu \| +\delta^2 \cdot \OPT\right) \\
		=~&\Pr\left( \exists j\in \{0,1,\ldots, k\}: \bu\cdot \one_{S_j}-(1-\delta)^j \cdot \bu \cdot  \one_{S_0} >\frac{\delta^2 \cdot \OPT}{\tol(\bu)}\cdot \tol(\bu) \right)\\
		\leq~& \delta^{-2}\cdot\exp  \left( -\frac{2\cdot \delta^4\cdot \left( \frac{\delta^2 \cdot \OPT}{\tol(\bu)}\right)^2 }{\OPT}\right)\\
		\leq~ & \delta^{-2} \cdot \exp\left( -\frac{2\cdot \delta^8 \cdot \OPT}{\psi^2}\right),
	\end{aligned}
  \end{equation}
  where the last inequality uses $\tol(\bu)\leq \psi$. We note that the second inequality in~\eqref{eq:so_concentration} assumes $\tol(\bu)\neq 0$, but the same outcome (i.e., the first expression is at most the last expression) can be trivially shown in case $\tol(\bu)=0$ (that is, $\bu$ is the zero vector).
	
  As $|\cS|\leq K$, we can use \eqref{eq:so_concentration} and the union bound to get
  \begin{equation}
  \label{eq:so_concentration_combined}
  \begin{aligned}
    &\Pr\left( \exists \bu \in \cS , j\in \{0,1,\ldots, k\}: \bu\cdot \one_{S_j}>(1-\delta)^j  \| \bu \| +\delta^2  \OPT\right)
	\leq K  \delta^{-2}  \exp\left( -\frac{2 \delta^8  \OPT}{\psi^2}\right)\enspace.
	\end{aligned}
  \end{equation}
  For the remainder of the proof we assume that 
  \begin{equation}
  \label{eq:so_first_assumption}
    \forall  j\in \{0,1,\ldots, k\}, ~\bu\in \cS: ~~~\bu\cdot \one_{S_j}\leq(1-\delta)^j  \| \bu \| +\delta^2  \OPT\enspace.
  \end{equation}
  By \eqref{eq:so_concentration_combined}, this assumption holds with probability at least  $1-K\cdot  \delta^{-2}\cdot  \exp\left( -\frac{2\cdot \delta^8 \cdot  \OPT}{\psi^2}\right)$. 
	
  Recall that $\OPT(I,v)\geq \max_{\bu\in \cS} \|\bu\|$ (\Cref{def:subset_oblivious}). Thus,
  \begin{equation}
  \label{eq:so_uniform_bound}
    \forall j\in \{0,1,\ldots, k\}, ~\bu\in \cS: ~~~\bu\cdot \one_{S_j}\leq (1-\delta)^j  \| \bu \| +\delta^2  \OPT\leq (1-\delta)^{j}\cdot \OPT+\delta^2\dot \OPT,
  \end{equation}
  where the first inequality is by \eqref{eq:so_first_assumption}. Hence, by \Cref{def:subset_oblivious}, for every $j=1,\ldots, k$ it holds that
  \begin{equation}
  \label{eq:so_appr_bound}
    \begin{aligned}
	  \OPT(S_{j-1},v)&\leq \appr(I,v,S_{j-1})\\
		&\leq~ \beta \cdot  \max_{\bu \in \cS} \one_{S_{j-1}} \cdot \bu +\delta^2 \cdot \OPT +K \\
		&\leq ~\beta \cdot \left( (1-\delta)^{j-1}\cdot \OPT+\delta^2\dot \OPT \right)  +\delta^2 \cdot\OPT +K\\
		&\leq~\beta\cdot (1-\delta)^{j-1}\cdot \OPT +2\cdot \beta \delta^2 \cdot\OPT +K,
	\end{aligned} 
  \end{equation}
  where the second inequality is by \eqref{eq:so_uniform_bound}. 
  Since we assumed \eqref{eq:so_first_assumption} holds, \eqref{eq:so_appr_bound} holds with probability at least $1-K \delta^{-2} \exp\left( -\frac{2\cdot \delta^8\cdot\OPT}{\psi^2}\right)$.
\end{proof}

We note that \Cref{obs:trivial_bound} and  \Cref{lem:so_bound,lem:first_fit_bound} suffice to show that \Cref{alg:basic_round_and_round} achieves an asymptotic  approximation ratio arbitrarily close to $(1+\ln \beta)$, which matches the {\em Round\&Approx} framework. To show \Cref{alg:basic_round_and_round} is strictly better we use some additional components.

We say a configuration $C\in \cC$ has {\em $\delta$-full slack} if $v_t(C)\leq 1-\delta$ for all $t=1,\ldots, d$. Define $\kappa(\delta) = \exp(\exp(\delta^{-3}))$. 
\begin{lemma}[Weak Structural Property]
\label{lem:weak_structural}
  Let $B_1,\ldots, B_s\in \cC$ be configurations such that $B_{\ell}$ has $\delta$-full slack for all $\ell=1,\ldots ,s$, and let $R=\bigcup_{\ell=1}^{s} B_\ell$.
  Then there exists a set $\cS\subseteq \mathbb{R}_{\geq 0}^I$ such that 
  \begin{itemize}
 	\item $\abs{\cS}\leq \kappa(\delta)$,
    \item $\supp(\bu)\subseteq R$ for all $\bu\in \cS$,\footnote{We define $\supp(\bu)=\{ i \in I~|~\bu_i>0\}$.}
    \item and for all $Q\subseteq R$ and $\gamma\in[0,1]$ which satisfy
      \begin{equation*}
        \forall \bu \in \cS:~~~~\one_{Q} \cdot \bu\leq \gamma \cdot \one_{R} \cdot \bu + \frac{\delta ^{20}}{\kappa(\delta)}\cdot \OPT(I,v) \cdot  \tol(\bu),
      \end{equation*}
      it holds that $\OPT(Q,v)\leq \gamma (1+d\cdot \delta)\cdot s + \delta^{10}\cdot \OPT+\kappa(\delta)$.
  \end{itemize}
\end{lemma}
We refer to $\cS$ as the {\em weak structure} of $B_1,\ldots, B_s$.
We defer the proof of \Cref{lem:weak_structural} to \Cref{sec:weak_structural}.  Intuitively, \Cref{lem:weak_structural} can be interpreted as follows.
If $R$ can be packed using $s$ configurations with $\delta$-full slack, and $Q\subseteq R$ is a random subset of $R$ such that $\Pr(i\in Q)\leq \gamma$ then $\OPT(Q,v)\lesssim \gamma s$, assuming $Q$ satisfies some concentration bounds.

We also utilize the existence of a nearly optimal solution of $(I,v)$ satisfying some additional properties.
We say an item $i\in I$ is {\em $\delta$-large} if there is a $t\in \{1,\hdots,d\}$ such that $v_t(i) \geq \delta$; otherwise, the item is {\em small}.
Observe that these notions extend the ones given in \Cref{sec:2vbp_improved} for the special case of $d=2$. 
Thus, we also use $L$ to denote the set of $\delta$-large items in the instance~$(I,v)$.
\begin{lemma}[Arranged solution]
\label{lem:sol_with_prop}
  For any $(I,v)$ there exists a solution $A_1,\ldots,A_m$ and sets\linebreak \mbox{$W_1,\ldots ,W_m\subseteq I$} such that
  \begin{itemize}
	\item $m\leq (1+d^2\cdot 14\cdot \delta)\cdot \OPT+1$,
	\item $W_b\subseteq A_b\cap L$ for  $b=1,\ldots ,m$,
	\item $|W_b|\leq d$ for $b=1,\ldots, m$,
	\item and $A_b\setminus W_b$  has $\delta$-full slack for $b=1,\ldots, m$. 
  \end{itemize}
\end{lemma}
We refer to $A_1,\ldots, A_m$ and $W_1,\ldots,W_m$ as an {\em arranged solution} of $(I,v)$.
We use \Cref{lem:sol_with_prop} as a means to utilize \Cref{lem:weak_structural}.
The main observation is that if $Z\subseteq [m]$ is a subset of configurations in the arranged solution such that $W_b\cap S_j =\emptyset$ for every $b\in Z$, then there is a weak structure of the configurations $(A_b\cap S_j)_{b\in Z}$ which can be used to bound $\OPT(S_{r},v)$ for $r\geq j$.

The proof of \Cref{lem:sol_with_prop} utilizes the following technical lemma of Bansal et al.~\cite{BansalEK2016}.
\begin{lemma}
\label{lem:d_small_items}
  Let $C\in \cC$ and let $Z\subseteq [d]$ be a set of coordinates such that $v_t(C) >1-\delta$ for all $t\in Z$  and $v_t(i)\leq \delta$ for all $i\in C$ and $t\in Z$.
  Then there is $Q\subseteq C$ such that $v_t(Q)\geq \delta$ for all $t\in Z$ and $v_t(Q)\leq 7\cdot d^2 \cdot \delta$ for all $t\in [d]$. 
\end{lemma}
\begin{proof}[Proof of \Cref{lem:sol_with_prop}]
  Let $A'_1,\ldots, A'_{m'}$ be an optimal solution for $(I,v)$.
  That is, $m'=\OPT(I,v)$. 
	
  For every $b\in [m']$ we define a set $W_b$ as follows.
  Start with  $W_b=\emptyset$ and while there is a coordinate $t\in [d]$  and $i\in A'_b\setminus W_b$ such that $v_t(A'_b\setminus W_b)>1-\delta $  and $v_t(i)\geq \delta$ add the item $i$ to $W_b$.
  Clearly, at the end of the process $\abs{W_b}\leq d$ and $W_b\subseteq A_b\cap L$.
  Furthermore, let $Z_b=\{t\in \{1,\hdots,d\}\mid v_t(A'_b\setminus W_b) > 1-\delta\}$.
  By construction of $W_b$ it holds that $v_t(i) \leq \delta$ for all $i\in A'_b\setminus W_b$ and $t\in Z_b$.
  Thus, by \Cref{lem:d_small_items}, for $b = 1,\hdots,m'$ there exists $Q_b\subseteq A'_b\setminus W_b$ such that $A'_b\setminus W_b \setminus Q_b$ has full slack and $v_t(Q_b) \leq 7\cdot d^2 \cdot \delta$ for all $t=1,\ldots,d$.
	
  Define $A_b = A'_b\setminus Q_b$.
  By the above $A_b\setminus W_b$ has $\delta$-full slack for $b=1,\ldots, m'$. 
  Let $\eta = \floor{\frac{ 1-\delta }{7\cdot d^2 \delta}}$, then the union of every $\eta$ of the sets among $Q_1,\ldots, Q_{m'}$ is a configuration with $\delta$-full slack.
  We simply iteratively pack $\eta$ of the sets $Q_1,\ldots,Q_{m'}$ into a single configuration.
  Thus there are configurations $A_{m'+1},\ldots, A_{m'+r}$ such that $A_b$ is with $\delta$-full slack for every $b=m'+1,\ldots, m'+r$, $r\leq \frac{m'}{\eta} +1$ and $A_{m'+1}\cup\ldots\cup A_{m'+r}=Q_1\cup\ldots, \cup Q_{m'}$.  We define $W_{m'+1},\ldots, W_{m'+r} =\emptyset$ and $m=m'+r$. 
	
  Since $\delta< \frac{1}{28 \cdot d^2}$, it holds that 
  \begin{equation*}
	\eta ~=~\floor{\frac{ 1-\delta }{7\cdot d^2 \delta}} ~\geq~ {\frac{ 1-\delta }{7\cdot d^2 \delta}} - 1 = \frac{1-\delta - 7\cdot d^2\cdot \delta }{7\cdot d^2 \cdot \delta } \geq \frac{\frac{1}{2}}{7\cdot d^2 \cdot \delta}= \frac{1}{14\cdot d^2\cdot \delta}\enspace.
  \end{equation*}
  Therefore, $r\leq \frac{m'}{\eta}+1\leq 14 \cdot d^2 \cdot \delta m'+1 = 14\cdot d^2 \cdot \delta \OPT+1$.
  Hence, $m=(1+14\cdot d^2\cdot \delta)\cdot \OPT +1$. 
\end{proof}

Let $A_1,\ldots, A_m$ and $W_1,\ldots ,W_m\subseteq I$  be an arranged solution of $(I,v)$. For every $j=0,1,\ldots,k$ define 
\begin{equation}
\label{eq:ddim_T}
  T_j=\{ b\in [m]~|~W_b\cap S_j=\emptyset\}
\end{equation}
to be the (indices of) configurations in the arranged solution such that $A_b\cap S_j$ is guaranteed to have $\delta$-full slack.
Define $j_1=\ceil{\frac{1}{2}\log_{1-\delta} \frac{1}{\beta}}$.
\begin{lemma}
\label{lem:Tj_based_analysis}
  With probability at least $1-K\cdot\kappa(\delta) \cdot\delta^{-4}\cdot  \exp\left(-\frac{\delta^{50}}{\psi^2 \cdot \kappa^2(\delta)}\cdot \OPT\right)$, it holds that
  \begin{equation*}
    \delta \sum_{j=1}^{k } \OPT(S_{j-1},v) ~\leq~ (1+\ln \beta )\OPT +|T_{j_1}|\cdot \left( 1-\frac{1}{\sqrt{\beta}} -\frac{1}{2} \ln \beta \right)+60\cdot d^2  \beta  \delta \cdot \OPT+\delta^{-3} K\cdot \beta\cdot \kappa(\delta)\enspace.
  \end{equation*}
\end{lemma}
The implication of \Cref{lem:Tj_based_analysis} is that if we show that $\abs{T_{j_1}}$ is at least a constant fraction of $\OPT$ (with high probability), then \Cref{alg:basic_round_and_round} attains an asymptotic approximation ratio which is strictly better than the $(1+\ln \beta)$ of {\em Round\&Approx}.
Indeed, such an assertion about $T_{j_1}$ will be proved later on in \Cref{lem:Tj_bound}.
We also note that the value of $j_1$ was selected arbitrarily.
A more refined analysis may consider  $\abs{T_{j}\setminus T_{j-1}}$ for {\em all} values of $j$. This concept is ingrained into our tighter analysis for the special case of $2$DVP given in \Cref{sec:new_analysis}.

The proof of \Cref{lem:Tj_based_analysis} partitions the sum $\delta \sum_{j=1}^{k } \OPT(S_{j-1},v)$ into three parts.
The first part is $\delta \sum_{j=1}^{j_1 } \OPT(S_{j-1},v)$, which is trivially bounded via \Cref{obs:trivial_bound}.
The last part is $\delta \sum_{j=j_2+1}^{k } \OPT(S_{j-1},v)$, where $j_2=\ceil{\log_{1-\delta} \frac{1}{\beta}}$.
Using the subset-oblivious algorithm based bound in \Cref{lem:so_bound}, this sum can be bounded by roughly $\OPT$.
The (remaining) middle part, $\delta \sum_{j=j_1+1}^{j_2 } \OPT(S_{j-1},v)$, utilizes a weak structure of the configuration in $T_{j_1}$ to attain a bound on $\OPT(S_{j-1},v)$, which is better than the trivial bound of $\OPT$ (and also better than the bound of \Cref{lem:so_bound} which is worse for those values of $j$).

\begin{proof}[Proof of \Cref{lem:Tj_based_analysis}]
  By \Cref{obs:trivial_bound},
  \begin{equation}
  \label{eq:up_to_j1}
    \begin{aligned}
	  \delta\sum_{r=1}^{j_1} \OPT(S_{r-1},v) &\leq~\delta \sum_{r=1}^{j_1} \OPT\\
	                                         &= ~\delta \cdot  j_1\cdot  \OPT \\
	                                         &\leq  ~\delta \cdot\frac{1}{2} \cdot \frac{\ln\beta}{-\ln(1-\delta) } \cdot \OPT+\delta\cdot  \OPT\\
	                                         &\leq ~\frac{1}{2} \left(\ln \beta\right) \OPT +\delta\cdot  \OPT\enspace.
    \end{aligned}
  \end{equation}
  The second inequality follows from the definition of $j_1$, and the third inequality holds since\linebreak \mbox{$-\ln(1-\delta )\geq \delta$}.

  Assume that 
  \begin{equation}
  \label{eq:so_assumption}
    \OPT(S_j,v)\leq \beta \cdot (1-\delta)^j\cdot \OPT +2\cdot \delta^2 \cdot \beta \cdot \OPT + K
  \end{equation}
  for all $j=0,1,\ldots, k-1$.
  By \Cref{lem:so_bound}, Assumption~\ref{eq:so_assumption} holds with probability at least $1-K\cdot  \delta^{-2}\cdot  \exp\left( -\frac{2\cdot \delta^8 \cdot  \OPT}{\psi^2}\right)$.
  Also, define $j_2 =  \ceil{\log_{1-\delta} \frac{1}{\beta}}$; therefore,
  \begin{equation}
  \label{eq:j2_onward}
    \begin{aligned}
      \delta\sum_{r=j_2+1}^{k} &\OPT(S_{r-1},v) \leq~ \delta\sum_{r=j_2+1}^{k } \left(\beta\cdot (1-\delta)^{r-1} \cdot \OPT+2\cdot \delta^2 \cdot \beta \cdot \OPT + K \right) \\
      \leq~& \beta\delta\cdot (1-\delta)^{j_2} \cdot \OPT \sum_{r=0}^{\infty} (1-\delta)^{r}   + 2\cdot k\cdot \delta\cdot\delta^2\cdot  \beta\cdot  \OPT+ \delta\cdot k\cdot K\\
      \leq~& \beta \cdot (1-\delta)^{j_2} \cdot \delta \cdot\frac{1}{1-(1-\delta)} \cdot \OPT +2\cdot \delta \beta \cdot\OPT + \delta^{-2}K\\
      \leq ~& \beta\cdot \frac{1}{\beta} \cdot \OPT + 2\cdot \delta \beta \cdot \OPT+\delta^{-2}\cdot K\\
      \leq~& \OPT + 2\cdot \delta\cdot\beta \cdot\OPT+ \delta^{-2}\cdot K\enspace.
    \end{aligned}
  \end{equation}
  The third inequality uses $k\leq \delta^{-2}$ and the forth inequality follows from $(1-\delta)^{j_2} \leq \frac{1}{\beta}$.

  Let $Q^* = \bigcup_{b\in T_{j_1}} A_b\cap S_{j_1}$.
  That is, $Q^*$ is the set of all items in  configurations which are guaranteed to have $\delta$-full slack in iteration $j_1$. 
  Since $\left(A_b\cap S_{j_1}\right)_{b\in T_{j_1}}$ is a collection of configuration with $\delta$-full slack, by \Cref{lem:weak_structural} there is a weak structure $\cS$ of $\left(A_b\cap S_{j_1}\right)_{b\in T_{j_1}}$. In particular, $\cS$ is $\cF_{j_1}$-measurable.
  Since $\supp(\bu )\subseteq Q^*$ for all $\bu \in \cS$, it follows that $\bu \cdot \one_{S_{r}} =\bu\cdot \one_{Q^*\cap S_r}$ for all $\bu\in \cS$ and $r=j_1,j_1+1,\ldots , k$.

  By \Cref{lem:concentration_multistep_prelim}, for every $\bu \in \cS$ it holds that 
  \begin{equation*}
    \begin{aligned}
	  \Pr&\left( \exists r\in \{j_1,\hdots,k\}:~\bu \cdot \one_{Q^*\cap S_{r}} - (1-\delta)^{r-j_1}\cdot  \bu \cdot \one_{Q^*} > \frac{\delta^{20}}{\kappa(\delta)}\cdot \OPT \cdot \tol(\bu) \right)\\
	  &=~\Pr\left( \exists r\in \{j_1,\hdots,k\}:~\bu \cdot \one_{ S_{r}} - (1-\delta)^{r-j_1}\cdot  \bu \cdot \one_{S_{j_1}} > \frac{\delta^{20}}{\kappa(\delta)}\cdot \OPT \cdot \tol(\bu) \right)\\
	  &\leq~ \delta^{-2}\cdot  \exp \left( -\frac{2\cdot\delta^{4} \cdot \frac{\delta^{40}}{\kappa^2(\delta)}\cdot \OPT^2}{\OPT} \right)\\
	  &\leq~ \delta^{-2}\cdot  \exp \left( -\frac{\delta^{50}}{\kappa^2(\delta)}\cdot \OPT \right)\enspace.
    \end{aligned}
  \end{equation*}
  Therefore, 
  \begin{equation}
  \label{eq:weak_prob}
    \begin{aligned}
      \Pr&\left( \forall \bu\in \cS,   r\in \{j_1,\hdots,k\}:~\bu \cdot \one_{S_{r} \cap Q^*} \leq  (1-\delta)^{r-j_1}\cdot  \bu \cdot \one_{Q^*} +\delta^{20}\cdot \OPT \cdot \tol(\bu) \right)\\
      & \geq ~1-\abs{\cS}\cdot \delta^{-2}\cdot \exp\left(-\frac{\delta^{50}}{\kappa^2(\delta)}\cdot \OPT\right)\\
      &\geq ~1- \kappa(\delta)\cdot \delta^{-2}\cdot \exp\left(-\frac{\delta^{50}}{\kappa^2(\delta)}\cdot \OPT\right)\enspace.
    \end{aligned} 
  \end{equation}

  For the remainder of the proof we assume that 
  \begin{equation}
  \label{eq:weak_assumption}
    \forall \bu\in \cS,   r\in \{j_1,\hdots,k\}:~\bu \cdot \one_{S_{r}\cap Q^*} \leq  (1-\delta)^{r-j_1}\cdot  \bu \cdot \one_{Q^*} +\frac{\delta^{20}}{\kappa(\delta)}\cdot   \OPT \cdot \tol(\bu).
  \end{equation}
  By \eqref{eq:weak_prob}, this assumption holds with probability at least $1-\kappa(\delta)\cdot\delta^{-2}\cdot \exp\left(-\frac{\delta^{50}}{\kappa^2(\delta)}\cdot \OPT\right)$. 

  By \eqref{eq:weak_assumption} it holds that 
  \begin{equation}
  \label{eq:sr_cap_qstar}
    \OPT(S_r\cap Q^*,v) \leq (1-\delta)^{r-j_1} \cdot (1+d\cdot \delta ) \abs{T_{j_1}} + \delta^{10}\cdot \OPT +\kappa(\delta)
  \end{equation}
  for all $r=j_1,j_1+1,\ldots k$.
  It trivially holds that 
  \begin{equation}
  \label{eq:j1_to_j2_decomp}
    \delta \sum_{r=j_1+1}^{j_2} \OPT(S_{r-1},v) =~\delta \sum_{r=j_1+1}^{j_2} \OPT(S_{r-1}\cap Q^*,v) + \delta \sum_{r=j_1+1}^{j_2} \OPT(S_{r-1}\setminus Q^*,v)\enspace. 
  \end{equation}
  By \eqref{eq:sr_cap_qstar} we have
  \begin{equation}
  \label{eq:with_qstar}
    \begin{aligned}
      \delta \sum_{r=j_1+1}^{j_2} &\OPT(S_{r-1}\cap Q^*,v)
      \leq~ \delta \sum_{r=j_1+1}^{j_2}\left( (1-\delta)^{r-1-j_1} \cdot (1+d\cdot \delta ) \abs{T_{j_1}} + \delta^{10}\cdot \OPT +\kappa(\delta)\right) \\
     &\leq~ \delta (1+d\cdot \delta ) \cdot |T_{j_1}|  \sum_{r=j_1+1}^{j_2 } (1-\delta)^{r-1-j_1} +k\cdot \delta^{11}\cdot \OPT+\delta \cdot  k\cdot\kappa(\delta)
 \\
&\leq ~(1+d\cdot \delta)\cdot |T_{j_1}| \cdot \delta \cdot\frac{1- (1-\delta)^{j_2-1 -j_1+1} }{1-(1-\delta)} + \delta^{9}\cdot \OPT+ \delta^{-1}\cdot \kappa(\delta) \\
&\leq ~(1+d\cdot \delta)\cdot |T_{j_1}| \left( 1- \frac{1}{\sqrt{\beta}} (1-\delta) \right) +\delta^9\cdot \OPT +\delta^{-1}\cdot \kappa(\delta)  \\
&\leq~ |T_{j_1}|  \cdot \left( 1-\frac{1}{\sqrt{\beta}}\right)  + 10\cdot d  \cdot \delta\cdot   \OPT +\delta^{-1}\cdot \kappa(\delta)\enspace.
    \end{aligned}
  \end{equation}
  The second inequality holds as $j_2 -j_1 \leq k$ .
  The third inequality uses $k\leq \delta^{-2}$.
  The forth inequality holds, as 
  \begin{equation}
  \label{eq:j_diff}
  j_2 -j_1   \leq \log_{1-\delta}\frac{1}{\beta}  +1 -\frac{1}{2} \cdot \log_{1-\delta}\frac{1}{\beta} = \frac{1}{2}\cdot\log_{1-\delta }\frac{1}{\beta} +1,
  \end{equation}
  thus $(1-\delta)^{j_2-j_1} \geq \frac{1}{\sqrt{\beta}} \cdot (1-\delta)$.
  The fifth inequality holds, as $|T_{j_1}| \leq m \leq (1+14\cdot d^2\cdot \delta )\OPT+1\leq 2\cdot \OPT+1$.

  It trivially holds that $\OPT(S_{r-1}\setminus Q^* ,v)\leq m-|T_j|$ for $r=j_1+1,\ldots ,j_2$ via the configurations $(A_b\cap S_{r-1})_{b\in \{1,\hdots,m\}\setminus T_j}$.
  Therefore,
  \begin{equation}
  \label{eq:no_qstar}
    \begin{aligned}
    \delta \sum_{r=j_1+1}^{j_2} \OPT(S_{r-1}\setminus Q^*,v) &\leq~\delta (j_2-j_1 ) (m -|T_{j_1}|) \\
                                                             &\leq~ \delta \left( \frac{1}{2}\cdot \log_{1-\delta} \frac{1}{\beta}+1\right) \cdot \left(m-|T_{j_1}|\right) \\
                                                             &= ~ \delta\cdot   \frac{1}{2}\cdot\frac{\ln \beta}{-\ln(1-\delta ) }  \cdot \left(m-|T_{j_1}|\right) + \delta (m-|T_{j_1}|)\\
                                                             &\leq~\frac{1}{2} \left( \ln \beta\right)  (m-|T_{j_1}|) +\delta m \\
                                                             &\leq~\frac{1}{2} \left( \ln \beta\right)  (m-|T_{j_1}|) +2\cdot \delta \OPT\enspace.
    \end{aligned}
  \end{equation}
  The second inequality follows from \eqref{eq:j_diff}.
  The third inequality holds, as $-\ln(1-\delta ) \geq \delta$. 

  By \eqref{eq:j1_to_j2_decomp}, \eqref{eq:with_qstar} and \eqref{eq:no_qstar} we have 
  \begin{equation}
  \label{eq:j1_to_j2}
    \begin{aligned}
	  \delta \sum_{r=j_1+1}^{j_2} & \OPT(S_{r-1},v) \\
	  &\leq~|T_{j_1}|   \left( 1-\frac{1}{\sqrt{\beta}}\right)  + 10\cdot d\cdot  \delta  \OPT +\delta^{-1}\cdot \kappa(\delta)+\frac{1}{2} \left( \ln \beta\right)  (m-|T_{j_1}|) +2\cdot \delta  \OPT \\
	  & \leq ~\frac{m}{2} \ln \beta  +|T_{j_1}| \cdot \left( 1-\frac{1}{\sqrt{\beta}} -\frac{1}{2} \ln \beta \right) +20\cdot d\cdot  \delta \cdot \OPT+\delta^{-1}\cdot \kappa(\delta)\\
	  & \leq ~\frac{(1+14\cdot d^2\delta  )\cdot \OPT+1}{2} \ln \beta  +|T_{j_1}|  \left( 1-\frac{1}{\sqrt{\beta}} -\frac{1}{2} \ln \beta \right) +20\cdot   d \delta  \OPT+\delta^{-1} \kappa(\delta)\\
	  & \leq ~\frac{\OPT}{2} \ln \beta  +|T_{j_1}| \cdot \left( 1-\frac{1}{\sqrt{\beta}} -\frac{1}{2} \ln \beta \right) +50\cdot d^2 \cdot \beta \cdot \delta \cdot \OPT+\delta^{-1}\cdot \kappa(\delta) + \frac{\ln \beta}{2}\enspace.
  \end{aligned}
  \end{equation}

  By \eqref{eq:up_to_j1},  \eqref{eq:j2_onward}, and \eqref{eq:j1_to_j2} we have
  \begin{equation*}
    \begin{aligned}
      \delta \sum_{r=1}^{k}& \OPT(S_{r-1},v) \\
      &\leq~ \frac{1}{2} \left(\ln \beta\right) \OPT +\delta \OPT
      \\& ~~~~~~~~~+\frac{\OPT}{2} \ln \beta  +|T_{j_1}| \cdot \left( 1-\frac{1}{\sqrt{\beta}} -\frac{1}{2} \ln \beta \right) +50\cdot d^2  \beta  \delta \cdot \OPT+\delta^{-1}\cdot \kappa(\delta)+\frac{\ln \beta}{2}\\
      & ~~~~~~~~~ +\OPT + 2\cdot \delta\cdot\beta \cdot\OPT+ \delta^{-2}\cdot K\\
      &\leq ~ (1+\ln \beta )\OPT +|T_{j_1}|\cdot \left( 1-\frac{1}{\sqrt{\beta}} -\frac{1}{2} \ln \beta \right)+60\cdot d^2  \beta  \delta \cdot \OPT+\delta^{-3} K\cdot \beta\cdot \kappa(\delta)\enspace.
    \end{aligned}
  \end{equation*}
  As we assumed that \eqref{eq:so_assumption} and \eqref{eq:weak_assumption} hold, the statement holds with probability
  \begin{equation*}
    \begin{aligned}
      &1-\kappa(\delta)\cdot\delta^{-2}\cdot \exp\left(-\frac{\delta^{50}}{\kappa^2(\delta)}\cdot \OPT\right)-K\cdot  \delta^{-2}\cdot  \exp\left( -\frac{2\cdot \delta^8 \cdot  \OPT}{\psi^2}\right) \\
      \geq~& 1-K\cdot\kappa(\delta) \cdot\delta^{-4}\cdot  \exp\left(-\frac{\delta^{50}}{\psi^2 \cdot \kappa^2(\delta)}\cdot \OPT\right) \enspace.\hfill*\qedhere
    \end{aligned}
  \end{equation*}
\end{proof}

To attain the statement of \Cref{lem:better_than_randa}, we show that $|T_{j_1}|$ is at least a constant fraction of $\OPT$ (with high probability), and combine this result with \Cref{lem:Tj_based_analysis}.

\begin{lemma}
\label{lem:Tj_bound}
  With probability at least $1-\delta^{-2}\cdot \exp(-\delta^{50}\cdot \OPT)$, it holds that
  \begin{equation*}
  |T_{j_1} | \geq \left(1-\beta^{-\frac{1}{2d}}\right)^d\cdot m -4\cdot d\cdot \delta \cdot \OPT\enspace.
  \end{equation*}
\end{lemma}
\begin{proof}
  For $j=0,1,\ldots, k$ and $\ell =0,1,\ldots , d$ define $V_{j,\ell} =\left\{b\in [m]~\middle|~\abs{W_b\cap S_j}=\ell\right\}$ and $V_{j,\leq \ell } = \bigcup_{h=0}^{\ell } V_{j,\ell}$.
  Since $A_b\setminus W_b$ is guaranteed to have $\delta$-full slack, the set $V_{j,\ell}$ ($V_{j,\leq \ell}$) can be intuitively interpreted as  (the indices of) the set of configurations among $A_1\cap S_j,\ldots, A_m\cap S_j$  which  have $\delta$-full slack if (at most) $\ell$ specific large items are removed from them. 
  Since $S_0\supseteq S_1 \supseteq \ldots \supseteq S_k$ it holds that $V_{0,\leq \ell}\subseteq V_{1,\leq \ell }\subseteq \ldots \subseteq V_{k,\leq \ell}$.
  Observe that $V_{j,\ell}$ is $\cF_{j}$-measurable and $T_j =V_{j,0} = V_{j,\leq 0}$. 

  Observe that for every $b=1,2,\ldots, m$  and $\ell =0,1,\ldots, d$ it  holds that $\{j~|~b\in V_{j,\ell}\} $ is a set of consecutive integers. That is, $b$ belong to $V_{j,\ell}$ from some iteration $r_1$ up to some iteration $r_2$. 
  The next claim essentially states that the difference $r_2-r_1$ is not expected to be too large.
  \begin{claim}
  \label{claim:expectation_diminishing_levels}
    Let $j\in \{0,1\ldots, k-1\}$, $\ell\in\{ 1,\ldots, d\}$ and let $Z\subseteq V_{j,\ell}$ be an $\cF_{j}$-measurable subset.
    Then it holds that
    \begin{equation*}
      \E\left[~\abs{Z\cap V_{j,\ell+1}}~ \middle|~ \cF_{j}\right]\leq (1-\delta)\cdot \abs{Z} \enspace .
    \end{equation*}
  \end{claim}
  \begin{claimproof}
    For $b = 1,\hdots,m$ let $i_b$ be an arbitrary item in $W_b\cap S_j$ (or an arbitrary item in $I$ in case $W_b\cap S_j=\emptyset$).
    In particular, $i_b$ is an $\cF_j$-measurable random variable.
    For $b = 1,\hdots,m$ it holds that
    \begin{equation*}
    \label{eq:Vjl_first}
	  \begin{aligned}
	    \Pr&\left(b\in Z\cap V_{j+1,\ell} ~\middle|~\cF_j\right)  = \Pr\left(b\in Z \textnormal{ and } W_b\cap S_j \subseteq S_{j+1} ~\middle|~\cF_j\right) \\
	       &\leq~ \Pr\left( b\in Z \textnormal{ and } i_b\in S_{j+1}~\middle|~\cF_j \right) \\ 
	       &\leq~ \one_{ b\in Z } \cdot (1-\delta ) \cdot \one_{i_b \in S_j} \\
	       &=~ (1-\delta )\cdot  \one_{ b\in Z}\enspace.
	  \end{aligned}
    \end{equation*}
    The second inequality follows from \Cref{lem:step_bound}.
    That last equality holds since if $b\in Z\subseteq V_{j,\ell }$ then $i_b\in S_j$ as $\ell \neq 0$.
    Thus,
    \begin{equation*}
	  \E\left[~\abs{Z\cap V_{j,\ell+1}}~ \middle|~ \cF_{j}\right]  = ~\sum_{b\in [m]}  \Pr\left(b\in Z\cap V_{j+1,\ell} ~\middle|~\cF_j\right)  \leq ~\sum_{b\in [m]} (1-\delta )\cdot  \one_{ b\in Z } = (1-\delta) |Z|\enspace.
    \end{equation*}
  \end{claimproof}
  We use \Cref{lem:Generalized_McDiarmid} to show that $\abs{Z\cap V_{j,\ell+1}}$ cannot be significantly larger than the bound on its expectation as stated in \Cref{claim:expectation_diminishing_levels}.
	
  \begin{claim}
  \label{claim:concentrated_diminishing_levels}
    Let $j\in \{0,1\ldots, k-1\}$, $\ell\in\{ 1,\ldots, d\}$ and let $Z\subseteq V_{j,\ell}$ be an $\cF_{j}$-measurable subset.
    Then $\abs{Z\cap V_{j,\ell+1}}\leq (1-\delta)\cdot \abs{Z}+\delta^{20} \cdot \OPT$ with probability at least $1-\exp\left(-\delta^{50}\cdot \OPT\right)$.
  \end{claim}
  \begin{claimproof}
    For every $S\subseteq I$, $\rho \in [\OPT]$  and $X\subseteq [m]$ define a function $f_{S,\rho,X}:\cC^{\OPT}\rightarrow \mathbb{R}$ by
    \begin{equation*}
      f_{S,\rho,X } (C_1,\ldots, C_{\OPT}  ) = \sum_{b\in X} \one_{W_b\cap S\cap  \left( \bigcup_{s =1}^{\rho} C_s \right)=\emptyset}\enspace.
    \end{equation*}
    Observe that 
    \begin{equation*}
   	  f_{S_j, \rho_{j+1}, Z}(C^{j+1}_1,\ldots, C^{j+1}_{\OPT} ) ~=~ \sum_{b\in Z} \one_{W_b\cap S_j\cap  \left( \bigcup_{s =1}^{\rho_{j+1}} C^{j+1}_s \right)=\emptyset} ~=~ \sum_{b\in Z} \one_{Z\in V_{j+1,\ell}} ~=~ \abs{Z\cap V_{j+1,\ell}}\enspace .
    \end{equation*}
    Moreover, as $S_j$, $\rho_{j+1}$ and $Z$ are $\cF_j$ measurable it follows that $f_{S_j,\rho_{j+1},Z}$ is $\cF_j$-measurable as well (note that $\rho_{j+1}$ is determined before $C^{j_1},_1,\ldots, C^{j+1}_{\rho_{j+1}}$ are sampled in \Cref{round:sample} of \Cref{alg:basic_round_and_round}).
 	
    Define $D=\{f_{S,\rho, X}~|~S\subseteq I ,~ \rho \in [\OPT], ~X\subseteq [m]\}$.
    It follows that $D$ is a finite set.
    In order to use \Cref{lem:Generalized_McDiarmid} we need to show that the functions in $D$ are of bounded difference.
 	
    Let $f_{S,\rho, X} \in D$,  $(C_1,\ldots, C_{\OPT}),~(C'_1,\ldots, C'_{\OPT})\in \cC^{\OPT} $ and $r\in [\OPT]$ such that $C_{s} =C'_{s}$ for $s = 1,\hdots,r-1,r+1,\hdots,\OPT$ (i.e., $(C_1,\ldots, C_{\OPT})$ and $(C'_1,\ldots, C'_{\OPT})$ are identical in all coordinates expect the $r$-th).
    If $r>\rho$ then
    \begin{equation*}
      \left|f_{S,\rho, X} (C_1,\ldots, C_\OPT) -f_{S,\rho,X}(C'_1,\ldots, C'_{\OPT})\right| =0\enspace.
    \end{equation*}
    Otherwise, 
    \begin{equation*}
      \begin{aligned}
        \bigg|f_{S,\rho, X}& (C_1,\ldots, C_\OPT) -f_{S,\rho,X}(C'_1,\ldots, C'_{\OPT})\bigg|=\left| \sum_{b\in X} \one_{W_b\cap S\cap  \left( \bigcup_{s =1}^{\rho} C_s \right)=\emptyset}- \sum_{b\in X} \one_{W_b\cap S\cap  \left( \bigcup_{s =1}^{\rho} C'_s \right)=\emptyset}\right|\\
 			&\leq \sum_{b\in X} \one_{W_b\cap S\cap C_r \neq \emptyset}+ \sum_{b\in X} \one_{W_b\cap S\cap C'_r \neq \emptyset}\\
 			&\leq 2\cdot d\cdot \delta^{-1}
 			 \enspace .
 	  \end{aligned}
    \end{equation*}
    The last inequality holds, since the sets $W_1,\ldots, W_m$ are pairwise disjoint and only contain large items, and furthermore, a configuration $C\in \cC$ may contain at most $d\cdot \delta^{-1}$ large items.
    Thus, $f_{S,\rho,X}$ is of $(2\cdot d\cdot \delta^{-1})$-bounded difference.

    By \Cref{claim:expectation_diminishing_levels} and \Cref{lem:Generalized_McDiarmid} we have
    \begin{equation*}
      \begin{aligned}
        \Pr&\left(\abs{Z\cap V_{j,\ell+1}} ~>~(1-\delta ) \abs{Z} +\delta^{20}\cdot  \OPT \right)
       \\&\leq~ \Pr\bigg(\abs{Z\cap V_{j,\ell+1}}   - \E\big[ \abs{Z\cap V_{j,\ell+1}}~\big|~\cF_j\big] ~>~ \delta^{20} \cdot \OPT \bigg)\\
      &\leq~\Pr\left(f_{S_j, \rho_{j+1}, Z}(C^{j+1}_1,\ldots, C^{j+1}_{\OPT} )  - \E\left[ 	f_{S_j, \rho_{j+1}, Z}(C^{j+1}_1,\ldots, C^{j+1}_{\OPT} )  ~\middle|~\cF_j\right]~>~\delta^{20}\cdot \OPT \right)\\
      &\leq~ \exp \left(- \frac{2\cdot \delta^{40 }\cdot \OPT^2}{\OPT \cdot 4 \cdot d^2 \cdot \delta^{-2}}\right) ~\leq \exp\left(-\delta^{50}\cdot \OPT\right) \enspace.
      \end{aligned}
    \end{equation*}
    The last inequality holds as $\frac{1}{d^2} \geq 28 \delta\geq \delta$.
  \end{claimproof}

  Define $\eta = \floor{\frac{1}{2d}\cdot\log_{1-\delta}\frac{1}{\beta} }$. 
  We use \Cref{claim:concentrated_diminishing_levels} to prove the following.
  \begin{claim}
  \label{claim:levels_first_induction}
    Let $\ell \in \{0,1,\ldots, d-1\}$.
    Then
    \begin{equation*}
      \abs{V_{\ell \cdot \eta, d-\ell } \cap V_{ (\ell +1) \cdot \eta,  d-\ell }}\leq \beta^{-\frac{1}{2d}}(1+2\delta) \cdot \abs{V_{\ell\cdot \eta,  d-\ell} }  +\eta \cdot \delta^{20} \cdot \OPT
    \end{equation*}
    with probability at least $1-\eta\cdot \exp\left(-\delta^{50 }\cdot \OPT\right)$.
  \end{claim}
  \begin{claimproof}
    We use induction on $j=0,1,\ldots, \eta$ to show that 
    \begin{equation*}
      \abs{V_{\ell \cdot \eta, d-\ell } \cap V_{ \ell \cdot \eta+j,  d-\ell }}\leq (1-\delta)^j\cdot \abs{V_{\ell\cdot \eta,  d-\ell} }  +j \cdot \delta^{20} \cdot \OPT
    \end{equation*}
    with probability at least $1-j\cdot \exp\left(-\delta^{50 }\cdot \OPT\right)$.

    \noindent{\bf Base case:} For $j=0$, it holds that  $\abs{V_{\ell\cdot \eta,  d-\ell} \cap V_{\ell\cdot \eta+j,  d-\ell}}= \abs{V_{\ell\cdot \eta,  d-\ell}}$ with probability $1$.

    \noindent{\bf Induction Step:} Assume the induction hypothesis holds for some $j \geq 0$.
    Define $Z=V_{\ell \cdot \eta, d-\ell } \cap V_{ \ell \cdot \eta+j,  d-\ell } $, and observe that $Z$ is $\cF_{\ell\cdot \eta +j}$-mesuarable.
    By the induction hypothesis and \Cref{claim:expectation_diminishing_levels}, it holds that 
    \begin{equation}
    \label{eq:levels_inner_induction}
      \begin{aligned}
        &\abs{Z}&=&~\abs{V_{\ell \cdot \eta, d-\ell } \cap V_{ \ell \cdot \eta+j,  d-\ell }}
        &&\leq~ (1-\delta)^j\cdot \abs{V_{\ell\cdot \eta,  d-\ell} }  +j \cdot \delta^{20} \cdot \OPT\\
        \textnormal { and }&&&~\abs{Z\cap V_{\eta \cdot \ell + j +1,d-\ell}}
        &&\leq~ (1-\delta)\cdot \abs{Z}+\delta^{20} \cdot \OPT
     \end{aligned}
    \end{equation}
    with probability at least $1-(j+1)\cdot \exp\left(-\delta^{50 }\cdot \OPT\right)$.
    Furthermore, if \eqref{eq:levels_inner_induction} holds, then
    \begin{equation*}
      \begin{aligned}
        \abs{V_{\ell \cdot \eta, d-\ell } \cap V_{ \ell \cdot \eta+j+1,d-\ell }}~&=~ \abs{Z \cap V_{ \ell \cdot \eta+j+1,  d-\ell }}\\
        &\leq~ (1-\delta)\cdot \abs{Z } + \delta^{20} \cdot \OPT\\
        &\leq ~(1-\delta)^{j+1} \cdot \abs{V_{\eta\cdot \ell, d-\ell}} + (j+1)\cdot \delta^{20} \cdot \OPT,
      \end{aligned}
    \end{equation*}
    where the first equality holds since for all $b\in V_{\ell \cdot \eta, d-\ell } \cap V_{ \ell \cdot \eta+j+1,  d-\ell }$ it also must hold that~$b\in V_{ \ell \cdot \eta+j,  d-\ell }$.
    This completes the induction step.

    Therefore, using the definition of $\eta$, 
    \begin{equation*}
      \begin{aligned}
        \abs{V_{\ell \cdot \eta, d-\ell } \cap V_{ (\ell+1) \cdot \eta,  d-\ell }}~&\leq ~(1-\delta)^\eta\cdot \abs{V_{\ell\cdot \eta,  d-\ell} }  +\eta\cdot \delta^{20} \cdot \OPT\\
        &\leq ~ \frac{\beta^{-\frac{1}{2d}} }{1-\delta}\cdot  \abs{V_{\ell\cdot \eta,  d-\ell} }+\eta \cdot \delta^{20}\cdot\OPT\\
        &\leq ~ {\beta^{-\frac{1}{2d}} }\cdot (1+2\delta)\cdot  \abs{V_{\ell\cdot \eta,  d-\ell} }+\eta \cdot \delta^{20}\cdot \OPT
      \end{aligned}
    \end{equation*}
    with probability at least $1-\eta\cdot \exp\left(-\delta^{50 }\cdot \OPT\right)$.
  \end{claimproof}

  Using \Cref{claim:levels_first_induction} and a simple induction, we attain the following.
  \begin{claim}
  \label{claim:levels_outer_induction}
    Let $\ell \in \{0,1,\ldots,d\}$.
    Then $\abs{V_{\ell \cdot \eta, \leq d-\ell  }} \geq \left( 1- \beta^{-\frac{1}{2d}} (1+2\delta)\right)^\ell\cdot m - \ell \cdot \eta \cdot \delta^{20}\cdot \OPT$ with probability at least $1-\ell\cdot \eta \cdot \exp(-\delta^{50}\cdot \OPT)$.
  \end{claim}
  \begin{claimproof}
    We prove the claim by induction over $\ell$. 

  \noindent {\bf Base Case:} For $\ell=0$ it holds that
  \begin{equation*}
    \abs{V_{0, \leq d  }} = \abs{[m]} = \left( 1- \beta^{-\frac{1}{2d}}\cdot (1+2\delta)\right)^0\cdot m - 0 \cdot \eta \cdot \delta^{10}\cdot \OPT \enspace .
  \end{equation*}
  \noindent{\bf Induction Step:}  Assume the  claim holds for $\ell<d$.
  Then, by the induction hypothesis and \Cref{claim:levels_first_induction} it holds that, with probability at least $1-(\ell+1)\cdot \eta\cdot \exp\left(-\delta^{50 }\cdot \OPT\right)$,
  \begin{equation}
  \label{eq:levels_outer_induction}
	\begin{aligned} 
	  &\abs{V_{\ell \cdot \eta, \leq d-\ell  }} &\geq~ &\left( 1- \beta^{-\frac{1}{2d}} (1+2\delta)\right)^\ell\cdot m - \ell \cdot \eta \cdot \delta^{20}\cdot \OPT\\
	  \textnormal{ and }&\abs{V_{\ell \cdot \eta, d-\ell } \cap V_{ (\ell +1) \cdot \eta,  d-\ell }}&\leq~& \beta^{-\frac{1}{2d}} \cdot(1+2\delta) \abs{V_{\ell\cdot \eta,  d-\ell} }  +\eta \cdot \delta^{20} \cdot \OPT \enspace .
	\end{aligned}
  \end{equation}
  Assuming \eqref{eq:levels_outer_induction} holds, we have 
  \begin{equation*}
    \begin{aligned}
	\big|&V_{(\ell+1)\eta,\leq d-\ell -1 }\big| ~\geq~ \abs{V_{\ell\eta,\leq d-\ell -1}}  +\abs{ V_{\ell\eta, d-\ell} \setminus V_{ (\ell+1)\eta, d-\ell} }\\
	\big|&V_{(\ell+1)\eta,\leq d-\ell -1 }\big| ~\geq~ \abs{V_{\ell\eta,\leq d-\ell -1}}  +\abs{ V_{\ell\eta, d-\ell} \setminus V_{ (\ell+1)\eta, d-\ell} }\\
		&=~\abs{V_{\ell\eta,\leq d-\ell -1}} +\abs{V_{\ell\eta, d-\ell} }  -\abs{V_{\ell\eta, d-\ell} \cap V_{ (\ell+1)\eta, d-\ell } } \\
		&\geq ~ \abs{V_{\ell\eta,\leq d-\ell -1}} +\abs{V_{\ell\eta, d-\ell} }   -\left( \beta^{-\frac{1}{2d}} \cdot(1+2\delta) \abs{V_{\ell\cdot \eta,  d-\ell} }  +\eta \cdot \delta^{20} \cdot \OPT\right) \\
		&= ~\abs{V_{\ell\eta,\leq d-\ell -1}} +\left( 1-\beta^{-\frac{1}{2d}} \cdot(1+2\delta) \right) \cdot \abs{ V_{\ell\cdot \eta,  d-\ell} } -\eta \cdot \delta^{20} \cdot \OPT\\
		&\geq~  \left( 1-\beta^{-\frac{1}{2d}} \cdot(1+2\delta) \right) \cdot \abs{V_{\ell \cdot \eta, \leq d-\ell}} -\eta \cdot \delta^{20} \cdot \OPT\\
			&~\geq  \left( 1-\beta^{-\frac{1}{2d}} \cdot(1+2\delta) \right) \cdot \left( \left( 1- \beta^{-\frac{1}{2d}} (1+2\delta)\right)^\ell\cdot m - \ell \cdot \eta \cdot \delta^{20}\cdot \OPT \right) -\eta \cdot \delta^{20} \cdot \OPT\\
			&\geq~\left( 1- \beta^{-\frac{1}{2d}} (1+2\delta)\right)^{\ell+1}\cdot m - (\ell+1) \cdot \eta \cdot \delta^{20}\cdot \OPT,
    \end{aligned}
  \end{equation*}
  which completes the induction step.
  \end{claimproof}

  By \Cref{claim:levels_outer_induction} it follows that with probability at least $1-d \cdot \eta \exp\left(-\delta^{50}\cdot \OPT\right) \geq 1-\delta^{2}\cdot \exp\left(-\delta^{50} \cdot \OPT\right)$ it holds that
  \begin{equation*}
    \begin{aligned}
    \abs{V_{j_1,0}} ~&\geq~ \abs{V_{\eta \cdot d, \leq 0 }} \\
    &\geq~ \left( 1- \beta^{-\frac{1}{2d}} (1+2\delta)\right)^d\cdot m - d \cdot \eta \cdot \delta^{20}\cdot \OPT\\
    &\geq ~  \left( 1- \beta^{-\frac{1}{2d}}\right)^d\cdot m -2d\cdot \delta \cdot m+ \cdot \delta^{18}\cdot \OPT\\
    &\geq ~\left( 1- \beta^{-\frac{1}{2d}}\right)^d\cdot m -4d\cdot \delta \cdot \OPT \enspace .
    \end{aligned}
  \end{equation*}
  The first inequality holds since $\eta \cdot d = d \cdot \floor{\frac{1}{2d} \cdot \log_{1-\delta}\frac{1}{\beta}} \leq \frac{1}{2} \cdot \log_{1-\delta}\frac{1}{\beta}\leq j_1$.
  The third inequality holds since $ \left( 1- \beta^{-\frac{1}{2d}} (1+2\delta)\right)^d \geq \left( 1- \beta^{-\frac{1}{2d}}\right)^d-2\delta \cdot d$ and  $d\cdot \eta \leq k \leq \delta^{-2}$.
  The last inequality holds as $m\leq 2\cdot \OPT$.
\end{proof}

To complete the proof of \Cref{lem:better_than_randa} we only need to combine the results of \Cref{lem:first_fit_bound,lem:Tj_based_analysis,lem:Tj_bound}.
Assume the inequalities 
\begin{equation}
\label{eq:ddim_assump}
  \begin{aligned}
  &\rho^* &&\leq ~8\cdot d\cdot \delta \cdot \OPT+1\\	
  &\delta \sum_{j=1}^{k } \OPT(S_{j-1},v) &&\leq~ (1+\ln \beta )\OPT +|T_{j_1}|\cdot \left( 1-\frac{1}{\sqrt{\beta}} -\frac{1}{2} \ln \beta \right)\\
  &&&~~~~~~~~~~+60\cdot d^2  \beta  \delta \cdot \OPT+\delta^{-3} K\cdot \beta\cdot \kappa(\delta)\\
  &|T_{j_1} | &&\geq \left(1-\beta^{-\frac{1}{2d}}\right)^d\cdot m -4\cdot d\cdot \delta \cdot \OPT
  \end{aligned}
\end{equation}
hold.
By \Cref{lem:first_fit_bound,lem:Tj_based_analysis,lem:Tj_bound}, these inequalities hold with probability at least 
\begin{equation*}
  \begin{aligned}
  &1-\delta^{-2}\cdot  \exp\left(-\delta^7 \cdot \OPT\right)-K\cdot\kappa(\delta) \cdot\delta^{-4}\cdot  \exp\left(-\frac{\delta^{50}}{\psi^2 \cdot \kappa^2(\delta)}\cdot \OPT\right)- \delta^{-2 }\exp\left(-{\delta^{50}}\cdot \OPT\right) \\
  \geq~& 1- K\cdot \delta^{-5}\cdot \kappa(\delta)\cdot \exp\left(-\frac{\delta^{50}}{\psi^2\cdot \kappa^2(\delta)}\cdot \OPT\right) \enspace .
  \end{aligned}
\end{equation*}

Thus, if $\OPT$ is sufficiently large then \eqref{eq:ddim_assump} occurs with probability at least $\frac{1}{2}$.
Furthermore, in this case it also holds that 
\begin{equation}
\label{eq:ddim_sum_bound}
  \begin{aligned}
    \delta \sum_{j=1}^{k }& \OPT(S_{j-1},v) \\
    &\leq~ (1+\ln \beta )\OPT +|T_{j_1}|\cdot \left( 1-\frac{1}{\sqrt{\beta}} -\frac{1}{2} \ln \beta \right)+60\cdot d^2  \beta  \delta \cdot \OPT+\delta^{-3} K\cdot \beta\cdot \kappa(\delta)\\
    &\leq~ (1+\ln \beta )\OPT +\left(\left(1-\beta^{-\frac{1}{2d}}\right)^d\cdot m -4\cdot d\cdot \delta \cdot \OPT \right) \cdot \left( 1-\frac{1}{\sqrt{\beta}} -\frac{1}{2} \ln \beta \right)\\
    &~~~~~~~~~~~+60\cdot d^2  \beta  \delta \cdot \OPT+\delta^{-3} K\cdot \beta\cdot \kappa(\delta)\\
    &\leq~ (1+\ln \beta )\OPT -\chi(\beta,d) \cdot m  + 90\cdot d^2\cdot \delta \cdot \beta \cdot \OPT+ \delta^{-3} K\cdot \beta\cdot \kappa(\delta) \\
    &\leq~ \left(1+\ln \beta  -\chi(\beta,d) +90\cdot d^2\cdot \delta \cdot \beta\right )\OPT + \delta^{-3} K\cdot \beta\cdot \kappa(\delta) \enspace .
   \end{aligned}
\end{equation}
The second and third inequalities hold as $-\beta \leq \left(1-\frac{1}{\sqrt {\beta} }  -\frac{1}{2} \ln \beta\right) \leq 0$.
The third inequality uses the definition of $\chi(\beta,d)$ as given in the statement of \Cref{thm:better_than_randa}.
The forth inequality holds as $\chi(\beta,d) \geq 0$ and $m\geq \OPT$.

By \eqref{eq:sum_rho}, \eqref{eq:ddim_assump} and \eqref{eq:ddim_sum_bound}, the size of the solution  returned  by \Cref{alg:basic_round_and_round} is
\begin{equation*}
  \begin{aligned}
    \sum_{j=1}^{k}&\rho_j +\rho^* \leq~ \delta^{-2} + (1+4\delta) \cdot \delta \sum_{j=1}^{k} \OPT(S_{j-1},v) + 8\cdot d\cdot \delta \cdot \OPT+1\\
    &\leq ~\delta^{-2}+(1+4\delta)\cdot \left( \left(1+\ln \beta  -\chi(\beta,d) +90\cdot d^2\cdot \delta \cdot \beta\right )\OPT + \delta^{-3} K\cdot \beta\cdot \kappa(\delta) \right) + 8\cdot d\cdot \delta \cdot \OPT\\
    &\leq    \left(1+\ln \beta  -\chi(\beta,d) +200\cdot d^2\cdot \delta \cdot \beta\right )\OPT + \delta^{-5} K\cdot \beta\cdot \kappa(\delta) \enspace .
  \end{aligned}
\end{equation*}
That is, the algorithm is a randomized asymptotic $(1+\ln \beta -\chi(\beta,d) + 200 \cdot d^2  \cdot \delta \cdot \beta )$-approximation algorithm for $d$VBP. \qed

\subsection{The Weak Structural Property}
\label{sec:weak_structural}
In this section we prove \Cref{lem:weak_structural}.
The lemma relies on an implicit rounding of the large items volumes to multiplicities of $\frac{\delta^2}{2d}$.
While the volume of the items is rounded up, the slack of the configurations $B_1,\ldots, B_s$ ensures that these remain feasible configurations with respect to the rounded weight.
Subsequently, the proof of the lemma views items of the same rounded volume as interchangeable, which is key in attaing the bound on $\OPT(Q,v)$ as stated in \Cref{lem:weak_structural}. 

The lemma is utilizes some ideas from Bansal et al.~\cite{BansalEK2016}.
However, the rounding procedure in their work only requires each of the configurations $B_1,\ldots, B_{s}$ to have slack in $d-1$ dimension, and combines a shifting argument as part of the rounding.
As mentioned in the introduction (see also \Cref{sec:flaw}), the approach taken by Bansal et al.~\cite{BansalEK2016} has a flaw in the analysis, and hence cannot be used.
Requiring the configurations to have $\delta$-full slack is a simple way to work around the flaw. 
When possible, the notations used in both lemmas are kept similar. 

\begin{proof}[Proof of \Cref{lem:weak_structural}]
  We assume that $d\in \mathbb{N}_{> 0 }$ and $\delta \in (0,0.1)$.
  Throughout the proof, consider an instance $(I,v)$ of $d$VBP.
  Furthermore, we assume $\delta \leq \frac{1}{d^2}$ and $\delta^{-1}\in \mathbb{N}$.
  As in the statement of \Cref{lem:weak_structural},  let $B_1,\ldots, B_s\in \cC$ be collection of configurations with $\delta$-full slack, and define $R= B_1 \cup B_2 \cup \ldots \cup B_s$. 

  Recall that $L$ is the set of large items of the instance $(I,v)$.
  Set $h=\delta^{-2}$ and $\cG = \{1,\ldots, 2\cdot d\cdot h\}^d$.
  For every $\ba \in \cG$ define
  \begin{equation}
  \label{eq:weak_classes}
    I_\ba = \left\{ i\in L \cap R ~\middle|~\forall r\in [d]:~~ ~ \frac{\delta^2}{2\cdot d }\cdot (\ba_r -1) ~<~ v_r(i)~\leq~ \frac{\delta^2}{2\cdot d} \cdot \ba_r~ \right\}\enspace.
  \end{equation}
  Also,  define the {\em rounded volume} of $\ba\in \cG$ by
  \begin{equation}
  \label{eq:weak_rounded_volume}
	\tv(\ba) = \frac{\delta^2}{2\cdot d}\cdot \ba\enspace.
  \end{equation}
  Implicitly, we round the volume of all items in $I_{\ba}$ to $\tv(\ba)$.  Since $v(i)\in (0,1]^d$ for every $i\in I$, it follows that $\bigcup_{\ba\in\cG} I_{\ba} = R\cap L$. 

  The {\em type} of a configuration $C\in \cC$, denoted $\type(C)$, is the vector $\bt\in \mathbb{N}^\cG $ defined by $\bt_\ba  = \abs{ I_\ba \cap C}$ for every $\ba\in \cG$. That is, $\bt_\ba$ is the number of items from $I_{\ba}$ in the configuration $C\in \cC$.
  Define $\cT= \{\type(B_\ell)~|~\ell = 1,2,\ldots, s\}$ to be the set of all types of configurations in $B_1,\ldots, B_s$.
  As a configuration $C$ may contain up to $d\cdot \delta^{-1}$ large items, it follows that
  \begin{equation}
	\label{eq:weak_type_bound}
	\begin{aligned}
		\abs{\cT} \leq~& \left(d\cdot \delta^{-1}\right)^{\abs{\cG}}\\ 
		\leq ~&\left(d \cdot\delta^{-1}\right)^{\left( 2\cdot d  \cdot h\right)^d}\\
        =~& \exp \left( \left( 2\cdot d \cdot \delta^{-2}\right)^d  \ln \left(d\cdot \delta^{-1}\right) \right) \\
        \leq~& \exp \left( \delta^{-4\cdot \delta^{-1}} \ln (\delta^{-2})\right)\\
        \leq~&\exp\left( \delta^{-5\cdot \delta^{-1}}\right)\\
        \leq ~ &\frac{\kappa(\delta)}{3},
        \leq \exp\left( \left(d^2\cdot \delta^{-2}\right)^{d+1}\right).
    \end{aligned}s
  \end{equation}
  where the third inequality holds as $d^2 \leq \delta^{-1}$.
  Similarly to \eqref{eq:weak_rounded_volume}, we define the {\em rounded volume} of $\bt \in \cT$ by 
  \begin{equation}
  \label{eq:weak_rounded_volume_type}
	\tv(\bt) = \sum_{\ba\in \cG} \bt_\ba\cdot \tv(\ba)\enspace.
  \end{equation}

  For every $\bt\in \cT$ define
  \begin{equation*}
    L_\bt = \bigcup_{\ell \in [s] \textnormal{ s.t. } \type(B_{\ell})= \bt} B_{\ell }\cap L \textnormal{~~~~~~and ~~~~~~}S_\bt = \bigcup_{\ell \in [s] \textnormal{ s.t. } \type(B_{\ell})= \bt} B_{\ell }\setminus  L,
  \end{equation*}
  as the set of large items  and the set of small items in configuration of type $\bt$ among $B_1,\ldots, B_s$, respectively.
  Also, for every $r=1,\ldots, d$ define $\bv^r\in [0,1]^I$ by $\bv^r_i = v_r(i)$ for all $i\in I$. That is, $\bv^r$ is a representation of the volume of the items in the $r$-th dimension as a vector. 
  For every $\bt\in \cT$ define 
  \begin{equation*}
    \cS_{\lr, \bt} =\left\{ \one_{I_\ba \cap L_{\bt}}~\middle|~\ba\in \cG\right\} \textnormal{~~~~~and~~~~} 
 	\cS_{\sm, \bt} = \left\{ \one_{S_{\bt}}\wedge \bv^r ~\middle|~ r=1,2,\ldots, d \right\}.
  \end{equation*}
  Finally, define $\cS =\bigcup_{\bt\in \cT} \left(  \cS_{\lr, \bt} \cup \cS_{\sm,\bt}\right)$. 

  \begin{claim}
    It holds $|\cS| \leq \kappa\left( \delta\right)$.
  \end{claim}
  \begin{claimproof}
	By a simple counting argument, 
	\begin{equation*}
	\begin{aligned}\abs{\cS} \leq~& \sum_{\bt \in \cT} \left(\abs{\cS_{\lr, \bt} } +\abs{ \cS_{\sm,\bt}} \right)\\
		 \leq~& \abs{\cT} \cdot \left(\abs{\cG}+d\right) \\
		 \leq~&\exp\left( \delta^{-5\delta^{-1}} \right) \left((d\cdot h )^d  +d \right)\\
		 \leq~& \exp \left( \delta ^{-6\cdot \delta^{-1}}\right)\\
		 \leq~&\exp\left(\exp \left( -5\cdot \delta^{-1} \cdot \ln (\delta)\right)\right)\\
		 \leq~&{\kappa(\delta)}\enspace.
		\end{aligned}
	\end{equation*}
	The third inequality uses \eqref{eq:weak_type_bound} and the forth inequality uses $d\leq \delta^{-1}$.
  \end{claimproof}

  We are left to show the constructed structure $\cS$ satisfies the condition in \Cref{lem:weak_structural}. The following claims provide some basic properties which will assist us in achieving this goal.
  \begin{claim}
  \label{claim:weakly_volume_upper_bound}
    Let $\bt\in \cT$ and let $C\in \cC$ be such that $C\subseteq L\cap R$ and $\type(C)\leq \bt$.
    That is, for all $\ba \in \cG$ it holds that $\type_\ba(C)\leq \bt_\ba$.
	Then $v(C)\leq \tv(\bt)$.
  \end{claim}
  \begin{claimproof}
	For  $r= 1,\ldots, d$ it holds that 
	\begin{equation*}
	\begin{aligned}
		v_r(C) &=~ \sum_{\ba\in \cG}  \sum_{i\in C\cap I_\ba} v_r(i)\\
		&\leq~\sum_{\ba \in \cG} \sum_{i\in C\cap I_{\ba} } \tv_r(\ba) \\
		&= ~\sum_{\ba \in \cG} \type_{\ba} (C) \cdot \tv_r(\ba) \\
		&\leq~\sum_{\ba \in \cG} \bt_\ba \cdot \tv_r(\ba) ~=\bv_r(\bt)\enspace.
	\end{aligned}
	\end{equation*}
	The first inequality holds since $v_r(i)\leq \tv_r(\ba)$ for every $i\in I_{\ba}$ by \eqref{eq:weak_classes} and \eqref{eq:weak_rounded_volume}.
	The second inequality follows from the assumptions of the claim. The last equality follows from the definition of $\tv(\bt)$ in \eqref{eq:weak_rounded_volume_type}.
  \end{claimproof}

  \begin{claim}
  \label{claim:weakly_volume_lower_bound}
	Let $\bt \in \cT$ and $C\in \cC$ such that $C\subseteq R\cap L$ and $\type(C) =\bt$.
	Then $v_r(C)\geq \tv_r(C)-\frac{\delta}{2}$ for $r=1,\ldots ,d$. 
  \end{claim}
  \begin{claimproof}
    For  $r=1,\ldots,  d$ it holds that 
    \begin{equation*}
      \begin{aligned}
        v_r(C) &=~ \sum_{\ba\in \cG } \sum_{\ba \in C\cap I_{\ba}} v_r(i)\\
        &\geq~\sum_{\ba\in \cG } \sum_{\ba \in C\cap I_{\ba}}\left(  \tv_r(\ba) - \frac{\delta^2}{2\cdot d}\right)\\
        & =~ \tv_r(\bt) - \abs{C\cap L} \cdot \frac{\delta^2}{2d} \\
        &\geq~ \tv_r(\bt) - \frac{\delta}{2}\enspace.
      \end{aligned}
    \end{equation*}
    The first inequality follows from \eqref{eq:weak_classes} and \eqref{eq:weak_rounded_volume}.
    The last inequality holds, as $\abs{C\cap \L}\leq d\cdot \delta^{-1}$.
  \end{claimproof}
  \begin{claim}
  \label{claim:weaky_small_items_volume}
    Let $\ell \in \{1,\hdots,s\}$ and $\bt = \type(B_{\ell})$.
	Then $v_r(B_{\ell}\setminus L )\leq 1- \tv_r(\bt) -\frac{\delta}{2}$ for $r=1,\ldots, d$. 
  \end{claim}
  \begin{claimproof}	
  For $r=1,\ldots ,d$ we have
  \begin{equation*}
    v_r(B_{\ell}\setminus L) = v_r(B_{\ell}) - v_r(B_{\ell }\cap L) \leq 1-\delta - \left( \tv_r(\bt)-\frac{\delta}{2}\right) =1-\tv_r(\bt)-\frac{\delta}{2}\enspace.
  \end{equation*}
  The inequality holds since $B_{\ell}$ has $\delta$-full slack  and by \Cref{claim:weakly_volume_lower_bound}. 
  \end{claimproof}

  The following is an immediate consequence of \Cref{claim:weaky_small_items_volume}. 
  \begin{cor}
  \label{cor:weak_type_volume} 
    For all $\bt\in \cT$ and $r\in \{1,2,\ldots, d\}$ it holds that $\tv_r(\bt) \leq 1-\frac{\delta}{2}$.
  \end{cor}

  Let $Q\subseteq R$ and $\gamma \in (0,1)$ be such that 
  \begin{equation}
  \label{eq:Q_cond}
    \forall \bu \in \cS:~~~~
	\one_{Q} \cdot \bu\leq \gamma \cdot \one_{R} \cdot \bu + \frac{\delta ^{20}}{\kappa(\delta)}\cdot \OPT(I,v) \cdot  \tol(\bu) \enspace .
  \end{equation}
  To complete the proof, we need to show that $\OPT(Q,v)\leq \gamma (1+d\cdot \delta)\cdot s + \delta^{10}\cdot \OPT + \kappa(\delta)$. Towards this end, we will construct a separate packing of $Q\cap (L_\bt \cup S_{\bt})$ for every $\bt\in \cT$.

  Define the {\em prevalence} of type $\bt \in \cT$ by $p_\bt = \abs{ \{\ell \in \{1,\hdots,s\}~|~\type(B_{\ell})=\bt\}}$.
  That is, $p_{\bt}$ is the number of configuration among $B_1,\ldots ,B_s$ of type $\bt$.
  For every $\bt \in \cT$, define
  \begin{equation}
  \label{eq:weak_eta_def}
    \eta_\bt = \ceil{ \gamma\cdot  p_\bt +\frac{\delta^{15}}{\kappa(\delta)} \cdot \OPT}\enspace.
  \end{equation}
  We will show that $\OPT(Q\cap (L_\bt \cup S_{\bt}) \setminus X_{\bt},v)\leq \eta_\bt$ where $X_\bt$ is a set that  satisfies $\OPT(X_t,v)\leq \delta \cdot \eta_{\bt}$.

  \begin{claim}
  \label{claim:weak_pack_large_items}
    For every $\bt\in \cT$ there exists $D^{\bt}_1,\ldots, D^{\bt}_{\eta_{\bt}}\subseteq I$ such that $\bigcup_{\ell=1}^{\eta_\bt} D^{\bt}_{\ell} = Q\cap L_{\bt}$  and $v(D^{\bt}_{\ell}) \leq \tv(\bt)$ for $\ell =1,\ldots, \eta_{\bt}$.
  \end{claim}
  By \Cref{claim:weak_pack_large_items} we can pack the items in $Q\cap L_{\bt}$ into $\eta_{\bt}$ configuration with volume at most $ \tv(\bt)$.
  The unused volume of $1-\tv_r(\bt)$ in each coordinate $r=1,\ldots, d$ will be used to pack the set  small items~$Q\cap S_{\bt}$. 
  \begin{claimproof}[Proof of \Cref{claim:weak_pack_large_items}]
    Let $\cG_{\bt} = \{\ba \in \cG~|~\bt_{\ba}\neq 0\}$.
    For every $\ba \in \cG\setminus \cG_{\bt}$, we have $\bt_{\ba}=0$, and therefore $L_{\bt}\cap I_\ba = \emptyset$ (configurations of type $\bt$ do not contain items from $I_\ba$, and $L_{\bt}$ is a set of items in configurations of type $\bt$). Thus $Q\cap L_{\bt} \cap I_{\ba} = \emptyset$ for all $\ba \in \cG\setminus \cG_{\bt}$.
	
    For all $\ba \in \cG_{\bt}$ it holds that $\one_{I_\ba \cap L_{\bt}} \in \cS_{\lr, \bt} \subseteq \cS$.
    Thus, by \eqref{eq:Q_cond} we have
    \begin{equation}
    \label{eq:weak_large_items_bound_first}
      \abs{Q\cap I_{\ba} \cap L_{\bt} } ~=~ \one_{Q} \cdot \one_{I_\ba \cap L_{\bt}} 
  	 ~\leq~ \gamma \cdot \one_{R}\cdot \one_{I_\ba \cap L_{\bt}} +\frac{\delta^{20}}{\kappa(\delta)}\cdot \OPT(I,v) \cdot\tol\left( \one_{I_\ba \cap L_{\bt}}\right)\enspace.
	\end{equation}
    Furthermore, for all $C\in \cC$ it holds that $\sum_{i\in C}  \left(\one_{I_\ba \cap L_{\bt}} \right)_i \leq \sum_{i\in C} \one_{i\in L} \leq d\cdot \delta^{-1} \leq \delta^{-2}$, thus $\tol(\one_{I_\ba \cap L_{\bt}})\leq \delta^{-2}$. By plugging the last inequality into \eqref{eq:weak_large_items_bound_first} we obtain,
    \begin{equation}
    \label{eq:weak_large_items_bound_second}
	  \abs{Q\cap I_{\ba} \cap L_{\bt} }
	  ~\leq~ \gamma \cdot \one_{R}\cdot \one_{I_\ba \cap L_{\bt}} +\frac{\delta^{18}}{\kappa(\delta)}\cdot \OPT(I,v) 
	  ~\leq~ \gamma \cdot \abs{R\cap I_{\ba} \cap L_{\bt} } + \frac{\delta^{18}}{\kappa(\delta)}\cdot \OPT(I,v)\enspace.
    \end{equation}
	
    Observe that
    \begin{equation}
    \label{eq:weak_class_bound}
	  \abs{R\cap I_{\ba} \cap L_{\bt}} ~=~ \sum_{\ell\in [s]\textnormal{ s.t. } \type(B_{\ell})=\bt} \abs{B_{\ell}\cap I_{\ba} }  =  \sum_{\ell\in [s]\textnormal{ s.t. } \type(B_{\ell})=\bt} \bt_\ba = p_{\bt}\cdot \bt_{\ba}\enspace.
    \end{equation}
    By \eqref{eq:weak_large_items_bound_second} and \eqref{eq:weak_class_bound}, it holds that
    \begin{equation*}
		\abs{Q\cap I_{\ba} \cap L_{\bt} }
	~\leq~ \gamma \cdot p_{\bt}\cdot \bt_{\ba}+ \frac{\delta^{18}}{\kappa(\delta)}\cdot \OPT(I,v) \leq \bt_{\ba} \cdot \eta_{\bt}.
    \end{equation*}
    Therefore, for every $\ba \in \cG_{\bt}$ we can partition $Q\cap I_{\ba} \cap L_{\bt}$ into $\eta_{\bt}$ sets $D^{\bt}_{1,\ba},\ldots ,D^{\bt}_{\eta_{\bt},\ba}$ such that $\abs{D^{\by}_{\ell,\ba} } \leq \bt_{\ba}$ (we allow sets in the partition to be empty).
	Define sets $D^{\bt}_{1},\ldots ,D^{\bt}_{\eta_{\bt}}$ by $D^{\bt}_{\ell} = \bigcup_{\ba\in\cG_{\bt}} D^{\bt}_{\ell,\ba}$ for all $\ell =1,2,\ldots, \eta_{\bt}$.
	It follows that 
    \begin{equation*}
    \bigcup_{\ell=1}^{\eta_{\bt}} D^{\bt}_{\ell} 
    ~=~ \bigcup_{\ell=1}^{\eta_{\bt}}  \bigcup_{\ba\in\cG_{\bt}} D^{\bt}_{\ell,\ba} 
    ~=~  \bigcup_{\ba\in\cG_{\bt}}\left(  Q\cap I_{\ba}\cap L_{\bt} \right)= Q\cap L_{\bt} \enspace.
    \end{equation*}
 
    For all $\ba\in \cG\setminus \cG_\bt$  and $\ell =1,\ldots, \eta_{\bt}$ it holds that $\type_{\ba} (D^{\bt}_{\ell} ) =\abs{D^{\bt}_{\ell} \cap I_{\ba} } = 0=\bt_\ba$. Furthermore, for all $\ba\in  \cG_\bt$  and $\ell =1,\ldots, \eta_{\bt}$  it holds that t $\type_{\ba} (D^{\bt}_{\ell} ) =\abs{D^{\bt}_{\ell} \cap I_{\ba} } = \abs{D^{\bt}_{\ell,\ba}}\leq\bt_\ba$. Thus, $\type(D^{\bt}_{\ell})\leq \bt$ for all $\ell =1,\ldots, \eta_{\bt}$.
    By \Cref{claim:weakly_volume_upper_bound}, it follows that $v(D^{\bt}_{\ell}) \leq \tv(\bt)$ for all $\ell=1,\ldots, \eta_{\bt}$.
  \end{claimproof}

  While \Cref{claim:weak_pack_large_items} handles the large items in $Q$, the next claim deals with the small items in $Q$.
  \begin{claim}
  \label{claim:weak_pack_small_items}
    For all $\bt \in \cT$  there exists $F^{\bt}_1,\ldots,F^\bt_{\eta_{\bt}}\subseteq I$ and $X_\bt\subseteq I$ such that 
    \begin{itemize}
      \item	$\bigcup_{\ell=1}^{\eta_{\bt}} F^{\bt}_{\ell} = \left(Q\cap S_{\bt}\right) \setminus X_{\bt}$,
	  \item $\OPT(X_{\bt},v)\leq \delta \cdot d \cdot \eta_{\bt}+1$,
      \item and $v_r(F^{\bt}_{\ell})\leq 1-\tv_r(\bt) $ for all $\ell=1,\ldots, \eta_{\bt}$ and $r=1,\ldots, d$.
    \end{itemize}
  \end{claim}
  \begin{claimproof}
    For all $r=1,\ldots, d$ it holds that $\one_{S_{\bt}}\wedge \bv^r\in \cS$.
	Thus, by \eqref{eq:Q_cond} it holds that 
    \begin{equation}
	\label{eq:weak_small_items_first}
      v_r(Q\cap S_{\bt} ) = \one_{Q}\cdot \left(\one_{S_{\bt}}\wedge \bv^r\right)~ \leq ~\gamma \cdot \one_{R} \cdot \left(\one_{S_{\bt}}\wedge \bv^r\right) + \frac{\delta^{20}}{\kappa(\delta)} \cdot \OPT \cdot \tol(\one_{S_{\bt}}\wedge \bv^r)\enspace.
    \end{equation}
    For all $C\in \cC$  it holds that  $\sum_{i\in C} \left(\one_{S_{\bt}}\wedge \bv^r\right)_i \leq \sum_{i\in C} v_r(i) \leq 1$, hence $\tol(\one_{S_{\bt}}\wedge \bv^r)\leq 1$.
	Thus, we can rewrite \eqref{eq:weak_small_items_first} as 
    \begin{equation}
	\label{eq:weak_small_items_second}
	  v_r(Q\cap S_{\bt} )  ~\leq ~\gamma \cdot \one_{R} \cdot \left(\one_{S_{\bt}} \wedge \bv^r\right) + \frac{\delta^{20}}{\kappa(\delta)} \cdot \OPT ~=~ \gamma \cdot v_r(R\cap S_{\bt})+\frac{\delta^{20}}{\kappa(\delta)} \cdot \OPT\enspace.
    \end{equation}
    By the definition of $S_{\bt}$ we also have
    \begin{equation}
	\label{eq:weak_st_volume}
      v_r(R\cap S_{\bt}) = \sum_{\ell \in [s]\textnormal{ s.t.} \type(B_{\ell})=\bt} v_r(B_{\ell} \setminus L) \leq  \sum_{\ell \in [s]\textnormal{ s.t.} \type(B_{\ell})=\bt}  \left( 1- \tv_r(\bt) -\frac{\delta}{2}\right) \leq  p_{\bt} \cdot  \left( 1- \tv_r(\bt) \right),
    \end{equation}
    where the first inequality follows from \Cref{claim:weaky_small_items_volume}.
	By \eqref{eq:weak_small_items_second} and \eqref{eq:weak_st_volume} we have
    \begin{equation}
    \label{eq:weak_small_items_third}
      v_r(Q\cap S_{\bt})~\leq~ \gamma  p_{\bt}   \left( 1- \tv_r(\bt) \right) + \frac{\delta^{20} }{\kappa(\delta)} \OPT ~\leq~ \left( 1- \tv_r(\bt) \right)  \cdot \left( \gamma  p_{\bt} + \frac{\delta^{15}}{\kappa(\delta)} \OPT\right) \leq (1-\tv_r(\bt))\cdot \eta_{\bt},
    \end{equation}
    where the second inequality follows from \Cref{cor:weak_type_volume}.

    Our construction utilizes integrality properties of the polytope $P$ defined by
    \begin{equation}
    \label{eq:weak_P_def}
      P= \left\{ \bmu \in [0,1]^{ Q\cap S_{\bt }\times [\eta_{\bt}]}~\middle|~
	  \begin{aligned}
        &\sum_{\ell =1}^{\eta_{\bt}} 	 \bmu_{i, \ell} = 1 & &\forall i\in Q\cap S_{\bt}\\
        &\sum_{i\in Q\cap S_\bt} v_r(i)\cdot \bmu_{i,\ell} \leq 1-\tv_r(\bt) &~~~~~& \forall r\in [d],~\ell\in [\eta_{\bt}],
      \end{aligned}\right\}\enspace.
    \end{equation}
    That is, an entry in $P$ is a vector with entries of the form $\bmu_{i,\ell}$, where $i\in Q\cap S_{\bt}$ and $\ell\in \{1,\ldots, \eta_{\bt}\}$.
    The entry $\bmu_{i,\ell}$ can be interpreted as the fractional assignment of the item $i$ to the $\ell$-th bin.
	The first constraint in \eqref{eq:weak_P_def} ensures all the items are fully assigned, and the second constraint enforces an upper bound on the total volume of items assigned to a specific bin in each coordinate.
    It is well-known (see, e.g., \cite{BansalEK2016}) that a vertex of $P$ contains at most $d\cdot \eta_{\bt}$ fractional entries.
	Formally, if $\bmu^* \in P$ is a vertex of $P$ then $\abs{\left\{(i,\ell) \in Q\cap S_t\times \{1,\hdots,\eta_{\bt}\}\mid \bmu^*_{i,\ell}\in (0,1)\right\}}\leq d\cdot \eta_{\bt}$.

    In order to exploit the above-mentioned property of $P$, we first need to show $P\neq \emptyset$.
    Define $\bx \in  [0,1]^{ Q\cap S_{\bt }\times \{1,\hdots,\eta_{\bt}\}}$ by $\bx_{i,\ell} = \frac{1}{\eta_{\bt}}$ for all $i\in Q\cap S_{\bt}$ and $\ell =1,\ldots, \eta_{\bt}$.
	For all $i\in Q\cap S_{\bt}$ it holds that
    \begin{equation}
	\label{eq:weak_poly_first}
      \sum_{\ell=1}^{\eta_{\bt}} \bx_{i,\ell} = \sum_{\ell=1}^{\eta_{\bt}}  \frac{1}{\eta_{\bt}} = 1\enspace.
    \end{equation}
    Furthermore, for every $\ell =1,\ldots ,\eta_{\bt}$ and $r=1,\ldots,d$ we have
    \begin{equation}
	\label{eq:weak_poly_second}
      \begin{aligned}
        \sum_{i\in Q\cap S_\bt} v_r(i)\cdot \bx_{i,\ell} = \sum_{i\in Q\cap S_\bt} v_r(i)\cdot \frac{1}{\eta_{\bt}} =\frac{1}{\eta_{\bt}} \cdot v_r(Q\cap S_{\bt}) \leq 1-\tv_r(\bt),
      \end{aligned}
    \end{equation}
    where the last inequality follows from \eqref{eq:weak_small_items_third}.
    By \eqref{eq:weak_poly_first} and \eqref{eq:weak_poly_second} we have $\bx \in P$, and thus $P\neq \emptyset$. 

    Therefore, there exists a vertex $\bmu^*$ of the polytope $P$ and it holds that
	\begin{equation*}
	  \abs{\left\{(i,\ell) \in Q\cap S_t\times \{1,\hdots,\eta_{\bt}\}\mid\bmu^*_{i,\ell}\in (0,1)\right\}}\leq d\cdot \eta_{\bt}\enspace.
	\end{equation*}
    Define $X_{\bt} = \left\{i\in Q\cap S_{\bt}~\middle| \exists \ell \in \{1,\hdots,s\}:~\bmu^*_{i,\ell}\in (0,1)\right\}$.
	It thus holds that $\abs{X_{\bt}}\leq d\cdot \eta_\bt$.
	Since all items in~$X_{\bt}$ are small, it holds that every subset of  $\delta^{-1}$ items of $X_{\bt}$ form a configuration, thus $\OPT(X_{\bt}, v) \leq \delta \abs{X_{\bt} }+1 \leq  \delta \cdot d \cdot \eta_{\bt}+1$. 

    For $\ell = 1,\hdots,\eta_{\bt}$ define $F^{\bt}_{\ell } = \left\{i\in Q\cap S_{\bt}~\middle|~\bmu^*_{i,\ell} = 1 \right\}$.
	As $\bmu^*\in P$ \eqref{eq:weak_P_def} it holds follows that $v_r(F^{\bt}_\ell) \leq \sum_{i\in Q\cap S_{\bt} } v_r(i)\cdot \bmu^*_{i,\ell} \leq 1-\tv_r(\bt)$ for all $r=1,\ldots, d$.
    Furthermore,
	\begin{equation*}
      \bigcup_{\ell=1}^{\eta_{\bt}} F^{\bt}_{\ell} = \{i\in Q\cap S_{\bt}\mid\forall \ell = 1,\hdots,\eta_{\bt}:~\bmu^*_{i,\ell}\in\{0,1\}\}  = \left(Q\cap S_{\bt}\right)\setminus X_{\bt},
	\end{equation*}
    which completes the proof of the claim.
  \end{claimproof}

  For every $\bt$ let $D^{\bt}_1,\ldots, D^{\bt},_{\eta_\bt}$ be the sets from \Cref{claim:weak_pack_large_items} and let $X_\bt$ and $F^{\bt}_1,\ldots, F^{\bt}_{\eta_{\bt}}$ be the sets from \Cref{claim:weak_pack_small_items}.
  It follows that $v_r(D^{\bt}_{\ell}\cup F^{\bt}_{\ell} )\leq \tv_r(\bt)+1-\tv_r(\bt)=1$ for all $\ell =1,\ldots, \eta_{\bt}$ and  $r=1,\ldots, d$. Thus $D^{\bt}_{\ell}\cup F^{\bt}_{\ell}\in \cC$ for all $\ell =1,\ldots, \eta_{\bt}$.
  It also holds that $\bigcup_{\ell=1}^{\eta_{\bt}} \left(D^{\bt}_{\ell}\cup F^{\bt}_{\ell}\right) = \left( Q\cap (L_{\bt}\cup S_{\bt})\right)\setminus X_{\bt}$.
  Therefore,
  \begin{equation*}
    \OPT\left( Q\cap (L_{\bt}\cup S_{\bt}),v\right) ~\leq~ \OPT\left( \left( Q\cap (L_{\bt}\cup S_{\bt}) \right)\setminus X_{\bt},v\right)  + \OPT(X_{\bt},v) ~\leq~ \eta_{\bt} + \delta \cdot d \cdot \eta_{\bt} +1,
  \end{equation*}
  and thus,
  \begin{equation*}
    \begin{aligned}
	  \OPT(Q,v) \leq~& \sum_{\bt\in \cT} \OPT\left( Q\cap (L_{\bt}\cup S_{\bt}),v\right)  \\
	  \leq ~& \sum_{\bt\in \cT}\left( \eta_{\bt}+\delta \cdot d \cdot \eta_{\bt}+1 \right) \\
	  \leq~& \abs{\cT} + (1+\delta \cdot d )\sum_{\bt\in \cT} \eta_{\bt}\\
	  =~&\abs{\cT} + (1+\delta \cdot d) \sum_{\bt\in \cT} \ceil{ \gamma\cdot  p_\bt +\frac{\delta^{15}}{\kappa(\delta)} \cdot \OPT}\\
	  \leq~&3\cdot \abs{\cT} + (1+\delta \cdot d) \sum_{\bt\in \cT} \left( \gamma\cdot  p_\bt +\frac{\delta^{15}}{\kappa(\delta)} \cdot \OPT\right)\\
	  =~&3\cdot \abs{\cT} + \gamma\cdot (1+\delta \cdot d) \sum_{\bt\in \cT}  p_\bt + (1+\delta\cdot d)\cdot \abs{\cT}\cdot \frac{\delta^{15}}{\kappa(\delta)}\cdot \OPT\\
	  \leq~&\kappa(\delta)+ \gamma\cdot  (1+\delta\cdot d)\cdot s +\delta^{10}\cdot\OPT 
	  \enspace.
  \end{aligned}
  \end{equation*}
  The first equality follows from \eqref{eq:weak_eta_def}.
  The last inequality uses \eqref{eq:weak_type_bound} and $d\leq \delta^{-1}$.
\end{proof}

\section{Asymptotic $\left(\frac{4}{3}+\eps\right)$ approximation for $2$VBP}
\label{sec:twodim}

\label{sec:tmp_prelim}

In this section we prove \Cref{lem:weakly_asym}. That is, we show that \Cref{alg:match_and_round} is a randomized asymptotic $\left( \frac{4}{3}+\eps\right)$-approximation algorithm for $2$VBP. The analysis of the algorithm utilizes a variant of the Configuration-LP \eqref{eq:config_LP} in which each item  $i\in I$ has a {\em demand}  $\bd_i\in [0,1]$. That is, given a $2$VBP instance and  for every 
 {\em demand} vector $\bd\in [0,1]^I$ define
\begin{equation}
	\label{eq:demand_config_LP}
	\begin{aligned}
		\demandLP(\bar{d}):~~~  & \min && \sum_{C\in \cC} \bx_C,\\
		&\forall i\in I: &&\sum_{C\in \cC} \bx_C\cdot C(i)= \bar{d}_i,\\\ 
		&\forall C\in \cC:~~~&& \bx_C\geq 0\enspace.
	\end{aligned}
\end{equation}
Observe that for every $S\subseteq I$ it holds that $\LP(S)$ is identical to $\demandLP(\one_{S})$.
We use $\OPTf(\bd)$ to denote the value of an optimal solution for $\demandLP(\bd)$

We extend the definition of configuration to allow multiple occurrences of items.
Let $(I,v)$ be a 2VBP instance. 
A  multi-set over $I$  is a function $C:I\rightarrow \mathbb{N}$. For $i\in I$ we say that $i\in C$ if $C(i)>0$. 
A {\em multi-configuration} is a multi-set $C$ over $I$ such that $v(C)=\sum_{i\in I} C(i)\cdot v(i)\leq (1,1)$. 
We use $\cC^*$ to denote the set of all multi-configurations. 
We identify the set $C\subseteq I$ with the multi-set $C'$  in which $C'(i)=C(i)$.  

Given $\bx\in [0,1]^{\cC}$ ($\bx\in [0,1]^{\cC^*}$) the {\em coverage} of $\bx$ is the vector $\by\in [0,1]^{I}$ defined by $\by_i =\sum_{C\in \cC} \bx_C \cdot C(i)$ ($\by_i =\sum_{C\in \cC^*} \bx_C \cdot C(i)$) for every $i\in I$.
We say that $\by\in [0,1]^{I}$ is {\em small-items integral} if $\by_i\in \{0,1\}$ for any $i\in I\setminus L$.
Similarly, we say that $\bx\in [0,1]^{\cC}$ ($\bx\in [0,1]^{\cC^*}$)  is {\em small-items integral} if its coverage is small-items integral.

Recall that $\OPT(I,v)$ is the minimum solution size for the instance $(I,v)$.
Our analysis relies on the existence of ``linear structures''. 
\begin{defn}[Linear Structure]
\label{def:linear_structure}
  Let $\delta, K>0$.
  Let $(I,v)$ be a $\delta$-2VBP instance, let $\blam \in [0,1]^{\cC^*}$, and let~$\bw\in [0,1]^{I}$ be the coverage of $\blam$.
	A {\em $(\delta,K)$-linear structure} of $\blam$ is a subset $\cS\subseteq \mathbb{R}^I_{\geq 0}$ of size at most~$K$ which satisfies the following property.
	For any small-items integral vector $\bz \in [0,1]^I$  and $\beta \in \left[\delta^5,1\right]$ such that $\supp(\bz)\subseteq \supp(\bw)$ and
	\begin{equation}
		\label{eq:structural_prop}
		\bz \cdot \bu\leq \beta \cdot \bw \cdot \bu + \frac{1}{K^{10}} \cdot \OPT(I,v) \cdot  \tol(\bu),
	\end{equation}
	for all $\bu\in \cS$, it holds that $\OPT_f(\bz)\leq \beta\cdot  (1+10\delta)\cdot \|\blam\|+K +\delta^{10}\cdot  \OPT(I,v)$. 
\end{defn}
Observe that a linear structure has properties similar to a weak structure (\Cref{lem:weak_structural}).
Intuitively, a linear structure implies that if a demand vector $\bz$ satisfies a `small' number of constraints with respect to $\beta$ (where $K$ is a constant, as defined in \Cref{lem:structural}) then we obtain a decrease in $\OPT_f(\bz)$ by factor of $\beta$.
While linear structures do not necessarily exist for arbitrary vectors $\blam$, we show that such structures exist for vectors which only select configurations with {\em slack}.
We say that $C\in \cC^*$ {\em has $\delta$-slack in dimension $d\in \{1,2\}$} if $v_d(C)\leq 1-\delta$.
We say that $C\in \cC^*$ {\em has $\delta$-slack} if there is $d\in \{1,2\}$ such that $C$ has $\delta$-slack in dimension $d$.
Finally, we say that $\blam\in [0,1]^{\cC^*}$ is {\em with $\delta$-slack} if every configuration $C\in \supp(\blam)$ has $\delta$-slack.

\begin{lemma}[Structural Property]
	\label{lem:structural}
	Let $(I,v)$ be a  $\delta$-2VBP instance, where $\delta \in (0, 0.1)$, and $\delta^{-1}\in \mathbb{N}$. 
	There is a set $\cS^*\subseteq \mathbb{R}_{\geq 0}^{I}$ such that $|\cS^*|\leq \varphi(\delta)\cdot |L|^4$,  where  $\varphi(\delta) =\exp\left(\delta^{-20}\right)$, which satisfies the following property.
	For any small-items integral $\blam\in [0,1]^{\cC^*}$ with $\delta$-slack, there is a $(\delta,\varphi(\delta))$-linear structure $\cS$  of $\blam$ where for all $\bu\in \cS$: if $\supp(\bu)\cap L \neq \emptyset $ then $\bu\in \cS^*$.
\end{lemma}

The proof of the lemma (given in \Cref{sec:structure}) uses some of the structural features shown by Bansal et al.~\cite{BansalEK2016}, along with the recent concept of fractional grouping, adopted from Fairstein et al.~\cite{FairsteinKS2021}.
While the set $\cS^*$ does not limit the number of structures which may be generated by the lemma, it limits the set of vectors these structures may use.
This attribute is crucial for our analysis (specifically, in the proof of \Cref{lem:undercovered_simple_and_small}).

To show the existence of linear structure we often need to convert an arbitrary configuration to a vector $\blam$ with a slack.
To this end, we use the following definition and lemmas. 
\begin{defn}
	\label{def:relax}
	Given $C\in \cC$ and $\psi \geq 1$, we say that $\blam\in [0,1]^{\cC^*}$ is a {\em $\psi$-relaxation} of $C$ if the following conditions simultaneously hold:
	\begin{enumerate}
		\item $\blam$ is with $\delta$-slack,
		\item $\| \blam \|\leq \psi$,
		\item and $\sum_{C'\in \cC^*} \blam_{C'}\cdot  C'(i) =C(i)$ for every $i\in I$. 
	\end{enumerate}
\end{defn}

\begin{lemma}
	\label{lem:delta_relaxed}
	Let $\delta \in (0,0.1)$ be such that $\delta^{-1}\in \mathbb{N}$ and let $(I,v)$ be a $\delta$-$2$VBP instance. Then 
	for any $C\in \cC_0$, there is a $(1+4\delta)$-relaxation of $C$.
\end{lemma} 

\begin{lemma}
	\label{lem:auxilary_struct}
	Let $\delta \in (0,0.1)$ and let $(I,v)$ be a $\delta$-$2$VBP instance.  Then
	for any $h=2,\ldots,  2\delta^{-1}$ and $C\in \cC_h$ there is an $\frac{h}{h-1}$-relaxation of~$C$.
\end{lemma}

\begin{lemma}
	\label{lem:relax_small}
	Let $\delta \in (0,0.1)$,  let $(I,v)$ be a $\delta$-$2$VBP instance, and let $C\in \cC$ such that $v(C)\leq(\delta, \delta)$. Then  there is a $4\delta$-relaxation of~$C$.
\end{lemma}	

The proofs of \Cref{lem:delta_relaxed},~\Cref{lem:auxilary_struct}, and~\Cref{lem:relax_small} are given in \Cref{sec:alpha_relax}.
Some of the statements and techniques used in the proofs can be viewed as variants of \cite[Lemma 5.3]{BansalEK2016}. We proceed to the analysis of \Cref{alg:match_and_round} in \Cref{sec:new_analysis}. The PTAS for the Matching Configuration LP \eqref{eq:matching_LP} (\Cref{lem:matching_ptas}) is given in \Cref{sec:match_lp}.
\subsection{The Analysis of \textsf{Match\&Round}}
\label{sec:new_analysis}
Throughout this section, we fix a $\delta$-2VBP instance $(I,v)$ and $\delta \in (0,0.1)$ such that $\delta^{-1}\in \mathbb{N}$.
Thus, notations such as $\rho_j$, $S_j$ $C^j_\ell$, and $\cM$ refer to the corresponding variables in the execution of \Cref{alg:match_and_round} (and the call to \Cref{alg:basic_round_and_round} as part of its execution), with $(I,v)$ as its input and~$\delta$ as the parameter.
We also use $\varphi(\delta)=\exp(\delta^{-20})$ as in \Cref{lem:structural} and $\OPT=\OPT(I,v)$.
We commonly use $k= \ceil{\ln_{1-\delta}(\delta)}\leq \delta^{-2}$.

The core of the analysis is in \Cref{sec:main_rnr}, in which we derive a bound on the number of configurations sampled by \Cref{alg:basic_round_and_round}.
\Cref{sec:asymptotic} gives the proof of~\Cref{lem:weakly_asym}.  
The analysis involves the use of several concentration bounds whose proofs are simple yet technical.
To avoid diversion from the main flow of the analysis, we defer the proofs of the concentration bounds to \Cref{sec:rnr_deferred}.

\label{sec:prob_space}

We use the probabilistic space $(\Omega, \cF, \Pr)$ as defined \Cref{sec:prelim}.
Recall that \Cref{lem:first_fit_bound} provides an upper bound on $\rho^*$, the size of the solution returned by First-Fit in \Cref{round:first_fit} of \Cref{alg:basic_round_and_round}. Also,
observe that $\E\left[|\cM|\right] = (1-\delta^4 )\cdot  \bx^0\cdot \one_{\cC_2}$ (recall $\cC_2$ is defined in \eqref{eq:Ch_def}).
We use the concentration bounds of Chekuri, Vondr{\'a}k and  Zenklusen~\cite{ChekuriVZ2011} to show that, with high probability, $|\cM|$ is close to its expectation.
\begin{lemma}
\label{lem:M_concentration}
	It holds that $|\cM|\leq \bx^0 \cdot \one_{\cC_2} +\delta^2 \cdot \OPT$ with probability at least $1- \exp\left(-\delta^{10}\cdot \OPT\right)$.
\end{lemma}
\noindent
The proof of the lemma is given in \Cref{sec:rnr_deferred}. 

The size of the solution returned by \Cref{alg:match_and_round} is $|\cM| +\sum_{j=1}^{k} \rho_j + \rho^*$.
As \Cref{lem:first_fit_bound} and \Cref{lem:M_concentration} give upper bounds for $|\cM|$ and $\rho^*$, it remains to derive an upper bound on $\sum_{j=1}^{k} \rho_j$, the total number of configurations sampled by $\iterround$.

\subsubsection{A Refined Analysis of the Iterattve Rounding}
\label{sec:main_rnr}
Our analysis relies on the key notion of ``untouched'' configurations. Recall the sets of configurations~$\cC_j$ were define in \eqref{eq:Ch_def}, and $\cC_0 =\cC\setminus \left(\bigcup_{h=2}^{2\cdot \delta^{-1}} \cC_h\right)$.
For iteration $j\in \{0,1,\ldots, k\}$, define the set of \emph{untouched configurations} as
\begin{equation*}
  U_j=\left\{ C\in \cC~|~C\cap S_j\notin \cC_0\right\}=\left\{ C\in \cC~|~v(C\cap S_j\cap L) >(1-\delta, 1-\delta)\right\} \enspace .
\end{equation*}
Since $S_0 \supseteq S_1 \supseteq \ldots \supseteq S_k$, it follows that $U_0\supseteq U_1 \supseteq \ldots \supseteq U_k$.
We denote by $T_0= \cC \setminus U_0$ the initial set of {\em touched} configurations, and by $T_j = U_{j-1}\setminus U_j$ the configurations that become touched in iteration~$j$, for $j = 1,\hdots,k$.
Observe that $\cC_0\subseteq  T_0$.
We refine the sets $U_j$ and $T_j$ by defining $U_{j,h} = U_{j} \cap \cC_h$ and $T_{j,h} = T_j \cap \cC_h$ for $j = 0,\hdots,k$ and $h = 0,\hdots,2\cdot \delta^{-1}$.

Intuitively, we view configurations in $\cC_0$ as ``easy'' compared to configurations in $\cC\setminus \cC_0$. 
Indeed, we can construct linear structures only for configurations with a slack (\Cref{lem:structural}),
and a slack can be obtained with negligible overhead for configurations in $\cC_0$.
Thus, configurations in $U_j$ ``remain difficult'' after iteration $j$, while configurations in $T_j$ ``become easy'' in iteration $j$.
Observe that 
\begin{equation}
\label{eq:rho_to_opt}
  \sum_{j=1}^{k} \rho_j \leq k+ \alpha(1+\delta^2) \sum_{j=0}^{k-1} \OPTf(\one_{S_j}) \leq k+  (1+2\delta)	\delta \sum_{j=0}^{k-1} \OPT_f(\one_{S_j}),
\end{equation}
where the first inequality uses $\rho_j =\ceil{\alpha z_j} \leq \alpha(1+\delta^2) \OPTf(\one_{S_{j-1}}) +1$, and the second inequality uses $\alpha(1+\delta^2)\leq (1+2\delta)\delta$. 
Next, we derive an upper bound on $\delta \sum_{j=0}^{k-1} \OPT_f(\one_{S_j})$.
By~\eqref{eq:rho_to_opt}, this would imply a bound on $\sum_{j=1}^{k} \rho_j$, the number of configurations sampled by \Cref{alg:basic_round_and_round}.

Recall that $\bx^0$ is the solution for $\MLP$ found in \Cref{match:MatchingLP} of \Cref{alg:match_and_round}. 
We define $\bx^*\in [0,1]^{\cC}$~by
\begin{equation*} 
  \bx^*_C = \sum_{C'\in U_0\setminus \cC_2 \textnormal{ s.t. } C'\cap L= C} \bx^0_{C'}
\end{equation*}
for each $C\in \cC$.
Inzuitively, $\bx^*$ can be viewed as selecting all the configurations in $U_0 \setminus \cC_2$ as in $\bx^0$, and then discarding the small items.
Since $U_0$ is $\cF_0$-measurable and $\bx^0$ is $\cF_{-1}$-measurable, it follows that $\bx^*$ is $\cF_0$-measurable.
It can be easily verified that $\bx^* \cdot \one_{\cC_h}= \bx^0 \cdot \one_{U_{0,h}}$ for every $3\leq h \leq 2\cdot \delta^{-1}$ and $\bx^* \cdot \one_{\cC_0}= \bx^* \cdot \one_{\cC_2}=0$.
Furthermore, for any $C\in \supp(\bx^*)$ it holds that~$C\subseteq S_0\cap L$. 

Let $\by^* \in [0,1]^{I}$ be the coverage of $\bx^*$.
Then $\supp(\by^*)\subseteq S_0 \cap L$.
We note that our definition of~$\bx^*$ does not include the coverage of items by configurations in $T_0 \cup \cC_2$ in $\bx^0$.
The coverage of these items is given by $\one_I-\by^*$.
In the analysis we consider these coverage vectors separately, using the inequality
\begin{equation}
\label{eq:initial_decomp}
  \delta \sum_{j=0}^{k-1} \OPT_f(\one_{S_j})\leq \delta \sum_{j=0}^{k-1} \OPT_f(\one_{S_j} \wedge \by^* ) +\delta \sum_{j=0}^{k-1} \OPT_f\left(\one_{S_j}\wedge (\one_I-\by^*)\right) \enspace .
\end{equation}

The configurations in $\supp(\bx^*)$ are those that remain ``difficult'' after the sampling of $\cM$; thus,~$\by^*$ represents the coverage of items by these difficult configurations.
Other configurations are either in $T_0$, or in $\cC_2$. As the configurations in $T_0$ are ``easy'', we use them to compensate for items not selected by the matching $\cM$. 
Due to a technical limitation of linear structures, we eliminate the small items from $\by^*$. 

Our analysis relies on the following application of linear structures in conjunction with \Cref{lem:concentration_multistep_prelim}.
\begin{lemma}
\label{lem:iterative_mcdiarmid}
  For $j\in \{0,1,\ldots, k \}$, let $\blam\in [0,1]^{\cC^*}$ be an $\cF_j$-measurable random vector, $\bw$ the coverage of $\blam$, $\cS$ an $\cF_{j}$-measurable random $(\delta,\varphi(\delta))$-linear structure of $\blam$, and $\bd\in [0,1]^{I}$ a small-items integral $\cF_j$-measurable random demand vector.
  Then
  \begin{equation*}
  	\forall r = j,\hdots,k:~~~
  \OPTf\left(\bd \wedge \one_{S_r} \right)\leq (1-\delta)^{r-j} (1+10\delta ) \|\blam\| + \varphi(\delta) + \delta^{10} \OPT
  \end{equation*}
  with probability at least $\xi-\varphi(\delta)^2\cdot  \exp\left( -\frac{\OPT}{\varphi^{25}(\delta)}\right)$, where
  \begin{equation}
  	\label{eq:xi_prob}
  \xi = \Pr\left( \forall \bu \in \cS:~(\one_{S_j}\wedge \bd)\cdot \bu \leq \bw \cdot \bu+\frac{1}{\varphi^{11}(\delta)}\cdot  \OPT \cdot \tol(\bu)\right).
  \end{equation}
\end{lemma}
The proof of the lemma is given in \Cref{sec:rnr_deferred}.

We proceed to separately bound the quantities $\delta \sum_{j=0}^{k-1} \OPT_f(\one_{S_j} \wedge \by^* )$ (see \Cref{lem:main_rnr}) and \linebreak $\delta \sum_{j=0}^{k-1} \OPT_f\left(\one_{S_j}\wedge (\one_I-\by^*)\right)$ (see \Cref{lem:undercovered_simple_and_small}).
The bound on $ \delta \sum_{j=0}^{k-1} \OPT_f(\one_{S_j}\wedge \by^*)$ is derived using the next lemmas.
\begin{lemma} 
\label{lem:concentration_Tjh}
  With probability at least $1-\delta^{-10}\exp\left(-\delta^{50} \cdot \OPT\right)$  it holds that
  \begin{equation}
  \label{eq:tjh_bound}
	\forall h = 2,\hdots, 2\cdot \delta^{-1}, j=1,\hdots,k:~~~~~
	  \bigg|\E\left[ \bx^* \cdot \one_{T_{j,h}}~\middle| ~\cF_{j-1} \right]-\bx^* \cdot \one_{T_{j,h}} \bigg|\leq  \delta^{20} \cdot \OPT \enspace .
	\end{equation}
\end{lemma}
The proof (given in \Cref{sec:rnr_deferred}) is a simple application of a \Cref{lem:Generalized_McDiarmid}.
\begin{lemma}
\label{lem:Ujh_lowerbound}
  There exists $\mu:(0,0.1)\rightarrow \mathbb{R}_+$, independent of the instance $(I,v)$ and $\delta$, such that
  \begin{equation}
  \label{eq:Ujh_lowerbound}
  	\forall h = 2,\hdots, 2\cdot \delta^{-1}, j = 1,\hdots,k:~
    \bx^*\cdot \one_{U_{j,h}} \geq (1-\delta)^
    {h\cdot j} \cdot \bx^* \cdot \one_{U_{0,h}} -\delta^{10}  \cdot \OPT \textnormal{ or } \OPTf(\one_{S_j})\leq \mu(\delta)
  \end{equation}
  with probability at least $1-\delta^{-10}\cdot \exp\left(-\delta^{50} \cdot \OPT \right)$. 
\end{lemma}
The lemma follows from the inequality $\Pr\left(C\in U_{j,h}~\middle|~\cF_{j-1} \right) \geq \one_{C\in U_{j-1,h}} \cdot  \left(1-\frac{h}{z_j} \right)^{\alpha \cdot z_j +1}$ implied by \Cref{lem:step_bound}, the observation that 
$\left(1-\frac{h}{z} \right)^{\alpha \cdot z +1} \rightarrow (1-\delta)^{h}$ as $z\rightarrow \infty$, and \Cref{lem:concentration_Tjh}.  
The dependence on $\mu$ in the lemma arises as the observation holds only if $z$ is sufficiently large.
The proof  is given in \Cref{sec:rnr_deferred}.    
Henceforth, we use $\mu$ to denote the function in \Cref{lem:Ujh_lowerbound}. 
\begin{lemma}
\label{lem:main_rnr}
  Assuming $\OPT> \delta^{-30}\cdot \left(\varphi(\delta)+\mu(\delta)\right)$, with  probability at least $1-\varphi^4(\delta)\cdot \exp\left( -\frac{\OPT}{\varphi^{25}(\delta)}\right)$ it holds that
  \begin{equation*}
    \delta \sum_{j=0}^{k-1} \OPT_f(\one_{S_j}\wedge \by^*) \leq  \frac{4}{3}\cdot \bx^0 \cdot \one_{U_0\setminus \cC_2}+30\cdot \delta \cdot \OPT \enspace .
  \end{equation*}
\end{lemma}
\begin{proof}
  For $j = 1,\hdots,k$, define $\bd^j\in [0,1]^{I}$, the touched demand of iteration $j$, as the coverage of~$\bx^* \wedge \one_{T_{j}}$. 
  This is the coverage of items in configurations that become touched in iteration $j$, given by $\bd^j_i = \sum_{C\in T_j} \bx^*_C \cdot C(i)$ for all $i\in I$. 
  For every $i\in I$ and $r\in \{0,1,\ldots, k-1\}$ we have 
  \begin{equation*}
    \by^*_i -\sum_{j=1}^{r} \bd^j_i = \sum_{C\in \cC} \bx^*_C \cdot C(i) - \sum_{j=1}^{r} \sum_{C\in T_j} \bx^*_C \cdot C(i) = \sum_{C\in U_r} \bx^*_C\cdot C(i),
  \end{equation*}
  where the last equality follows from $\supp(\bx^*) \cap T_0 =\emptyset$ (by the definition of $\bx^*$).  Hence, $\bx^*\wedge \one_{U_r}$ is a solution for $\LP\left(\by^* -\sum_{j=1}^{r}\bd^j\right)$, and thus $\OPTf\left(\by^* -\sum_{j=1}^{r}\bd^j\right)\leq \bx^* \cdot \one_{U_r}$.
  It follows that for $r = 0,1,\ldots,k-1$, 
  \begin{equation}
  \label{eq:basic_decomposition}
  \begin{split}
    \OPTf(\by^*\wedge \one_{S_r}) \leq \sum_{j=1}^{r} \OPTf\left(\bd^j\wedge \one_{S_r}\right) + \OPTf\bigg(\by^* -\sum_{j=1}^{r}\bd^j\bigg)\\
                                  \leq \sum_{j=1}^{r} \OPTf\left(\bd^j\wedge \one_{S_r}\right) + \bx^*\cdot \one_{U_r} \enspace .
  \end{split}
  \end{equation}

  We use \Cref{lem:iterative_mcdiarmid} to bound the above terms $\OPTf\left(\bd^j\wedge \one_{S_r}\right)$.  
  We note that a natural candidate for the construction of the vector $\blam$ in \Cref{lem:iterative_mcdiarmid} for iteration $j=1,\ldots, k$ is the vector $\bmu^j\in [0,1]^{\cC^*}$ defined by $\bmu^j_C= \sum_{C'\in T_j \textnormal{ and } C'\cap S_j = C} \bx^*_{C'}$ for all $C\in \cC$ (and $\bmu^j_C=0$ for $C\in \cC^*\setminus \cC$).
  It is easy to verify that $\bmu^j$ is with $\delta$-slack and its coverage is  $\bd^j\wedge \one_{S_j}$. However, using this construction in the analysis leads to a sub-optimal approximation ratio.
  To some extent, this sub-optimality can be attributed to the fact that $\supp(\bmu^j)$ may contain configurations which use only a small fraction of the available volume. 
  For example, in case  $C\in T_{j,h}$ for some large $h$ and $|C\cap S_j\cap L|=1$, we may have that $\bmu^j_{C\cap L\cap S_j} >0$, while $v(C\cap L\cap S_j)$ is very small (e.g, $(0,1.1\cdot\delta)$). Due to dependencies between items, such events may have non-negligible  probability. 
  To overcome this sub-optimality, we use for the construction of $\blam^j \in \mathbb{R}^{\cC^*}$ conditional probabilities as described below.
   
  For $h = 2,\hdots,2\cdot \delta^{-1}$ and $C\in \cC_h$, let $\bgam^C\in [0,1]^{\cC^*}$ be an $\frac{h}{h-1}$-relaxation of $C$.
  The existence of~$\bgam^C$ is guaranteed by \Cref{lem:auxilary_struct}.
  We define, for $j = 1,\hdots,k$,
  \begin{equation}
  	\label{eq:blamj_def}
    \blam^j =\sum_{C\in \cC\setminus \cC_0} \bx^*_C \cdot \left( \Pr\left( C\in T_j~|~\cF_{j-1}\right) - \left(1- \left(1-\frac{1}{z_j} \right)^{\rho_j} \right)\cdot \one_{C\in U_{j-1}}\right)\cdot \bgam^C,
  \end{equation}
  and let $\bw^j$ be the coverage of $\blam^j$.
  Since $U_{j-1}$, $\rho_j$ and $z_j$  are $\cF_{j-1}$-measurable, it follows that $\blam^j$ is $\cF_{j-1}$-measurable (and thus also $\cF_j$-measurable).  
  Furthermore, since  $\bgam^C$ is with  $\delta$-slack for every $C\in \cC\setminus \cC_0$, it follows that $\blam^j$ is with $\delta$-slack for $j = 1,\hdots,k$. 
  
  \begin{claim}
  	\label{claim:lambda_coverage}
  	For $j=1,\ldots,k $ and $i\in I$ it holds that $\E\left[\bd^j_i \cdot \one_{i\in S_j}~\middle|~\cF_{j-1} \right]  = \bw^j_i$.
  \end{claim}
  \begin{claimproof}
	For any $i\in I\setminus L$ and $j = 1,\hdots,k$ it holds that $\E\left[\bd^j_i \cdot \one_{i\in S_j} ~|~\cF_{j-1}\right]=0= \bw^j_i$, as \mbox{$\supp(\by^*)\subseteq L$} and $\by^*$ is the coverage of $\bx^*$.
	Thus, it remains to handle the case in which $i\in L$. 
	
    Now, for every $i\in L$ and $j = 1,\hdots,k$, we have
    \begin{equation}
  	\label{eq:demand_expectation_first}
      \begin{aligned}
        \E\left[\bd^j_i \cdot \one_{i\in S_j}~\middle|~\cF_{j-1} \right] &= \E\left[
    	\sum_{C\in \cC} \one_{C\in T_j} \cdot \one_{i\in S_j} \cdot \bx^*_C 	\cdot C(i)~\middle|~\cF_{j-1}\right]\\
    	&=\E\left[\sum_{C\in \cC\setminus \cC_0} \left(\one_{C\in T_j}  -\one_{C\in T_j}\cdot\one_{i\notin S_j} \right)\cdot \bx^*_C 	\cdot C(i)~\middle|~\cF_{j-1}\right]\\
    	&=\sum_{C\in \cC\setminus \cC_0}  \left( \Pr\left( C\in T_j ~|~\cF_{j-1}\right) -\E\left[ \one_{i\notin S_j} \one_{C\in U_{j-1}} ~|~\cF_{j-1}\right] \right) \cdot  \bx^*_C \cdot C(i) \enspace .
      \end{aligned}
    \end{equation}
    The second equality uses $T_j\cap \cC_0=\emptyset$ for $j\geq 1$, and the third equality uses that
    \begin{equation*}
      \one_{C\in T_j} \one_{i\notin S_j} = \one_{C\in U_{j-1}} \cdot \one_{C\notin U_j} \cdot \one_{i\notin S_j}= \one_{C\in U_{j-1}} \cdot \one_{i\notin S_j}
    \end{equation*}
    for any configuration $C$ for which $i\in C$.   
    By \Cref{lem:step_bound}, we have
    \begin{equation*}	 
	  \begin{aligned}
  	    \E\left[ \one_{i\notin S_j}\one_{C\in U_{j-1}} ~\middle| \cF_{j-1} \right]&= \one_{C\in U_{j-1}} \cdot \E\left[ \one_{i\notin S_j} ~\middle| \cF_{j-1} \right] \\
  	    &=  \one_{C\in U_{j-1} } \left(1- \one_{i\in S_{j-1}} \left(1-\frac{1}{z_j}\right)^{\rho_j} \right)\\
  	    &=   \one_{C\in U_{j-1} } \left(1-  \left(1-\frac{1}{z_j}\right)^{\rho_j} \right)
  	  \end{aligned}
	\end{equation*}
    for any $C\in \cC\setminus \cC_0$ and $i\in C\cap L$.
    Furthermore, since $\bgam^C$ is a relaxation of $C$, we have that\linebreak $C(i)=\sum_{C'\in \cC^*}\bgam^{C}_{C'}\cdot C'(i)$.
	Therefore, for any $C\in \cC\setminus \cC_0$ and $i\in L$, it holds that
    \begin{equation}
  	\label{eq:demand_expect_aux}
  	  \begin{aligned}
  	    \bigg( \Pr&\left( C\in T_j ~|~\cF_{j-1}\right) -\E\left[ \one_{i\notin S_j} \one_{C\in U_{j-1}} ~|~\cF_{j-1}\right] \bigg) \cdot  \bx^*_C \cdot C(i) \\
  	    &=  \left( \Pr\left( C\in T_j ~|~\cF_{j-1}\right) - \left(1-  \left(1-\frac{1}{z_j}\right)^{\rho_j} \right)\cdot \one_{C\in U_{j-1}}\right) \cdot  \bx^*_C \cdot C(i) \\
  	    &= \left( \Pr\left( C\in T_j ~|~\cF_{j-1}\right) - \left(1-  \left(1-\frac{1}{z_j}\right)^{\rho_j} \right) \cdot \one_{C\in U_{j-1}}\right) \cdot  \bx^*_C \cdot \sum_{C' \in \cC^*} \bgam^{C} _{C'}\cdot C'(i) \enspace.
      \end{aligned}
    \end{equation}
  
    By incorporating \eqref{eq:demand_expect_aux} into \eqref{eq:demand_expectation_first}, we have (for every $i\in L$ and $j = 1,\hdots,k$) that
    \begin{equation*}
      \begin{aligned}
  	  &\E\left[\bd^j_i \cdot \one_{i\in S_j}~\middle|~\cF_{j-1} \right]\\
  	  =~& \sum_{C\in \cC\setminus \cC_0}  \bx^*_C \cdot \left(  \Pr\left( C\in T_j ~|~\cF_{j-1}\right)- \left(1-  \left(1-\frac{1}{z_j}\right)^{\rho_j} \right)\cdot \one_{C\in U_{j-1}} \right)\cdot \sum_{C'\in \cC^*}  \bgam^C_{C'} \cdot C'(i)\\
  	  =~& \sum_{C'\in \cC^*}   C'(i)  \cdot \sum_{C\in \cC\setminus \cC_0} \bx^*_C \cdot \left(  \Pr\left( C\in T_j ~|~\cF_{j-1}\right)- \left(1-  \left(1-\frac{1}{z_j}\right)^{\rho_j} \right)\cdot \one_{C\in U_{j-1}}    \right)\cdot \bgam^C_{C'}\\
  	  =~& \sum_{C'\in \cC^*}   C'(i) \cdot \blam^j_{C'} =\bw^j_i \enspace .
      \end{aligned}
    \end{equation*}
  \end{claimproof} 

  To show the existence of a linear structure for $\blam$ using \Cref{lem:structural}, we also need the following claim.
  \begin{claim}
	For $j=1,\ldots, k$ it holds that $\blam^j \in [0,1]^{\cC^*}$, $\bw^j \in [0,1]^I$, and $\blam^j$ is small-items integral. 
  \end{claim}
  \begin{claimproof}
    We first show that $\blam^j \in \mathbb{R}^{\cC^*}_{\geq 0}$. 
    Let $C\in \cC\setminus \cC_0$, thus there is $i\in C\cap L$. It therefore holds that
    \begin{equation}
    \label{eq:lambda_nonneg_argument}
      \begin{aligned}
		\Pr(C\in T_j~|~\cF_{j-1})& = \E\left[\one_{C\in U_{j-1}}\cdot \one_{C\notin U_j}~|~\cF_{j-1}\right]\\
			&=\one_{C\in U_{j-1}}\cdot \Pr( {C\notin U_j}~|~\cF_{j-1})\\
			&\geq \one_{C\in U_{j-1}}\cdot \Pr( {i \notin S_j}~|~\cF_{j-1})\\
			&= \one_{C\in U_{j-1}} \left(1-  \left(1-\frac{1}{z_j} \right)^{\rho_j} \cdot \one_{i\in S_{j-1}}\right)\\
			&= \one_{C\in U_{j-1}} \left(1-  \left(1-\frac{1}{z_j} \right)^{\rho_j}\right) \enspace.
		\end{aligned}
	\end{equation}
	The  inequality  holds since  $i\notin S_j$ implies $C\notin U_j$, the third equality is by \Cref{lem:step_bound}, and the last equality holds since $\one_{C\in U_{j-1}} \cdot \one_{i\in S_{j-1}} = \one_{C\in U_{j-1}}$.
	By \eqref{eq:lambda_nonneg_argument} it follows that $\blam^j\in \mathbb{R}_{\geq 0}^{\cC^*}$. 
	
   Since $\bw^j_i = \E\left[\bd^j_i \cdot \one_{i\in S_j} ~|~\cF_{j-1}\right] \leq 1$ for every $i\in I$ for $j=1,\ldots, k$ (\Cref{claim:lambda_coverage}) it follows that $\bw^j\in [0,1]^{I} $ and subsequently $\blam^{j} \in [0,1]^{\cC^*}$ for $j=1,\ldots, k$.
   Furthermore, $\bw^j_i =0$ for every $i\in I\setminus L$ (as $\by^*_i=0$, $\by^*$ is the coverage of $\bx^*$ and \eqref{eq:blamj_def}), hence $\bw^j$ and $\blam^j$ are small-items integral.
  \end{claimproof} 
  
  By \Cref{lem:structural} there is a $(\delta,\varphi(\delta))$-linear structure $\cS_j$ of $\blam^j$ for $j = 1,\hdots,k$.
  \begin{claim}
  \label{claim:concentration_gen_iteration}
    For any $j\in \{1,\hdots,k\}$ it holds that 
    \begin{equation*}
    \Pr\left( \forall \bu \in \cS_j:~(\one_{S_j}\wedge \bd^j)\cdot \bu ~\leq ~\E\left[ \left(\one_{ S_j }  \wedge\bd^j  \right)\cdot \bu~\middle|~\cF_{j-1} \right]+\frac{\OPT}{\varphi^{11}(\delta)}   \cdot \tol(\bu)\right) \geq  1- \varphi(\delta)\cdot  \exp \left(-\frac{\OPT}{\varphi^{25}(\delta)}\right) \enspace .
    \end{equation*}
  \end{claim}
  The proof of \Cref{claim:concentration_gen_iteration}, given in \Cref{sec:rnr_deferred}, follows from \Cref{lem:Generalized_McDiarmid}.  
  By \Cref{claim:lambda_coverage} it holds that $\E\left[  \bu \cdot \left( \bd^j \wedge \one_{ S_j } \right) \middle| \cF_{j-1}\right]= \bu \cdot \bw^j$ for $j = 1,\hdots,k$ and $\bu \in \cS_j$; therefore, 
  \begin{equation*}
    \begin{aligned}
    \Pr&\left( \forall \bu \in \cS_j:~(\one_{S_j}\wedge \bd^j)\cdot \bu \leq \bw^j \cdot \bu+\frac{\OPT}{\varphi^{11}(\delta)}   \cdot \tol(\bu)\right)\\
 & = \Pr\left( \forall \bu \in \cS_j:~(\one_{S_j}\wedge \bd^j)\cdot \bu \leq \E\left[ \left( \one_{ S_j} \wedge  \bd^j \right)\cdot \bu\middle|\cF_{j-1} \right]+\frac{\OPT\cdot \tol(\bu)}{\varphi^{11}(\delta)}   \right)\\ &\geq  1-\varphi(\delta)\cdot \exp \left(-\frac{\OPT}{\varphi^{25}(\delta)}\right) \enspace .
  \end{aligned}
  \end{equation*}
  Here, the last inequality follows from \Cref{claim:concentration_gen_iteration}. 
  Thus, by \Cref{lem:iterative_mcdiarmid}, with probability at least $$1-k\cdot \varphi(\delta)\cdot  \exp\left( -\frac{\OPT}{\varphi^{25}(\delta)}\right)
  -k\cdot \varphi^2(\delta)\cdot  \exp\left( -\frac{\OPT}{\varphi^{25}(\delta)}\right) \geq
  1- \varphi^3(\delta) \cdot \exp\left( -\frac{\OPT}{\varphi^{25}(\delta)}\right),$$ 
  it holds that
  \begin{equation}
  \label{eq:OPT_dj_sr}
    \forall j=1,\hdots,k,r=j,\hdots,k:~
  \OPTf\left(\bd^j \wedge \one_{S_r} \right)\leq (1-\delta)^{r-j} (1+10\delta ) \|\blam^j\| + \varphi(\delta) + \delta^{10} \OPT \enspace.
  \end{equation}
  We henceforth assume that \eqref{eq:OPT_dj_sr}, \eqref{eq:tjh_bound} and \eqref{eq:Ujh_lowerbound}  hold. 

  Observe that, for $j = 1,\hdots,k$,
  \begin{equation}
  \label{eq:blamj_norm}
  	\begin{aligned}
      \|\blam^j\| &= \sum_{C\in \cC\setminus \cC_0} \bx^*_C \cdot \left( \Pr\left( C\in T_j~|~\cF_{j-1}\right) -\left(1- \left(1-\frac{1}{z_j} \right)^{\rho_j} \right)\cdot  \one_{C\in U_{j-1}}\right)\cdot \|\bgam^C\|\\
      &\leq \sum_{C\in \cC\setminus \cC_0} \bx^*_C \cdot \left( \Pr\left( C\in T_j~|~\cF_{j-1}\right) -\delta\cdot  \one_{C\in U_{j-1}}\right)\cdot \|\bgam^C\|\\
      &\leq\sum_{h=2}^{2\cdot \delta^{-1}}  
      \sum_{C\in \cC_h} \bx^*_C \cdot \left( \Pr\left( C\in T_j~|~\cF_{j-1}\right) -\delta\cdot \one_{C\in U_{j-1}}\right)\cdot \frac{h}{h-1}\\
      &=\sum_{h=2}^{2\cdot \delta^{-1}}  
       \frac{h}{h-1} \left(\E\left[ \bx^* \cdot \one_{T_{j,h}} ~\middle|~\cF_{j-1} \right]- \delta \cdot \bx^*\cdot \one_{U_{j-1,h}} \right)\\
      &\leq \sum_{h=2}^{2\cdot \delta^{-1}}  
       \frac{h}{h-1} \left(\bx^* \cdot \one_{T_{j,h}}+ \delta^{10}\cdot \OPT- \delta \cdot \bx^*\cdot \one_{U_{j-1,h}} \right)\\
      & \leq \sum_{h=2}^{2\cdot \delta^{-1}}   \frac{h}{h-1}\left( (1-\delta )\bx^* \cdot \one_{U_{j-1, h}} - \bx^*\cdot \one_{U_{j,h}} \right) + \delta^8 \cdot \OPT \enspace .
    \end{aligned}
  \end{equation}
  The first inequality follows from $\left(1-\frac{1}{z_j} \right)^{\rho_j} \leq (1-\delta)$ (\Cref{lem:step_bound}).
  The second inequality holds, since $\bgam^C$ is an $\frac{h}{h-1}$-relaxation of $C$ for any $C\in \cC_h$; the third inequality follows from the assumption that~\eqref{eq:tjh_bound} holds; and the last inequality uses $T_{j,h} = U_{j-1,h} \setminus U_{j,h}$. 
  
  Combining \eqref{eq:OPT_dj_sr} and \eqref{eq:blamj_norm} with $\OPT> \delta^{-30} \varphi(\delta)$,  we have 
  \begin{equation*}
  	\frac{\OPTf\left(\bd^j \wedge \one_{S_r} \right)}{1+10\delta }\leq (1-\delta)^{r-j}  \sum_{h=2}^{2\cdot \delta^{-1} } \frac{h}{h-1}\left( (1-\delta )\bx^* \cdot \one_{U_{j-1, h}} - \bx^*\cdot \one_{U_{j,h}} \right)  + \delta^{7} \OPT
  \end{equation*}
  for $j = 1,\hdots,k$ and $r = j,\hdots,k$.
  Using the last inequality and \eqref{eq:basic_decomposition}, we obtain 
  \begin{equation*}
  \begin{aligned}
  	\frac{\OPTf(\by^*\wedge \one_{S_r})}{1+10\delta} 
  	&\leq
  	\sum_{j=1}^{r} (1-\delta)^{r-j}  \sum_{h=2}^{2\cdot \delta^{-1} } \frac{h}{h-1}\left( (1-\delta )\bx^* \cdot \one_{U_{j-1, h}} - \bx^*\cdot \one_{U_{j,h}} \right)  
  	+ 
  	\bx^*\cdot \one_{U_r} + \delta^{5}\OPT
  	\\
  	&=\sum_{h=2}^{2\cdot \delta^{-1}} \frac{h}{h-1} \left(  (1-\delta)^r \cdot \bx^* \cdot \one_{U_{0,h}} -\bx^*\cdot  \one_{U_{r,h}}\right) + 
  		 \bx^*\cdot \one_{U_r}+ \delta^{5}\OPT\\
  	&= \sum_{h=2}^{2\cdot \delta^{-1}} \frac{1}{h-1} \left(  (1-\delta)^r \cdot \bx^* \cdot \one_{U_{0,h}} -\bx^*\cdot \one_{U_{r,h}}\right) +
  		(1-\delta)^r \bx^*\cdot \one_{U_{0}}+ \delta^{5}\OPT
  \end{aligned}
  \end{equation*}
  for every $r\in \{0,1,\ldots, k-1\}$. 
  Observe that $\OPTf(\by^*\wedge \one_{S_r})\leq \OPTf(\one_{S_r})\leq \OPTf(\one_{S_j}) $ for $j = 1,\hdots,k$ and $r = j,\hdots,k$; thus, if $\OPT(\one_{S_j})\leq \mu(\delta)\leq \delta^{30}\OPT$ for some $j\in\{1,\hdots,k\}$, then for every $r\geq j$ it holds that $ \OPTf(\by^*\wedge \one_{S_r})\leq \delta^{30} \OPT$ . 
  Using the above inequality and \eqref{eq:Ujh_lowerbound}, we have
  \begin{equation*}
  	\frac{\OPTf(\by^*\wedge \one_{S_r})}{1+10\delta} 
  	\leq \sum_{h=2}^{2\cdot \delta^{-1}} \frac{ (1-\delta)^r  -(1-\delta)^{h\cdot r}}{h-1}\cdot \bx^* \cdot \one_{U_{0,h}}   + 
  	(1-\delta)^r \bx^*\cdot \one_{U_{0}} + \delta^{4}\OPT \enspace .
  \end{equation*}
  Thus,
  \begin{equation*}
    \begin{aligned}
    & \frac{ \delta \sum_{j=0}^{k-1} \OPTf(\one_{S_j}\wedge \by^*)}{1+10\delta}
    \\
      \leq~&
    \delta \sum_{j=0}^{k-1}
    \sum_{h=2}^{2\cdot \delta^{-1}} \frac{ (1-\delta)^j  -(1-\delta)^{h\cdot j}}{h-1}\cdot \bx^* \cdot \one_{U_{0,h}}   + 
    \delta\cdot  \sum_{j=0}^{k-1} (1-\delta)^j \bx^*\cdot \one_{U_{0}} + \delta^{3}\OPT\\
     =~& \delta \sum_{h=2}^{2\cdot \delta^{-2}} \frac{\bx^* \cdot \one_{U_{0,h}}}{h-1} \left( \frac{1-(1-\delta)^k}{1-(1-\delta)} - \frac{1-(1-\delta)^{k\cdot h}}{1-(1-\delta)^h} \right) + \delta\cdot  \frac{1-(1-\delta)^k}{1-(1-\delta)} \cdot \bx^*\cdot  \one_{U_0}  + \delta^{3}\OPT\\
     \leq~&
     \sum_{h=2}^{2\cdot \delta^{-2}} \frac{\bx^* \cdot \one_{U_{0,h}}}{h-1} \left( 1- \frac{1-\delta}{h} \right) + \bx^*\cdot \one_{U_0}  + \delta^{3}\OPT\\
     \leq~& 
    \sum_{h=3}^{2\cdot \delta^{-2}} \frac{h+1}{h}\cdot \bx^0 \cdot \one_{U_{0,h}} + \delta^{3}\OPT+\delta \|\bx^*\| \enspace .
    \end{aligned}
  \end{equation*}
  The second inequality holds, since $(1-\delta)^{k} \leq \delta$ and $(1-\delta)^h \geq 1-\delta h$.
  The last inequality uses $\bx^*\cdot \one_{U_{0,h}} =\bx^0 \cdot \one_{U_{0,h}}$ for $h\geq 3$, and $\bx^* \cdot \one_{\cC_2}=0$ by the definition of $\bx^*$.
  Since $\|\bx^*\|\leq \|\bx^0 \| \leq (1+\delta^2)\OPT\leq 1.01\cdot \OPT$, we have
  \begin{equation}
  	\label{eq:non_matching_upper}
    \delta \sum_{j=0}^{k-1} \OPTf(\one_{S_j}\wedge \by^*) \leq \sum_{h=3}^{2\cdot \delta^{-2}} \frac{h+1}{h}\cdot \bx^0 \cdot \one_{U_{0,h}} + 30\cdot \delta\cdot \OPT
                                                          \leq  \frac{4}{3} \cdot \bx^0\cdot \one_{U_0\setminus \cC_2}+30\delta \cdot \OPT,
  \end{equation}                                                          
  as in the statement of the lemma.
  As we assumed that~\eqref{eq:OPT_dj_sr}, \eqref{eq:tjh_bound} and \eqref{eq:Ujh_lowerbound} hold, by \Cref{lem:concentration_Tjh} and \Cref{lem:Ujh_lowerbound} it follows that \eqref{eq:non_matching_upper} holds with probability at least
  \begin{equation*}
    1- \varphi^{3}(\delta)\cdot \exp \left( -\frac{\OPT}{\varphi^{25}(\delta)}\right)- 2\cdot \delta^{-10} \exp \left(- \delta^{50}\cdot\OPT\right)\geq 1-\varphi^4(\delta)\cdot \exp\left(-\frac{\OPT}{\varphi^{25}(\delta)}\right) \enspace .\qedhere
  \end{equation*}
\end{proof}

Define $\by^{\cM}$ as the coverage of $\bx^0 \wedge \one_{\cC_2}$; that is, $\by^{\cM}_i = \sum_{C \in \cC_2} \bx^0_C \cdot C(i)$ for all $i \in I$.
To obtain a bound on $\delta \sum_{j=0}^{k-1} \OPTf\left( \one_{S_j} \wedge (\one_I-\by^0)\right)$, we use the next lemma.

\begin{lemma}
\label{lem:match_prob}
  For any $i\in I$ it holds that $\Pr(i\notin S_0 ) =(1-\delta^4)\by^{\cM}_i$ if $i\in L$, and $\Pr(i\notin S_0 ) = 0$ otherwise.
\end{lemma}
\begin{proof}
  Let $G=(L,E)$ be the $\delta$-matching graph of the instance.
  We use $N(i)$ to denote the set of neighbors of $i\in L$ . 
  Since $\cM$ is a matching, for every $i\in L$ it holds that $\one_{i\notin S_0} =  \sum_{i'\in N(i)} \one_{ \{i,i'\} \in \cM}$.
  Therefore, for any $i\in L$ it holds that  
  \begin{equation*}
    \begin{aligned}
    \Pr(i\notin S_0)  = \E[\one_{i\notin S_0} ] 
    = \sum_{i'\in N(i)}  \E\left[ \one_{\{i,i'\} \in \cM}\right]  &=(1-\delta^4)\sum_{~i' \in N(i)~ } \sum_{~C\in \cC_2 \textnormal{ s.t. } \{i,i'\}\subseteq C~} \bx^0_C \\
    &= (1-\delta^4) \sum_{C\in \cC_2} \bx^0_C \cdot C(i) = (1-\delta^4 )\cdot \by^{\cM}_i \enspace .
    \end{aligned} 
  \end{equation*}
  The third equality holds, since $\Pr(e \in \cM) = (1-\delta^4) \sum_{C\in \cC_2 \textnormal{ s.t. } e\subseteq C} \bx^0_C$.
  Also, for any $i\in I\setminus L$ it holds that $i\notin \bigcup_{e\in \cM} e$; thus, $i\in S_0$, i.e., $\Pr(i\notin S_0)=0$. 
\end{proof}

We now derive an upper bound for $ \delta \sum_{j=0}^{k-1} \OPTf\left( \one_{S_j} \wedge (\one_I-\by^*)\right)$. 
\begin{lemma}
\label{lem:undercovered_simple_and_small}
  Assuming $\OPT>\delta^{-30} \varphi(\delta)$, with probability at least $1- \exp \left(-\frac{\OPT}{\varphi^{25}(\delta)} + \varphi^2(\delta )\cdot \ln \OPT \right)$ it holds that
  \begin{equation*} 
	\delta \sum_{j=0}^{k-1} \OPTf\left(\one_{S_j} \wedge (\one_I-\by^*)\right) \leq \frac{4}{3}\cdot \bx^0\cdot \one_{T_0\setminus \cC_2} +\frac{1}{3}\cdot |\cM| + 50\delta\cdot (\OPT+|\cM|) \enspace .
  \end{equation*}
\end{lemma}
\begin{proof}
  Similar to the proof of \Cref{lem:main_rnr}, we use \Cref{lem:iterative_mcdiarmid} also in this proof.
  To this end, we construct a vector $\blam$ that is used to derive a linear structure $\cS$.
  Subsequently, we show that $\blam$ and~$\cS$  admit the conditions of \Cref{lem:iterative_mcdiarmid} with respect to the demand vector $\one_{S_0} \wedge (\one_{I}-\by^*)$.

  For any $h = 2,\hdots,2\cdot \delta^{-1}$ and $C \in \cC_h$, let $\bgam^C$ be an $\frac{h}{h-1}$-relaxation of $C$, and for any $C\in \cC_0$ let $\bgam^C$ be a $(1+4\delta)$-relaxation of $C$.
  Furthermore, for any $C\in \cC$ such that $v(C)\leq (\delta, \delta)$ let $\btau^C$ be a $4\delta$-relaxation of $C$.
  The existence of these relaxations is guaranteed by \Cref{lem:delta_relaxed},~\Cref{lem:auxilary_struct}, and~\Cref{lem:relax_small}. 
  Define
  \begin{equation*}
	  \blam = \delta^4 \sum_{i\in L} \by^{\cM}_i \cdot  \one_{\{\{i\}\}}+\sum_{C\in T_0\setminus \cC_2} \bx^0_C  \cdot \bgam^C+\sum_{C\in U_0\cup \cC_2} \bx^0_C \cdot  \btau^{C\setminus L},
  \end{equation*}
  where $\one_{\{\{i\}\}}\in [0,1]^{\cC^*}=\bz$ such that $\bz_{\{i\}}=1$, and $\bz_{C}=0$ for $C\in \cC^*\setminus\{\{i\}\}$.
  Observe that $\cC_0\subseteq T_0$ by definition; thus, $v(C\setminus L)\leq (\delta, \delta)$ for every $C\in U_0\cup \cC_2$.
  That is, $\blam$ is well-defined. 
  Since the instance does not contain $\delta$-huge items, it follows that $\one_{\{\{i\}\}}$ is with $\delta$-slack.
  Hence, $\blam$ is with $\delta$-slack as well.
  As $T_0$ and $U_0$ are $\cF_0$-measurable, it follows that $\blam$ is $\cF_0$-measurable.
  Let $\bw$ be the coverage of $\blam$ and define $\bd = \one_{S_0}\wedge (1-\by^*)$.
  Observe that we may have $\bw_i > 0$ (i.e., $i \in \supp(\bw)$) for items already selected by the matching, that is, items in $L\setminus S_0$.
  The coverage of these items can intuitively be viewed as a placeholder for items $i\in L \cap S_0$ for which $\bw_i<\bd_i$. 
	
  For any $i\in I\setminus L$, it holds that
  \begin{equation}
  \label{eq:small_items_diff}
    \begin{aligned}
      \bw_i = \sum_{C\in \cC^*} \blam_C\cdot C(i) & = \sum_{C\in T_0 \setminus \cC_2} \bx^0_C \cdot C(i) + \sum_{C\in U_0 \cup \cC_2} \bx^0_C  \cdot C(i) \\
    		    &=\sum_{C\in \cC} \bx^0_C \cdot C(i) =1 = \one_{i\in S_0}(1-\by^*_i) = \bd_i \enspace .
	\end{aligned}
  \end{equation}
  The fourth equality holds, as $\bx^0$ is a solution for $\MLP$.
  The fifth equality holds, since $\by^*_i=0$ for all $i\in I\setminus L$ and by \Cref{lem:match_prob}.
  In particular, it follows that $\bw$ and $\blam$ are small-items integral, and $\bw_i- \bd_i = 0$ for any $i\in I\setminus L$.
  Furthermore, for any $i\in L$ it holds that 
  \begin{equation*}
    \bw_i =\delta^4\cdot \by^{\cM}_i + \sum_{C\in T_0\setminus \cC_2} \bx^0_C\cdot C(i)\leq\delta^4\cdot \by^{\cM}_i + \sum_{C\in \cC\setminus \cC_2} \bx^0_C\cdot C(i) =\delta^4\cdot \by^{\cM}_i + (1-\by^{\cM}_i)\leq 1,
  \end{equation*}
  thus $\bw\in [0,1]^I$ and we can infer that $\blam \in [0,1]^{\cC^*}$. 
	
  For any $i\in L$, we have
  \begin{equation*}
  	\begin{aligned}
      \bd_i &= \one_{i\in S_0} \left( 1-\sum_{C\in U_0\setminus \cC_2} \bx^0_C \cdot C(i) \right) \\
    	    &= \one_{i\in S_0} - \left(1-\one_{i\notin S_0} \right) \sum_{C\in U_0\setminus \cC_2} \bx^0_C \cdot C(i) \\
          	&= \one_{i\in S_0} - \sum_{C\in U_0 \setminus \cC_2} \bx^0_C \cdot C(i) -\sum_{C\in \cC\setminus \cC_2} \one_{i\notin S_0} \cdot \one_{C\in U_0} \cdot \bx^0_C \cdot C(i)\\
          	&= \one_{i\in S_0} -  \sum_{C\in U_0 \setminus \cC_2} \bx^0_C \cdot C(i),
  	\end{aligned}
  \end{equation*}
  where the the fourth equality holds since for every $C\in \cC$ such that $i\in C$, if $i\notin S_0$ then $C\notin U_0$.
  Thus, for every $i\in L$,
  \begin{equation}
  \label{eq:large_items_diff}
    \begin{aligned}
      \bw_i -\bd_i &= \delta^4 \cdot \by^{\cM}_i +\sum_{C \in T_0\setminus \cC_2} \bx^0_C \cdot C(i) -\left(\one_{i\in S_0} -  \sum_{C\in U_0 \setminus \cC_2} \bx^0_C \cdot C(i) \right)\\
    	&= \delta^4 \cdot \by^{\cM}_i +\sum_{C\in \cC\setminus \cC_2} \bx^0_C \cdot C(i) -\one_{i\in S_0} \\
    	&= \delta^4 \cdot \by^{\cM }_i +1- \by^{\cM}_i -\one_{i\in S_0} \\
    	&= \one_{i\notin S_0} - (1-\delta^4)\cdot \by^{\cM}_i,
  	\end{aligned}
  \end{equation}
  where the third equality holds since
  \begin{equation*}
	1=\sum_{C\in \cC} \bx^0_C \cdot C(i) = \sum_{C\in \cC\setminus \cC_2} \bx^0_C \cdot C(i) +\sum_{C\in  \cC_2} \bx^0_C \cdot C(i)  = \sum_{C\in \cC\setminus \cC_2} \bx^0_C \cdot C(i) + \by^{\cM}_i\enspace.
  \end{equation*}	
  By \eqref{eq:small_items_diff}, \eqref{eq:large_items_diff} and \Cref{lem:match_prob}, it holds that $\E[\bw_i] =\E[\bd_i ]$ for every $i\in I$.

  Using the concentration bounds for $\CVZ$, as given by Chekuri et al.~\cite{ChekuriVZ2011}, we can show that, with high probability, $\bu \cdot \bd \lesssim \bu \cdot \bw$ for every $\bu \in \mathbb{R}_{\geq 0}^I$. 
  \begin{claim}
  \label{claim:concentration_first_iteration}
    For any $\bu \in \mathbb{R}_{\geq 0}^I$ it holds that
    \begin{equation*}
      \Pr\left(\bd \cdot \bu > \bw \cdot \bu+\frac{\OPT}{\varphi^{11}(\delta)} \cdot \tol(\bu)\right) \leq \exp \left(-\frac{\OPT}{\varphi^{25}(\delta)}\right) \enspace .
    \end{equation*}
  \end{claim}
  \noindent
  The proof of \Cref{claim:concentration_first_iteration} is given in \Cref{sec:rnr_deferred}.
 	
  Let $\cS^*\subseteq \mathbb{R}^{I}_{\geq 0}$ be the set defined in \Cref{lem:structural}.
  Also, by \Cref{lem:structural}, there exists a $(\delta,\varphi(\delta))$-linear structure $\cS$ of $\blam$ such that for any $\bu \in \cS$  which satisfies $\supp(\bu)\cap L\neq \emptyset$ it holds that $\bu \in \cS^*$.
  Observe that $\cS^*$ is non-random while $\cS$ is an $\cF_0$-measurable random set, as $\blam$ is $\cF_0$-measurable.  
 
  \Cref{claim:concentration_first_iteration} requires that the vector $\bu\in \mathbb{R}^I_{\geq 0}$ is deterministic, and thus we cannot directly use the claim  with a random vector $\bu \in \cS$.
  Instead, we use the set $\cS^*$ to circumvent this issue.
  Observe that for any $\bu \in \cS$, if $\supp(\bu) \cap L=\emptyset$ then $\bd\cdot \bu = \bw \cdot \bw$ by \eqref{eq:small_items_diff}, and if $\supp(\bu) \neq \emptyset $ then $\bu \in \cS^*$. 
  Thus,
  \begin{equation*}
    \begin{aligned} 
      \Pr&\left( \forall \bu \in \cS:~\bd \cdot \bu \leq \bw \cdot \bu+\frac{\OPT}{\varphi^{11}(\delta)} \cdot \tol(\bu)\right) \geq
      \Pr\left( \forall \bu \in \cS^*:~\bd \cdot \bu \leq \bw \cdot \bu+\frac{\OPT}{\varphi^{11}(\delta)}   \cdot  \tol(\bu)\right)\\
       & \geq 1- |\cS^*|\cdot \exp \left(-\frac{\OPT}{\varphi^{25}(\delta)}  \right)\geq 1- \exp \left(-\frac{\OPT}{\varphi^{25}(\delta)}+\varphi(\delta) \cdot \ln\OPT \right) \enspace .
     \end{aligned} 
  \end{equation*}

  The second inequality is by the union bound, and \Cref{claim:concentration_first_iteration}.
  The third inequality holds, since $|\cS^*|\leq \varphi(\delta)\cdot |L|^4 \leq \varphi(\delta)\cdot 2^4\cdot \delta^{-4}\cdot \OPT^4$ as $\OPT\geq \frac{\delta}{2} |L|$. 
  Therefore, by \Cref{lem:iterative_mcdiarmid}, it holds that  
  \begin{equation}
  \label{eq:complement_bound}
	\forall j = 0,\hdots,k:~~~~
		\OPTf\left(\one_{S_j} \wedge(\one_I -\by^*)\right)\leq (1-\delta)^{j} \cdot (1+10\delta)\|\blam\| +\varphi(\delta)+\delta^{10}\OPT 
  \end{equation}
  with probability at least
  \begin{equation*}
    1- \exp \left(-\frac{\OPT}{\varphi^{25}(\delta)} + \varphi(\delta )\cdot \ln \OPT \right) -\varphi^2(\delta)\cdot  \exp\left(-\frac{\OPT}{\varphi^{25}(\delta)} \right)\geq 1- \exp \left(-\frac{\OPT}{\varphi^{25}(\delta)} + \varphi^2(\delta )\cdot \ln \OPT \right) \enspace .
  \end{equation*}
  We henceforth assume that~\eqref{eq:complement_bound} holds. 
	
  We note that
  \begin{equation}
  \label{eq:blam_norm_first}
	\begin{aligned} \| \blam \| &\leq \delta^4 \cdot\one_L \cdot \by^{\cM} + 
	  \sum_{h=3}^{2\cdot \delta^{-1}} \frac{h}{h-1}\cdot  \bx^0 \cdot \one_{T_{0,h}}  +(1+4\delta)\cdot \bx^0 \cdot \one_{\cC_0} +4\delta \|\bx^0\|\\
	  &\leq \frac{4}{3}\cdot  \one_{T_{0}\setminus \cC_2} \cdot \bx^0 + \frac{1}{6}\cdot \sum_{h=3}^{2\cdot \delta^{-1}} \bx^0 \cdot \one_{T_{0,h}}  +10 \delta \cdot \OPT,
	\end{aligned} 
  \end{equation}
  where the second inequality uses
  \begin{equation*}
    \one_L\cdot \by^{\cM} = \sum_{i\in L} \by^{\cM}_i= \sum_{i\in L} \sum_{C\in \cC_2} \bx^0_C\cdot C(i) = \sum_{C\in \cC_2} \bx^0_C \cdot 2 \leq 2\cdot \bx^0\cdot \one_{\cC_2}\leq 2 \cdot(1+\delta^2)\OPT \enspace .
  \end{equation*} 
  It also holds that
  \begin{multline*}
		 \sum_{h=3}^{2\cdot \delta^{-1}} \bx^0 \cdot \one_{T_{0,h}} =    \sum_{C\in \cC\setminus \cC_0\setminus \cC_2} \bx^0_C\cdot \one_{C\in T_0}\\
		                                                            \leq \sum_{C\in \cC\setminus \cC_0 \setminus \cC_2} \bx^0_C\sum_{i\in C\cap L}\one_{i\notin S_0}
		 \leq \sum_{i\in L} \one_{i\notin S_0}\sum_{C\in \cC\setminus \cC_2} \bx^0_C \cdot C(i)
		\leq \sum_{i\in L} \one_{i\notin S_0}  \leq 2\cdot  |\cM| \enspace .  
  \end{multline*}
  Plugging the above inequality into \eqref{eq:blam_norm_first}, we obtain
  \begin{equation}
  \label{eq:blam_norm}
    \begin{aligned} \| \blam \| 
    &\leq \frac{4}{3}\cdot  \bx^0\cdot \one_{T_0\setminus \cC_2} +\frac{1}{3}\cdot  |\cM|  + 10\cdot\delta\cdot  \OPT \enspace .
	\end{aligned} 
  \end{equation}
	
  By~\eqref{eq:complement_bound} and \eqref{eq:blam_norm}, we have
  \begin{equation}
  \label{eq:matching_part}
    \begin{aligned}
	  \delta&\sum_{j=0}^{k-1} \OPTf\left(\one_{S_j}\wedge (\one_I -\by^*))\right) \leq\delta \sum_{j=0}^{k-1} \left( (1-\delta)^{j}\cdot (1+10\delta)\|\blam\| +\varphi(\delta)+\delta^{10}\OPT  \right)\\
	  &\leq (1+10\delta) \|\blam\| + \delta^{8}\OPT\\
		&\leq  (1+10\delta )\left(\frac{4}{3} \cdot \bx^0\cdot \one_{T_0\setminus \cC_2} +\frac{1}{3}\cdot  |\cM|  + 10\delta\cdot  \OPT \right) +\delta^{8} \OPT\\
		&\leq\frac{4}{3}\cdot  \bx^0\cdot \one_{T_0\setminus \cC_2} +\frac{1}{3}\cdot  |\cM|  + 50\delta (\OPT+|\cM|),
	\end{aligned} 
  \end{equation}
  where the second inequality uses $\OPT>\delta^{-30} \varphi(\delta)$, and the last inequality holds since\linebreak $\|\bx^0 \|\leq 1.01 \cdot \OPT$.
  As we assumed that \eqref{eq:complement_bound} holds, it follows that inequality \eqref{eq:matching_part} holds with probability at least $1- \exp \left(-\frac{\OPT}{\varphi^{25}(\delta)} + \varphi^2(\delta )\cdot \ln \OPT \right)$, as stated in the lemma.	
\end{proof}

\subsubsection{Asymptotic Approximation Ratio}
\label{sec:asymptotic}
\begin{proof}[Proof of \Cref{lem:weakly_asym}]
  Note that we may assume $\OPT$ is larger than any function which depends on~$\delta$ (but not on the instance). 
  Assume that the statements of \Cref{lem:first_fit_bound,lem:M_concentration,lem:main_rnr,lem:undercovered_simple_and_small} hold.
  This occurs with probability at least
  \begin{equation*}
    1- \delta^{-2}\cdot \exp(-\delta^7\cdot \OPT) -\exp(-\delta^{10}\cdot \OPT) - \varphi^{4}(\delta )\cdot \exp\left(-\frac{\OPT}{\varphi^{25}(\delta)}\right)- \exp\left(-\frac{\OPT}{\varphi^{25}(\delta)} +\varphi^2(\delta)\cdot \ln\OPT \right)\geq \frac{1}{2},
  \end{equation*}
  assuming that $\OPT$ is sufficiently large.  
  
  We also assume that $\OPT>\delta^{-30}\left( \varphi(\delta)+\mu(\delta)\right)$. 
  By \Cref{lem:main_rnr,lem:undercovered_simple_and_small}, we have
  \begin{equation*}
    \begin{aligned}
	  \sum_{j=1}^{k} \rho_j & \leq k + (1+2\delta) \delta \sum_{j=0}^{k-1} \OPT_f(\one_{S_j})\\
	                        & \leq k + (1+2\delta)\left(\delta \sum_{j=0}^{k-1} \OPT_f(\one_{S_j} \wedge \by^* ) +\delta \sum_{j=0}^{k-1} \OPT_f\left(\one_{S_j}\wedge (\one_I-\by^*)\right)\right)\\
                            & \leq k + (1+2\delta)\left(\frac{4}{3} \cdot \bx^0\cdot \one_{U_0\setminus \cC_2} +30\delta \OPT+  \frac{4}{3}\cdot  \bx^0\cdot \one_{T_0\setminus \cC_2} +\frac{1}{3}\cdot  |\cM| + 50\delta (\OPT+|\cM|)\right)\\
                            & \leq k + (1+2\delta )\left( \frac{4}{3}\cdot \bx^0\cdot  \one_{\cC\setminus \cC_2 }+ \frac{1}{3} \cdot  |\cM|  +80\delta (\OPT+ |\cM|)\right)\\
                            & \leq \frac{4}{3}\cdot \bx^0 \one_{\cC\setminus \cC_2 }+ \frac{1}{3} \cdot  |\cM|  +90\delta (\OPT+|\cM|) \enspace .
    \end{aligned}
  \end{equation*}
  The first inequality uses~\eqref{eq:rho_to_opt}, and the last inequality assumes $\OPT>\frac{k}{\delta}$.
  The number of configurations returned by the algorithm (assuming the statement of the lemmas hold) is
  \begin{equation*}
    \begin{aligned} 
      |\cM|+\sum_{j=1}^{k} \rho_j +\rho^* & \leq |\cM| + \frac{4}{3}\cdot \bx^0\cdot  \one_{\cC\setminus \cC_2 }+ \frac{1}{3} \cdot  |\cM| +90 \delta( \OPT+|\cM|)+16\delta\OPT + 1\\
                                          & \leq \frac{4}{3} \cdot \bx^0 \cdot \one_{\cC\setminus \cC_2} +110\delta\cdot  \OPT +  \left(\frac{4}{3} + 90 \delta \right)|\cM|\\
                                          & \leq \frac{4}{3} \cdot \bx^0 \cdot \one_{\cC\setminus \cC_2} +110\delta\cdot  \OPT +  \left(\frac{4}{3} + 90 \delta \right)\cdot \left( \bx^0\cdot \one_{\cC_2} + \delta^2 \OPT \right)\\
                                          & \leq \left( \frac{4}{3} + 90\delta\right) \|\bx^0\| +110\delta \OPT +90\delta^3 \OPT\\
                                          & \leq\left(  \frac{4}{3} +250\delta\right) \OPT,
     \end{aligned}
  \end{equation*}
  where the last inequality uses $\|\bx\|^0\leq (1+\delta^2)\OPT$.
\end{proof}

\subsubsection{Concentration}
\label{sec:rnr_deferred}
\label{sec:concentration}
In this section we give the missing proofs of \Cref{sec:prob_space} and~\Cref{sec:main_rnr}.

\begin{proof}[{\bf Proof of \Cref{lem:iterative_mcdiarmid}}]
  Let $\cS= \{ \bu^1,\ldots ,\bu^{\floor{\varphi(\delta)}}\}$, where $\bu^\ell$ is an $\cF_{j }$-measurable random vector for~$\ell \in [\varphi(\delta)]$ (in case $|\cS|<\floor{\varphi(\delta)}$ the same vector may appear several times in $\bu^1,\ldots ,\bu^{\floor{\varphi(\delta)}}$).
  As $\cS$ is a $(\delta,\varphi(\delta))$ linear structure, it holds that
  \begin{equation*}
    \begin{aligned}
      \Pr&\left(\forall  r = j,\hdots,k:~~
      \OPTf\left(\bd \wedge \one_{S_r} \right)\leq (1-\delta)^{r-j} (1+10\delta ) \|\blam\| + \varphi(\delta) + \delta^{10}\cdot\OPT\right)\\
      & \geq \Pr \left(\forall r = j,\hdots,k,\ell =1,\hdots,\varphi(\delta):~~(\one_{S_r}\wedge \bd) \cdot \bu^{\ell}
	    \leq (1-\delta)^{r-j}\cdot \bw \cdot \bu^{\ell} +\frac{\OPT}{\varphi^{10}(\delta)}\cdot   \tol(\bu^{\ell})\right)\\
      & \geq \Pr\left(
	  \begin{aligned}
        & \forall\ell = 1,\hdots,\varphi(\delta ):~~~&&(\one_{S_j}\wedge \bd)\cdot \bu^{\ell} \leq \bw \cdot \bu^{\ell}+\frac{1}{\varphi^{11}(\delta)}\cdot  \OPT \cdot \tol(\bu^{\ell})\\
        & \forall\ell = 1,\hdots,\varphi(\delta), ~r = j,\hdots,k:&& 
         (\one_{S_r}\wedge \bd) \cdot \bu^{\ell}  \leq (1-\delta)^{r-j}\cdot (\one_{S_j}\wedge \bd) \cdot \bu^{\ell} +\frac{\OPT}{\varphi^{11}(\delta)}\cdot   \tol(\bu^{\ell})
      \end{aligned}
      \right)\\
      &\geq \xi- \sum_{\ell=1}^{\floor{\varphi(\delta)}} \Pr\left(\exists r\in \{j,\hdots,k\}: (\one_{S_r}\wedge \bd) \cdot \bu^{\ell}  > (1-\delta)^{r-j}\cdot (\one_{S_j}\wedge \bd) \cdot \bu^{\ell} +\frac{\OPT}{\varphi^{11}(\delta)}\cdot   \tol(\bu^{\ell}) \right) \\
      &\geq \xi -\varphi(\delta)\cdot \delta^{-2} \cdot \exp\left(- \frac{2\cdot \delta^4 \cdot \left(\frac{\OPT^2}{\varphi^{11}(\delta)}\right)^2 }{\OPT} \right)\\
      &\geq \xi -\varphi^2(\delta )\cdot \exp\left( -\frac{\OPT}{ \varphi^{25}(\delta)}\right) \enspace .
    \end{aligned} 
  \end{equation*}
  The fourth equality follows from the union bound and the definition of $\xi$ in \eqref{eq:xi_prob}.
  The fifth inequality is by \Cref{lem:concentration_multistep_prelim}. 
\end{proof} 

The following technical lemma will be used to prove \Cref{lem:concentration_Tjh}.
\begin{lemma}
\label{lem:tjh_specific}
  Let $j\in \{1,\hdots,k\}$ and $h \in \{2,\hdots,  2\cdot \delta^{-1}\}$.
  Then
  \begin{equation*}
    \Pr\bigg( \Big|\E\left[\bx^* \cdot \one_{T_{j,h}}~\middle| ~\cF_{j-1} \right]-\bx^* \cdot \one_{T_{j,h}} \Big|>  \delta^{20} \cdot \OPT\bigg) 
	\leq 2\cdot \exp\left(-\delta^{50  }\cdot \OPT\right) \enspace .
 \end{equation*}
\end{lemma}
\begin{proof}
  Let $\mathcal{V}\subseteq [0,1]^{\cC}$ be the set of values that $\bx^*$ can take, that is, $\mathcal{V} =\{ \bx^*(\omega)~|~\omega\in \Omega\}$.
  Since $\Omega$ is finite, it follows that $\mathcal{V}$ is finite as well.
  Furthermore, since $\sum_{C\in \cC} \bx^*_C \cdot C(i) \leq \sum_{C\in \cC} \bx^0_C \cdot C(i) =1$ for every $i\in I$, it follows that $\sum_{C\in \cC} \bx_C\cdot C(i)\leq 1$ for every $\bx\in \mathcal{V}$ and $i\in I$. 

  For any $U\subseteq \cC$, $\rho = 1,\hdots,\OPT$ and $\bx \in \mathcal{V}$ define $f_{U,\rho,\bx}:\cC^{\OPT} \rightarrow \mathbb{R}$ by
  \begin{equation*}
    f_{U,\rho,\bx}\left(C_1,\ldots, C_{\OPT}\right) = \bx \cdot \one_{\{C\in U~|~C\cap \left( \bigcup_{\ell=1}^{\rho} C_{\ell}\right)\cap L \neq \emptyset\}}
	                                                = \sum_{C\in U} \bx_C \cdot \one_{C\cap \left( \bigcup_{\ell=1}^{\rho} C_{\ell}\right)\cap L \neq \emptyset} \enspace .
  \end{equation*}
  Define $D=\{ f_{U,\rho,\bx}~|~U\subseteq \cC,~\rho = 1,\hdots,\OPT, ~\bx \in\mathcal{V}\}$.
  It follows that $D$ is a finite set.

  Let $f_{U,\rho,\bx}\in D$, $(C_1,\ldots, C_{\OPT}),~(C'_1,\ldots, C'_{\OPT})\in \cC^{\OPT}$, and $r = 1,\hdots,\OPT$ such that $C_{\ell} = C'_{\ell}$ for $\ell = 1,\hdots,r-1,r+1,\hdots,\OPT$.
  If $r>\rho$ then $\left|f_{U,\rho, \bx} (C_1,\ldots, C_\OPT) -f_{U,\rho,\bx}(C'_1,\ldots, C'_{\OPT})\right| = 0$.
  Otherwise, let $T= \bigcup_{\ell \in \{1,\hdots,\rho\}\setminus\{r\}} C_\ell = \bigcup_{\ell \in \{1,\hdots,\rho\}\setminus\{r\}} C'_\ell$.
  It holds that 
  \begin{equation*}
    \begin{aligned}
      \bigg|f_{U,\rho, \bx}& (C_1,\ldots, C_\OPT) -f_{U,\rho,\bx}(C'_1,\ldots, C'_{\OPT})\bigg| \\
	  & = \left| \bx \cdot \left(  \one_{\{C\in U~|~C\cap \left( T\cup C_r\right)\cap L \neq \emptyset\}}- \one_{\{C\in U~|~C\cap \left( T\cup C'_r\right)\cap L \neq \emptyset\}} \right) \right|\\
	  & = \left| \sum_{C\in U} \bx_C \cdot \one_{C\cap (T\cup C'_r) \cap L = \emptyset} \cdot \one_{C\cap C_r\cap L \neq \emptyset } - \sum_{C\in U} \bx_C \cdot \one_{C\cap (T\cup C_r) \cap L = \emptyset} \cdot \one_{C\cap C'_r\cap L \neq \emptyset } \right|\\
	  & \leq \max\left\{\sum_{C\in U} \bx_C \cdot \one_{C\cap (T\cup C'_r) \cap L = \emptyset} \cdot \one_{C\cap C_r\cap L \neq \emptyset }, ~ \sum_{C\in U} \bx_C \cdot \one_{C\cap (T\cup C_r) \cap L = \emptyset} \cdot \one_{C\cap C'_r\cap L \neq \emptyset }  \right\}\\
      & \leq  \max\left\{\sum_{C\in \cC} \bx_C  \cdot \one_{C\cap C_r\cap L \neq \emptyset }, ~ \sum_{C\in \cC} \bx_C \cdot \one_{C\cap C'_r\cap L \neq \emptyset }  \right\} \enspace .
    \end{aligned}
  \end{equation*}
  Furthermore,
  \begin{equation*}
    \sum_{C\in \cC} \bx_C \cdot \one_{C\cap C_r\cap L\neq \emptyset} \leq \sum_{C\in \cC} \bx_C \sum_{i\in C_r\cap L} C(i)
	                                                                    = \sum_{i\in C_r \cap L} \sum_{C\in \cC} \bx_C\cdot C(i)
																	 \leq |C_r\cap L |\leq 2\cdot \delta^{-1},
  \end{equation*}
  and by a symmetric argument $\sum_{C\in \cC} \bx_C \cdot \one_{C\cap C'_r\cap L\neq \emptyset}  \leq2\cdot\delta^{-1}$.
  Thus,
  \begin{equation*}
    \bigg|f_{U,\rho, \bx} (C_1,\ldots, C_\OPT) -f_{U,\rho,\bx}(C'_1,\ldots, C'_{\OPT})\bigg| \leq 2\cdot \delta^{-1}\enspace.
  \end{equation*}
  That is, all functions in $D$ are of $(2\delta^{-1})$-bounded difference.

  Define $g = f_{U_{j-1,h},\rho_j, \bx^*}$. Since $U_{j-1,h}$, $\rho_j$ and $\bx^*$ are $\cF_{j-1}$-measurable, we have that $g$ is a $\cF_{j-1}$-measurable random function. For every $C\in \cC$ it holds that $C\in T_{j,h}$ if and only if $C\in U_{j-1,h}$ and $C\cap L \cap \left( \bigcup_{\ell \in [\rho_j]} C^j_{\ell} \right)\neq \emptyset  $. Thus,
  \begin{equation*}
    g(C^j_1,\ldots, C^j_{\OPT}) =  \bx^* \cdot \one_{\{C\in U_{j-1,h}~|~C\cap \left( \bigcup_{\ell=1}^{\rho_j} C^j_{\ell}\right)\cap L \neq \emptyset\}}= \bx^* \cdot \one_{T_{j,h}} \enspace .
  \end{equation*}
  Therefore,
  \begin{equation*}
    \begin{aligned}
     \Pr&\bigg( \Big|\E\left[ \bx^* \cdot \one_{T_{j,h}}~\middle| ~\cF_{j-1} \right]-\bx^* \cdot \one_{T_{j,h}} \Big|>  \delta^{20} \cdot \OPT\bigg)\\
     & = \Pr\bigg( \Big|\E\left[ g(C^j_1,\ldots, C^j_{\OPT})~\middle| ~\cF_{j-1} \right]- g(C^j_1,\ldots, C^j_{\OPT}) \Big|>  \delta^{20} \cdot \OPT\bigg)\\
     & = \Pr\bigg( \E\left[ g(C^j_1,\ldots, C^j_{\OPT})~\middle| ~\cF_{j-1} \right]- g(C^j_1,\ldots, C^j_{\OPT}) >  \delta^{20} \cdot \OPT\bigg)\\
     &~~~~+ \Pr\bigg( \E\left[ -g(C^j_1,\ldots, C^j_{\OPT})~\middle| ~\cF_{j-1} \right]+ g(C^j_1,\ldots, C^j_{\OPT}) > \delta^{20} \cdot \OPT\bigg)\\
	 &\leq 2 \cdot \exp\left( -\frac{2\cdot \delta^{40} \cdot \OPT^2}{ (2\delta^{-1})^2 \cdot \OPT} \right)\leq 2 \cdot \exp\left( -\delta^{50} \cdot \OPT \right),
    \end{aligned} 
  \end{equation*}
  where the  inequality follows from \Cref{lem:Generalized_McDiarmid}.
\end{proof}

The proof of \Cref{lem:concentration_Tjh} follows directly from \Cref{lem:tjh_specific}. 
\begin{proof}[{\bf Proof of \Cref{lem:concentration_Tjh}}]
  By the union bound, we have
  \begin{equation*}
    \begin{aligned}
	\Pr&\left( \forall j=1,\hdots,k, h=2,\hdots,2\cdot \delta^{-1}:~\bigg|\E\left[ \bx^* \cdot \one_{T_{j,h}}~\middle| ~\cF_{j-1} \right]-\bx^* \cdot \one_{T_{j,h}} \bigg|\leq  \delta^{20} \cdot \OPT\right)  \\
	&\geq  1- \sum_{j=1}^k \sum_{h=2}^{2\cdot \delta^{-1}} \Pr \left(~\bigg|\E\left[ \bx^* \cdot \one_{T_{j,h}}~\middle| ~\cF_{j-1} \right]-\bx^* \cdot \one_{T_{j,h}} \bigg|>  \delta^{20} \cdot \OPT  \right)\\
	&\geq 1- k\cdot 2\cdot \delta^{-1} \cdot 2\cdot \exp\left(-\delta^{50}\cdot \OPT \right)\\
	&\geq 1- \delta^{-10} \cdot \exp(-\delta^{50}\cdot \OPT),
	\end{aligned}
  \end{equation*}
  where the second inequality follows from \Cref{lem:tjh_specific} and the last inequality uses $k\leq \delta^{-2}$. 
\end{proof}

We use \Cref{lem:concentration_Tjh} to prove \Cref{lem:Ujh_lowerbound}.
\begin{proof}[{\bf Proof of \Cref{lem:Ujh_lowerbound}}]
  For every $\eps\in (0,0.1)$ and $h\in \mathbb{N}$, it holds that $\lim_{z\rightarrow \infty} \left( 1-\frac{h}{z}\right)^{\ceil{-z\cdot \ln(1-\eps)}} = (1-\eps)^h$; thus, there is $M_{\eps,h}>1$ such that for every $z>M_{\eps, h}$ it holds that $\left( 1-\frac{h}{z}\right)^{\ceil{-z\cdot \ln(1-\eps)}} \geq  (1-\eps)^h-\eps^{20}$. 
  Define $\mu:(0,0.1)\rightarrow \mathbb{R}_+$ by $\mu(\eps) = \max\left\{ M_{\eps, h} ~|~h\in [2,2\cdot\eps^{-1}]\cap \mathbb{N} \right\}$ for every $\eps\in (0,0.1)$.
  Note that since the maximum is taken over a finite set of numbers, each greater than one, it follows that $\mu(\eps)\in (1,\infty)$ for every $\eps\in (0,0.1)$. 
		
  Assume the event in \eqref{eq:tjh_bound} occurs.
  Let $j\in \{1,\hdots,k\}$ and $h\in\{2,\hdots,2\delta^{-1}\}$.
  For any $C\in \cC_h$ it holds that
  \begin{equation}
  \label{eq:ujh_bound}
    \begin{aligned}
	  \Pr&\left(C\in U_{j,h}~|~\cF_{j-1}\right) = \one_{C\in U_{j-1,h}}\cdot  \Pr\left( \forall \ell \in 1,\hdots,\rho_j:~C^j_{\ell} \cap C \cap L=\emptyset ~\middle|~\cF_{j-1}\right)\\
      &= \one_{C\in U_{j-1,h}}\cdot  \left( 1-\frac{\sum_{C'\in \cC} \bx^j_{C'}\cdot \one_{C'\cap C\cap L\neq \emptyset}}{z_j}\right)^{\ceil{ -z_j\cdot \ln(1-\delta)}}\\
      &\geq \one_{C\in U_{j-1,h}}\cdot   \left(1-\frac{h}{z_j}\right)^{ \ceil{-z_j\cdot \ln(1-\delta)}}\\
      &\geq \one_{\OPTf(\one_{S_{j-1}}) > \mu(\delta)}\cdot  \one_{C\in U_{j-1,h}} \cdot \left( (1-\delta)^{h}-\delta^{20} \right) \enspace .
    \end{aligned} 
  \end{equation}
  The first inequality holds, since, for every $C\in \cC$,
  \begin{equation*}
    \sum_{C'\in \cC} \bx^j_{C'}\cdot \one_{C'\cap C\cap L\neq \emptyset}\leq \sum_{C'\in \cC} \bx^j_{C'}\cdot \sum_{i\in C\cap L} C'(i) = \sum_{i\in C\cap L} \sum_{C'\in \cC} \bx^j_{C'}\cdot C(i) \leq h \enspace .
  \end{equation*}
  The last inequality in \eqref{eq:ujh_bound} holds by definition of $\mu$ and since  $z_j \geq \OPTf(\one_{S_{j-1}})$.
		
  We therefore have
  \begin{equation*}
    \begin{aligned}
      \one_{U_{j,h}} \cdot \bx^* & =    \one_{U_{j-1,h} }\cdot \bx^* - \one_{T_{j,h}}\cdot \bx^*\\
			                     & \geq \one_{U_{j-1,h} } \cdot \bx^* -\E\left[ \one_{T_{j,h}}\cdot \bx^*~|~\cF_{j-1} \right] -\delta^{20}\cdot \OPT\\
			                     & =    \E\left[ \one_{U_{j,h}}\cdot \bx^*~|~\cF_{j-1} \right] -\delta^{20}\cdot \OPT\\
			                     & \geq \one_{\OPTf(\one_{S_{j-1}})> \mu(\delta)} \cdot \one_{U_{j-1,h}} \cdot  \bx^* \left( (1-\delta)^h - \delta^{20} \right)-\delta^{20}\cdot \OPT\\
			                     & \geq \one_{\OPTf(\one_{S_{j-1}})> \mu(\delta)} \cdot \one_{U_{j-1,h}}  \cdot \bx^* \cdot (1-\delta)^{h} - \delta^{19}\cdot \OPT \enspace .
    \end{aligned}
  \end{equation*}
  The first inequality is due to~\eqref{eq:tjh_bound}, the second inequality follows from~\eqref{eq:ujh_bound}, and the last inequality uses $\one_{U_{j-1,h}} \cdot \bx^* \leq \|\bx^*\| \leq \|\bx^0\|\leq 2\OPT$. 
  Overall, we showed that 
  \begin{equation}
  \label{eq:ujh_step_lowerbound}
    \one_{U_{j,h}} \cdot \bx^*   \geq \one_{\OPTf(\one_{S_{j-1}})> \mu(\delta)} \cdot \one_{U_{j-1,h}}  \cdot \bx^* \cdot (1-\delta)^{h} - \delta^{19}\cdot \OPT
  \end{equation}
  for $j = 1,\hdots,k$ and $h = 2,\hdots,2\cdot\delta^{-1}$.
		
  \begin{claim}
  \label{clim:ujh_induction}
    For $h = 2,\hdots,2\cdot \delta^{-1}$ and $j = 0,1,\ldots,k$ it holds that
    \begin{equation*}
      \bx^*\cdot \one_{U_{j,h}} \geq (1-\delta)^{h\cdot j} \cdot \bx^* \cdot \one_{U_{0,h}} -j\cdot \delta^{19}  \cdot \OPT \textnormal{ or }   \OPTf(\one_{S_j})\leq \mu(\delta) \enspace .
    \end{equation*}
  \end{claim}
  \begin{claimproof} 
    Fix $h \in \{2,\hdots,2\cdot \delta^{-1}\}$.
    We show the claim by induction over $j$. 

	\noindent{\bf Base case:}
    For $j=0$ it clearly holds that $\bx^* \cdot \one_{U_{0,h}} \geq (1-\delta)^{h\cdot 0 }\cdot \bx^*\cdot \one_{U_{0,h}} - 0\cdot \delta^{19}\cdot \OPT$. 
			
    \noindent{\bf Induction step:} Assume the induction hypothesis holds for $j-1$.
    If $\OPTf(\one_{S_j})\leq\mu(\delta)$ then the statement holds for $j$.
    Otherwise, $\OPTf(\one_{S_j})>\mu(\delta)$, and so \mbox{$\OPTf(\one_{S_{j-1}})\geq \OPTf(\one_{S_j})>\mu(\delta)$}.
    By the induction hypothesis, we have
    \begin{equation}
    \label{eq:induction}
      \bx^*\cdot \one_{U_{j-1,h}} \geq (1-\delta)^{h\cdot (j-1)} \cdot \bx^* \cdot \one_{U_{0,h}} -(j-1)\cdot \delta^{19}  \cdot \OPT \enspace .
    \end{equation}
    Therefore, 
    \begin{equation*}
      \begin{aligned}
		\bx^*\cdot \one_{U_{j,h}} & \geq \one_{\OPTf(\one_{S_{j-1}})> \mu(\delta)} \cdot \one_{U_{j-1,h}}  \cdot \bx^* \cdot (1-\delta)^{h} - \delta^{19}\cdot \OPT\\
				                  & =    \one_{U_{j-1,h}}  \cdot \bx^* \cdot (1-\delta)^{h} - \delta^{19}\cdot \OPT\\
                                  & \geq (1-\delta)^h \left( (1-\delta)^{h\cdot (j-1)} \cdot \bx^* \cdot \one_{U_{0,h}} -(j-1)\cdot \delta^{19}  \cdot \OPT\right) -\delta^{19} \cdot \OPT\\
				                  & \geq (1-\delta)^{h\cdot j }\cdot \bx^* \cdot \one_{U_{0,h}} -j\cdot \delta^{19}\cdot \OPT \enspace .
      \end{aligned}
    \end{equation*}
	The first inequality holds by \eqref{eq:ujh_step_lowerbound}, and the second inequality is by~\eqref{eq:induction}.  
  \end{claimproof}
		
  By \Cref{clim:ujh_induction}, for $j = 1,\hdots,k$ and $h = 2,\hdots,2\cdot \delta^{-2}$, either $\OPTf(\one_{S_j}) \leq \mu(\delta)$, or 
  \begin{equation*}
    \bx^*\cdot \one_{U_{j,h}} \geq (1-\delta)^{h\cdot j} \cdot \bx^* \cdot \one_{U_{0,h}} -j\cdot \delta^{19}  \cdot \OPT
	                          \geq (1-\delta)^{h\cdot j} \cdot \bx^* \cdot \one_{U_{0,h}} - \delta^{10}  \cdot \OPT,
  \end{equation*}
  as required (the last inequality uses $j\leq k \leq \delta^{-2}$).
  Since we assumed \eqref{eq:tjh_bound} occurs, this property holds with probability at least $1-\delta^{-10}\cdot \exp(-\delta^{50}\cdot \OPT)$ by \Cref{lem:concentration_Tjh}.
\end{proof}

We now proceed to the proof of \Cref{claim:concentration_gen_iteration}.
We use the same notation as in the proof of \Cref{lem:main_rnr}, where the claim is stated.
\begin{claimproof}[{\bf Proof of \Cref{claim:concentration_gen_iteration}}]
  As in the proof of \Cref{lem:tjh_specific}, let $\mathcal{V}\subseteq [0,1]^{\cC}$ be all the values $\bx^*$ can take (formally, $\mathcal{V} =\{ \bx^*(\omega)~|~\omega\in \Omega\}$).
  It follows that $\sum_{C\in \cC} \bx_C\cdot C(i)\leq 1$ for every $i\in I$ and $\bx\in \mathcal{V}$. 
  Also, let $A\subseteq \mathbb{R}_{\geq 0}^I$ be the set of all values the vectors in $\cS_j$ can take (formally, $A=\{ \bu~|~\exists \omega\in \Omega: ~\bu\in \cS_j(\omega) \} $)
  As $\Omega$ is finite, it follows that $\mathcal{V}$ and $A$ are finite. 
  
  For any $U\subseteq \cC$, $S\subseteq I$, $\bx\in \mathcal{V}$, $\rho\in [\OPT]$ and $\bu \in A$, we define $f_{U,S,\bx, \rho, \bu }:\cC^{\OPT}\rightarrow \mathbb{R}$ by 
  \begin{equation*}
    f_{U,S,\bx, \rho, \bu }(C_1,\ldots, C_{\OPT}) =
    \begin{cases}
  	  \displaystyle
  	  \frac{1}{\tol(\bu)}\cdot \sum_{C\in U} \bx_C \cdot \one_{C\cap \left( \bigcup_{\ell\in [\rho] } C_{\ell }\right)\cap L\neq \emptyset} \cdot \sum_{i\in C\cap S} \one_{i\notin \bigcup_{\ell = 1,\hdots,\rho} C_{\ell} }\cdot  \bu_i~~~  	&  \tol(\bu) \neq 0	\\
  	  0& \textnormal{otherwise}
    \end{cases}
  \end{equation*}
  Let $D=\{f_{U,S,\bx, \rho ,\bu}~|~U\subseteq \cC,~S\subseteq I, ~\bx \subseteq \mathcal{V},~\rho \in [\OPT],~ \bu \in A\}$.
  It follows that $D$ is finite. 
  
  Let $f_{U,S, \bx, \rho, \bu} \in D$,  $(C_1,\ldots, C_{\OPT}),~(C'_1,\ldots, C'_{\OPT})\in \cC^{\OPT} $ and $r\in \{1,\hdots,\OPT\}$ be such that $C_{\ell} =C'_{\ell}$ for $\ell = 1,\hdots,r-1,r+1,\hdots\OPT$.
  If $\tol(\bu)=0$ or $r>\rho$,
  \begin{equation*}
    \left|f_{U,S, \bx, \rho, \bu} (C_1,\ldots, C_\OPT) -f_{U,S, \bx, \rho, \bu}(C'_1,\ldots, C'_{\OPT})\right| =0 \enspace .
  \end{equation*}
  Otherwise, let $T=\bigcup_{\ell \in \{1,\hdots,\rho\}\setminus \{r\}} C_{\ell }=  \bigcup_{\ell \in \{1,\hdots,\rho\}\setminus \{r\}} C'_{\ell}$.
  Then
  \begin{equation*}
    \begin{aligned}
  	  \bigg| &f_{U,S, \bx, \rho, \bu}  (C_1,\ldots, C_\OPT) -f_{U,S, \bx, \rho, \bu} (C'_1,\ldots, C'_{\OPT})\bigg|\\
      &= \frac{1}{\tol(\bu)} \cdot \bigg|\sum_{C\in U} \bx_C \cdot \one_{C\cap \left( T\cup C_r\right)\cap L\neq \emptyset} \cdot \sum_{i\in C\cap S} \one_{i\notin T\cup C_r}\cdot  \bu_i
- \sum_{C\in U} \bx_C \cdot \one_{C\cap \left( T\cup C'_r\right)\cap L\neq \emptyset} \cdot \sum_{i\in C\cap S} \one_{i\notin T\cup C'_r}\cdot  \bu_i\bigg| \\
      &= \frac{1}{\tol(\bu)} \cdot \bigg|  \sum_{C\in U} \sum_{i\in C\cap S} \bx_C\cdot \bu_i \cdot \left(\one_{C\cap \left( T\cup C_r\right)\cap L\neq \emptyset}\cdot \one_{i\notin T\cup C_r} 
  -\one_{C\cap \left( T\cup C'_r\right)\cap L\neq \emptyset}\cdot \one_{i\notin T\cup C'_r}\right)\bigg|\\
  	&\leq \frac{1}{\tol(\bu)} \cdot  \sum_{C\in U} \sum_{i\in C\cap S} \bx_C\cdot \bu_i \cdot \left| \one_{C\cap \left( T\cup C_r\right)\cap L\neq \emptyset}\cdot \one_{i\notin T\cup C_r} 
  	-\one_{C\cap \left( T\cup C'_r\right)\cap L\neq \emptyset}\cdot \one_{i\notin T\cup C'_r} \right|\\
  	&\leq \frac{1}{\tol(\bu)} \cdot  \sum_{C\in U} \sum_{i\in C\cap S} \bx_C\cdot \bu_i\cdot  \left( \one_{C\cap (C'_r\cup C_r) \cap L \neq \emptyset } + \one_{i\in C_r\cup C'_r}\right)\\
& \leq \frac{1}{\tol(\bu)} \sum_{C\in U} \one_{C\cap (C'_r\cup C_r) \cap L \neq \emptyset }\cdot \bx_C \cdot \sum_{i\in C} \bu_i ~ +~ \frac{1}{\tol(\bu)}\cdot \sum_{i\in C_r \cup C'_r}\bu_i \cdot \sum_{C\in U} \bx_C\cdot C(i)\\
	&\leq \frac{1}{\tol(\bu)} \cdot \tol(\bu)\cdot  4\cdot \delta^{-1}  + \frac{1}{\tol(\bu)} \sum_{i\in C_r\cup C'_r} \bu_i\\
	&\leq 4\cdot \dot \delta^{-1} + \frac{1}{\tol(\bu)}\cdot 2\cdot \tol(\bu)\\
	&\leq \delta^{-2},
	 \end{aligned}
  \end{equation*}
  where the fourth inequality uses
  \begin{equation*}
    \sum_{C\in U} \one_{C\cap (C'_r\cup C_r) \cap L \neq \emptyset } \cdot \bx_C\leq \sum_{i\in (C_r\cup C'_r) \cap L } \sum_{C\in \cC} \bx_C\cdot C(i)\leq \sum_{i\in (C_r\cup C'_r) \cap L } 1 \leq 4\cdot \delta^{-1} \enspace .
  \end{equation*}
  We conclude that all functions in $D$ are of $\delta^{-2}$-bounded difference. 
  
  Recall $\cS_j$ is a $(\delta, \varphi(\delta))$-linear structure of $\blam^j$.
  Since $\blam^j$ is $\cF_{j-1}$-measurable, it follows that $\cS_j$ is also $\cF_{j-1}$-measurable. 
  As in the proof of \Cref{lem:iterative_mcdiarmid}, we denote $\cS_j = \{ \bu^1,\ldots ,\bu^{\floor{\varphi(\delta)}}\}$ where $\bu^s$ is an $\cF_{j-1}$-measurable random vector for $s = 1,\hdots,\floor{\varphi(\delta)}$ (in case $|\cS_j|<\floor{\varphi(\delta)}$ the same vector may appear several times in $\bu^1,\ldots ,\bu^{\floor{\varphi(\delta)}}$).  
  
  For $s = 1,\hdots,\floor{\varphi(\delta)}$ define a random function $g^s = f_{U_{j-1}, S_{j-1}, \bx^*,\rho_j,\bu^s}$.
  Since $U_{j-1},S_{j-1},\bx^*,\rho_j$ and~$\bu^s$ are all $\cF_{j-1}$-measurable, it follows that $g^s$ is $\cF_{j-1}$-measurable as well. 
  Furthermore,
  \begin{equation*}
    \begin{aligned}
      \tol(\bu^s)&\cdot g^s(C^j_1,\ldots, C^j_{\OPT} ) = \sum_{C\in U_{j-1}} \bx^*_C \cdot \one_{C\cap \left( \bigcup_{\ell = 1,\hdots,\rho_j} C^j_{\ell }\right)\cap L\neq \emptyset} \cdot \sum_{i\in C\cap S_{j-1}} \one_{i\notin \bigcup_{\ell \in [\rho_j]} C^j_{\ell} }\cdot  \bu^s_i\\
      &=\sum_{i\in I}\one_{i\in S_{j} } \cdot \bu^s_i \cdot \sum_{C\in T_j} \bx^*_C \cdot C(i) = \sum_{i\in I} \one_{i\in S_{j} } \cdot \bu^s_i \cdot \bd^j_i = (\one_{S_j} \wedge \bd^j)\cdot \bu^s,
    \end{aligned}
  \end{equation*}
  where the third equality follows from the definition of $\bd^j$.
  Thus, for $s = 1,\hdots,\floor{\varphi(\delta)}$ it holds that
  \begin{equation*}
    \begin{aligned}
      \Pr&\left( (\one_{S_j}\wedge \bd^j)\cdot \bu^s ~> ~\E\left[\bu^s \cdot \left( \bd^j \wedge \one_{ S_j } \right)~\middle|~\cF_{j-1} \right]+\frac{\OPT}{\varphi^{11}(\delta)}   \cdot \tol(\bu)\right)\\
      &=\Pr\left(g^s(C^j_1,\ldots,C^{j}_{\OPT}) ~> ~\E\left[g^s(C^j_1,\ldots,C^{j}_{\OPT})~\middle|~\cF_{j-1} \right]+\frac{\OPT}{\varphi^{11}(\delta)}  \right)\\
      &\leq \exp\left( -\frac{2\cdot \left( \frac{\OPT}{\varphi^{11}(\delta)}\right)^2 }{\delta^{-4} \cdot \OPT}\right)\leq \exp \left(-\frac{\OPT}{\varphi^{25}(\delta)}\right),
    \end{aligned}
  \end{equation*}
  where the last inequality is by \Cref{lem:Generalized_McDiarmid}.
  
  Thus, using the union bound we have that
  \begin{equation*}
    \begin{aligned}
      \Pr&\left( \forall \bu \in \cS_j:~(\one_{S_j}\wedge \bd^j)\cdot \bu ~\leq ~\E\left[\bu \cdot \left( \bd^j \wedge \one_{ S_j } \right)~\middle|~\cF_{j-1} \right]+\frac{\OPT}{\varphi^{11}(\delta)}\cdot \tol(\bu)\right)\\
      &\geq 1-\sum_{s=1}^{\floor{\varphi(\delta)}}\Pr\left( (\one_{S_j}\wedge \bd^j)\cdot \bu^s ~> ~\E\left[\bu^s \cdot \left( \bd^j \wedge \one_{ S_j } \right)~\middle|~\cF_{j-1} \right]+\frac{\OPT}{\varphi^{11}(\delta)}   \cdot \tol(\bu)\right)\\
      &\geq 1-\varphi(\delta )\cdot \exp\left( -\frac{\OPT}{\varphi^{25}(\delta)}\right) \enspace . 
    \end{aligned} 
  \end{equation*}
\end{claimproof}

It remains to prove \Cref{lem:M_concentration} and \Cref{claim:concentration_first_iteration}.
We use $G=(L,E)$ to denote the $\delta$-matching graph of $(I,v)$, and $P_{\cM}(G)$ to denote the matching polytope of $G$.
Both proofs rely on the concentration bounds of $\CVZ$ given below.  
\begin{lemma}[\cite{ChekuriVZ2011}]
\label{lem:cvz11}
  Let $\bar{\beta}\in P_{\cM}(G)$ and $\gamma>0$.
  Also, denote $\cM=\CVZ(\bar{\beta}, \gamma)$.
  Then $\cM$ is a matching, and for any $\ba\in [0,1]^{E}$ the following holds:
  \begin{enumerate}
    \item $\Pr(e\in \cM)= (1-\gamma )\bar{\beta}_e$ for any $e\in E$. 
    \item For any $\xi\leq \E\left[ \sum_{e\in\cM} \ba_e\right]$ and $\eps>0$, it holds that $\Pr\left(\sum_{e\in\cM} \ba_e\leq (1-\eps)\cdot \xi \right)\leq \exp\left(-\frac{\xi \cdot \eps^2 \cdot \gamma}{20}\right)$.
    \item For any $\xi\geq \E\left[ \sum_{e\in\cM} \ba_e\right]$ and $\eps>0$, it holds that $\Pr\left(\sum_{e\in\cM} \ba_e\geq (1+\eps)\cdot \xi \right)\leq \exp\left(-\frac{\xi \cdot \eps^2 \cdot \gamma}{20}\right)$.
  \end{enumerate}
\end{lemma}

\begin{proof}[{\bf Proof of \Cref{lem:M_concentration}}] 
  As $\cM = \CVZ(\cE(\bx^0),\delta^4)$, it follows that 
  \begin{equation*}
    \E\left[|\cM|\right] = \sum_{e\in E}  \Pr(e\in \cM ) = (1-\delta^4)\cdot\sum_{e\in E} \cE_e(\bx^0) = (1-\delta^4)\cdot \sum_{e\in E} \sum_{C\in \cC \textnormal{ s.t. } e\in C} \bx^0_C = (1-\delta^4) \cdot \one_{\cC_2}\cdot \bx^0 \enspace .
  \end{equation*}
  If $\one_{\cC_2}\cdot \bx^0=0$, then $|\cM|=0$, and the statement of the lemma holds. 

  Otherwise, by \Cref{lem:cvz11},
  \begin{equation*}
    \begin{aligned}
      \Pr&\left( |\cM|>\one_{\cC_2}\cdot \bx^0 +\delta^2\cdot \OPT \right) = \Pr\left( |\cM| >\one_{\cC_2}\cdot  \bx^0 \cdot \left(1+ \frac{\delta^2\cdot \OPT}{\one_{\cC_2}\cdot \bx^0} \right) \right)\\
      &\leq \exp\left( -\frac{1}{20}\cdot \delta^4\cdot ( \one_{\cC_2}\cdot \bx^0)\cdot \left(\frac{\delta^2\cdot \OPT}{\one_{\cC_2}\cdot \bx^0} \right)^2 \right) \leq \exp\left( -\delta^{10} \cdot \OPT\right),
    \end{aligned} 
  \end{equation*}
  where the last inequality uses $\one_{\cC_2}\cdot \bx^0 \leq (1+\delta^2)\OPT \leq 2\OPT$.
  Therefore,
  \begin{equation*}
    \Pr\left( |\cM|\leq \one_{\cC_2}\cdot \bx^0 +\delta^2\cdot \OPT \right) \geq 1-  \exp\left( -\delta^{10} \cdot \OPT\right) \enspace .\qedhere
  \end{equation*}
\end{proof} 

\begin{claimproof}[{\bf Proof of \Cref{claim:concentration_first_iteration}}]
  We use the same notation as in the proof of \Cref{lem:undercovered_simple_and_small}, where the claim is stated.
  If $\tol(\bu)=0$ the claim trivially  holds.
  Thus, we may assume that $\tol(\bu)\neq \emptyset$.
	
  Observe that
  \begin{equation*}
	\begin{aligned}
	  \bw\cdot \bu - \bd\cdot \bu = \sum_{i\in I}\left(\bw_i -\bd_i\right) \bu_i
	                              = \sum_{i\in L}\left( \one_{i\notin S_0} -(1-\delta^4)\cdot \by^{\cM}_i \right) \bu_i  = \sum_{i\in L} \one_{i\notin S_0} \cdot \bu_i-\E\left[\sum_{i\in L} \one_{i\notin S_0} \cdot \bu_i \right],
	  \end{aligned} 
  \end{equation*}
  where the second equality is by \eqref{eq:small_items_diff} and \eqref{eq:large_items_diff}, and the last equality is by \Cref{lem:match_prob}.
  Furthermore,
  \begin{equation*}
    \sum_{i\in L}\one_{i\notin S_0} \cdot \bu_i = \sum_{ \{i_1, i_2\} \in \cM} \left( \bu_{i_1} + \bu_{i_2} \right) \enspace .
  \end{equation*}
  Thus, 
  \begin{equation}
  \label{eq:first_concentration_internal}
    \begin{aligned}
	  \Pr&\left( \bd \cdot \bu > \bw \cdot \bu +\frac{\OPT}{\varphi^{11} (\delta) } \tol(\bu)\right)\\
	  &=\Pr\left( \sum_{i\in L} \one_{i\notin S_0} \cdot \bu_i  <  \E\left[ \sum_{i\in L} \one_{i\notin S_0} \cdot \bu_i \right]-
	\frac{\OPT}{\varphi^{11} (\delta) }\cdot  \tol(\bu)\right)\\
	  &=\Pr\left( \sum_{ \{i_1, i_2\} \in \cM} \frac{ \bu_{i_1} + \bu_{i_2}}{\tol(\bu)} <  \E\left[ \ \sum_{ \{i_1, i_2\} \in \cM} \frac{ \bu_{i_1} + \bu_{i_2}}{\tol(\bu)} \right]-
	\frac{\OPT}{\varphi^{11} (\delta)}\right)\\
	  &\leq \exp\left( -\frac{1}{20}\cdot \delta^4 \cdot  \E\left[ \ \sum_{ \{i_1, i_2\} \in \cM} \frac{ \bu_{i_1} + \bu_{i_2}}{\tol(\bu)} \right] \cdot \left( \frac{\OPT}{\varphi^{11}(\delta) \cdot  \E\left[ \ \sum_{ \{i_1, i_2\} \in \cM} \frac{ \bu_{i_1} + \bu_{i_2}}{\tol(\bu)} \right] }\right)^2 \right)\\
	  &\leq \exp\left( -\frac{\OPT}{\varphi^{25}(\delta)}\right) \enspace .
	\end{aligned}
  \end{equation}
  The first inequality is by \Cref{lem:cvz11}; observe that $\cM\subseteq E\subseteq \cC$, therefore $\frac{\bu_{i_1}+\bu_{i_2}}{\tol(\bu)}\leq 1$ for any $\{i_1,i_2\}\in E$.
  The last inequality uses
  \begin{equation*}
    \E\left[ \ \sum_{ \{i_1, i_2\} \in \cM} \frac{ \bu_{i_1} + \bu_{i_2}}{\tol(\bu)} \right] \leq \frac{|L|}{2} \leq \delta^{-1}\cdot \OPT \enspace .
  \end{equation*}

  We implicitly assumed in \eqref{eq:first_concentration_internal} that $ \E\left[ \ \sum_{ \{i_1, i_2\} \in \cM} \frac{ \bu_{i_1} + \bu_{i_2}}{\tol(\bu)} \right] \neq 0$.
  In case \mbox{$\E\left[\ \sum_{\{i_1, i_2\} \in \cM} \frac{\bu_{i_1} + \bu_{i_2}}{\tol(\bu)} \right] = 0$}, we have $\ \sum_{ \{i_1, i_2\} \in \cM} \frac{\bu_{i_1} + \bu_{i_2}}{\tol(\bu)} = 0$, and
  \begin{equation*}
    \begin{aligned}
      \Pr&\left( \bd \cdot \bu > \bw \cdot \bu +\frac{\OPT}{\varphi^{11} (\delta) } \tol(\bu)\right)= \Pr\left( \sum_{\{i_1, i_2\} \in \cM} \frac{ \bu_{i_1} + \bu_{i_2}}{\tol(\bu)} <  \E\left[ \ \sum_{ \{i_1, i_2\} \in \cM} \frac{ \bu_{i_1} + \bu_{i_2}}{\tol(\bu)} \right]-\frac{\OPT}{\varphi^{11} (\delta)}\right)\\
			&= \Pr\left(0< -\frac{\OPT}{\varphi^{11} (\delta)}\right)  =0 \leq \exp\left( -\frac{\OPT}{\varphi^{25}(\delta)}\right) \enspace .
    \end{aligned}
  \end{equation*}
\end{claimproof}

\subsection{Proof of the Structural Lemma}
\label{sec:structure}
In this section we give the proof of \Cref{lem:structural}. 
Let $\delta \in (0,0.1)$ such that $\delta^{-1}\in \mathbb{N}$, and let $(I,v)$ be a $\delta$-2VBP instance.
As in \Cref{sec:new_analysis}, we use $\OPT=\OPT(I,v)$.  

We first need to construct  the set $\cS^*\subseteq \mathbb{R}_{\geq 0}^I$.
The construction is technical; its components will become clearer below.
The terms $\preceq_d$, $I_{d,j}$ , $h$ and $\hd$ defined as part of the construction of $\cS^*$ are also used in the construction of the linear structure $\cS$.

Let $\succeq^*$ be an arbitrary total order\footnote{We refer the reader to Appendix B.2 of Cormen et al.~\cite{CormenLRS2001} for a formal definition of total order.} over $I$.
For $d\in \{1,2\}$ we define a total order $\succeq_d$ on $I$ by $i_1 \succeq_d i_2$ if and only if  $v_d(i_1) > v_d(i_2)$ or ($v_d(i_1)=v_d(i_2)$ and $i_1\succeq^* i_2$).
Let $h=\delta^{-2}$. For any $d\in \{1,2\}$ and $j = 1,\hdots,2h$ we define a set $I_{d,j}=\left\{i\in L~|~\frac{\delta^2}{2} \cdot (j-1)<v_d(i)\leq \frac{\delta^2}{2} \cdot j \right\} $.
The construction of the linear structure $\cS$ implicitly rounds the volume in dimension $d$ of items in $I_{d,j}$ to $j\cdot \frac{\delta^2}{2}$, and applies {\em fractional grouping} to round the volume of the items in the dimension other than $d$, i.e., $\hd= 3-d$. 
For $d\in \{1,2\}$ define $\cS_d^* = \left\{ \one_{\{i\in I_{d,j}~|~ q_1~ \preceq_{\hd} ~i~ \preceq_{\hd} ~q_2 \}}~\middle|~j \in[2h],~ q_1,q_2\in L\right\}$.
The set~$\cS^*_d$ contains an indicator vector for every possible group which may be generated by the fractional grouping for $I_{d,j}$.
Finally, the set $\cS^*$ is defined by $\cS^* = \left\{ \bu^1 \wedge \bu^2 ~\middle|~\bu^1\in \cS^*_1,~\bu^2\in \cS^*_2\right\}$.
Observe that $|\cS^* | \leq |\cS^*_1| \cdot |\cS^*_2| \leq \left( 2h\cdot |L|^2\right)^2= \delta^{-5}\cdot |L|^4 \leq \varphi(\delta) \cdot |L|^4$.

Let $\blam\in [0,1]^{\cC^*}$ be a small-items integral vector with $\delta$-slack, and let $\bw\in [0,1]^I$ be the coverage of~$\blam$. 
In \Cref{sec:construction} we construct the linear structure $\cS$ of $\blam$, and in \Cref{sec:correctness} we show the structure indeed satisfies the requirements in \Cref{def:linear_structure}.
The construction and proof of correctness rely on a technical {\em refinement} lemma whose proof is given in \Cref{sec:refinement}.

\subsubsection{Construction of $\cS$}
\label{sec:construction}
Our construction uses a partition of $\blam$ into two parts: $\blam^1$ and $\blam^2$, such that for any~$d\in \{1,2\}$ and~$C\in \supp(\blam^d)$ it holds that $C$ has $\delta$-slack in dimension $d$.
Formally, we define $\blam^1 \in [0,1]^{\cC^*}$ by
\begin{equation*}
  \forall C\in \cC^*: ~~~~~ \blam^1_C =
  \begin{cases}
	\blam_C & \textnormal{if $C$ has $\delta$-slack in dimension $1$} \\
		0 &\textnormal{otherwise}
  \end{cases}
\end{equation*}
Also, we define $\blam^2 \in [0,1]^{\cC^*}$ by $\blam^2 = \blam - \blam^1$.
Indeed, as $\blam$ is with $\delta$-slack, for every $d\in \{1,2\}$ and $C\in \supp(\blam^d)$, it holds that $C$ has $\delta$-slack in dimension $d$. For $d\in \{1,2\}$ let $\bw^d$ be the coverage of~$\blam^d$.

As mentioned above, for each $d\in \{1,2\}$ we implicitly give a rounding scheme for the large items, in which the volume in dimension $d$ of all items in $i\in I_{d,j}$ is rounded up to $j\cdot \frac{\delta^2}{2}$.
The slack of configurations in $\supp(\blam)$ is used to compensate for the possible volume increase.
In the other dimension, $\hd$, we apply {\em fractional grouping}, defined as follows.
\begin{defn}
\label{def:grouping}
  Let $E\neq \emptyset$ be a finite set, $ \bgam\in [0,1]^E$, $\succeq$ be a total order over $E$ and $\xi \in \mathbb{N}_+$.
  A partition $G_1,\ldots, G_{\tau}$ of $E$ is a {\em $\xi$-fractional grouping} with respect to $\bgam$ and $\succeq$ if the following conditions hold:
  \begin{enumerate}
	\item For every $1\leq \ell_1<\ell_2 \leq \tau$, $i_1 \in G_{\ell_1}$ and $i_2\in G_{\ell_2}$ it holds that $i_1\succeq i_2$. 
	\item For $\ell = 1,\hdots,\tau-1$ it holds that $\one_{G_{\ell}} \cdot \bgam \geq \frac{\| \bgam\|}{\xi}$.
	\item For $\ell = 1,\hdots,\tau$ it holds that $\one_{G_{\ell}} \cdot \bgam \leq \frac{\| \bgam\|}{\xi} + 1$.
  \end{enumerate}
\end{defn}

The proof of the next lemma utilizes arguments from Fairstein et al.~\cite{FairsteinKS2021}. 
\begin{lemma}
\label{lem:fractional_grouping_exists}
  For any finite  set $E\neq \emptyset$ , $\bgam\in [0,1]^E$, a total order $\succeq$ over $E$ and $\xi \in \mathbb{N}_+$, there is a $\xi$-fractional grouping $G_1,\ldots ,G_{\tau}$ of $E$ with respect to $\bgam$ and $\succeq$ for which $\tau \leq \xi$.
\end{lemma}
\begin{proof}
  If $\bgam=\bzero$ then the partition $G_1=E$ is a $\xi$-fractional grouping.
  We henceforth assume $\bgam\neq \bzero$.

  Assume, without loss of generality, that $E=\{1,2,\ldots, \nu \}=[\nu]$ and $a\succeq b$ if and only if $a\leq b$.
  Define a sequence $(q_\ell)_{\ell =0}^{\infty}$ by $q_0=0$, and $q_{\ell} = \min\left\{e\in E~\middle|~\sum_{f=q_{\ell-1}+1}^{e} \bgam_f > \frac{\| \bgam\|}{\xi}\right\} \cup \{\nu \}$.
  Also, define $\tau = \min\{ \ell\in\mathbb{N} ~|~q_{\ell }=\nu \}$. 
  Since $\|\bgam\|>0$, it follows that  $(q_{\ell})_{\ell=0}^{\tau}$ is monotonically increasing. 
	
  We define $G_{\ell } = \left\{ e\in E~\middle|~ q_{\ell-1}< e\leq q_{\ell}\right\}$.
  Then $G_{\ell} = \{1,\hdots,q_{\ell}\}\setminus\{1,\hdots,q_{\ell-1}\}$ for $\ell = 1,\hdots,\tau$.
  As $q_0=0$, $q_{\tau}=\nu$ and $(q_{\ell})_{\ell=0}^{\tau}$ is monotonically increasing, it follows that $G_1,\ldots, G_{\tau}$ is a partition of $E$.
  Clearly, for any $1\leq \ell_1 <\ell_2 \leq \tau$, $i_1 \in G_{\ell_1} $ and $i_2 \in G_{\ell_2}$ it holds that $i_1 \leq q_{\ell_1} \leq q_{\ell_2-1} <i_2$ thus $i_1\succeq i_2$.  
	
  Let $\ell \in \{1,\hdots,\tau\}$.
  By definition of $q_{\ell}$ it holds that $\sum_{ f = q_{\ell -1}+1 }^{q_{\ell}-1} \bgam_{f} \leq \frac{\|\bgam\|}{\xi}$.
  Hence, as $\bgam_{q_{\ell}}\leq 1$, it also holds that $\one_{G_{\ell}}\cdot \bgam=\sum_{e\in G_{\ell}} \bgam_e = \bgam_{q_{\ell}} + \sum_{ e= q_{\ell -1}+1 }^{q_{\ell}-1}\bgam_e \leq \frac{\|\bgam\|}{\xi}+1$.
	
  Let $\ell\in\{1,\hdots,\tau-1\}$.
  Then $q_\ell \neq \nu$ and $q_{\ell} = \min\left\{e\in E~\middle|~\sum_{f=q_{\ell-1}+1}^{e} \bgam_f > \frac{\| \bgam\|}{\xi}\right\}$. 
  Therefore, $\one_{G_{\ell}}\cdot \bgam = \sum_{e\in G_{\ell}}\bgam_e = \sum_{ e= q_{\ell -1}+1 }^{q_{\ell}} \bgam_{e} > \frac{\| \bgam \| }{\xi}$.
	
  Thus, we showed that $G_1,\ldots ,G_{\tau}$ is a $\xi$-fractional grouping of $E$ with respect to $\bgam$ and $\succeq$. 
  It also holds that
  \begin{equation*}
    \|\bgam\|  = \sum_{e\in E} \bgam_e = \sum_{\ell=1}^{\tau} \sum_{e\in G_\ell } \bgam_e  \geq   \sum_{\ell=1}^{\tau-1} \sum_{e\in G_\ell } \bgam_e  > \sum_{\ell=1}^{\tau-1} \frac{ \|\bgam\|}{\xi}  =(\tau-1) \frac{\| \bgam \|}{\xi} \enspace .
  \end{equation*}
  Hence, $\tau-1 <\xi$, and as both $\tau $ and $\xi$ are integral it follows that $\tau \leq \xi$.
  This completes the proof.
\end{proof}

For any $d\in \{1,2\}$ and $j = 1,\hdots,2h$ define a vector $\bgam^{d,j}\in [0,1]^{I_{d,j}}$ by $\bgam^{d,j}_i = \bw^d_{i}$ for $i\in I_{d,j}$.
By \Cref{lem:fractional_grouping_exists}, for any $d\in \{1,2\}$ and $j = 1,\hdots,2h$ such that $I_{d,j}\neq \emptyset$ there is an $h$-fractional grouping~$\left(G^{d,j}_\ell \right)_{\ell=1}^{\tau_{d,j}}$ of $I_{d,j}$ with respect to $\bgam^{d,j}$ and the total order $\succeq_{\hd}$ with $\tau_{d,j}\leq h$.  
For $d\in \{1,2\}$ let $\cG_d = \left\{(j,\ell)~\middle |~j \in [2h],~I_{d,j}\neq \emptyset \textnormal{ and } \ell\in [\tau_{d,j}]\right\}$.
It follows that $\cG_1,\cG_2 \subseteq \{1,\hdots,2h\}\times \{1,\hdots,h\}$ and thus $|\cG_1|,~ |\cG_2|\leq 2\delta^{-4}$. 

Our objective is to add to the structure $\cS$ vectors $\bu$ to ensure that if $\bz\in [0,1]^{I}$ satisfies \eqref{eq:structural_prop} then we can decompose $\bz\wedge \one_{L}$ to $\bz^1,\bz^2\in [0,1]^{I}$ such that $\bz\wedge \one_{L}=\bz^1+\bz^2$ and $\bz^d \cdot \one_{G^{d,j}_{\ell} }\lesssim \beta \cdot \bw^d \cdot \one_{G^{d,j}_{\ell}}$ for any $d\in \{1,2\}$ and $(j,\ell)\in \cG_d$. 
This can be intuitively interpreted as a decrease  in demand for items in $G^{d,j}_{\ell}$ by a factor of $\beta$.
As we have a rounding scheme for each dimension, an item $i\in L$ may belong to two groups $G^{d,j}_{\ell}$- one from the scheme for dimension~$1$ and another from the scheme of dimension~$2$.
We therefore add into $\cS$ vectors which represent the intersection of each pair of such groups, and therefore impose a decrease in demand by a factor of $\beta$ for each intersection. 

Formally, our linear structure will contain the set $\cS_{\textnormal{large}}$, which we define as
\begin{equation}
\label{eq:Slarge}
  \cS_{\textnormal{large}}=\left\{\one_{G^{1,j_1}_{\ell_1}} \wedge \one_{G^{2,j_2}_{\ell_2}}  ~\middle|~(j_1,\ell_1) \in \cG_1,~(j_2,\ell_2)\in \cG_2\right\} \enspace .
\end{equation}
In \Cref{sec:correctness} we show that if $S_{\textnormal{large}}\subseteq \cS$ and $\bz$ satisfies \eqref{eq:structural_prop} then we can find the decomposition $\bz^1$ and $\bz^2$ as mentioned above. 
Furthermore, to show the correctness of the structure we (implicitly) use a {\em shifting} argument (see, e.g., \cite{DeLaVegaL1981}) in which items in $G^{d,j}_{\ell}$ take the place of items in $G^{d,j}_{\ell-1}$. 

We use the rounding schemes for the large items to define a {\em type} for each configuration.
We then fractionally associate each small item $i\in I\setminus L$ with the various types, and use this association as a basis for the linear structure.
For $d\in \{1,2\}$, the {\em $d$-type} of a multi-configuration $C\in \cC^*$, denoted by $\type^d(C)$, is the vector $\bt\in \mathbb{N}^{\cG_d}$ defined by $\bt_{(j,\ell )}= \sum_{i\in G^{d,j}_{\ell} } C(i)$ for any $(j,\ell)\in \cG_d$.
That is, $\bt_{(j,\ell)}$ is the number of items in $C$ which belong to $G^{d,j}_{\ell}$.
Since the set $G^{d,j}_{\ell}$ contains only large items, it follows that $\bt_{(j,\ell)}\leq 2\delta^{-1}$. Let $\cT_d=\left\{ \type^d(C)~|~C\in \cC^*\right\}$ be the set of all possible $d$-types.
It follows that $\cT_d \subseteq \{0,1,\ldots, 2\cdot \delta^{-1}\}^{\cG_d}$, and therefore $|\cT_d|\leq \left(1+2\cdot \delta^{-1}\right)^{2\cdot \delta^{-4}}\leq \exp(\delta^{-6})$.

The {\em small item association} of $d\in \{1,2\}$ and the $d$-type $\bt\in \cT_d$ is the vector $\ba^{d,\bt}\in [0,1]^{I}$ defined by
\begin{equation}
\label{eq:association}
  \ba^{d,\bt}_{i}= \sum_{C\in\cC^{*} \textnormal{ s.t.} \type^d(C)=\bt } \blam_C^{d} \cdot C(i) ,
\end{equation}
for $i\in I\setminus L$ and $\ba^{d,\bt}_i = 0$ for $i\in L$.
Intuitively, $\ba^{d,\bt}_i$ is the fraction of $i\in I\setminus L$ selected by configurations of type $\bt$ in $\blam^d$. 

For $d\in \{1,2\}$ define $\bv^d\in [0,1]^I$ by $\bv^d_i=v_d(i)$ for all $i\in I$.
Also, we use $\bullet$ to denote element-wise multiplication of two vectors. That is, for $\ba,\bb\in \mathbb{R}^I$ let $\ba \bullet \bb = \bc$, where $\bc_i =\ba_i\cdot \bb_i$ for every~$i\in I$. 
The next lemma will be useful towards adding more vectors to the linear structure.
\begin{lemma}[Small Items Refinement]
\label{lem:refinement}
  Let $\ba \in [0,1]^{I}$ be such that $\supp(\ba)\subseteq I\setminus L$, let $d\in \{1,2\}$, and let $q\in \mathbb{N}_{\geq 4}$.
  Then there are subsets $H_1,\ldots, H_q\subseteq I\setminus L$ such that for any $Q\subseteq I\setminus L$ and $\beta\in \left[\frac{1}{q},1\right]$ which satisfy
  \begin{equation}
  \label{eq:refinement_prop}
    \forall j=1,\hdots,q:~~~
	\left\| \one_{Q\cap H_j} \bullet \ba \bullet \bv^d \right\| \leq \beta \left\| \one_{H_j} \bullet \ba \bullet \bv^d \right\| + \frac{\OPT}{q^5} \max\left\{ v_d(C\cap H_j)~\middle |~C\in \cC\right\}
  \end{equation}
  there is a set $X\subseteq Q$ which admits the following properties:
  \begin{enumerate}
    \item \label{refinment:prop1} $\left\|\one_{X}\bullet  \ba \bullet (\bv^1+\bv^2) \right\| \leq \frac{16}{q}\cdot \OPT +2q\delta$.
	 \item \label{refinement:prop2} $\left\|\one_{Q\setminus X}\bullet  \ba \bullet \bv^d \right\| \leq \beta \cdot \ba \cdot \bv^d$.
  \end{enumerate}
\end{lemma}
We refer to $H_1,\ldots, H_q$ as the {\em refinement} of $\ba$ and $q$ in dimension $d$.
Indeed, the condition in~\eqref{eq:refinement_prop} is essentially a variant of \eqref{eq:structural_prop}. 
\Cref{lem:refinement} plays a central role in showing the correctness of the structure $\cS$ (see the proof of \Cref{lem:P_not_empty}).
We defer the proof of \Cref{lem:refinement} to \Cref{sec:refinement}.

We select $q=\ceil{\exp\left(\delta^{-10}\right)}$. 
For any $d,d' \in \{1,2\}$ and $\bt\in \cT_d$ let $H^{d,\bt,d'}_1,\ldots, H^{d,\bt, d'}_{q}$ be the refinement of $\ba^{d,\bt}$ and $q$ in dimension $d'$.
We use the small items association and its refinement to define additional vectors as follows:
\begin{equation*}
  \cS_{\textnormal{small}}= \left\{  \one_{H^{d,\bt,d'}_{j}}\bullet  \ba^{d,\bt}\bullet \bv^{d'}~|~d,d'\in\{1,2\},~\bt\in\cT_d, ~j = 1,\hdots,q \right\} \enspace .
\end{equation*}

Finally, the structure is $\cS= \cS_{\textnormal{large}} \cup \cS_{\textnormal{small}}$.

\subsubsection{Correctness}
\label{sec:correctness}
We first observe that 
\begin{equation*}
  |\cS| = |\cS_{\textnormal{large}}| + |\cS_{\textnormal{small}}| \leq |\cG_1| \cdot|\cG_2| +2\cdot q \cdot \left( |\cT_1|+|\cT_2|\right) \leq \exp(\delta^{-20})=\varphi(\delta) \enspace .
\end{equation*}
Let $\bu\in \cS$ such that $\supp(\bu)\cap L\neq \emptyset$, then $\bu\in \cS_{\textnormal{large}}$.
Therefore, by \eqref{eq:Slarge} there is $(j_1,\ell_1)\in \cG_1$ and $(j_2,\ell_2)\in \cG_2$ such that $\bu =\one_{G^{1,j_1}_{\ell_1} }\wedge \one_{G^{2,j_2}_{\ell_2}}$. 
By \Cref{def:grouping}, for $d\in \{1,2\}$ there are  $q^d_1, q^d_2 \in I_{d,j_d}$ such that $G^{d,j_d}_{\ell_d} = \{ i\in I_{d,j_d}~|~q^d_1 \preceq_{\hd} i \preceq_{\hd} q^d_2 \}$; thus, $\one_{G^{d,j_d}_{\ell_d}} \in \cS^*_d$. It follows that~\mbox{$\bu = \one_{G^{1,j_1}_{\ell_1} }\wedge \one_{G^{2,j_2}_{\ell_2}} \in \cS^*$.}

Let $\beta \in [\delta^5, 1]$ and $\bz\in [0,1]^{I}$ such that $\bz$ is small-items integral, $\supp(\bz) \subseteq \supp(\bw)$, and 
\begin{equation}
\label{eq:z_condition}
  \bz \cdot \bu\leq \beta \cdot \bw \cdot \bu + \frac{1 }{\varphi^{10}(\delta)} \cdot \OPT \cdot \tol(\bu)
\end{equation}
for all $\bu\in \cS$.
To verify that $\cS$ is a $(\delta,\varphi(\delta))$ linear structure, it remains to show that $\OPT_f(\bz)\leq \beta (1+10\delta)\cdot \|\blam\|+\varphi(\delta) +\delta^{10}\cdot  \OPT(I,v)$.

We first generate two vectors $\bz^1$ and $\bz^2$ such that $\bz\wedge \one_{L}$ and $\bz^d\cdot \one_{G^{d,j}_{\ell}}\lesssim 
\beta \bw^d\cdot \one_{G^{d,j}_{\ell}}$ for every $d\in \{1,2\}$ and $(j,\ell)\in\cG_d$. Each item $i\in L$ belongs to  groups $G^{1,j_1}_{\ell_1}$ and $G^{2,j_2}_{\ell_2}$. The demand $\bz_i$ of $i$ is partitioned between $\bz^1$ and $\bz^2$ with the same proportion that $\bw^1$ and $\bw^2$ contributed to the total demand of items in $G^{1,j_1}_{\ell_1}\cap G^{2,j_2}_{\ell_2}$. Specifically,
for $d\in \{1,2\}$, define $\bz^d\in [0,1]^{I}$ by
\begin{equation}
\label{eq:z_decomp_def}
  \forall (j_1,\ell_1)\in \cG_1,~(j_2,\ell_2)\in \cG_2,~i\in G^{1,j_1}_{\ell_1}\cap G^{2,j_2}_{\ell_2}\cap \supp(\bz):~~~
  \bz^d_i =
  \bz_i\cdot 
  \frac{\left (\one_{G^{1,j_1}_{\ell_1}} \wedge\one_{G^{2,j_2}_{\ell_2}} \right) \cdot \bw^d}{\left (\one_{G^{1,j_1}_{\ell_1}} \wedge\one_{G^{2,j_2}_{\ell_2}} \right) \cdot \bw},
\end{equation}
and $\bz^d_i=0$ for any other $i\in I$.  
Observe that since $\supp(\bz)\subseteq \supp(\bw)$ we never get in~\eqref{eq:z_decomp_def} a division by zero.
Since for every $i\in L$ there is a unique $(j_1,\ell_1)\in \cG_1$ and a unique $(j_2, \ell_2)\in \cG_2$ such that $i\in G^{1,j_1}_{\ell_1} \cap G^{2,j_2}_{\ell_2}$, it follows that $\bz\wedge \one_L= \bz^1+\bz^2$. 
For every $d\in \{1,2\}$ and $(j,\ell)\in \cG_d$ it holds that
\begin{equation}
\label{eq:z_decomp_first_ineq}
  \begin{aligned}
    \bz^d \cdot \one_{G^{d,j}_{\ell} }  &=
    \sum_{~(j',\ell')\in \cG_{\hd} ~}
    \sum_{~i\in G^{d,j}_{\ell }\cap G^{\hd, j'}_{\ell'}\cap \supp(\bz) ~}
     \bz_i\cdot 
    \frac{\left (\one_{G^{d,j}_{\ell}} \wedge\one_{G^{\hd,j'}_{\ell'}} \right) \cdot \bw^d}{\left (\one_{G^{d,j}_{\ell}} \wedge\one_{G^{\hd,j'}_{\ell'}} \right) \cdot \bw} \\
    &= \sum_{~(j',\ell')\in \cG_{\hd} \textnormal{ s.t.}~\left (\one_{G^{d,j}_{\ell}} \wedge\one_{G^{\hd,j'}_{\ell'}} \right) \cdot \bw\neq 0~}
    \left(  \left(\one_{G^{d,j}_{\ell}} \wedge\one_{G^{\hd,j'}_{\ell'}} \right) \cdot \bz \right)\cdot 
    \frac{\left (\one_{G^{d,j}_{\ell}} \wedge\one_{G^{\hd,j'}_{\ell'}} \right) \cdot \bw^d}{\left (\one_{G^{d,j}_{\ell}} \wedge\one_{G^{\hd,j'}_{\ell'}} \right) \cdot \bw} \enspace .\\
  \end{aligned}
\end{equation}
Since $ \one_{G^{d,j}_{\ell}} \wedge\one_{G^{\hd,j'}_{\ell'}} \in \cS$, by  \eqref{eq:z_condition} it holds
\begin{equation}
\label{eq:z_decomp_second_ineq} 
  \begin{aligned}
    \left(\one_{G^{d,j}_{\ell}} \wedge\one_{G^{\hd,j'}_{\ell'}} \right) \cdot \bz
    &\leq \beta \left(\one_{G^{d,j}_{\ell}} \wedge\one_{G^{\hd,j'}_{\ell'}} \right) \cdot \bw +\frac{\OPT}{\varphi^{10}(\delta)} \cdot \tol\left(\one_{G^{d,j}_{\ell}} \wedge\one_{G^{\hd,j'}_{\ell'}}   \right)\\
    &\leq \beta \left(\one_{G^{d,j}_{\ell}} \wedge\one_{G^{\hd,j'}_{\ell'}} \right) \cdot \bw +\frac{2\cdot \delta^{-1} }{ \varphi^{10}(\delta)} \cdot \OPT.
  \end{aligned}
\end{equation}
The second inequality holds since there are at most $2\delta^{-1}$ large items in a configuration.
Plugging~\eqref{eq:z_decomp_second_ineq} into \eqref{eq:z_decomp_first_ineq}, we have
\begin{equation}
\label{eq:z_decomp_third_ineq}
  \begin{aligned}
    &\bz^d \cdot \one_{G^{d,j}_{\ell} } \\
	\leq& \sum_{~(j',\ell')\in \cG_{\hd}  \textnormal{ s.t.}~\left (\one_{G^{d,j}_{\ell}} \wedge\one_{G^{\hd,j'}_{\ell'}} \right) \cdot \bw\neq 0~}
	\left(  \beta \left(\one_{G^{d,j}_{\ell}} \wedge\one_{G^{\hd,j'}_{\ell'}} \right) \cdot \bw +\frac{2\cdot \delta^{-1}\OPT}{\varphi^{10}(\delta)} \right)\cdot 
	\frac{\left (\one_{G^{d,j}_{\ell}} \wedge\one_{G^{\hd,j'}_{\ell'}} \right) \cdot \bw^d}{\left (\one_{G^{d,j}_{\ell}} \wedge\one_{G^{\hd,j'}_{\ell'}} \right) \cdot \bw} \\
	\leq& \sum_{~(j',\ell')\in \cG_{\hd}  \textnormal{ s.t.}~\left (\one_{G^{d,j}_{\ell}} \wedge\one_{G^{\hd,j'}_{\ell'}} \right) \cdot \bw\neq 0~}
	\beta \left (\one_{G^{d,j}_{\ell}} \wedge\one_{G^{\hd,j'}_{\ell'}} \right) \cdot \bw^d ~+~ \frac{ \delta^{-6}}{\varphi^{10}(\delta)} \cdot \OPT \\
	\leq& \beta \cdot \one_{G^{d,j}_{\ell} } \cdot \bw^d + \frac{ \delta^{-6}}{\varphi^{10}(\delta)} \cdot\OPT,
  \end{aligned}
\end{equation}
where the second inequality holds since $|\cG_{\hd}|\leq 2\cdot \delta^{-4}$. 

Therefore, for every $d\in \{1,2\}$ there is a vector $\br^{d}\in [0,1]^I$ such that, for any $(j,\ell)\in \cG_d$, 
\begin{equation}
\label{eq:z_minus_r_conds}
  \left(\bz^d - \br^d\right) \cdot \one_{G^{d,j}_{\ell } } \leq \max\left\{ \beta \cdot \one_{G^{d,j}_{\ell} }\cdot \bw^d -2,0 \right\},
\end{equation}
for every $i\in I$ it holds that $r^d_i\leq z^d_i$, and $\| \br^d  \| \leq \left( 2 +\frac{\delta^{-6}}{\varphi^{10}(\delta)}\cdot \OPT\right)\cdot |\cG_d|\leq   \delta ^{-5} + \frac{ \delta^{-11}}{\varphi^{10}(\delta)} \OPT$.
Hence, $\OPT_f(\br^d) \leq \|\br^d\| \leq \delta ^{-5} + \frac{ \delta^{-11}}{\varphi^{10}(\delta)} \OPT$, as $\sum_{i\in I} \br^d_i \cdot \one_{\{i\}}$ is a solution for $\LP(\br^d)$. 
 
For any $d\in \{1,2\}$, let $F_d =\bigcup_{j\in [2h] \textnormal{ s.t. } (j,1)\in \cG_d } G^{d,j}_1$ be the set of all items that belong to a first group in one of the  fractional groupings $G^{d,j}_1,\ldots, G^{d,j}_{\tau_{d,j}}$.  By \eqref{eq:z_minus_r_conds},
\begin{equation*}
  \begin{aligned}
    (\bz^d-\br^d) \cdot \one_{F_d} &\leq \sum_{j\in \{1,\hdots,2h\} \textnormal{ s.t. } (j,1)\in \cG_d }   (\bz^d-\br^d) \cdot \one_{G^{d,j}_1}  
    \leq \sum_{j\in \{1,\hdots,2h\} \textnormal{ s.t. } (j,1)\in \cG_d }   \max\left\{ \beta \cdot \one_{G^{d,j}_{1}}\cdot \bw^d-2,0 \right\}\\
    &\leq  \beta \sum_{j\in \{1,\hdots,2h\} \textnormal{ s.t. } (j,1)\in \cG_d } \frac{  \one_{I_{d,j} } \cdot \bw^d}{h} =\beta  \frac{\bw^d \cdot \one_L}{h}  \leq  2\cdot \beta \cdot \delta \cdot \|\blam^d\| 
  \end{aligned}
\end{equation*}
where the third inequality is by \Cref{def:grouping}, and the last inequality follows from $h=\delta^{-2}$ and  
\begin{equation*}
  \sum_{i\in L} \bw^d_i = \sum_{C\in \cC^*} \blam^d_C \cdot \sum_{i\in L} C(i) \leq \sum_{C\in \cC^*} \blam^d_C \cdot 2\delta^{-1} =2\cdot \delta^{-1} \|\blam^{d}\| \enspace .
\end{equation*}
  
Define  $Q=\supp(\bz)\setminus L = \{i\in I\setminus L~|~\bz_i=1\}$ and 
\begin{equation}
\label{eq:y_def}
  \by = \sum_{d\in \{1,2\}} \left((\bz^d -\br^d) \wedge \one_{L\setminus F_d} \right) +\one_Q \enspace .
\end{equation}
Then,
\begin{equation}
\label{eq:first_scraping}
  \begin{aligned}
    \OPTf(\bz)&\leq
  	\sum_{d\in \{1,2\} } \left( \OPTf( \br^d) +\OPTf( (\bz^d-\br^d) \wedge \one_{F_d} )\right) + \OPTf\left(\by\right)\\
  	&\leq \OPTf\left(\by\right)\ +  2\beta\delta \| \blam\| +2\delta^{-5}+ \frac{2\cdot \delta^{-11}}{\varphi^{10}(\delta)}\OPT \enspace .
  \end{aligned}
\end{equation}
We proceed to derive an upper bound on~$\OPTf(\by)$, which in turn implies an upper bound on~$\OPTf(\bz)$. 

Given $d\in \{1,2\}$ we  define the {\em $d$-size} of  $(j,\ell)\in \cG_d$, denoted $s^{d}(j,\ell) \in [0,1]^2$,  by $s^d_d(j,\ell) = \frac{\delta^2}{2} j$ and $s^d_{\hd} = \min\{ v_\hd(i)~|~i\in G^{d,j}_{\ell}\}$. 
The value $s^d(j,\ell)$ can be viewed a the rounded volume of items in~$G^{d,j}_{\ell}$. 

The next lemma gives the basis for our {\em shifting} argument.
\begin{lemma}
\label{lem:shifted_size}
  Let $d\in \{1,2\}$,  $(j,\ell)\in \cG_d$ and $i\in G^{d,j}_{\ell}$.
  If $\ell\neq 1$ then $v(i)\leq s^d(j,\ell- 1)$.
\end{lemma}
\begin{proof}
  As $i\in G^{d,j}_{\ell}\subseteq I_{d,j}$, it follows that  $v_d(i)\leq  \frac{\delta^2}{2}\cdot j=s^d_d(j,\ell-1)$.
  Furthermore, 	$v_{\hd}(i')\geq v_{\hd}(i) $ for every $i'\in G^{d,j}_{\ell-1}$ as $(G^{d,j}_{\ell'})_{\ell'=1}^{\tau_{d,j}}$ is an $h$-fractional grouping with respect to the relation $\succeq_{\hd}$.
  Hence, 
  \begin{equation*}
	v_{\hd}(i) \leq \min \left\{v_{\hd}(i')~|~i'\in G^{d,j}_{\ell-1} \right\} = s^d_{\hd}(j,\ell-1) \enspace .\qedhere
  \end{equation*}
\end{proof}

We extend the definition of size to $d$-types by $s^{d}(\bt) = \sum_{(j,\ell)\in \cG_d} \bt_{(j,\ell) }\cdot s^d(j,\ell) $ for any $d\in \{1,2\}$ and $\bt \in \cT_d$.  
\begin{lemma}
  \label{lem:rounded_bounds}
  Let $d\in \{1,2\}$ and $C\in \cC^*$ with $\blam^d_C > 0$.
  Then $\sum_{i\in I\setminus L} v(i)\cdot C(i) \leq \bone-s^d\left(\type^d(C)\right)$.
\end{lemma}
\begin{proof}
  For any $i\in L$  such that $C(i)>0$ there is a unique $(j,\ell)\in \cG_d$  for which $i\in G^{d,j}_{\ell}$.
  Thus,
  \begin{equation}
  \label{eq:round_types_basic}
    \sum_{i\in I\setminus L} v(i)\cdot C(i) = \sum_{i\in I} v(i)\cdot C(i)  - \sum_{i\in L} v(i)\cdot C(i)= v(C) - \sum_{(j,\ell)\in \cG_d} \sum_{i\in G^{d,j}_{\ell} }  v(i)\cdot C(i) \enspace .
  \end{equation}
  Therefore, we have
  \begin{equation}
    \begin{aligned}
    \label{eq:round_types_dim_d}
      \sum_{i\in I\setminus L} v_d(i)\cdot C(i) &= v_d(C) - \sum_{(j,\ell)\in \cG_d} \sum_{i\in G^{d,j}_{\ell} }  v_d(i)\cdot C(i)\\
      &\leq 1-\delta - \sum_{(j,\ell)\in \cG_d} \sum_{i\in G^{d,j}_{\ell} }  \left(s^d_d(j,\ell ) -\frac{\delta^2}{2}\right)\cdot C(i)\\
      & =1-\delta - \sum_{(j,\ell)\in \cG_d} \sum_{i\in G^{d,j}_{\ell} }  s^d_d(j,\ell )\cdot C(i)+\frac{\delta^2}{2}\sum_{(j,\ell)\in \cG_d} \sum_{i\in G^{d,j}_{\ell} } \cdot C(i)\\
      &=1-\delta - \sum_{(j,\ell)\in \cG_d}  \type^d_{(j,\ell)} (C)\cdot s^d_d(j,\ell)+\frac{\delta^2}{2}
      \sum_{i\in L} C(i)\\
      &\leq 1- s_d^d(\type^d(C)) \enspace .
    \end{aligned}
  \end{equation}
  The first equality is by \eqref{eq:round_types_basic}.
  The first inequality holds, as $C$ has $\delta$-slack in dimension $d$ since $\blam^d_C>0$, and since $v_d(i)> \frac{\delta^2}{2}(j-1)$ for any $i\in G^{d,j}_{\ell}\subseteq I_{d,j}$.
  The last inequality holds as there are at most $2\delta^{-1}$ large items in a multi-configuration.
  Similarly,
  \begin{equation}
    \begin{aligned}
    \label{eq:round_types_other_dim}
      \sum_{i\in I\setminus L} v_{\hd}(i)\cdot C(i) &= v_{\hd}(C) - \sum_{(j,\ell)\in \cG_d} \sum_{i\in G^{d,j}_{\ell} }  v_{\hd}(i)\cdot C(i) \\
      &\leq 1- \sum_{(j,\ell)\in \cG_d} \sum_{i\in G^{d,j}_{\ell} }  s^d_{\hd}(j,\ell )  \cdot C(i)\\
	  &=1 - \sum_{(j,\ell )\in \cG_d}  \type^{d}_{(j,\ell)} (C)\cdot s^d_{\hd}(j,\ell)\\
      &\leq 1- s_{\hd}^{d}(\type^d(C)) \enspace .
    \end{aligned}
  \end{equation}
  The first equality follows from \eqref{eq:round_types_basic} and the first inequality  is by the definition of $s^{d}_{\hd}(j,\ell)$.
  The statement of the lemma follows from~\eqref{eq:round_types_dim_d} and~\eqref{eq:round_types_other_dim}.
\end{proof}

For any $d\in\{1,2\}$ and $\bt\in \cT_d$, the {\em prevalence} of type $\bt$ is $\eta_d(\bt) = \sum_{C\in \cC^* \textnormal{ s.t. } \type^d(C)=\bt} \blam^d_C$.
Informally, $\eta_d(\bt)$ is the number of configurations of type $\bt$ selected by $\blam^d$.
Also, define $\kappa_d(\bt) = \ceil{\beta\cdot \eta_d(\bt) } + 2\cdot \delta^{-1} $ for any $d\in\{1,2\}$ and $\bt\in \cT_d$.
We construct a solution of $\LP(\by)$ in which there are $\kappa_d(\bt)$ configurations with large items of total size at most $s^d(\bt)$. 
For the assignment of large items we use the next lemma.

\begin{lemma}
\label{lem:shifting_application}
  There are vectors $\bx^{d,\bt}\in [0,1]^{\cC}$ for $d\in \{1,2\}$ and $\bt\in \cT_d$ such that
  \begin{enumerate}
    \item for any $d\in\{1,2\}$	the coverage of $\sum_{\bt\in \cT_d} \kappa_d(\bt) \cdot \bx^{d,\bt}$ is $\left(\bz^d-\br^d\right) \wedge \one_{L\setminus F_d}$,
	\item for any $d\in\{1,2\}$ and $\bt\in \cT_d$ it holds that $\|\bx^{d,t}\| = 1$,
	\item and for any $d\in \{1,2\}$, $\bt\in \cT_d$ and $C\in \supp(\bx^{d,\bt})$, it holds that $v(C)\leq s^d(\bt)$. 
  \end{enumerate}
\end{lemma}
The proof of \Cref{lem:shifting_application} relies on the following combinatorial claim (we omit the proof).
\begin{claim}
\label{claim:comb}
  Let $E$ be an arbitrary finite set,  $\xi \in \mathbb{N}_+$ and  $\bgam\in \left[0,\frac{1}{\xi}\right]^{E}$ such that $\|\bgam\|\leq 1$.
  Then there exists a random set $K\subseteq E$ such that $|K| \leq \xi$ and $\Pr(e\in K) = \xi \cdot \bgam_e$ for every $e\in E$. 
\end{claim}
\begin{proof}[Proof of \Cref{lem:shifting_application}]
  Let $d\in \{1,2\}$ and for any $(j,\ell) \in \cG_d$, define $\rho_{(j,\ell)} = \sum_{\bt\in \cT_d} \bt_{(j,\ell)} \cdot\kappa_{d}(\bt)$. 
  Then $\rho_{(j,\ell)} \geq 2\cdot\delta^{-1}$.
  For any $(j,\ell)\in \cG_d$  and $i\in G^{d,j}_{\ell}$ such that $\ell\neq 1$, define $p_i= \frac{\bz^d_i -\br^d_i}{\rho_{(j,\ell-1)}}\leq \frac{1}{2\cdot \delta^{-1}}$.
	
  For every $(j,\ell)\in\cG_d$ with $\ell\neq 1$ it holds that
  \begin{equation*}
	\begin{aligned}
	  \rho_{(j,\ell-1)} &=\sum_{\bt\in \cT_d} \bt_{(j,\ell-1)}\cdot \kappa_d(\bt)\\  &\geq \beta \sum_{\bt\in \cT_d} \bt_{(j,\ell-1)}\cdot \eta_d(\bt) = \beta \one_{G^{d,j}_{\ell-1}} \cdot \bw^d\\ &\geq \beta \frac{\bw^d \cdot \one_{I_{d,j}} }{h} \\
	 &\geq  \max\left\{ \beta \cdot \bw^d\cdot \one_{G^{d,j}_{\ell}}-1,~0 \right\} \geq  \left(\bz^d - \br^d  \right)\cdot \one_{G^{d,j}_{\ell}} \enspace .
	\end{aligned}
  \end{equation*}
  The second and third inequalities hold since $G^{d,j}_1,\ldots, G^{d,j}_{\tau_{d,j}}$ is an  $h$-fractional grouping of $I_{d,j}$.
  The last inequality is by \eqref{eq:z_minus_r_conds}.
  Therefore, $\sum_{i\in G^{d,j}_{\ell}} p_i \leq 1$. 
	
  Fix $\bt\in \cT_d$, and for any $(j,\ell)\in \cG_d$ with $\ell\neq 1$ let $K_{(j,\ell)}\subseteq G^{d,j}_{\ell}$ be a random set such that $|K_{(j,\ell )}| \leq \bt_{(j,\ell-1)}$ and $\Pr(i\in K_{(j,\ell)})=\bt_{(j,\ell-1)} \cdot p_i$ for every $i\in G^{d,j}_{\ell}$.
  The random sets $K_{(j,\ell)}$ exist by \Cref{claim:comb}.
  Furthermore, we may assume the random sets $\left(K_{(j,\ell)} \right)_{(j,\ell)\in \cG_d,~\ell\neq 1}$ are independent. Define $R=\bigcup_{(j,\ell)\in \cG_d \textnormal{ s.t. }\ell\neq 1} K_{(j,\ell)}$ and $\bx^{d,\bt}_C = \Pr(R=C)$ for all $C\in \cC$.
  It follows that $\|\bx^{d,\bt}\|=\sum_{C\in \cC^* }\Pr(R=C)=1$.
  Observe that 
  \begin{equation*}
    v(R)\leq \sum_{(j,\ell)\in \cG_d \textnormal{ s.t. }\ell\neq 1} v(K_{(j,\ell)})\leq \sum_{(j,\ell)\in \cG_d \textnormal{ s.t. }\ell\neq 1} \bt_{(j,\ell-1)} \cdot s^d(j,\ell-1)\leq s^d(\bt) \enspace .
  \end{equation*}
  The second inequality holds since $|K_{(j,\ell)}|\leq \bt_{(j,\ell-1)}$ and for every $i\in K_{(j,\ell)}$ it holds that $v(i)\leq s^d(j,\ell-1)$ by Lemma~\ref{lem:rounded_bounds}.
  Thus, for every $C\in \supp(\bx^{d,\bt})$ we have that $v(C)\leq s^{d}(\bt)$.
  Finally, for every $i\in \supp\left( (\bz^d-\br^d)\wedge \one_{L\setminus F_d} \right)$, there is $(j,\ell) \in \cG_d$ with $\ell\neq 1$ such that $i\in G^{d,j}_{\ell}$.
  Hence,
  \begin{equation}
  \label{eq:specific_type_dist}
    \sum_{C\in \cC} \bx^{d,\bt}_C\cdot  C(i)  = \Pr( i\in R) = \bt_{(j,\ell-1)} \cdot   \frac{\bz^d_i -\br^d_i}{\rho_{(j,\ell-1)}} \enspace .
  \end{equation}

  Let $ \bw'$ be the coverage of $\sum_{\bt\in \cT_d} \kappa_d(\bt) \cdot \bx^{d,\bt}$.
  By construction, we have $\bw'_i=0$ for any $i\in I$ such that $i\not\in \supp\left( (\bz^d-\br^d)\wedge \one_{L\setminus F_d} \right)$. For any  $i\in \supp\left( (\bz^d-\br^d)\wedge \one_{L\setminus F_d} \right)$, it holds that
  \begin{equation*}
    \bw'_i = \sum_{C\in \cC} \sum_{\bt\in \cT_d}  \kappa_d(\bt) \cdot \bx^{d,\bt}_C \cdot C(i) = \sum_{\bt \in \cT_d} \kappa_d(\bt) \cdot \bt_{(j,\ell-1)} \cdot   \frac{\bz^d_i -\br^d_i}{\rho_{(j,\ell-1)}} = \bz^d_i -\br^d_i,
  \end{equation*}
  where the second equality is by \eqref{eq:specific_type_dist}, and the last equality is by the definition of $\rho_{(j,\ell)}$. 
\end{proof}

Recall that $Q = \supp(\bz)\setminus L$.
The assignment of items in $Q$ relies on integrality properties of polytopes.
Define $M = \exp(-\delta^{-9})\cdot  \OPT+ \exp(\delta^{-11})$  and 
\begin{equation*}
  B= \left\{(d,\bt,m)~|~d\in\{1,2\},~\bt \in \cT_d,~m\in [\kappa_d(\bt)] \right\} \cup \{1,\hdots,M\} \enspace .
\end{equation*}
We consider $B$ as a set of bins, and define a polytope
\begin{equation}
\label{eq:P_def}
  P=\left\{\bmu \in [0,1]^{Q\times B} ~\middle|~
  \begin{aligned} 
	&\sum_{b\in B} \bmu_{i,b} = 1&~~ &\forall i\in Q\\
    &\sum_{i\in Q} \bmu_{i,(d,\bt,m)} \cdot v(i)\leq \bone-s^d(\bt) && \forall d\in \{1,2\},~\bt \in \cT_d,~m\in \{1,\hdots,\kappa_d(\bt)\}\\
    &\sum_{i\in Q} \bmu_{i,m} \cdot v(i)\leq \bone && \forall m\in \{ 1,\ldots, M\}
  \end{aligned} 
   \right\} 
\end{equation}
The entry $\bmu_{i,b}$ in $P$ represents a fractional assignment of an item $i\in Q$ to bin $b$.
The first constraint in \eqref{eq:P_def} represents the requirement that each item is fully assigned, and the remaining constraints represent a volume limit for each bin. 

The following is a well known integrality property of $P$ (see, e.g., Bansal et al.~\cite{BansalEK2016}).
\begin{lemma}
\label{lem:integrality}
  Let $\bmu$ be a vertex of $P$.
  Then $|\{i\in Q~|~\exists b\in B:~0<\bmu_{i,b}<1\}|\leq 2\cdot |B|$.
\end{lemma}
Before we use \Cref{lem:integrality}, we need to show that $P$ has a vertex.
\begin{lemma}
\label{lem:P_not_empty}
  It holds $P\neq \emptyset$.
\end{lemma}
\begin{proof}
  Ideally, we would like to define $\bmu_{i,(d,\bt,m)} = \frac{a^{d,\bt}_i}{\kappa_d(\bt)}$ for any $i\in Q$, $d\in \{1,2\}$, $\bt\in\cT_d$ and\linebreak $m\in \{1,\hdots,\kappa_d(\bt)\}$.
  Using~\eqref{eq:z_condition} we can show that  $\sum_{i\in Q} \bmu_{i,(d,\bt,m)} \cdot v_{d'}(i)$ is not significantly larger than $\bone-s^d(\bt)$; however, we cannot show it is smaller (or equal) to $\bone-s^d(\bt)$.
  Thus, the suggested vector $\bmu$ may not satisfy the properties in \eqref{eq:P_def}.
  We use \Cref{lem:refinement} to overcome this difficulty.
  Specifically, we define $\bmu_{i,(d,\bt,m)} = \frac{a^{d,\bt}_i}{\kappa_d(\bt)}$ for items $i\in Q\setminus X_1 \setminus X_2$, where the sets $X_1$ and $X_2$ are obtained via \Cref{lem:refinement}.
  The value of $\bmu_{i,m}$ is subsequently  increased  for $i\in X_1\cup X_2$ to ensure the first constraint in \eqref{eq:P_def} holds.
  Property~\ref{refinment:prop1} of \Cref{lem:refinement} is used to show that $\sum_{i\in Q} \bmu_{i,m} \cdot v(i)\leq \bone $, and property~\ref{refinement:prop2} of the lemma is used to show that $\sum_{i\in Q} \bmu_{i,(d,\bt,m)} \cdot v(i)\leq \bone-s^d(\bt)$.
  We now proceed to the formal proof. 

  Recall that $H^{d,\bt,d'}_1,\ldots ,H^{d,\bt,d'}_q$ is the refinement of $\ba^{d,\bt}$ and $q=\ceil{\exp(\delta^{-10})}$ in dimension $d'$. 
  For every $d,d'\in \{1,2\}$, $\bt \in \cT_d$ and $j = 1,\hdots,q$ it holds that 
  \begin{equation*}
    \begin{aligned}
      \sum_{i\in H^{d,\bt,d'}_{j} \cap Q}& \ba^{d,\bt}_{i} \cdot v_{d'}(i)= \bz \cdot \left(\one_{H^{d,\bt,d'}_{j}}\bullet  \ba^{d,\bt}\bullet \bv^{d'} \right)\\
      &\leq \beta\cdot  \bw \cdot  \left(\one_{H^{d,\bt,d'}_{j}}\bullet  \ba^{d,\bt}\bullet \bv^{d'} \right) + \frac{1}{\varphi^{10}(\delta)} \cdot \OPT \cdot \max\left\{ \sum_{i\in C}\one_{i\in H^{d,\bt,d'}_{j}}\cdot   \ba^{d,\bt}_i \cdot \bv_{d'}(i) ~\middle|~C\in \cC\right\}\\
      &\leq \beta \cdot  \| \one_{H^{d,\bt,d'}_{j}}\bullet  \ba^{d,\bt}\bullet \bv^{d'} \| + \frac{1}{\varphi^{10}(\delta)} \cdot \OPT\cdot \max\left\{v_{d'}( H^{d,\bt,d'}_j \cap C)~\middle|~C\in \cC\right\} \enspace .
    \end{aligned}
  \end{equation*}
  The equality follows from the definition of $Q$.
  The first inequality follows from \eqref{eq:z_condition} and the fact that\linebreak $\one_{H^{d,\bt,d'}_{j}}\bullet  \ba^{d,\bt}\bullet \bv^{d'}\in \cS_{\textnormal{small}}\subseteq \cS$.
  The second inequality holds, as $\bw$ is small-items integral and $\supp(\ba^{d,\bt}) \subseteq  \supp(\bw)\setminus L$. 
  Thus, by \Cref{lem:refinement}, for every $d,d'\in \{1,2\}$, $\bt \in \cT_d$ and $j = 1,\hdots,q$ there is a set $X^{d,\bt,d'}\subseteq Q$ such that  
  \begin{equation}
  \label{eq:X_prop}
    \left\|\one_{X^{d,\bt,d'}}\bullet  \ba^{d,\bt} \bullet (\bv^1+\bv^2) \right\| \leq \frac{16}{q}\cdot \OPT +2q\delta \textnormal{~~~and~~~} \left\|\one_{Q\setminus X^{d,\bt,d'}}\bullet  \ba^{d,\bt} \bullet \bv^{d'} \right\| \leq \beta \cdot \ba^{d,\bt} \cdot \bv^{d'} \enspace .
  \end{equation} 

  Define $\bmu \in [0,1]^{Q\times B}$ by
  \begin{equation*}
    \bmu_{i,(d,\bt, m)} =
    \begin{cases}
      \frac{\ba^{d,\bt}_i}{\kappa_d(\bt)},& i\in Q\setminus X^{d,\bt,1}\setminus X^{d,\bt,2},\\
	  0,&\textnormal{otherwise} \enspace .
    \end{cases}
  \end{equation*}
  for every $i\in Q$, $d\in \{1,2\}$, $\bt\in \cT_d$ and $m = 1,\hdots,\kappa_d(\bt)$.
  Also, for every $i\in Q$ and $m = 1,\hdots,M$ define
  \begin{equation*}
    \bmu_{i,m} = \sum_{d\in \{1,2\}}\sum_{\bt\in\cT_d} \frac{\ba^{d,\bt}_i \cdot \one_{i\in X^{d,\bt,1} \cup X^{d,\bt,2}}}{M} \enspace .
  \end{equation*}

  Next, we show that $\bmu\in P$.
  For every $i\in Q$ it holds that
  \begin{equation*}
    \begin{aligned}
	  \sum_{b\in B} &\bmu_{i,b}= \sum_{d\in \{1,2\}} \sum_{\bt\in \cT_d} \sum_{m\in [\kappa_d(\bt)]} \bmu_{i,(d,\bt,m)} +\sum_{m\in [M]} \bmu_{i,m}\\
	  & =\sum_{d\in \{1,2\}} \sum_{\bt\in \cT_d}   \kappa_d(\bt)\cdot  \frac{\ba^{d,\bt}_i}{\kappa_d(\bt)} \cdot \one_{i\in Q\setminus X^{d,\bt,1}\setminus X^{d,\bt,2}} +  \sum_{d\in \{1,2\}} \sum_{\bt\in \cT_d}   M\cdot \frac{a^{d,\bt}_i}{M} \cdot \one_{ i\in X^{d,\bt,1}\cup X^{d,\bt,2} }\\
	  &= \sum_{d\in \{1,2\}} \sum_{\bt\in \cT_d} \ba^{d,\bt}\cdot \one_{i\in Q\setminus X^{d,\bt,1}\setminus X^{d,\bt,2}}+\sum_{d\in \{1,2\}} \sum_{\bt\in \cT_d} a^{d,\bt }\cdot \one_{ i\in X^{d,\bt,1}\cup X^{d,\bt,2} } \\
	  &= \sum_{d\in \{1,2\}} \sum_{\bt\in \cT_d} \ba^{d,\bt}\\
	  &=\bw^1_i+\bw_i^2=1,
    \end{aligned}
  \end{equation*}
  where the fifth equality follows from \eqref{eq:association}. 

  For every $d,d'\in \{1,2\}$, $\bt\in \cT_d$ we have
  \begin{equation*}
    \begin{aligned}
      \ba^{d,\bt} \cdot \bv^{d'} &=\sum_{i\in I\setminus L} v_{d'}(i)\sum_{C\in \cC^* \textnormal{ s.t }\type^d(C)=\bt} \blam_C^d\cdot C(i) = \sum_{C\in \cC^* \textnormal{ s.t }\type^d(C)=\bt} \blam_C^d\cdot \sum_{i\in I\setminus L} v_{d'}(i)\cdot C(i)\\
	  &\leq \sum_{C\in \cC^* \textnormal{ s.t }\type^d(C)=\bt} \blam^d_C\cdot  \left( 1-s^d_{d'}(\bt)\right) =\left(1-s^{d}_{d'}(\bt)\right)\cdot \eta_d(\bt),
    \end{aligned}
  \end{equation*}
 where the first equality is by \eqref{eq:association} and the inequality is by \Cref{lem:rounded_bounds}. 
  Thus, for $m = 1,\hdots,\kappa_d(\bt)$ we have
  \begin{equation*}
    \sum_{i\in Q} \bmu_{i,(d,t,m)} \cdot v_{d'}(i) = \sum_{i\in Q\setminus X^{d,\bt, 1}\setminus X^{d,\bt,2}} \frac{\ba^{d,\bt}_i  \cdot v_{d'}(i)}{\kappa_d(\bt)}   \leq   \frac{\beta\cdot\ba^{d,\bt} \cdot\bv^{d'}}{\kappa_d(\bt)} \leq 
    \frac{\beta\cdot \left(1-s^{d}_{d'}(\bt)\right) \eta_d(\bt)}{\kappa_d(\bt)}\leq1-s^{d}_{d'}(\bt),
  \end{equation*}
  where the first inequality is by \eqref{eq:X_prop}. 

  Finally, for every $m = 1,\hdots,M$  and $d'\in \{1,2\}$ we have
  \begin{equation*}
    \begin{aligned}
      \sum_{i\in Q}& \bmu_{i,m} \cdot v_{d'}(i)  =\sum_{i\in Q} v_{d'}(i) \sum_{d\in \{1,2\}} \sum_{\bt \in \cT_d} \frac{\ba^{d,\bt}_i \cdot \one_{i\in X^{d,\bt,1} \cup X^{d,\bt,2}} }{M}\\
      & \leq \frac{1}{M}  \sum_{d\in \{1,2\}} \sum_{\bt \in \cT_d} \left( \|\one_{X^{d,\bt,1}} \bullet \ba^{d,\bt}  \bullet \bv^{d'}\| +\|\one_{X^{d,\bt,2}} \bullet \ba^{d,\bt}  \bullet \bv^{d'}\|\right)\\
      &\leq \frac{1}{M }\sum_{d\in \{1,2\}} \sum_{\bt \in \cT_d} \left(  \frac{32}{q}\cdot \OPT +4q\delta \right) \leq 1,
    \end{aligned}
  \end{equation*}
  where the second inequality is by~\eqref{eq:X_prop} and the last inequality holds since $|\cT_d|\leq \exp(\delta^{-6})$, $q\geq \exp(\delta^{-10})$ and $M= \exp(-\delta^{-9}) \cdot \OPT+ \exp(\delta^{-11})$.
  Thus, $\bmu\in P$, i.e., $P\neq \emptyset$. 
\end{proof}

We now have the tools to prove the following.
\begin{lemma}
\label{lem:y_bound}
  It holds that $\OPTf(\by)\leq (1+8\delta)|B|+1$.
\end{lemma}
\begin{proof}
  Let $\bmu^*$ be a vertex of $P$, and let $Q_I =\{i\in Q~|~\exists b\in B:~\bmu^*_{i,b}=1 \}$.
  By \Cref{lem:integrality} it holds that $|Q\setminus Q_I|\leq 2|B|$.
  As $Q\subseteq I\setminus L$, it follows that the items of $Q\setminus Q_I$ can be packed into $4\delta |Q\setminus Q_I|+1 \leq 8\delta |B|+1$ bins using the First-Fit strategy (\Cref{lem:first_fit}). 
  Thus, $\OPTf(\one_{Q\setminus Q_I})\leq  8 \delta |B| +1$. 

  For every $b\in B$ define $C_b = \{ i\in Q~|~\bmu^*_{i} = 1\}$.
  It follows that $Q_I = \bigcup_{b\in B} C_b$. 
  Recall that $\bx^{d,\bt}$ are the vectors defined in \Cref{lem:shifting_application}. 
  For every $(d,\bt, m)\in B\setminus \{1,\hdots,M\}$ define a vector $\bgam^{d,\bt,m}\in [0,1]^{\cC}$ by $\bgam^{d,\bt,m}_{C\cup C_{d,\bt, m} } =\bx^{d,\bt}_{C}$ for any $C\in \supp(\bx^{d,\bt})$, and $\bgam^{d,\bt,m} _{C'}= 0$ for any other configuration $C'\in \cC$.
  By definition of $P$, it holds that $v(C_{d,\bt,m}) \leq \bone -s^d(\bt)$, and by \Cref{lem:shifting_application}, for every $C\in \supp(\bx^{d,\bt})$ it holds that $v(C)\leq s^d(\bt)$; thus, $C\cup C_{d,\bt,m}\in \cC$, and $\bgam^{d,\bt,m}$ is well defined.
  Also, for any $m = 1,\hdots,M$ define $\bgam^m\in [0,1]^C$ by $\bgam^m_{C_m} = 1$ and $\bgam^{m}_C=0$ for $C\in \cC\setminus\{C_m\}$.   

  Define $\bx= \sum_{b\in B} \bgam^b$.
  We show that $\bx$ is a solution for $\LP\left(\sum_{d\in \{1,2\}} \left((\bz^d -\br^d) \wedge \one_{L\setminus F_d} \right) +\one_{Q_I}  \right)$. For $i\in L$ we have
  \begin{equation*}
    \begin{aligned}
      \sum_{C\in \cC} \bx_C \cdot C(i) &=\sum_{C\in \cC} \sum_{b\in B} \bgam^b_C\cdot C(i) \\
      &= \sum_{C\in \cC} \sum_{d\in \{1,2\} } \sum_{\bt\in \cT_d} \sum_{m\in [\kappa_d(\bt) ]} \bx^{d,\bt}_C\cdot C(i) \\
      &=  \sum_{d\in \{1,2\} } \sum_{C\in \cC} \sum_{\bt\in \cT_d} \kappa_d(\bt) \cdot \bx^{d,\bt}_C\cdot C(i)\\
      &=  \sum_{d\in \{1,2\} }\left( (\bz^d -\br^d) \wedge \one_{L\setminus F_d} \right) \enspace .
    \end{aligned} 
  \end{equation*}
  The second equality holds by definition of $\bgam^b$, and since the sets $C_b$ do not contain large items.
  The last equality is by \Cref{lem:shifting_application}. 
  For any $i\in Q_I$ there is a unique $b\in B$ such that $i\in C_b$.
  Thus, $\sum_{C\in \cC} \bx_C\cdot C(i)  =\sum_{C\in \cC} \bgam^b_C\cdot C(i)=1$.
  Therefore, $\bx$ is a solution for the linear program $\LP\left(\sum_{d\in \{1,2\}} \left((\bz^d -\br^d) \wedge \one_{L\setminus F_d} \right) +\one_{Q_I}  \right)$.
  As $\|\bgam^b\|=1$ for every $b\in B$, it follows that $\|\bx\|=B$.
  Thus, 
  \begin{equation*}
    \OPTf\left(\sum_{d\in \{1,2\}} \left((\bz^d -\br^d) \wedge \one_{L\setminus F_d} \right) +\one_{Q_I}  \right)\leq \|\bx\|=B,
  \end{equation*}
  and by definition of $\by$ \eqref{eq:y_def}, we have
  \begin{equation*}
    \OPTf(\by)= \OPTf\left(\sum_{d\in \{1,2\}} \left((\bz^d -\br^d) \wedge \one_{L\setminus F_d} \right) +\one_{Q_I}  \right)+ \OPTf(\one_{Q\setminus Q_I}) \leq (1+8\delta) |B|+1 \enspace .\qedhere
  \end{equation*}
\end{proof}

Observe that 
\begin{equation}
	\label{eq:B_bound}
\begin{aligned}
	|B|&=\sum_{d\in \{1,2\}}\sum_{\bt\in \cT_d} \kappa_d(t) + M= \sum_{d\in \{1,2\}}\sum_{\bt\in \cT_d}\left( \ceil{\beta \eta_d(t)} + 2\delta^{-1}\right) + \exp(-\delta^{-9} )\cdot \OPT+\exp(\delta^{-11}) \\
	&\leq \beta \|\blam\| + (|\cT_1|+|\cT_2| ) \cdot (1+2\delta^{-1})+ \exp(-\delta^{-9} )\cdot \OPT+\exp(\delta^{-11})\\
	& \leq \beta \|\blam\| + \exp(-\delta^{-9} )\cdot \OPT+\exp(\delta^{-12}) \enspace .
	\end{aligned}
\end{equation}	
The first inequality holds since $\sum_{\bt\in \cT_d} \eta_d(\bt) =  \|\blam^d\|$, and the second inequality uses $|\cT_d|\leq \exp(\delta^{-6})$.
By \eqref{eq:first_scraping} we have
\begin{equation*}
  \begin{aligned}
    \OPTf(\bz)&\leq  \OPTf\left(\by\right)\ +  2\beta\delta \| \blam\| +2\delta^{-5}+ \frac{2\cdot \delta^{-11}}{\varphi^{10}(\delta)}\OPT\\
	&\leq (1+8\delta)|B|+1 +  2\beta\delta \| \blam\| +2\delta^{-5}+ \frac{2 \delta^{-11}}{\varphi^{10}(\delta)}\OPT\\
	&\leq (1+8\delta) \left(\beta \|\blam\| + \exp(-\delta^{-9} )\cdot \OPT+\exp(\delta^{-12})\right)+1 + 2\delta\beta  \| \blam\| +2\delta^{-5}+ \frac{2\cdot \delta^{-11}}{\varphi^{10}(\delta)}\OPT\\
    &\leq \beta (1+10\delta )\|\blam\| +\exp(\delta^{-20}) + \delta^{10} \OPT,
  \end{aligned}
\end{equation*}
where the second inequality is by \Cref{lem:y_bound}, the third inequality is by \eqref{eq:B_bound}, and the last inequality uses $\varphi(\delta)=\exp(\delta^{-20})$.
Thus, we showed that $\cS$ is a linear structure, which completes the proof of \Cref{lem:structural}. \qed

\subsubsection{Refinement for Small Items}
\label{sec:refinement}

\noindent
{\bf Proof of \Cref{lem:refinement}:}
Define $r(i) = \frac{v_d(i)}{v_{\hd}(i)}$ for any $i\in I$.
Assume, without loss of generality, that $I\setminus L=\{1,2,\ldots, s\}$ for some $s \in \mathbb{N}$, and $r(1)\leq r(2)\leq \ldots\leq r(s)$.

	If $\ba \cdot (\bv^1 + \bv^{2}) \leq \frac{1}{q^2}\OPT + 2q\delta$ define $H_1 =I\setminus L$ and $H_j= \emptyset$ for $j\in \{2,\ldots, q\}$. 
	 Let $Q\subseteq I\setminus L$ and $\beta\in[\frac{1}{q},1]$ which satisfies \eqref{eq:refinement_prop}. We can  select $X=I\setminus L$. It follows that  
	$\| \one_{Q\setminus X} \bullet \ba \bullet \bv^d \| =0 \leq \beta \cdot \ba \cdot \bv^d $ and $\|\one_{X}\bullet \ba \bullet (\bv^1+\bv^2)\| = \ba \cdot \left( \bv^1+ \bv^2\right)\leq \frac{16}{q}\OPT+2q\delta$.
	This shows the statement of the lemma in case $\ba \cdot (\bv^1 + \bv^{2}) \leq \frac{1}{q^2}\OPT + 2q\delta$.
	We henceforth assume that 
	\begin{equation}
		\label{eq:vplus_lb}
	\ba \cdot (\bv^1 + \bv^2) > \frac{1}{q^2}\OPT + 2q\delta \enspace .
	\end{equation}

	Define $h_0 = 0$, and for $j = 1,\hdots,q$ set
	\begin{equation} 
		\label{eq:hj_def}
		h_j = \min \left\{ i\in [s]~\middle| ~  
		\left(  \ba \wedge  \one_{[i]}\right) \cdot  (\bv^1+\bv^2)   \geq \frac{j}{q} \cdot \ba \cdot (\bv^1+\bv^2) \right\} \enspace .
	\end{equation}
	Observe that the set over which the minimum is taken is non-empty for all $j\in\{1,\hdots,q\}$.
	Hence, $h_j$ is well defined.
	Define $H_j = \{ i\in \{1,\hdots,s\}~|~h_{j-1}<i\leq h_j\}$; then $H_j = \{1,\hdots,h_j\}\setminus \{1,\hdots,h_{j-1}\}$ for $j = 1,\hdots,q$.

	Let $Q\subseteq I\setminus L $ and $\beta\in [\frac{1}{q},1]$ satisfy \eqref{eq:refinement_prop}. 	
	For $j = 1,\hdots,q$ and $C\in \cC$ it holds that $v_d(C\cap H_j) \leq 1$, and 
	\begin{equation*}
	  v_d(C\cap H_{j}) = \sum_{i\in C\cap H_j} v_d(i) = \sum_{i\in C\cap H_j}  v_{\hd}(i)\cdot r(i )\leq r(h_j)\sum_{i\in C\cap H_j} v_{\hd}(i)\leq r(h_j) \enspace .
	\end{equation*}
	Thus, $v_d(C\cap H_{j}) \leq \min\{ 1, r(h_j)\}$.
	We conclude that 
	\begin{equation}
		\label{eq:refinement_conf_bound}
		\max\left\{ v_d(C\cap H_j) ~|~C\in \cC\right\} \leq \min\{1,r(h_j)\}
	\end{equation}
	for $j = 1,\hdots,q$.
	
	We use in our proof the following inequality (that we prove later), for $j = 2,\hdots,q$:
	\begin{equation}
		\label{eq:hj_bound}
		\| \one_{H_j}\bullet \ba \bullet \bv^d\| \geq \frac{1}{2} \min\{1,r(h_{j-1})\} \cdot \frac{1}{q^3 }\OPT,
	\end{equation}
	
	For $j = 1,\hdots,q$ define
	\begin{equation*}
	  \beta_j = \max\left\{0, ~\|\one_{Q\cap H_j} \bullet \ba \bullet \bv^d \| - \beta  \| \one_{H_j}\bullet \ba \bullet \bv^d \| \right\} \enspace .
	\end{equation*}
	It follows from  \eqref{eq:refinement_prop} and \eqref{eq:refinement_conf_bound} that
	\begin{equation*}
	\beta_j \leq
	\frac{\OPT}{q^5} \cdot  
	\max\left\{ v_d(C\cap H_j) ~|~C\in \cC\right\} 
	\leq  \min\{r(h_j),1 \}\cdot
	 \frac{\OPT}{ q^5} \enspace .
	\end{equation*}
	For every $j\in [q]\setminus \{1\}$ we define a set $X_j\subseteq Q\cap H_j$. If 
	$\| \one_{Q\cap H_j} \bullet \ba \bullet \bv^d\| +\beta_{j-1}-\beta_{j}\leq \beta\cdot \| \one_{H_j} \bullet \ba \bullet \bv^d\|$ then we define $X_j= \emptyset$.
	Otherwise, we define $X_j$ to be an inclusion-minimal subset of $Q\cap H_j$ such that $\| \one_{Q\cap H_j\setminus X_j} \bullet \ba \bullet \bv^d\| +\beta_{j-1}-\beta_{j}\leq \beta\cdot \| \one_{H_j} \bullet \ba \bullet \bv^d\|$.
	Observe that
	\begin{equation*}
	\| \one_{Q\cap H_j\setminus (Q\cap H_j)} \bullet \ba \bullet \bv^d\|
	+\beta_{j-1} - \beta_{j} \leq \beta_{j-1} \leq \min\{1,
	\tau_{j-1}\} \cdot  \frac{\OPT}{q^5 }  
	\leq  
	\beta \| \one_{H_j}\bullet \ba \bullet \bv^d\|,
	\end{equation*}
	where the last inequality follows from $\beta \geq \frac{1}{q}$ and \eqref{eq:hj_bound}.
	Hence, there exists $X_j\neq \emptyset$.
	As the set is inclusion-minimal, it follows that there is $x_j\in X_j$ such that $\|\one_{X_j\setminus \{x_j\}} \bullet \ba \bullet \bv^d \|\leq \beta_{j-1}\leq \frac{\OPT}{q^5}$.
	Thus,
	\begin{equation*}
	\begin{aligned}
			\|\one_{X_j\setminus \{x_j\}}& \bullet \ba \bullet \bv^{\hd} \|
		= \sum_{i\in X_j\setminus \{x_j \}} \ba_i \cdot v_{\hd} (i) = 
		\sum_{i\in X_j\setminus \{x_j \}} \ba_i\cdot \frac{v_d(i) }{r(i)} \leq 
		 \sum_{i\in X_j\setminus \{x_j \}} \ba_i\cdot \frac{v_d(i) }{r(h_{j-1})}\\
		 &=
		\frac{\| \one_{X_j\setminus \{x_j\}} \bullet \ba \bullet \bv^d\| }{r(h_{j-1})} 
		\leq \frac{\beta_{j-1}}{r(h_{j-1})} \leq 
		\frac{1}{r(h_{j-1})}
		\min\{r(h_{j-1}),1 \}\cdot \frac{\OPT(I,v)}{ q^5}\leq \frac{\OPT(I,v)}{ q^5},
	\end{aligned}
	\end{equation*}
	where the first inequality holds as $X_j\subseteq H_j$.

	Define  $X= (H_q\cap Q) \cup \bigcup_{j=2}^{q}X_j$.
	It follows that 
	\begin{equation*}
	\begin{aligned}
	\|\one_{Q\setminus X} \cdot \ba \cdot \bv^d\| &=  \sum_{j=1}^{q-1}\|  \one_{ (Q\setminus X) \cap H_j} \cdot \ba \cdot \bv^d \|	\\
		&= \|  \one_{ (Q\setminus X) \cap H_1} \cdot \ba \cdot \bv^d \|	 -\beta_1 +\sum_{j=2}^{q-1}
		\left(
		\|  \one_{ (Q\setminus X) \cap H_j} \cdot \ba \cdot \bv^d \|	 + \beta_{j-1} - \beta_j
		\right) + \beta_{q-1}\\
		&\leq \beta \sum_{j=1}^{q-1} 	\|  \one_{  H_j} \cdot \ba \cdot \bv^d \|	 +\beta_{q-1}\\
		&\leq \beta \sum_{j=1}^{q-1} 	\|  \one_{  H_j} \cdot \ba \cdot \bv^d \|	 +\min\{r(h_{q-1}),1\} \cdot \frac{\OPT(I,v)}{ q^5} \\
		&\leq \beta \sum_{j=1}^{q} 	\|  \one_{  H_j} \cdot \ba \cdot \bv^d \|	 = 
		\beta \cdot \ba \cdot \bv^d \enspace .
	\end{aligned}
	\end{equation*}
	The first equality holds as $\supp(\ba )\subseteq \bigcup_{j\in [q]} H_j $. 
	The first inequality follows from the definitions of~$\beta_1$ and $X_j$ (for $j\in \{2,\ldots, q-1\}$). The last inequality follows from $\beta \geq\frac{1}{q}$ and \eqref{eq:hj_bound}.
	
	Note that $\|\one_{H_q} \cdot \ba \cdot \left(\bv^1+\bv^2\right) \| \leq \frac{\ba \cdot \bv^d}{q}\leq \frac{2\cdot  \OPT}{q}$.
	Thus, 
	\begin{equation*}
	\begin{aligned}
	\|\one_{X} \bullet \one_{A} \bullet(\bv^1+\bv^2)\| &\leq  	\|\one_{H_q} \cdot \ba \cdot \left(\bv^1+\bv^2\right) \|  +\sum_{j=2}^{q }\|\one_{X_j} \cdot \ba \cdot \left(\bv^1+\bv^2\right) \| \\
	&\leq \frac{2\cdot  \OPT}{q} + q \cdot 2\cdot \frac{\OPT}{q^5}+ 2\delta q\leq \frac{16}{q}\OPT+2\delta q \enspace .
	\end{aligned}
	\end{equation*}

	It remains to show that~\eqref{eq:hj_bound} holds.
	For $j = 1,\hdots,q$, we have  
	\begin{equation}
		\label{eq:vhj_first}
		\begin{aligned}
			\|\one_{H_j}\bullet \ba \bullet(\bv^1+\bv^2)\|&= 
				\|\one_{[h_j]}\bullet \ba \bullet(\bv^1+\bv^2)\|
				-	\|\one_{h_{j-1}}\bullet \ba \bullet(\bv^1+\bv^2)\|
			\\
			&\geq \frac{j}{q} \ba \cdot (\bv^1+\bv^2)- \frac{j-1}{q}\ba\cdot  (\bv^1+\bv^2)-2\delta \\
			&=\frac{1}{q}\ba \cdot (\bv^1+\bv^2)-2\delta\\
			& \geq \frac{1}{q}  \left( \frac{1}{q^2} \OPT +2\delta q\right) -2\delta\\
			& =\frac{1}{q^3} \OPT(I,v) \enspace .
		\end{aligned}
	\end{equation}
	The first inquality follows from~\eqref{eq:hj_def} and $v_1(i)+v_2(i)\leq 2\delta$ for all $i\in I\setminus L$.
	The second inequality follows from~\eqref{eq:vplus_lb}.
	Additionally, for $j = 2,\hdots,\ell$ we have
	\begin{equation}
		\label{eq:vhj_second}
		\begin{aligned}
		\|\one_{H_j}\bullet \ba \bullet(\bv^1+\bv^2)\|  &= \|\one_{H_j}\bullet \ba \bullet\bv^d\|+ \|\one_{H_j}\bullet \ba \bullet\bv^{\hd}\|\\
			&= \|\one_{H_j}\bullet \ba \bullet\bv^d\|+ \sum_{i\in H_j} \ba_i \cdot v_{\hd}(i) \\
			&= \|\one_{H_j}\bullet \ba \bullet\bv^d\| + \sum_{i\in H_j} \ba_i\cdot  \frac{v_d(i)}{ r(i)} \\
			&\leq \|\one_{H_j}\bullet \ba \bullet\bv^d\| + \sum_{i\in H_j}\ba_i\cdot  \frac{v_d(i)}{ r(h_{j-1}) } \\
			&= \|\one_{H_j}\bullet \ba \bullet\bv^d\| \cdot \left(1+ \frac{1}{r(h_{j-1})}\right),
		\end{aligned}
	\end{equation}
	where the inequality follows from $r(1)\leq r(2)\leq \ldots \leq r(p)$.
	Using \eqref{eq:vhj_first} and \eqref{eq:vhj_second}, we get
	\begin{equation*}
		\forall j = 2,\hdots,q:~~~~~\|\one_{H_j}\bullet \ba \bullet\bv^d\| \geq \left( 1+ \frac{1}{r(h_{j-1})}\right)^{-1} \cdot \frac{1}{q^3} \OPT \geq\frac{1}{2} \min\{1,\tau_{j-1 }\} \cdot \frac{1}{q^3 }\OPT	,
	\end{equation*}
	where the inequality follows from $\left( 1+x^{-1}\right)^{-1} \geq \frac{1}{2} \min \{ 1, x\}$ for every $x\geq 0$. Inequality \eqref{eq:hj_bound} follows from the last inequality. 
\qed

\subsection{Existence of $\psi$-Relaxations}
\label{sec:alpha_relax}
In this section we  prove \Cref{lem:delta_relaxed,lem:auxilary_struct,lem:relax_small}.
That is, we  show how to obtain relaxations for various configurations.

\noindent
{\bf Proof of \Cref{lem:delta_relaxed}:}
  Let $S\subseteq C \setminus L$ be an inclusion-minimal set such that either $v_1(C\setminus S)\leq 1-\delta$ or $v_2(C\setminus S)\leq 1-\delta$.
  As $S$ is inclusion-minimal, it holds that 
  \begin{equation}
  \label{eq:minimal_S_prop}
    \forall i \in S:~~~~ v\bigg(C\setminus \left(S\setminus \{i\}\right) \bigg)> (1-\delta, 1-\delta).
  \end{equation} 
  Such a set exists, since $C\in \cC_0$.
	
  In the following we show that $v(S)\leq (2\delta, 2\delta )$.
  Suppose, for sake of contradiction, that \mbox{$v_1(S)>2\delta$} or $v_2(S)>2\delta$.
  Then $S\neq \emptyset$ and there is an $i\in S$.
  Assume, without loss of generality, that $v_1(S)>2\delta$.
  Then $v_1(S\setminus \{i\})>\delta$ as all items in $S$ are small, and $i\in S$.
  Therefore,
  \begin{equation*}
    v_1\bigg(C\setminus (S\setminus \{i\})\bigg) = v_1(C)-v_1(S\setminus \{i\}) \leq 1 -\delta,
  \end{equation*}
  contradicting \eqref{eq:minimal_S_prop}.
  Thus, $v(S)\leq (2\delta, 2\delta )$. 
	
  Define $C_1 = C\setminus S$ and $C_2\in \cC^*$ by
  \begin{equation*}
	C_2(i) = \begin{cases}
	           \kappa, & i\in S,\\
	           0,      & i\not\in S \end{cases}
  \end{equation*}
  for $i\in I$, where $\kappa=\floor{\frac{1}{2} (\delta^{-1} -1 )}$.
  Observe that $C_1$ has $\delta$-slack by definition of $S$.
  Additionally,  
  \begin{equation*}
	v_1(C_2)\leq v_1(S)\cdot \kappa \leq 2\delta \kappa \leq 2\delta \cdot \frac{1}{2} (\delta^{-1}-1)\leq 1-\delta,
  \end{equation*}
  thus $C_2$ is a multi-configuration with $\delta$-slack. 
	
  Define $\blam\in [0,1]^{\cC^*}$ by $\blam_{C_1}= 1$, $\blam_{C_2}=\frac{1}{\kappa}$ and $\blam_{C'}=0$ for $C'\in \cC\setminus \{C_1,C_2\}$. 
  Clearly, for any $C'\in \cC^*$ such that $\blam_{C'}>0$ it holds that $C'$ has $\delta$-slack.
  Thus, $\blam$ has $\delta$-slack. 
	
  For any $i\in C\setminus S$ we have
  \begin{equation*}
	\sum_{C'\in \cC^*} \blam_{C'}\cdot C'(i)=C_1(i)+\frac{1}{\kappa}\cdot C_2(i)=1+0=1 \enspace .
  \end{equation*}
  For any $i\in S$ it holds that 
  \begin{equation*}
	\sum_{C'\in \cC^*} \blam_{C'}\cdot C'(i)=C_1(i)+\frac{1}{\kappa}\cdot C_2(i)=0+\frac{1}{\kappa}\cdot \kappa = 1 \enspace .
  \end{equation*}
  For any $i\in I\setminus C$ it holds that 
  \begin{equation*}
	\sum_{C'\in \cC^*} \blam_{C'}\cdot C'(i)=C_1(i)+\frac{1}{\kappa}\cdot C_2(i)=0+\frac{1}{\kappa}\cdot 0 = 0 \enspace .
  \end{equation*}
	
  Since $\delta^{-1}\in \mathbb{N}$, we have $\kappa\geq \frac{1}{2}(\delta^{-1}-1) -\frac{1}{2} = \frac{1}{2} \delta^{-1} -1$.
  Therefore, 
  \begin{equation*}
	\|\blam\|=\sum_{C'\in \cC^*} \blam_{C'} = \blam_{C_1}+\blam_{C_2} =1+\frac{1}{\kappa}\leq 1+ \frac{1}{ \frac{1}{2} \delta^{-1} -1}=  1+ \frac{2\delta }{ 1-2\delta}\leq 1 + 4\delta,
  \end{equation*}
  where the last inequality holds as $\delta \leq 0.1$
	
  We showed that $\blam$ is  a $(1+4\delta)$-relaxation of $C$.
  This completes the proof of the lemma.
\qed 

\medskip

\noindent
{\bf Proof of \Cref{lem:auxilary_struct}:}
  Let $C\cap L=\{i_1,\ldots, i_{h}\}$.
  Define $h$ configurations $C_1,\ldots, C_{h}$ by $C_{\ell} = C\setminus \{i_{\ell}\}$ for $ \ell=1,\ldots, h-1$  and $C_h = C\cap L \setminus \{i_h\}$.
  It can be easily shown that $C_1,\ldots, C_{h}$ are configurations.
  Define $\blam\in [0,1]^{\cC^*}$ by 
  \begin{equation*}
    \blam_{C'} = \begin{cases} \frac{1}{h-1},& C' = C_{\ell} \textnormal{ for some } h\in\{1,\hdots,\ell\}, \\ 0, &\textnormal{otherwise} \enspace .\end{cases}
  \end{equation*}
  For $\ell = 1,\hdots,h$ it holds that $i_{\ell}$ is large; thus, there is $d_{\ell}\in \{1,2\}$ such that $v_{d_{\ell}}(i_{\ell})\geq \delta$.
  Therefore,
  \begin{equation*}
    v_{d_{\ell}}(C_\ell)\leq v_{d_{\ell}}(C\setminus \{i_{\ell}\}) = v_{d_{\ell}}(C) - v_{d_{\ell}}(i_{\ell} )\leq 1-\delta \enspace .
  \end{equation*}
  That is, all configurations $C_1,\ldots, C_h$ have $\delta$-slack.
  Thus, for any $C'\in \cC^*$ with $\blam_{C'}>0$ it holds that~$C'$ has $\delta$-slack.
  Hence, $\blam$ has $\delta$-slack. 

  For any $i\in C\cap L$ there is an $\ell\in\{1,\hdots,h\}$ such that $i=i_{\ell}$.
  Thus, 
  \begin{equation*}
    \sum_{C'\in \cC^*} \blam_{C'}\cdot C'(i)=\sum_{j=1}^{h} \frac{1}{h-1}\cdot C_j(i_{\ell} )= \sum_{j\in [h]\setminus \{\ell\}} \frac{1}{h-1}= 1 \enspace .
  \end{equation*}
  For any $i\in C\setminus L$ it holds that $i\in C_{\ell} $ for $\ell = 1,\hdots,h-1$; thus, 
  \begin{equation*}
    \sum_{C'\in \cC^*} \blam_{C'}\cdot C'(i)=\sum_{j=1}^{h} \frac{1}{h-1}\cdot  C_j(i) =\sum_{j=1}^{h-1} \frac{1}{h-1}  = 1 \enspace .
  \end{equation*}
  For any $i\in I\setminus C$ we have $i\not\in C_{\ell}$ for $\ell = 1,\hdots,h$.
  Therefore,
  \begin{equation*}
    \sum_{C'\in \cC^*} \blam_{C'}\cdot C'(i)=\sum_{j=1}^{h} \frac{1}{h-1}\cdot  C_j(i) =0 \enspace .
  \end{equation*}

  Finally, 
  \begin{equation*}
    \|\blam\|=\sum_{C'\in \cC^*} \blam_{C'} = \sum_{\ell=1}^{h} \blam_{C_{\ell}} = \frac{h}{h-1} \enspace .
  \end{equation*}
  Thus, we showed that $\blam$ is a $\frac{h}{h -1}$-relaxation of $C$. 
\qed

\medskip

\noindent
{\bf Proof of \Cref{lem:relax_small}:}
  Define $C'\in \cC^*$ by
  \begin{equation*}
    C'(i) =\begin{cases} \kappa,& i\in C,\\ 0,&\textnormal{otherwise},  \end{cases}
  \end{equation*}
  where $\kappa= \ceil{\frac{1}{2}\delta^{-1}} $ and $\blam\in [0,1]^{\cC^*}$ by $\blam_{C'} = \frac{1}{\kappa }$ and $\blam_D=0$ for any $D\in \cC^*\setminus \{C'\}$.
  Observe that
  \begin{equation*}
    v_1(C')=\sum_{i\in I} v_1(i)\cdot C'(i)  = \kappa \cdot v_1(C) \leq \ceil{\frac{1}{2}\delta^{-1}} \cdot \delta \leq \left(\frac{1}{2}\cdot \delta^{-1} +1 \right)\cdot \delta \leq \frac{1}{2}+\delta \leq 0.6\leq 1-\delta,
  \end{equation*}
  where the last two inequalities follow from $\delta \in (0,0.1)$.
  Thus, $C'$ has $\delta$-slack and hence $\blam$ is with $\delta$-slack.

  For any $i\in C$ it holds that $\sum_{D\in \cC^*} \blam_D\cdot D(i) = \blam_{C'}\cdot C'(i)=  \frac{1}{\kappa} \cdot \kappa =1$.
  Also, for any $i\in I\setminus C$ it holds that $\sum_{D\in \cC^*} \blam_D\cdot D(i) =  \blam_{C'}\cdot C'(i)=0$.
  Finally,
  \begin{equation*}
    \|\blam\| = \frac{1}{\kappa} \leq \frac{1}{ \ceil{\frac{1}{2} \delta}} \leq 2\delta \leq 4\delta \enspace .
  \end{equation*}
  Thus, $\blam$ is a $4\delta$-relaxation of $C$, as required.
\qed

\subsection{Solving the Matching-LP}
\label{sec:match_lp}
\newcommand{\primal}{{\textnormal{PRIMAL}}}
\newcommand{\dual}{{\textnormal{DUAL}}}
\newcommand{\bbeta}{{\bar{\beta}}}
\newcommand{\rMLP}{{\textnormal{rMLP}}}
\newcommand{\ellipsoid}{\textnormal{{\textsf{Ellipsoid}}}}
\newcommand{\elliq}{\textnormal{{\textsf{Ellipsoid\_Q}}}}
\newcommand{\ellir}{{\textnormal{\textsf{Ellipsoid\_R}}}}
\newcommand{\sepq}{\textnormal{{\textsf{Q\_separator}}}}
\newcommand{\sepr}{\textnormal{{\textsf{R\_separator}}}}
\newcommand{\conv}{\textnormal{\texttt{conv}}}
\newcommand{\cone}{\textnormal{\texttt{cone}}}

In this section we present a PTAS for the $\MLP$ problem, thus proving \Cref{lem:matching_ptas}.
Let $\delta \in (0,0.1)$ and $\eps\in (0,0.1)$.
Our objective is to obtain a polynomial-time $(1+O(\eps))$-approximation for $\MLP$.
To this end we use a result of Gr{\"o}tschel,  Lov{\'a}sz, and  Schrijver \cite{GrotschelLS1981}, which outlines the ellipsoid method via separation oracles.
A separation oracle for a polytope $P\subseteq \mathbb{R}^n$ accepts as input 
a point $\bx\in \mathbb{R}^n$, and either determines that $\bx\in P$ or finds $\bc \in \mathbb{R}^n$ such that $\bx\cdot \bc < \by\cdot \bc$ for any $\by\in P$.
That is, the oracle finds a hyperplane which separates between $\bx$ and the polytope $P$. It is also required that the encoding size of the returned hyperplane is polynomial in the query encoding  size. 
Given a separation oracle, the ellipsoid method either determines that $P=\emptyset$ or finds $\bx\in P$ in time polynomial in $n$ and the {\em facet complexity} of $P$.
As a consequence, if $P= \emptyset$ then the execution of the ellipsoid method is comprised of invocations of the separation oracle that always result in a separating hyperplane.
If $P\neq \emptyset$, then at least one of the calls to the separation oracle results in $\bx\in P$.

We use an approximate variant of the separation oracle 
commonly used to solve linear programs similar to \eqref{eq:config_LP} (see, e.g., \cite{KarmarkarK1982}).
In the classic setting, the ellipsoid method is executed with the dual of the original linear program, as this program has a polynomial number of variables.
For example, the dual linear program of~\eqref{eq:config_LP} has $|I|$ variables.
This approach cannot be directly implemented for $\MLP$, since the number of variables in both the primal and dual linear programs is non-polynomial in the $\delta$-huge free 2VBP instance $(I,v)$, due to the number of linear constraints required to represent the matching polytop.
We overcome this difficulty by projecting polytopes in a vector space of non-polynomial dimension into polytopes with polynomial dimension.
A similar approach was recently used by Fairstein et al.~\cite{FairsteinKS2021}.

We use the following definitions and lemmas from Gr{\"{o}}tschel et al.~\cite{GrotschelLS1988}. 

\begin{defn}[{\cite[Definition 6.2.2]{GrotschelLS1988}}]
\label{def:polyhedron}
  Let $P\subseteq \mathbb{R}^n$ be a polyhedron, and $\varphi\geq n+1$ a positive integer. 
  \begin{enumerate}
    \item We say that $P$ has {\em facet complexity at most $\varphi$} if there exists a system of linear inequalities with rational coefficients that has a solution set $P$ 
	  that the encoding length of each inequality in the system is at most $\varphi$.  
	\item We say that $P$ has a {\em vertex complexity at most $\varphi$} if there exist finite sets $V_1,V_2$ of rational vectors such that $P=\conv(V_1)+\cone(V_2)$ and each of the vectors in $V_1\cup V_2$ has encoding length at most $\varphi$.\footnote{$\conv(V_1)$ is the {\em convex hull} of $V_1$ and $\cone(V_2)$ is the {\em conic hull} of $V_2$.
	We refer the reader to Gr{\"{o}}tschel et al.~\cite{GrotschelLS1988} for the formal definitions.}   
	\item A {\em well-described polyhedron} is a triplet $(P, n, \varphi)$ where $P\subseteq \mathbb{R}^n$ is a polyhedron with facet complexity at most $\varphi$. 
  \end{enumerate}
\end{defn}
\begin{lemma}[{\cite[Lemma 6.2.4]{GrotschelLS1988}}]
\label{lem:facet_to_vertex}
  Let $P\subseteq \mathbb{R}^n$ be a polyhedron with facet complexity at most~$\varphi$.
  Then $P$ has vertex complexity at most $4n^2\cdot \varphi$. 
\end{lemma}
\begin{proposition}[{The {\em Ellipsoid Method}, \cite[Theorem 6.4.1]{GrotschelLS1988}}]
\label{lem:ellipsoid}
  There is an algorithm $\ellipsoid$ which given $n,\varphi$ and a separation oracle for a well-described polyhedron $(P,n,\varphi)$, determines that either $P=\emptyset$ or returns $\bx \in P$ in time polynomial in $n+\varphi$.
\end{proposition}

Throughout this section, we define multiple mathematical optimization problems.
We use $\OPT(\mathcal{P})$ to denote the value of the optimal solution for the problem $\mathcal{P}$. We use $\langle x\rangle$ to denote the encoding length of a number/vector/inequality  $x$. 
To simplify notation, we assume the $\delta$-2VBP instance $(I,v)$ is fixed throughout this section, and omit it from the input of the algorithms.
We use $G=(L, E)$ to denote the $\delta$-matching graph of $(I,v)$ as defined in \Cref{sec:alg_intro}, and $P_{\cM}(G)$ is the matching polytope of $G$.
Recall that $\cE$ is the projection function defined in \Cref{sec:alg_intro}.

We first simplify our problem. 
We relax the requirement $\sum_{C\in \cC} \bx_C\cdot C(i)=1$ in \eqref{eq:matching_LP} and use inequality instead.
That is,
\begin{equation}
\label{eq:relaxed_matching_LP}
  \begin{aligned}
	\rMLP:~~~  & \min && \sum_{C\in \cC} \bx_C\\
	&\forall i\in I: &&\sum_{C\in \cC} \bx_C\cdot C(i)\geq 1\\\
	& && \cE(\bx)\in P_{\cM}(G)\\
	&\forall C\in \cC:~~~&& \bx_C\geq 0
  \end{aligned}
\end{equation}
It can be easily shown that the optima of \eqref{eq:matching_LP} and \eqref{eq:relaxed_matching_LP} are equal; furthermore, a solution for \eqref{eq:relaxed_matching_LP} can be easily converted to a solution for \eqref{eq:matching_LP} of the same or lower value.  

Our objective is to find a variant of \eqref{eq:relaxed_matching_LP} in which the set $\cC$ is replaced by a polynomial-size set~$\cD\subseteq \cC$, while approximately preserving the optimal value.
To this end we use the following family of polytopes:
\begin{equation}
\label{eq:matching_extended}
  \forall \cD\subseteq \cC:~~P(\cD) =\left\{  (\bx, \by)~\middle|~
  \begin{aligned}
    &\bx\in \mathbb{R}_{\geq 0}^{\cD} ,~\by\in P_{\cM}(G)\\
	&\cE(\bx) \leq \by\\
	\forall i\in I :~&\sum_{C\in \cD} \bx_C \cdot C(i)\geq 1
  \end{aligned}
  \right\} \enspace .
\end{equation}

Given $\cD\subseteq \cC$, with a slight abuse of notation we refer to a vector $\bx\in \mathbb{R}_{\geq 0}^{\cD}$ as a vector in $\mathbb{R}_{\geq 0}^{\cC}$ where $\bx_C= 0$ for every $C\in \cC\setminus \cD$.
This ensures that the term $\cE(\bx)$ is well defined. 
Since $P_\cM(G)$ is downward closed, we have that $\rMLP$ is equivalent to the problem of finding $(\bx,\by)\in P(\cC)$ such that $\|\bx\|$ is minimized.\footnote{A polytope $P\subseteq \mathbb{R}^n_{\geq 0}$ is \emph{downward closed} if for any $\bx\in P$ and $\by\in \mathbb{R}^n_{\geq 0}$ such that $\by\leq \bx$ it holds that $\by\in P$.}
For $\cD\subseteq \cC$ we define $\rMLP(\cD)$ as the problem of finding $(\bx, \by)\in P(\cD)$ such that $\|\bx\|$ is minimized.
It follows that $\OPT(\rMLP(\cD))\geq \OPT(\rMLP)$ for any $\cD\subseteq \cC$. 

We use $P(\cD)$ to define a family of additional polytopes $Q(\cD,h)$ in $\mathbb{R}^E$, one for each $\cD\subseteq \cC$ and $h\in \mathbb{R}_{\geq 0}$:
\begin{equation}
\label{eq:poly_Q}
  Q(\cD,h) =\left\{  \by\in \mathbb{R}^{E}\ ~\middle|~
  \begin{aligned}
\exists \bx\in \mathbb{R}^{\cD}_{\geq 0}:~(\bx,\by)\in P(\cD) \textnormal{ and } \|\bx\| \leq h  
\end{aligned}\right\} \enspace .
\end{equation}
It thus follows that $Q(\cD,h)\neq \emptyset$ if and only if $\OPT(\rMLP(\cD))\leq h$.
Furthermore, $Q(\cD,h)$ is a polytope in a vector space of polynomial size.
We use the ellipsoid method to determine if 
$Q(\cC,h) =\emptyset$ 
for various values of $h$.
The separation oracle first checks if $\by\in P_\cM(G)$, and otherwise finds a separating hyperplane using a  separation oracle for the matching polytope.
If $\by\in P_\cM(G)$ we use the following linear program, which depends on $\by\in P_\cM(G)$ and $\cD\subseteq \cC$, to obtain a separating hyperplane:
\begin{equation}
\label{eq:primal}
  \begin{aligned}
	\primal(\by,\cD)~~~~~& \min && \sum_{C\in \cD} \bx_C,\\
	&\forall i\in I: &&\sum_{C\in \cD} \bx_C\cdot C(i)\geq 1,\\\
	&\forall e\in E: && \sum_{C\in S(e)\cap \cD}  \bx_C \leq \by_e,\\
	&\forall C\in \cC:&& \bx_C\geq 0 \enspace .
  \end{aligned}
\end{equation}
where for every $e\in E$ we define its superset of configurations as $S(e)=\{ C\in \cC~|~e\subseteq \cC\}$.
Using this notation it holds that $\left(\cE(\bx)\right)_e= \sum_{C\in S(e)}\bx_C$.  
It follows that $\by \in Q(\cD, h)$ if and only if $\by\in P_{\cM} (G)$ and  $\OPT(\primal(\by,\cD))\leq h$.

Recall the set $\cC_2$ is defined in \eqref{eq:Ch_def}.
For any $C\in \cC$ it holds that $C\in \cC_2$ if and only if there is $e\in E$ such that $C\in S(e)$.
We use this observation to derive the dual of $\primal(\by,\cD)$, which is the following linear program:
\begin{equation}
\label{eq:dual}
  \begin{aligned}
    \dual(\by,\cD)~~~~~& \max && \sum_{i\in I}\blam_i -\sum_{e\in E} \bbeta_e\cdot \by_e,\\
	&\forall C\in \cD\setminus\cC_2: &&\sum_{i\in C} \blam_i  \leq 1,\\
	&\forall e\in E,~C\in S(e)\cap \cD: && \sum_{i\in C} \blam_i \leq 1+ \beta_e,\\
	&\forall i\in I: && \blam_i\geq 0,\\
	&\forall e\in E: && \bbeta_e \geq 0\enspace .
  \end{aligned}
\end{equation}

Observe that the feasibility region of $\dual(\by,\cD)$ is independent of $\by$.
That is, for any $\cD\subseteq \cC$ we can define
\begin{equation}
\label{eq:polyR_def}
  R(\cD)= \left\{(\blam, \bbeta )\in \mathbb{R}_{\geq 0 }^I \times \mathbb{R}_{\geq 0 }^E ~\middle|~
  \begin{aligned}
    &\forall C\in \cD\setminus \cC_2: &&\sum_{i\in C} \blam_i  \leq 1\\
    &\forall e\in E,~C\in S(e)\cap \cD: && \sum_{i\in C} \blam_i \leq 1+ \beta_e
  \end{aligned}
  \right\}
  \enspace .
\end{equation}
Then $\dual(\by,\cD)$ is the problem of finding $(\blam, \bbeta)\in R(\cD)$ for which $\sum_{i\in I} \blam_i -\sum_{e\in E} \bbeta_e \cdot \by_e$ is maximized.

We use the following relation between $R(\cC)$ and $Q(\cC,h)$ to generate separating hyperplanes.
\begin{lemma}
\label{lem:Q_invariant}
  For any $h\in \mathbb{R}_{\geq 0}$, $\by \in Q(\cC,h)$  and $(\blam, \bbeta)\in R(\cC)$ it holds that 
  \begin{equation*}
    \sum_{i\in I} \blam_i-\sum_{e\in E} \bbeta_e \cdot \by_e \leq h \enspace .
  \end{equation*}
\end{lemma}
\begin{proof}
  As $\by\in Q(\cC,h)$ it follows that $\OPT(\dual(\by, \cC)) =\OPT(\primal(\by,\cC))\leq h$. Thus, as $(\blam,\bbeta)\in R(\cC)$ we have
  \begin{equation*}
    \sum_{i\in I} \blam_i-\sum_{e\in E} \bbeta_e \cdot \by_e \leq \OPT(\dual(\by, \cC)) \leq h \enspace .\qedhere
  \end{equation*}
\end{proof}

We also use $R(\cC)$ to bound the facet complexity of $Q(\cC,h)$. 
\begin{lemma}
\label{lem:Q_complexity}
  There is a polynomial $p_1$ (independent of the instance $(I,v)$) such that for any $h\geq 0$ the facet complexity of $Q(\cC,h)$ is at most $p_1(|I|+\langle  h\rangle )$. 
\end{lemma}
\begin{proof}
  By \eqref{eq:polyR_def}, the facet complexity of $R(\cC)$ is polynomial in the encoding of the input instance~$(I,v)$.
  Therefore, by \Cref{lem:facet_to_vertex}, the vertex complexity of $R(\cC)$ is at most $ 4 \cdot (|I|+|I|^2)$ times the facet complexity of $R(\cC)$.
  Thus, the vertex complexity of $R(\cC)$ is polynomial in $|I|$. Hence, there is a polynomial $q$ such that the vertex complexity of $R(\cC)$ is at most $q(|I|)$.

  By \Cref{def:polyhedron} there are $V_1, V_2\subseteq \mathbb{R}^I_{\geq 0} \times \mathbb{R}^E_{\geq 0}$ such that $R(\cC)= \conv(V_1)+\cone(V_2)$ and $\langle \bu \rangle \leq q(|I|)$ for every $\bu \in V_1 \cup V_2$.
  For any $h\geq 0$ define
  \begin{equation*}
    Q'(h)=\left\{\by\in P_\cM(G)~\middle|~
	\begin{aligned}
	  \forall (\blam,\bbeta)\in V_1:~~~~~&\sum_{i\in I} \blam_i-\sum_{e\in E} \bbeta_e\cdot \by_e \leq h\\
	  \forall (\blam,\bbeta)\in V_2:~~~~~&\sum_{i\in I} \blam_i-\sum_{e\in E} \bbeta_e\cdot \by_e \leq 0\\
    \end{aligned}\right\}
  \end{equation*}

  \begin{claim}
  \label{claim:Q_in_prime}
    For any $h\geq 0$ it holds that $Q(\cC,h)\subseteq Q'(h)$.
  \end{claim}
  \begin{claimproof}
    Let $\by \in Q(\cC,h)$.
	For any $(\blam,\bbeta)\in V_1$ it holds that $(\blam,\bbeta)\in R(\cC)$, thus $\sum_{i\in I}\blam_i-\sum_{e\in E} \bbeta_e\cdot \by_e \leq h$ by \Cref{lem:Q_invariant}.
	Suppose, for sake of contradiction, that there is $(\blam,\bbeta)\in V_2$  such that $\sum_{i\in I}\blam_i-\sum_{e\in E} \bbeta_e\cdot \by_e =\xi> 0$.
	It therefore holds that $(\frac{h+1}{\xi}\blam,\frac{h+1}{\xi}\bbeta)\in R(\cC)$.
	Thus
	\begin{equation*}
      h\geq \OPT(\primal(\by,\cC))=\OPT(\dual(\by,\cC))\geq  \sum_{i\in I}\frac{h+1}{\xi}\cdot\blam_i-\sum_{e\in E}\frac{h+1}{\xi}\cdot \bbeta_e\cdot \by_e\geq h+1,
	\end{equation*}
    a contradiction.
	Hence, $\sum_{i\in I}\blam_i-\sum_{e\in E} \bbeta_e\cdot \by_e \leq 0$ for every $(\blam,\bbeta)\in V_2$, and $\by \in Q'(h)$.
  \end{claimproof}
  \begin{claim}
  \label{claim:Qprime_in_Q}
	For any  $h\geq 0$ it holds that $Q'(h)\subseteq Q(\cC,h)$.	
  \end{claim}
  \begin{claimproof}
    Let $\by \in Q'(h)$ and $(\blam^*,\bbeta^*)\in R(\cC)$.
    As $R(\cC)=\conv(V_1)+\cone(V_2)$ there are numbers $\zeta_{\blam,\bbeta}\geq 0$ for all $(\blam,\bbeta)\in V_1$ and $\xi_{\blam,\bbeta}\geq 0$ for all $(\blam,\bbeta)\in V_2$ such that $\sum_{(\blam,\bbeta)\in V_1}\zeta_{\blam,\bbeta}=1$, and
    \begin{equation*}
      (\blam^*,\bbeta^*)= \sum_{(\blam,\bbeta)\in V_1}\zeta_{\blam,\bbeta} \cdot (\blam,\bbeta) +\sum_{(\blam,\bbeta)\in V_2}\xi_{\blam,\bbeta} \cdot (\blam,\bbeta) \enspace .
    \end{equation*}
    Thus,
    \begin{equation*}
      \begin{aligned}
        &\sum_{i\in I} \blam^*_i-\sum_{e\in E} \bbeta^*_e \cdot \by_e\\
        =~& \sum_{(\blam,\bbeta)\in V_1} \zeta_{\blam,\bbeta} \left(\sum_{i\in I} \blam_i-\sum_{e\in E} \bbeta_e \cdot \by_e \right)+\sum_{(\blam,\bbeta)\in V_2} \xi_{\blam,\bbeta} \left(\sum_{i\in I} \blam_i-\sum_{e\in E} \bbeta_e \cdot \by_e \right)\\
        \leq~& \sum_{(\blam,\bbeta)\in V_1} \zeta_{\blam,\bbeta} \cdot h+\sum_{(\blam,\bbeta)\in V_2} \xi_{\blam,\bbeta} \cdot 0\\
        \leq~& h \enspace.
      \end{aligned}
    \end{equation*}
    That is, we showed that $\sum_{i\in I} \blam^*_i-\sum_{e\in E} \bbeta^*_e \cdot \by_e\leq h$ for every $(\blam^*,\bbeta^*)\in R(\cC)$.
	Hence, \linebreak\mbox{$\OPT(\dual(\by,\cC))\leq h$}.
	As it also holds that $\by\in P_{\cM}(G)$, we conclude that $\by \in Q(\cC,h)$.
  \end{claimproof}

  By \Cref{claim:Q_in_prime} and \Cref{claim:Qprime_in_Q} it follows that $Q'(h) = Q(\cC,h)$.
  Furthermore, by {\em Edmonds' matching polytope theorem} (see, e.g., Corollary 25.1a in Schrijver's book~\cite{Schrijver2003}) it holds that 
  \begin{equation*}
    P_{\cM}(G) = \left\{\by \in \mathbb{R}_{\geq 0}^{E} ~\middle|~
	\begin{aligned}
	  &\forall i\in L&&:~~~\sum_{(i,i')\in E} \bx_{(i,i')} \leq 1 \\
      &\forall U\subseteq L \textnormal{ s.t. $U$ is odd}&&:~~~\sum_{(i,i')\in E~\textnormal{ s.t. } i,i'\in U} \bx_{(i,i')} \leq \floor{\frac{|U|}{2}}
	\end{aligned}
	\right\}.
  \end{equation*}
  Thus,
  \begin{equation*}  
    Q(\cC,h) = Q'(h) = \left\{\by\in \mathbb{R}_{\geq 0}^{E}~\middle|~
	\begin{aligned}
	  &\forall (\blam,\bbeta)\in V_1&&:~~~\sum_{i\in I} \blam_i-\sum_{e\in E} \bbeta_e\cdot \by_e \leq h\\
	  &\forall (\blam,\bbeta)\in V_2&&:~~~\sum_{i\in I} \blam_i-\sum_{e\in E} \bbeta_e\cdot \by_e \leq 0\\
	  &\forall i\in L&&:~~~\sum_{(i,i')\in E} \bx_{(i,i')} \leq 1 \\
	  &\forall U\subseteq L \textnormal{ s.t. $U$ is odd }&&:~~~\sum_{(i,i')\in E~\textnormal{s.t.} i,i'\in U} \bx_{(i,i')} \leq \floor{\frac{|U|}{2}} 
    \end{aligned}
	\right\} \enspace .
  \end{equation*}
  That is, $Q(\cC,h)$ is the solution set for a system of linear equations in which the encoding length of each inequality is at most $q(|I|) + \langle h\rangle + O(|I|^2)$. This completes the proof of \Cref{lem:Q_complexity}.
\end{proof}
 
Let $\cM^*$ be a maximum matching in the graph $G$.
Since each of the vertices in a matching polytope corresponds to a(n integral) matching, it holds that 
\begin{equation}
\label{eq:max_matching}
  \sum_{e\in E} \by_e \leq |\cM^*|\quad \mbox{for all } \by\in P_{\cM}(G) \enspace .
\end{equation}
Since for every $e\in \cM^*$ it holds that $e\in \cC_2$, i.e., $v_1(e)>(1-\delta)$, for every solution $\bx$ of $\rMLP$ we have
\begin{equation*}
  \begin{aligned}
    \sum_{C\in \cC} \bx_C&  \geq \sum_{C\in \cC} \bx_C \cdot v_1(C)
                            \geq \sum_{C\in \cC} \bx_C \sum_{i\in I} v_1(i) \cdot C(i) \\
                           &=\sum_{i\in I} v_1(i) \sum_{C\in \cC} \bx_C\cdot C(i) \geq \sum_{i\in I}v_1(i) \geq \sum_{e\in \cM^*} v_1(e) >(1-\delta)|\cM^*| \enspace .
  \end{aligned}
\end{equation*}
Hence,
\begin{equation*}
  \OPT(\rMLP) > (1-\delta ) |\cM^*| \enspace .
\end{equation*}

We combine \Cref{lem:Q_invariant} with the next lemma that is proved later in this section.
\begin{lemma}
\label{lem:dual_approx}
  There is a polynomial-time algorithm $\ellir$ which, given $\by\in P_\cM(G)$ and\linebreak $h>(1-\delta )|\cM^*|$, returns
  \begin{itemize}
	\item either a subset $\cD\subseteq \cC$ of size $|\cD|$ polynomial in the input size such that $
	 \OPT(\dual(\by,\cD))\leq \left(1+\eps \right)h$,
	\item or a point $(\blam, \bbeta)\in R(\cC)$  such that $\sum_{i\in I} \blam_i -\sum_{e\in E} \bbeta_e\cdot \by_e  >h$.
  \end{itemize} 
\end{lemma}

\begin{algorithm}[h]
  \SetAlgoLined
  \SetKwInOut{Output}{Output}
  \SetKwInOut{Input}{Input}
	
  \Input{  $\by \in \mathbb{R}_{\geq 0}^{E}$, $h>(1-\delta)|\cM^*| $.} 
  \Output{Either a separating hyperplane between $Q(\cC,h)$ and $\by$, or a subset $\cD\subseteq \cC$.}
	
  \DontPrintSemicolon
  \label{Qsep:match}
  If $\by\notin P_{\cM}(G)$, then find a separating hyperplane between $\by$ and $P_\cM(G)$ and return it.\; 
  \label{Qsep:ellir}
  Run $\ellir$ (\Cref{lem:dual_approx}) with $\by$ and $h$ as its inputs\;
	
  \uIf{ $\ellir$  returned $(\blam, \bbeta)\in R(\cC)$  such that $\sum_{i\in I} \blam_i -\sum_{e\in E} \bbeta_e\cdot \by_e > h$}
  {
    return $\sum_{i\in I} \blam_i -\sum_{e\in E} \bbeta_e\cdot \bz_e  =h$ as a separating hyperplane
  }
  \Else{ 
    notify the ellipsoid algorithm to abort, and 
	return the set $\cD\subseteq \cC$ returned by $\ellir$
  }

  \caption{$\sepq$\label{alg:Q_separator}}
\end{algorithm}

We use algorithm $\ellir$ in \Cref{lem:dual_approx} to derive a separation oracle for $Q(\cC ,h)$.
The pseudocode of the oracle is given in \Cref{alg:Q_separator}.
We note there is a polynomial-time separation oracle for the matching polytope (see, e.g, Schrijver~\cite{Schrijver2003}); thus, Step~\ref{Qsep:match} can be implemented in polynomial time.
While the algorithm does not formally qualify as a separation oracle, it gives the following guarantee:
\begin{lemma}
\label{lem:Q_seperator}
  Given $\by\in \mathbb{R}^E_{\geq 0}$ and $h>(1-\delta)|\cM^*|$, \Cref{alg:Q_separator},
  \begin{itemize}
    \item either returns a separating hyperplane between $Q( \cC,h) $ and $\by$,
	\item or notifies the ellipsoid method to abort and returns $\cD\subseteq \cC$ of polynomial cardinality such that\\ $\OPT(\dual(\by,\cD))\leq (1+\eps)h$.  In this case, it must hold that $\by\in P_{\cM}(G)$. 
  \end{itemize}
\end{lemma}
\begin{proof}
  If $\by \not\in P_\cM(G)$ then \Cref{alg:Q_separator} finds a separating hyperplane between $\by$ and $P_\cM(G)$.
  As $Q(\cC,h)\subseteq P_\cM(G)$, this hyperplane also separates between $\by$ and $Q(\cC,h)$.  
 	
  If the invocation of $\ellir$ returns $(\blam, \bbeta)\in R(\cC)$ such that $\sum_{i\in I} \blam_i -\sum_{e\in E} \bbeta_e\cdot \by_e  >h$, then  $\sum_{i\in I} \blam_i -\sum_{e\in E} \bbeta_e\cdot \bz_e = h$ is a separating hyperplane between $\by$ and $Q(\cC,h)$ by \Cref{lem:Q_invariant}. 
  Otherwise, by \Cref{lem:dual_approx}, the invocation of $\ellir$ returns a subset $\cD\subseteq \cC$ of polynomial cardinality  such that $ \OPT(\dual(\by,\cD))\leq \left(1+\eps\right) h $.
  It follows that in this case \Cref{alg:Q_separator} notifies the ellipsoid to abort and returns $\cD$.
\end{proof}

\Cref{alg:elliq} utilizes $\sepq$ as a separation oracle.
The algorithm may return a vector $\bx\in \mathbb{R}_{\geq 0}^{\cD}$ for some $\cD\subseteq \cC$.
Recall that we interpret such a vector as a  vector in $\mathbb{R}^{\cC}$ as well. 

\begin{algorithm}[h]
  \SetAlgoLined
  \SetKwInOut{Input}{Input}
  \SetKwInOut{Output}{Output}
  \Input{$h>(1-\delta)|\cM^*|$}
  \Output{Either determine that $\OPT(\rMLP)>h$, or return a solution $\bx'$ for $\rMLP$ with $\| \bx'\| \leq (1+\eps)h$.} 
  \DontPrintSemicolon
	
  Run $\ellipsoid$ with $n=|E|$, $\varphi = p_1(|I|+\langle h\rangle )$ and $\sepq$ (and $h$) as the separation oracle \label{elliq:elli}
	
  \uIf{the ellipsoid method returned that the polytope  is empty}
  {
    Return $\OPT(\rMLP)>h$
  } 
  \Else{
    This case can only happen if the $\sepq$ notified the ellipsoid to abort and returned a set $\cD\subseteq \cC$. 
	Find an optimal solution $(\bx',\by')$ for   $\rMLP(\cD)$ and return $\bx'$. 	\label{elliq:lp}\;
  }
  \caption{$\elliq$\label{alg:elliq}}
\end{algorithm}

\begin{lemma}
\label{lem:elliq}
  In polynomial time, \Cref{alg:elliq} either determines that $\OPT(\rMLP)>h$, or finds a solution $\bx'$ for $\rMLP$ satisfying $\|\bx'\|\leq (1+\eps)h$. 
\end{lemma}
\begin{proof} 
  By~\Cref{lem:Q_complexity} it holds that $(Q,n,\varphi)$ is a well-described polyhedron.
  As $n$, $\varphi$ are polynomial in the instance, it follows the execution time of the ellipsoid method is polynomial.
  Furthermore, if the algorithm solves $\rMLP(\cD)$ in \Cref{elliq:lp} then, by \Cref{lem:Q_seperator}, we have that $|\cD|$ is polynomial, and hence $\rMLP(\cD)$ can be solved in polynomial time (as there is a separation oracle for $\cE(\bx)\in P_{\cM}(G)$, and the number of variables and additional constraints is polynomial). 

  By \Cref{lem:Q_seperator}, if the ellipsoid method asserts that the polytope is empty, it holds that all invocations of $\sepq$ returned a separating hyperplane. Hence, this is a valid execution of $\ellipsoid$ with a separation oracle for $Q(\cC,h)$.
  It follows that $Q(\cC,h)=\emptyset$, implying that $\OPT(\rMLP)= \OPT(\rMLP(\cC)) >h$ due to~\eqref{eq:poly_Q}.  
 
  Otherwise, it must hold that the execution of the ellipsoid method was aborted by $\sepq$ at some iteration.
  Let $\by\in P_\cM(G)$ be the value of $\by$ used in the call to $\sepq$ in this iteration, let $\cD\subseteq \cC$ be the subset of configurations returned  by $\sepq$, and let $(\bx',\by')\in P(\cD)$  be the solution found in \Cref{elliq:lp}. 
  It holds that $\|\bx'\| \leq \OPT(\primal(\by,\cD)) =\OPT(\dual(\by, \cD))\leq \left(1+\eps\right) h$, where the last inequality is by \Cref{lem:Q_seperator}.
  Since $(\bx',\by') \in P(\cD)$, it holds that $\cE(\bx')\leq \by'\in P_{\cM}(G)$; thus, $\cE(\bx')\in P_\cM(G)$.
  For the same reason, we also have $\sum_{C\in \cC} \bx'_C \cdot C(i)\geq 1$ for all $i\in I$.
  Hence,~$\bx'$ is a solution for $\rMLP$ of value at most~$(1+\eps)h$. 
\end{proof}

\begin{algorithm}[h]
  \SetAlgoLined
  \SetKwInOut{Configuration}{Configuration}
  \SetKwInOut{Input}{Input}
  \SetKwInOut{Output}{Output}
  \Configuration{$\eps, \delta\in (0,0.1)$ }
  \Input{A 2VBP instance $(I,v)$.}
  \Output{A $(1+O(\eps))$-approximate solution $\bx$ for $\rMLP$.}
	
  \DontPrintSemicolon

  Run a binary search over the range $(\ell, u)=( (1-\delta)|\cM^*|, |I|)$: in each iteration call $\elliq(h)$ with $h=\frac{\ell+u}{2}$. 
  If $\elliq$ returned that $\OPT(\rMLP)>h$ update $\ell =h$;
  if $\elliq$ returned a solution $\bx$, set $\bx$ to be the best solution and $u=h$. Repeat the process until $u-\ell<\eps$.\\
	
  If $u\neq |I|$, return the best solution found; else, return a vector $\bx\in \{0,1\}^\cC$ where $\bx_{\{i\}}=1$ for every $i\in I$ and $\bx_C=0$ for any other $C\in \cC$. 
 
  \caption{\textsf{Matching-LP}\label{alg:matching_LP}}
\end{algorithm} 

Our algorithm for $\delta$-$\rMLP$, given in \Cref{alg:matching_LP}, uses $\elliq$ to perform a binary search.  
\begin{proof}[{\bf Proof of \Cref{lem:matching_ptas}}] 
  We show that \Cref{alg:matching_LP} is a polynomial time $(1+3\eps)$-approximation algorithm for $\rMLP$.
  This immediately implies a PTAS for the $\MLP$ problem due to the connection between $\MLP$ and $\rMLP$. 
	
  By \Cref{lem:elliq} it holds that $\OPT(\rMLP) >\ell$ throughout the binary search, and if $u\neq |I|$ then the best solution found  $\bx$ satisfies $\|\bx\|\leq (1+\eps)u$ throughout the execution of the binary search. 
  Thus, \Cref{alg:matching_LP} returns a solution $\bx$ satisfying	
  \begin{equation*}
    \|\bx\|\leq (1+\eps)u<(1+\eps)(\ell +\eps) <(1+\eps)(\OPT(\rMLP)+\eps)\leq (1+3\eps)  \OPT(\rMLP),
  \end{equation*}
  where the last inequality holds since $\OPT(\rMLP)\geq 1$ (otherwise $I=\emptyset$ and $\bx=\bzero$ is an optimal solution).
\end{proof}

It remains to prove~\Cref{lem:dual_approx}.
Similar to $\elliq$, the ellipsoid method is applied with an approximate separation oracle.
Consider the following family of polytopes.
For any $\ell\geq 0$, $\by\in P_{\cM}(G)$ and $\cD\subseteq \cC$, define
\begin{equation}
\label{eq:polyR_refined}
  \begin{aligned}
 	R(\ell,\by, \cD) &= \left\{ (\blam, \bbeta)\in R(\cD)~|~\sum_{i\in I}\blam_i -\sum_{e\in E} \bbeta_e \cdot \by_e \geq \ell \right\} \\
 	&= \left\{(\blam, \bbeta )\in \mathbb{R}_{\geq 0 }^I \times \mathbb{R}_{\geq 0 }^E ~\middle|~\begin{aligned}
 		&&&\sum_{i\in I}\blam_i -\sum_{e\in E} \bbeta_e \cdot \by_e \geq \ell\\
 		&\forall C\in \cD\setminus \cC_2: &&\sum_{i\in C} \blam_i  \leq 1\\
 		&\forall e\in E,~C\in S(e)\cap \cD: && \sum_{i\in C} \blam_i \leq 1+ \beta_e
 	\end{aligned}
 	\right\}  \enspace .
 	\end{aligned}
 \end{equation}

The ellipsoid method is used with polytopes in $R(\ell,\by, \cD)$. 
To derive a separation oracle for $R(\ell,\by,\cD)$ we use a PTAS for {\sc $2$-Dimensional Knapsack} ($2$DK)~\cite{FriezeC1984}.
Using the terminology in this paper, the input for $2$DK is a 2VBP instance $(S,v)$, a profit vector $\bar{p}\in\mathbb{R}_{\geq 0}^S$ and a two-dimensional budget $\bb\in \mathbb{R}_{\geq 0}^2$.
The objective is to find a subset $W\subseteq S$ of items such that $v(W)=\sum_{i\in W} v(i)\leq  \bb$, and  $p(W)\equiv \sum_{i\in W} \bar{p}_i$ is maximal.
Denote a $2$DK instance by $(S,v,\bar{p},\bb)$.
We also allow $\bar{p}\in \mathbb{R}^T_{\geq 0}$ where $S\subseteq T$.
The separation oracle is given in \Cref{alg:R_separator}.
The pseudocode  uses $N_G[j] = \{ i\in L~|~\{i,j\}\in E\}\cup\{j\}$ to denote the closed neighborhood of $j\in L$ in the $\delta$-matching graph $G$. 

\begin{algorithm}[h]
  \SetAlgoLined
  \SetKwInOut{Input}{Input}
  \SetKwInOut{Output}{Output}
	
  \Input{$(\blam,\bbeta)\in \mathbb{R}^{I}\times \mathbb{R}^E$, $\by\in P_{\cM}(G)$ and $\ell> (1-\delta)|\cM^*|$.}
  \Output{Either a separating hyperplane between $R(\ell, \by ,\cC)$ and $(\blam,\bbeta)$, or $(\blam',\bbeta')\in R\left( (1-\frac{\eps}{2} ) \ell, \by, \cC \right)$.}
	
  \DontPrintSemicolon
	
  If $\sum_{i\in I}\blam_i -\sum_{e\in E} \bbeta_e \cdot \by_e <\ell$, then return it as the separating hyperplane.
  \label{sepr:objective} \;
	
  Find a $(1-\frac{\eps}{8})$-approximate solution  $W$ for the $2$DK instance $(I\setminus L,v,\blam,\bone)$.
  If $\sum_{i\in W} \blam_i >1$, return $W$  as a separating hyperplane.\label{sepr:small_items}\;

  \ForEach{\label{sepr:loop_loose} $j\in L$}
  {
    Find a $(1-\frac{\eps}{8})$-approximate solution  $W$ for the $2$DK instance $(I\setminus N_G[j],v,\blam,\bone-v(j))$.
    If $\sum_{i\in W\cup \{j\}} \blam_i >1$ return $W\cup \{j\}$  as a separating hyperplane.\label{sepr:iter_loose}
  } 

  \ForEach{\label{sepr:loop_tight} $e\in E$}
  {
    Find a $(1-\frac{\eps}{8})$-approximate solution  $W$ for the $2$DK instance $(I\setminus L,v,\blam,\bone-v(e))$.
    If $\sum_{i\in W\cup e} \blam_i >1+\bbeta_e$ return $W\cup e$  as a separating hyperplane. \label{sepr:tight_iter}
  }
	
  Notify the ellipsoid method to abort, and return  $\left( \left( 1-\frac{\eps}{8}  \right) \blam,\bbeta'\right) $ where $\bbeta'_e= \min\{ 2,\bbeta_e \}$ for every $e\in E$. 
  \caption{$\sepr$\label{alg:R_separator}}
\end{algorithm}

As in the case of $\sepq$, we show that $\sepr$ has properties similar to those of a separation oracle. 
\begin{lemma}
\label{lem:R_seperator}
  On input $(\blam,\bbeta)\in \mathbb{R}^{I}\times \mathbb{R}^E$, $\by\in P_{\cM}(G)$ and $\ell\geq (1-\delta)|\cM^*|$, in polynomial time \Cref{alg:R_separator} either
  \begin{itemize}
    \item returns a separating hyperplane between $R(\ell , \by ,\cC) $ and $(\blam,	 \bbeta)$, or
    \item notifies the ellipsoid method to abort and returns $(\blam',\bbeta')\in R\left(  \left(1-\frac{\eps}{2}\right)\ell, \by, \cC \right)$.
  \end{itemize}
\end{lemma}
\begin{proof}
  Since $2$DK admits a PTAS~\cite{FriezeC1984}, it follows that \Cref{alg:R_separator} runs in polynomial time.

  If $\sum_{i\in I}\blam_i -\sum_{e\in E} \bbeta_e \cdot \by_e <\ell$ then the algorithm returns this inequality as a separating hyperplane in Step~\ref{sepr:objective}.
  This inequality indeed serves as a separating hyperplane by the definition of $R(\ell,\by ,\cC)$ in~\eqref{eq:polyR_refined}.
  Thus, for the remainder of the proof, we may assume that $\sum_{i\in I}\blam_i -\sum_{e\in E} \bbeta_e \cdot \by_e \geq \ell$.

  If the algorithm returns a set $W$ in Step~\ref{sepr:small_items}, then $W\subseteq I\setminus L$ and $v(W)\leq \bone$ as a  solution for $2$DK.
  Thus, $W\in \cC\setminus \cC_2$  and the inequality $\sum_{i\in W} \blam_i >1$ defines a separating hyperplane by \eqref{eq:polyR_refined} and \eqref{eq:polyR_def}.
  Hence, for the remainder of the proof we may assume that the algorithm did not return a set in Step~\ref{sepr:small_items}.
  This implies that the optimal solution for the $2$DK instance $(I\setminus L, v ,\blam,\bone)$ has value at most $\left( 1-\frac{\eps}{8 }\right)^{-1}$.
  Since every $C\in \cC$ such that $C\subseteq I\setminus L$ is a solution for $(I\setminus L, v ,\blam,\bone)$, it follows that 
  \begin{equation}
  \label{eq:sepr_small_items}
    \forall C\in \cC,~C\subseteq I\setminus L:~~\sum_{i\in C} \blam_i \leq \left( 1-\frac{\eps}{8 }\right)^{-1} \enspace .
  \end{equation}

  Consider the  case in which the algorithm returns the set $W\cup \{ j\}$ in Step~\ref{sepr:iter_loose}.
  It holds that\linebreak $v(W\cup \{j\})\leq v(W)+v(j)\leq  \bone -v(j)+v(j)=\bone$, as $W$ is a solution for the $2$DK instance $(I\setminus N_G[j] ,v , \blam,\bone-v(j))$.
  Thus, $W\cup \{j\}\in \cC$.
  Suppose, for the sake of contradiction, that $W\cup\{j\} \in \cC_2$.
  Thus, there is some $j'\in W\cap L$ such that $(j,j')\in E$, and we conclude that $W\cap N[j]\neq \emptyset$,  contradicting $W\subseteq I\setminus N[j]$ (see Step~\ref{sepr:loop_loose}).
  It therefore holds that $W\cup\{j\}\in \cC\setminus \cC_2$. 
  Since $\sum_{i\in W\cup \{j\}} \blam_i>1$, the configuration $W\cup \{j\}$ defines a separating hyperplane, by~\eqref{eq:polyR_refined} and~\eqref{eq:polyR_def}. 

  Hence, for the remainder of the proof we may assume that the algorithm did not return a separating hyperplane in Step~\ref{sepr:iter_loose}.
  Let $C\in \cC\setminus \cC_2$.
  If $C\subseteq I\setminus L$ then it holds that $\sum_{i\in C} \blam_i \leq \left( 1-\frac{\eps}{8}\right)^{-1}$ by \eqref{eq:sepr_small_items}.

Consider the iteration of the loop in Step~\ref{sepr:loop_loose} in which $j=j^*$, and let~$W$ be the set found in this iteration in Step~\ref{sepr:iter_loose}. 
It holds that  $C\setminus \{j\}$ is a solution  for the $2$DK instance $(I\setminus N_G[j] ,v , \blam,\bone-v(j))$; thus, $\sum_{i\in W} \blam_i \geq \left(1-\frac{\eps}{8}\right) \sum_{i\in C\setminus \{j\}} \blam_i$.
  Since the algorithm did not return $W\cup \{j\}$, we have that that $\sum_{i\in W\cup\{j\}}\blam_i\leq 1$.
  Therefore,
  \begin{equation*}
  \sum_{i\in C} \blam_i = \blam_j  + \sum_{i\in C\setminus \{j\}} \blam_i \leq \blam_j + \left(1-\frac{\eps}{8}\right)^{-1} \sum_{i\in W} \blam_i \leq \left(1-\frac{\eps}{8}\right)^{-1} \sum_{i\in W\cup \{j\}} \blam_i \leq\left(1-\frac{\eps}{8}\right)^{-1} \enspace .
  \end{equation*}
  Thus, 
  \begin{equation}
  	\label{eq:sepr_loose}
  	\forall C\in \cC\setminus \cC_2:~~~\sum_{i\in C}\blam_i \leq \left(1+\frac{\eps}{8}\right)^{-1}.
  \end{equation}

  Next, we consider the case in which the algorithms returns the set $W\cup e$ in Step~\ref{sepr:tight_iter}.
  Then $v(W\cup\{e\})=v(W)+v(e)\leq \bone - v(e)+v(e)=\bone$ since $W$ is a solution for $(I\setminus L ,v ,\blam ,\bone-v(e))$. Hence, $W\cup e\in \cC$. It follows that $W\cup \{e\} \in S(e)$.
  Since $\sum_{i\in W\cup e } \blam_i >1+\bbeta_e$, it follows that $W\cup e$ defines a separating hyperplane between $(\blam, \bbeta)$ and $R(\ell , \by, \cC)$ (by \eqref{eq:polyR_def} and \eqref{eq:polyR_refined}).  

  We may therefore assume that the algorithm does not return a set in Step~\ref{sepr:tight_iter} throughout its execution. Let $e^*\in E$ and $C\in S(e^*)$, and  consider the iteration of the loop in Step~\ref{sepr:loop_tight} in which $e=e^*$.
  It holds that  $C \setminus e \subseteq I\setminus L$ (otherwise, $v_d(C)>1$  for some $d\in \{1,2\}$)  and $v(C\setminus e)\leq \bone -v(e)$; thus, $C\setminus e$ is a solution for the $2$DK instance $(I\setminus L, v, \blam, \bone-v(e))$.
  Let $W$ be the approximate solution found for $(I\setminus L, v, \blam, \bone-v(e))$.
  It then holds that $\sum_{i\in W} \blam_i \geq \left( 1-\frac{\eps}{8 }\right) \sum_{i\in C\setminus e} \blam_e$.
  Also, since we assume that the algorithm does not return a set in Step~\ref{sepr:tight_iter}, it holds that $\sum_{i\in W\cup e} \blam \leq 1 +\beta _e$.
  Therefore, we have that
  \begin{equation}
  \label{eq:tight_first}
	\sum_{i\in C} \blam_{i}  = \sum_{i\in e} \blam_i + \sum_{i\in C\setminus e} \blam_i \leq \sum_{i\in e} \blam_i +\left( 1-\frac{\eps}{8}\right) ^{-1}\sum_{i\in W}\blam_i \leq  \left( 1-\frac{\eps}{8}\right) ^{-1} \sum_{i\in W\cup e} \blam_i \leq   \left( 1-\frac{\eps}{8}\right) ^{-1} (1+\beta_e) \enspace .
  \end{equation}
  Let $e=\{j_1,j_2\}$.
  Then $\{j_1\},\{j_2\},C\setminus e \in \cC\setminus \cC_2$.
  Therefore, by \eqref{eq:sepr_loose}, 
  \begin{equation}
	\label{eq:tigh_second}
    \sum_{i\in C} \blam_i \leq \blam_{j_1} + \blam_{j_2} + \sum_{i\in C\setminus e} \blam_i \leq 3\left(1+\frac{\eps}{8}\right)^{-1} \enspace .
  \end{equation}
  By \eqref{eq:tight_first} and \eqref{eq:tigh_second}, we have
  \begin{equation}
	\label{eq:sepr_tight}
	\forall e\in E,~C\in S(e):~~~\sum_{i\in C} \blam_i \leq \left( 1-\frac{\eps}{8}\right) ^{-1} (1+\min\{\bbeta_e,2\}) =\left( 1-\frac{\eps}{8}\right) ^{-1} (1+\bbeta'_e) \enspace .
  \end{equation}
  By \eqref{eq:sepr_loose} and \eqref{eq:sepr_tight} it holds that $\left( \left(1-\frac{\eps}{8}\right) \blam, \bbeta' \right) \in R(\cC)$. 
  Furthermore,  
  \begin{equation*}
    \begin{aligned}
     \sum_{i\in I} \left(1-\frac{\eps}{8}\right) \blam_i - \sum_{e\in E} \bbeta'_e\cdot \by_e &= 
\left(1-\frac{\eps}{8}\right) \cdot \left( 
\sum_{i\in I} \blam_i - \sum_{e\in E} \bbeta'_e\cdot \by_e \right)  - \frac{\eps}{8 } \sum_{e\in E} \bbeta'_e \cdot \by_e \\
 &\geq 
 \left(1-\frac{\eps}{8}\right) \cdot \left( 
 \sum_{i\in I} \blam_i - \sum_{e\in E} \bbeta_e\cdot \by_e \right)  - \frac{\eps}{4} \sum_{e\in E} \by_e 
 \\
& \geq  \left(1-\frac{\eps}{8}\right)  \ell  -\frac{\eps}{4} \frac{\ell}{1-\delta}  \\
&\geq \left(1-\frac{\eps}{2}\right) \ell \enspace .
\end{aligned}
\end{equation*}
The first inequality holds since $\bbeta_e' =\min\{\bbeta_e,2\}$, the second inequality uses $\sum_{e\in E} \by_e \leq |\cM^*| < \frac{\ell}{1-\delta}$ due to \eqref{eq:max_matching}.  Thus, $\left( \left(1-\frac{\eps}{8}\right) \blam, \bbeta' \right) \in R\left(  \left(1-\frac{\eps}{2}\right)\ell, \by, \cC \right)$.
\end{proof}

The facet complexity of $R(\ell, \by , \cD)$ can be trivially bounded by \eqref{eq:polyR_def}, as stated in the next lemma (we omit the proof).  
\begin{lemma}
\label{lem:R_facet}
  There is a polynomial $p_2$ (independent of the instance $(I,v)$) such that for any $\cD\subseteq \cC$, $\by \in P_{\cM}$ and $\ell \geq 0$ the facet complexity of $R(\ell, \by ,\cD)$ is at most $p_2(|I| + \langle \by\rangle+ \langle \ell \rangle )$.
\end{lemma}

\Cref{alg:ellir} uses the ellipsoid method with $\sepr$ as the separation oracle. 
\begin{algorithm}[h]
  \SetAlgoLined
  \SetKwInOut{Input}{Input}
  \SetKwInOut{Output}{Output}
  \Input{$\by \in P_{\cM}(G)$ and $h>(1-\delta)|\cM^*|$}
  \Output{Either a subset $\cD \subseteq \cC$ such that $\OPT(\dual(\by,\cD))\leq (1+\eps)h$ or a point~$(\blam,\bbeta)\in R(\cC)$ such that $\sum_{i\in I} \blam_i -\sum_{e\in E} \bbeta_e \cdot \by_e  >h$.}
  \DontPrintSemicolon

  Run  $\ellipsoid$ (\Cref{lem:ellipsoid}) with $n=|I|+|E|$, $\varphi=p_2(|I|+\langle \by \rangle +\langle \ell \rangle)$ and $\sepr$ as the separation oracle, where $\sepr$ is used with $\by$ and $\ell =\frac{h}{1-\frac{3\eps}{4}}$.
	
  \uIf{the ellipsoid method returned the polytope  is empty} {
    Let $\cD$ be the set of configurations returned by $\sepr$ as a separating hyperplanes throughout the execution of the ellipsoid method.
    Return $\cD$.
  } 
  \Else{
    \tcp{This only happens if $\sepr$ aborted the ellipsoid method.} 
    Return $(\blam, \bbeta)$, where  $(\blam,\bbeta)\in \left( \left(1-\frac{\eps}{2} \right)\ell, \by, \cC \right)$ is the value returned by $\sepr$. 
  }
  \caption{$\ellir$\label{alg:ellir}}
\end{algorithm}

\begin{proof}[{\bf Proof of \Cref{lem:dual_approx}}]
  Note that $\ellir$ runs in polynomial time.
  Furthermore, $\ell >h>(1-\delta) |\cM^*|$.
  Thus, $\sepr$ is used with parameters that match the conditions of \Cref{lem:R_seperator}.

  Consider the execution of \Cref{alg:ellir}.
  If the ellipsoid method returns that the polytope is empty then all separating hyperplanes returned by $\ellir$ are also separating hyperplanes with respect to the polytope $R(\ell, \by ,\cD)$.
  Thus, it must hold that $R(\ell, \by ,\cD)=\emptyset$.
  This implies that $\OPT(\dual(\by, \cD))\leq \ell = \frac{h}{1-\frac{3\eps}{4}} \leq (1+\eps )h$ .
  Since the execution of the ellipsoid is of polynomial time, it follows that $|\cD|$ is also polynomial.	
	
  If the ellipsoid method was aborted, then by \Cref{lem:R_seperator} it holds that $(\blam,\bbeta)\in \left( \left(1-\frac{\eps}{2} \right)\ell, \by, \cC \right)$.
  By \eqref{eq:polyR_refined} we have that $(\blam, \bbeta )\in R(\cC)$, and
  \begin{equation*}
    \sum_{i\in I} \blam - \sum_{e\in E} \bbeta_e \cdot \by_e \geq \left(1-\frac{\eps}{2}\right)\ell
	                                                         = \left(1-\frac{\eps}{2} \right) \frac{h}{1-\frac{3\eps}{4}} > h \enspace .\qedhere
  \end{equation*}
\end{proof}

\section{Basic Probabilistic Tools}
\label{sec:prob_basic}
In this section we prove \Cref{lem:concentration_multistep_prelim,lem:first_fit_bound}; the probabilistic lemmas which are used both in \Cref{sec:ddim} and \Cref{sec:twodim}.
The proof of \Cref{lem:concentration_multistep_prelim} follows from an iterative application of \Cref{lem:Generalized_McDiarmid}.
\Cref{lem:first_fit_bound} is an application of \Cref{lem:concentration_multistep_prelim}. 

We begin with the following technical lemma.
\begin{lemma}
\label{lem:concentration_step}
  Let $j\in \{0,1,\ldots, k-1\}$ and $t>0$. Also, let $\bu\in \mathbb{R}^I_{\geq 0}$ be an $\cF_{j}$-measurable random vector. Then,
  \begin{equation*}
    \Pr\left(\bu \cdot \one_{S_{j+1}} - (1-\delta) \bu \cdot \one_{S_j} > t \cdot \tol(\bu) \right) ~\leq ~\exp \left( -\frac{2\cdot t^2}{\OPT} \right) \enspace.
  \end{equation*}
\end{lemma}
\begin{proof}
  Let $A$ be the set of possible values the random vector $\bu$ can take, that is, $A=\{\bu(\omega)~|~\omega\in \Omega\}$.
  Since $\Omega$ is finite, it holds that $A$ is also finite.
	
  For any $S\subseteq I$, $\rho \in \{1,\hdots,\OPT\}$ and $\ba \in A$ define $f_{S,\rho, \ba}:\cC^{\OPT}\rightarrow \mathbb{R}$ by 
  \begin{equation*} 
		f_{S,\rho, \ba} (C_1,\ldots, C_{\OPT})=\begin{cases} 
			\frac{1}{\tol(\ba)}  \cdot \ba \cdot \one_{S\setminus \left( \bigcup_{\ell=1}^{\rho}C_{\ell}\right)},	~~~~&\tol(\ba)\neq 0,\\
			0, &\textnormal{otherwise} \enspace .
		\end{cases} 
  \end{equation*}
  Also, define $D= \{ f_{S,\rho,\ba} ~|~S\subseteq I, \rho \in \{1,\hdots,\OPT\}, \ba \in A\}$.
  It can be easily verified that $D$ is finite. 

  Let $f_{S,\rho, \ba} \in D$,  $(C_1,\ldots, C_{\OPT}),~(C'_1,\ldots, C'_{\OPT})\in \cC^{\OPT} $ and $r\in [\OPT]$ such that $C_{\ell} =C'_{\ell}$ for $\ell = 1,\hdots,r-1,r+1,\hdots,\OPT$.
  If $\tol(\ba)=0$ or $r>\rho$, then $$\left|f_{S,\rho, \ba} (C_1,\ldots, C_\OPT) -f_{S,\rho,\ba}(C'_1,\ldots, C'_{\OPT})\right| =0\enspace.$$
  Otherwise, let $T=\bigcup_{\ell \in [\rho]\setminus \{r\}} C_{\ell }=  \bigcup_{\ell \in [\rho]\setminus \{r\}} C'_{\ell}$.
  Then
  \begin{equation*}
    \begin{aligned}
      \bigg|f_{S,\rho, \ba} (C_1,\ldots, C_\OPT) -f_{S,\rho,\ba}(C'_1,\ldots, C'_{\OPT})\bigg|&=\left|
	  \frac{1}{\tol(\ba) } \cdot \ba \left(\one_{S\setminus T \setminus C_r} - \one_{S \setminus T \setminus C'_r} \right)\right|\\
	  &=\left|\frac{1}{\tol(\ba) }\left( \sum_{i\in (S\cap C'_r) \setminus (C_r \cup T) }\ba_i -  \sum_{i\in (S\cap C_r) \setminus (C'_r \cup T)} \ba_i \right)\right| \\
	  &\leq\frac{1}{\tol(\ba) } \cdot  \max \left\{\sum_{i\in (S\cap C'_r) \setminus (C_r \cup T) }\ba_i ,~ \sum_{i\in (S\cap C_r) \setminus (C'_r \cup T)} \ba_i \right\}\\
	  &\leq \frac{1}{\tol(\ba) } \cdot \tol(\ba) \leq 1 \enspace .
    \end{aligned}
  \end{equation*}
  The second equality holds, as $S\setminus T \setminus C_r\setminus \left( S\setminus T\setminus C'_r\right) = (S\cap C'_r)\setminus (C_r \cup T)$ and symmetrically $S\setminus T \setminus C'_r\setminus \left( S\setminus T\setminus C_r\right) = (S\cap C_r)\setminus (C'_r \cup T)$.
  Thus, $f_{S,\rho,\ba}$ is of $1$-bounded difference.
	
  Define a random function $g=f_{S_{j}, \rho_{j+1}, \bu}$.
  Since $S_j$, $\rho_{j+1}$ and $\bu$ are $\cF_j$-measurable, it follows that~$g$ is $\cF_j$-measurable.
  By definition of $g$, we have 
  \begin{equation*}
    \tol(\bu) \cdot g(C^{j+1}_1,\ldots, C^{j+1}_{\OPT}) = \bu \cdot \one_{S_j \setminus \bigcup_{\ell=1}^{\rho_{j+1}} C^{j+1}_{\ell}} = \bu \cdot \one_{S_{j+1}} \enspace .
  \end{equation*}
  Furthermore,
  \begin{equation*}
    \begin{aligned}
      \E[\tol(\ba)\cdot g(C^{j+1}_1,\ldots, C^{j+1}_{\OPT})~|~\cF_{j}] &= \E[\bu \cdot \one_{S_{j+1}}~|~\cF_j ] = \sum_{i\in I} \bu_i \cdot \Pr(i\in S_{j+1}~|~\cF_j)\\ & \leq (1-\delta )\sum_{i\in I } \bu_i \cdot \one_{i\in S_{j}} =(1-\delta)\cdot \ba \cdot \one_{S_j},
	\end{aligned}
  \end{equation*}
  where the inequality holds by \Cref{lem:step_bound}.
  Therefore,
  \begin{multline*}
	\Pr\left(\bu \cdot \one_{S_j+1} - (1-\delta) \bu \cdot \one_{S_j} > t \cdot \tol(\bu) \right)\\
	\leq \Pr\left(g(C^{j+1}_1,\ldots, C^{j+1}_{\OPT})- \E[g(C^{j+1}_1,\ldots, C^{j+1}_{\OPT}~|~\cF_{j} )] > t\right)
	\leq \exp \left( -\frac{2\cdot t^2}{\OPT} \right),
  \end{multline*}
  where the last inequality is by \Cref{lem:Generalized_McDiarmid}.
\end{proof}

We use \Cref{lem:concentration_step} as part of the proof of \Cref{lem:concentration_multistep_prelim} 
\begin{proof}[Proof of \Cref{lem:concentration_multistep_prelim}]
  We note that
  \begin{equation*}
    \begin{aligned}
     &\Pr\left( \exists r\in \{j,\hdots,k\}:~\bu \cdot \one_{S_{r}} - (1-\delta)^{r-j}\cdot  \bu \cdot \one_{S_j} > t \cdot \tol(\bu) \right)\\
     =~& \Pr\left( \exists r\in \{j,\hdots,k\}:~ \sum_{\ell = j+1}^{r} \left( \bu \cdot \one_{S_{\ell}} - (1-\delta)\cdot  \bu \cdot \one_{S_{\ell-1}}\right)\cdot  (1-\delta)^{r-\ell}   > t \cdot \tol(\bu)\right)\\
	 \leq~& \Pr \left(\exists r\in \{j+1,\hdots,k\}, ~\ell\in \{j+1,\hdots,r\}:~ \left( \bu \cdot \one_{S_{\ell}} - (1-\delta)\cdot \bu \cdot \one_{S_{\ell-1}}\right)\cdot  (1-\delta)^{r-\ell} > \frac{t}{r -j} \cdot \tol(\bu)\right)\\
		\leq~& \Pr \left( \exists~\ell\in \{j+1,\hdots,k\}:~ \bu \cdot \one_{S_{\ell}} - (1-\delta)\cdot  \bu \cdot \one_{S_{\ell-1}}  > \frac{t}{k} \cdot \tol(\bu)\right) \\
		\leq~& \sum_{\ell=j+1}^{k} \Pr \left( \bu \cdot \one_{S_{\ell}} - (1-\delta)\cdot  \bu \cdot \one_{S_{\ell-1}}  > \frac{t}{k} \cdot \tol(\bu)\right) \\
		\leq~& k \cdot \exp\left( -\frac{2 \cdot \left(\frac{t}{k}\right)^2}{\OPT}\right)\\
      \leq~& \delta^{-2} \exp\left( -\frac{2\cdot\delta^4\cdot t^2}{\OPT}\right) \enspace .
    \end{aligned} 
  \end{equation*}
  The first inequality holds, since if a sum of $n$ variables is greater than $T$ there most be a variable with value greater than $\frac{T}{n}$.
  The fourth inequality is by \Cref{lem:concentration_step}, and the last inequality uses~\mbox{$k\leq \delta^{-2}$}.
\end{proof}

\Cref{lem:first_fit_bound} is a simple application of \Cref{lem:concentration_multistep_prelim}.
\begin{proof}[{\bf Proof of \Cref{lem:first_fit_bound}}]
  Define $\bu \in [0,1]^I$ by $\bu_i = \sum_{t=1}^{d} v_t(i)$. 
  For any $C\in \cC$ it holds that $\sum_{i\in C} \bu_i = \sum_{t=1}^{d} v_t(C)\leq d$, therefore $\tol(\bu)\leq d$.  
  Furthermore, there is  partition $(Q_1,\ldots, Q_{\OPT})$ of $I$ such that~$Q_{\ell}$ is a configuration for $\ell =1,\hdots,\OPT$.
  Therefore,
  \begin{equation}
  \label{eq:S0_total_volume}
    \bu \cdot \one_{S_0} \leq	\bu \cdot \one_{I} = \sum_{\ell =1}^{\OPT} \bu \cdot \one_{Q_{\ell}} \leq \OPT \cdot \tol(\bu)\leq d\cdot \OPT \enspace .
  \end{equation}
  Recall that $\rho^*$ is the number of configurations used by First-Fit in \Cref{round:first_fit} of \Cref{alg:basic_round_and_round}. 
  Using \Cref{lem:first_fit}, we have
  \begin{equation*}
    \begin{aligned}
      \Pr(\rho^*  > 8\cdot d \cdot \delta \cdot \OPT+1) & \leq \Pr\left(\sum_{t=1}^{d} v_t(S_k)> 4 \cdot d\cdot \delta \cdot \OPT\right)\\
        &\leq \Pr(\bu \cdot \one_{S_k}  >4\cdot d\cdot \delta \cdot \OPT )\\
		&\leq  \Pr\left(\bu \cdot \one_{S_k}- (1-\delta)^{k}  \cdot \bu \cdot \one_{S_0} > 3 \cdot d \cdot  \delta \cdot \OPT \right)\\
		&\leq  \Pr\left( \exists r\in \{0,\hdots,k\}:~~\bu \cdot \one_{S_r}- (1-\delta)^{r}  \cdot \bu \cdot \one_{S_0} > \tol(\bu )\cdot \delta \cdot \OPT \right)\\
		&\leq \delta^{-2}\cdot \exp \left( -\frac{2\cdot \delta ^4 \cdot  \delta^2 \cdot \OPT^2 }{\OPT} \right)\\
		&\leq \delta^{-2}\cdot \exp \left( -\delta^{7}\cdot \OPT \right) \enspace .
    \end{aligned}
  \end{equation*}
  The third inequality uses \eqref{eq:S0_total_volume} and $(1-\delta)^{k}\leq \delta$.
  The fifth inequality is by \Cref{lem:concentration_multistep_prelim}.
  Hence, $\Pr(\rho^*  \leq 8 \cdot d\cdot\delta \cdot \OPT+1)\geq 1-\delta^{-2}\cdot \exp(-\delta^{7}\cdot \OPT)$. 
\end{proof}

\section{Discussion}
\label{sec:discussion}
In this paper we showed that a simple iterative randomized rounding scheme (\Cref{alg:basic_round_and_round}) improves the state-of-the-art algorithms for {\sc $d$-Dimentional Vector Bin Packing}, for any $d>3$.
We also showed that \Cref{alg:basic_round_and_round} outperforms any algorithm within the {\em Round\&Approx} framework of Bansal et al.~\cite{BansalCS2010}. 
Slight modifications in this algorithm  to include an initial matching phase (\Cref{alg:match_and_round}) led to an algorithm that yields an asymptotic $\left(\frac{4}{3} +\eps\right)$-approximation for {\sc $2$-Dimentional Vector Bin Packing}, improving upon the $\left(\frac{3}{2}+\eps\right)$-approximation algorithm of Bansal et al.~\cite{BansalEK2016}.
To the best of our knowledge, we use here for the first time iterative {\em randomized} rounding in the context of {\sc Bin Packing} problems.

For arbitrary $d>2$ we applied a fairly simple analysis of  \Cref{alg:basic_round_and_round}, which leaves much room for improvement.
Our analysis of \Cref{alg:match_and_round} is the result of multiple back-and-forth steps which led to new insights on the stochastic process generated by randomized rounding, and on
structural properties of $d$VBP which proved useful in the analysis.
The matching  subroutine in \Cref{alg:match_and_round} was introduced as part of this process. While this led to a significantly better asymptotic approximation ratio for $d=2$, our analysis for this case is more complex.

We note that many of the ideas used in the analysis for $d=2$ can be easily incorporated into the analysis for $d > 2$.
For example,  the sets $T_j$ (defined in \eqref{eq:ddim_T}) used in the proof of \Cref{thm:better_than_randa} are analogous to touched configurations in the analysis of \Cref{sec:main_rnr}.
While the analysis for $d=2$ considers the set $T_j$ for every iteration $j$ and attempts to exploit it to improve the approximation ratio, the analysis for arbitrary $d$ only considers the set $T_{j_1}$ for a specific value of $j_1$. 

As part of the analysis of \Cref{alg:match_and_round} we introduced a structural property for $2$VBP (\Cref{lem:structural}) which combines ideas of Bansal et al.~\cite{BansalEK2016} and Fairstein et al.~\cite{FairsteinKS2021}.
Intuitively, it should be possible to extend the lemma to arbitrary $d > 2$.
While the rounding scheme presented in the proof of \Cref{lem:structural} can be extended to $d > 2$, the Small Items Refinement (\Cref{lem:refinement}) is tailored to the two-dimensional case. 

The basic idea behind \Cref{alg:basic_round_and_round} is that covering items with some fixed probability via iterative randomized rounding requires sampling fewer configurations, in comparison to non-iterative rounding. In our proofs we used structural properties of $d$VBP (e.g, \Cref{lem:weak_structural,lem:structural}) to formalize this basic idea.
Intuitively, the same basic idea should also work for other {\sc Bin Packing} problems, such as  {\sc Geometric 2-Dimensional Bin Packing}~\cite{BansalK2014} and {\sc Generalized Multidimensional Bin Packing}~\cite{KhanSS2021}, for which the state-of-the-art algorithms use the {\em Round\&Approx} framework.
Formalizing this intuition requires an analog of the structural properties for each of these {\sc Bin Packing} variants.
We note that, even without a tailored structural property, following the outline of the proof of \Cref{lem:better_than_randa}, it can be easily shown that a simple adaptation of \Cref{alg:basic_round_and_round} yields an asymptotic approximation ratio which is at least as good as the ratio of any {\em Round\&Approx} algorithm for {\sc Geometric 2-Dimensional Bin Packing}~\cite{BansalK2014} and for {\sc Generalized Multidimensional Bin Packing}~\cite{KhanSS2021}.

\Cref{alg:basic_round_and_round} can be  used also to simplify existing results.
For example, in \Cref{lem:AFTPAS} we showed the algorithm is an AFPTAS for \textsc{Bin Packing}.
We conjecture that the algorithm is also an AFPTAS for \textsc{Bin Packing with Cardinality Constraints} \cite{EpsteinL2010}.

Finally, the number of configurations sampled in each iteration of \Cref{round:loop} in \Cref{alg:basic_round_and_round} was selected arbitrarily for an easier analysis.
One may consider selecting a {\em single} configuration per iteration.
We believe that such modification is unlikely to yield a better approximation ratio, but rather make the analysis more complicated.
A main cause for complication here is that the vanilla form of McDiarmid's concentration bound~\cite{McDiarmid1989} cannot be used, due to stronger dependencies between the sampled configurations.

\bibliographystyle{alpha}
\bibliography{twodim}
\newpage
\appendix
\section{The Flaw in Bansal, Eli{\'a}{\v{s}} and Khan  \cite{BansalEK2016}}
\label{sec:flaw}
The flaw we found in the work of Bansal et al.~\cite{BansalEK2016}  is in the proof of Theorem~6.1.
The theorem refers to properties of the residual items after sampling configurations using a solution for the Configuration-LP.
The proof of the theorem relies on McDiarmid's bound, given as Lemma 6.1 in~\cite{BansalEK2016}.
The flaw is in the use of Lemma 6.1, affecting the correctness of the analysis of the asymptotic approximation guarantees of Algorithm 3 and Algorithm 4 in~\cite{BansalEK2016}. 
We refer below to the third paragraph in the left column of page 1575 in~\cite{BansalEK2016} (starting with ``We now consider the small items'').
As some of the ingredients in the proof of Theorem~6.1 are missing, we expand steps and add details where necessary, while keeping the deviation from \cite{BansalEK2016} to a minimum.
 
Using the notation of~\cite{BansalEK2016}, let $\rho>1$, let $\bx$ be a solution for the Configuration-LP \eqref{eq:config_LP} of the $d$VBP instance $(I,v)$, and let $X_1,..., X_{r}\sim \bx$ be a tuple of $t = \ceil{\rho \cdot z^*}$ random configurations distributed by $\bx$, where $z^* = \|\bx\|$.
Also, define $J= I\setminus \left( \bigcup_{\ell=1}^{r} X_{\ell} \right)$ to be the items {\em not} selected by the sampled configurations $X_1,\ldots, X_r$. 

For $j=(h_1,\ldots,h_d)\in [0,1]^{d}$, $\cS_j\subseteq I$ is a set of items such that $v_k(i) \leq h_k$ for all $i\in \cS_j$ and $k = 1,\hdots,d$.
The set $\cS_j$ represents a class of small items.
Bansal et al.~\cite{BansalEK2016} define functions $f^k_{\cS_j}:\cC^r\rightarrow \mathbb{R}$ by 
\begin{equation}
\label{eq:f_bansal}
	f^k_{\cS_j}(C_1,\ldots , C_{r} )=  \sum_{i\in \cS_j \setminus\left( \bigcup_{\ell=1}^{r} C_{\ell} \right)} v_k(i) \cdot\frac{1}{h_k} 
\end{equation}
for $k = 1,\hdots,d$.
The definition in~\cite{BansalEK2016} is: ``Let function $f^{k}_{\cS_j}$ be $\sum_{i\in \cS_j \cap J} v_k(i)\cdot \frac{1}{h_k} $'' (up to a minor adaptation to our slightly different notation), which we can only interpret as \eqref{eq:f_bansal} due to the subsequent use of $f^k_{\cS_j}$
in \cite{BansalEK2016} as a function whose
domain is a tuple of configurations, and since 
\begin{equation*}
  f^k_{\cS_j} (X_1, \ldots, X_r)  = \frac{1}{h_k}\cdot  \sum_{i\in \cS_j \setminus\left( \bigcup_{\ell=1}^{r} X_{\ell} \right)} v_k(i)  = \frac{1}{h_k} \sum_{i\in \cS_j\cap J} v_k(i).
\end{equation*}

To use Lemma 6.1 
the authors of~\cite{BansalEK2016} attempt to show that $f^k_{\cS_j}$ is of
$1$-bounded difference (see the definition in \Cref{sec:prelim} of the preset paper) for $k = 1,\hdots,d$.
To this end, they consider $\ell^* \in \{1,\hdots,r\}$ and two vectors $x=(C_1,\ldots, C_r)\in \cC^r$ and $x'=(C'_1,\ldots, C'_r)\in \cC^r$ such that $C_\ell = C'_{\ell}$ for $\ell \in\{1,\hdots,r\}\setminus \ell^*$.
That is,~$x$ and~$x'$ differ only in one coordinate.
Subsequently, the authors state the following:
\begin{equation}
  \begin{aligned}
  \label{eq:bansal_flaw}
    &f^k_{\cS_j} (C_1,\ldots ,C_k) -f^k_{\cS_j} (C'_1,\ldots, C'_k) \\
\leq~& \max \left\{\sum_{i\in \cS_j \cap C_{\ell^*}} v_k(i)\cdot \frac{1}{h_k},~\sum_{i\in \cS_j \cap C'_{\ell^*}} v_k(i)  \cdot \frac{1}{h_k} \right\}\\
\textcolor{red}{\leq}~&\frac{1}{h_k}\cdot h_k ~\leq~ 1.
\end{aligned}
\end{equation}
The second inequality (marked is red) is incorrect.
With no explanation for this inequality, it appears that Bansal et al.~\cite{BansalEK2016} assumed that $v_k(\cS_j\cap C) \leq {h_k}$ for any $C\in \cC$.
However, there may be $C\in \cC$ such that $v_k(\cS_j\cap C) =1$.
For example, suppose that $h_k= \frac{1}{10}$, and let $v_k(i)=h_k$ and $v_{k'}(i)=0$  for every $i\in \cS_j$ and $k'\in \{1,\hdots,d\}\setminus \{k\}$.
Then a configuration $C$ containing $10$ items from $\cS_j$ satisfies $v_k(\cS_j\cap C) =1>{h_k}$. 

In the setting of the proof of Theorem~6.1 of \cite{BansalEK2016}, the items in $\cS_j$ are assigned to configurations $C^*_1, \ldots, C^*_m$  in a specific solution.
Indeed, it holds that $v_k(C^*_\ell \cap \cS_j)\leq h_k$ for $\ell = 1,\hdots,m$, and we believe this led the authors of 
\cite{BansalEK2016} to the conclusion that $v_k(C \cap \cS_j)\leq h_k$ for every configuration $C\in \cC$, and hence to the flawed inequality in \eqref{eq:bansal_flaw}. 

Thus, the proof that $f^k_{\cS_j}$ is of $1$-bounded difference is incorrect, and  the subsequent use of Lemma~6.1 fails.

A correct version of \eqref{eq:bansal_flaw} is 
\begin{equation}
  \begin{aligned}
  \label{eq:bansal_unflaw}
    &f^k_{\cS_j} (C_1,\ldots ,C_k) -f^k_{\cS_j} (C'_1,\ldots, C'_k) \\
	\leq~& \max \left\{\sum_{i\in \cS_j \cap C_{\ell^*}} v_k(i)\cdot \frac{1}{h_k},~\sum_{i\in \cS_j \cap C'_{\ell^*}} v_k(i)  \cdot \frac{1}{h_k} \right\}\\
	{\leq}~&\frac{1}{h_k}.
  \end{aligned}
\end{equation}
However, this inequality only shows that $f^k_{\cS_j}$ is of  $\frac{1}{h_k}$-bounded difference.
As $\frac{1}{h_k}$ may be large (for example, it may be that $\frac{1}{h_k} = (\OPT(I))^3 $), the concentration bound which can be derived from~\eqref{eq:bansal_unflaw} is too weak to complete the proof. 

Theorem~6.1 of \cite{BansalEK2016} is a central component in the proofs of the asymptotic $\left(1+\ln\left(\frac{3}{2}\right)+\eps\right)$-approximation for $2$VBP and of the asymptotic $\left(1.5+\ln\left(\frac{d+1}{2}\right)+\eps\right)$-approximation for $d$VBP.
By the above, the two results are incorrect.

\end{document}